\documentclass[11pt, a4paper]{article}

\usepackage{amsmath,amssymb,amsthm,amscd}
\usepackage{mathtools}
\usepackage{mathpazo}
\usepackage{mathrsfs}  
\usepackage{bm}
\usepackage{amsbsy}
\usepackage{slashed}
\usepackage{array}
\usepackage{hyperref}
\usepackage{graphicx}
\usepackage{csquotes}
\usepackage{tikz-cd}
\usetikzlibrary{shapes.misc}
\usepackage{pst-node}
\usepackage{dsfont}
\usepackage{framed}
\usepackage[all,cmtip]{xy}
\usepackage{verbatim}
\usepackage{simpler-wick}
\usepackage{minibox}

\DeclareFontFamily{U}{eus}{\skewchar\font'60}
\DeclareFontShape{U}{eus}{m}{n}{%
	<-6>eusm5%
	<6-8>eusm7%
	<8->eusm10%
}{}

\DeclareMathAlphabet\EuScript{U}{eus}{m}{n}

\newcommand{\C}{\mathbb{C}}

\newcommand{\R}{\mathbb{R}}

\newcommand{\Z}{\mathbb{Z}}

\newcommand{\pa}{\partial}
\newcommand{\pacr}{\bar{\partial}_{b}}
\newcommand{\pab}{\bar{\partial}}

\newcommand{\CE}{Chevalley-Eilenberg }

\newcommand{\MC}{Maurer-Cartan }
\newcommand\mycdots{\makebox[1em][c]{$\cdot\hss\cdot\hss\cdot$}}

\newcommand\sbullet[1][.5]{\mathbin{\vcenter{\hbox{\scalebox{#1}{$\bullet$}}}}}

\DeclareMathOperator{\Tr}{Tr}

\DeclareMathOperator{\Hom}{Hom}

\numberwithin{equation}{section}

\newtheorem{defn}{Definition}[section]
\newtheorem{theorem}{Theorem}[section]
\newtheorem{cor}{Corollary}[section]
\newtheorem{lem}{Lemma}[section]
\newtheorem{prop}{Proposition}[section]

\newtheorem{conj}{Conjecture}[section]
\newtheorem{remark}{Remark}[section]

\def\curved{\tikz[baseline=.1ex]{
		\draw[ ->] (0,0.32) arc (30:330:0.5);}
}

\tikzset{cross/.style={cross out, draw=black, fill=none, minimum size=2*(#1-\pgflinewidth), inner sep=0pt, outer sep=0pt}, cross/.default={3pt}}

\tikzset{crosss/.style={cross out, draw=black, fill=none, minimum size=2*(#1-\pgflinewidth), inner sep=0pt, outer sep=0pt}, crosss/.default={2pt}}

\title{Twisted Holography and Celestial Holography from Boundary Chiral Algebra}
\author{Keyou Zeng}

\date{\vspace{-3ex}}

\newcommand{\Addresses}{{
		\bigskip
		\footnotesize
		
		Keyou Zeng, \textsc{Perimeter Institute for Theoretical Physics, Waterloo, Canada }\par\nopagebreak
		\textit{E-mail address}: \texttt{kzeng@perimeterinstitute.ca}
		
}}

\date{\vspace{-3ex}}

\begin{document}
	\maketitle
	\begin{abstract}
We study the Kaluza-Klein reduction of various $6d$ holomorphic theories. The KK reduction is analyzed in the BV formalism, resulting in theories that come from the holomorphic topological twist of $3d$ $\mathcal{N} = 2$ supersymmetric field theories. Effective interactions of the KK theories at the classical level can be obtained at all orders using homotopy transfer theorem. We also analyze a deformation of the theories that comes from deforming the spacetime geometry to $SL_2(\C)$ due to the brane back-reaction. We study the boundary chiral algebras for the various KK theories. Using Koszul duality, we argue that by properly choosing a boundary condition, the boundary chiral algebra coincides with the universal defect chiral algebra of the original theory. This perspective provides a unified framework for accessing the chiral algebras that arise from both twisted holography and celestial holography programs.
		
	\end{abstract}
	\tableofcontents
	\section{Introduction}
	The AdS/CFT correspondence \cite{Maldacena:1997re,Witten:1998qj} is perhaps one of the greatest achievements of string theory. It has a deep influence on many branches of theoretic physics that goes far beyond string theory. A bridge between the AdS/CFT correspondence and mathematics had been missing until the recent formulation of twisted supergravity. In \cite{Costello:2016mgj}, Costello and Li proposed a twisted form of AdS/CFT correspondence that relates a twist of a supersymmetric gauge theory \cite{Witten:1988ze} with a twist of supergravity. This program throws light upon the mathematical structure hidden behind (part of) the AdS/CFT correspondence. The notion of Koszul duality is shown to play a crucial part in this story. In \cite{Costello:2017fbo}, Costello formulated the following twisted form of holography principle (or conjecture):
	\begin{conj}
		Consider a stack of branes in twisted string theory or M-theory. We can consider the following two algebras
		\begin{itemize}
			\item The algebras $\mathcal{A}_N$ of local operators of the twisted supersymmetric gauge theory living on a stack of $N$ branes, after sending  $N\to \infty$.
			\item The algebra $\mathcal{B}$ of local operators of the twisted supergravity restricted alone the location of the branes.
		\end{itemize}
	Then, these two algebras are Koszul dual:
	\begin{equation}
		\lim_{N \to \infty} \mathcal{A}_N = \mathcal{B}^!.
	\end{equation}
	\end{conj}

In this direction, there has been a surge of recent works exploring different aspects of twisted supergravity and twisted holography. Various twisted holography models are proposed \cite{Ishtiaque:2018str,Costello:2018zrm,Costello:2020jbh}. More examples of twisted supergravity are found \cite{Eager:2021ufo,Raghavendran:2021qbh}. Operators that create D-branes with non-trivial geometrical shapes called giant graviton in twisted holography are studied in \cite{Budzik:2021fyh}. In this paper, we will also study the model of \cite{Costello:2018zrm} in detail.

The notion of Koszul duality also enters the celestial holography program \cite{Strominger:2017zoo}. In a recent work \cite{Costello:2022wso}, Costello and Paquette explained various relationships between the scattering amplitudes of a class of $4d$ theories and the corresponding celestial chiral algebras using twistor correspondence. For $4d$ theories that come from local holomorphic field theories on twistor space, the corresponding chiral algebra is Koszul dual to the algebra of observables of the twistor theory restricted to $\C\mathbb{P}^1$.

In fact, the concept of Koszul duality is plentiful in the structure of quantum field theory, especially in the study of defects and boundaries \cite{Paquette:2021cij}. In this introduction, we briefly review two main sources of Koszul duality in physics. The bridge between these two pictures of Koszul duality will be the starting point of our analysis.

	\subsection{Two pictures of Koszul duality}
	
One important aspect of Koszul duality in quantum field theory arises from defects. Given a field theory $\mathcal{T}$ on some manifold $M$, one can consider coupling some other system along a submanifold $S \subset M$. Then, Koszul duality says that the algebra of the universal defect that we can couple to the theory $\mathcal{T}$ is Koszul dual to the algebra of the theory $\mathcal{T}$ along $S$. Here, the universal defect means that it is the most general possible consistent coupling to our given theory. We refer to \cite{Costello:2020jbh,Paquette:2021cij} for more examples of Koszul duality from universal defect. The mathematical incarnation of this property of Koszul duality for associative algebra (see also \cite{Gui:2022pnx} for chiral algebra) can be expressed as
\begin{equation}
	MC(A\otimes B) = \mathrm{Hom}(A^!,B),\text{ for any }B,
\end{equation}
where $MC(A\otimes B)$ is the set of \MC elements that control the deformation of the coupled systems.

	Another perspective of Koszul duality comes from considering boundary. In a bulk-boundary system, different boundary conditions lead to different boundary algebras. We can consider boundary conditions \footnote{We also need to require the boundary condition to be large enough. A boundary condition $\mathcal{B}$ is called a large, if the category $\mathcal{C}$ of the boundary condition is equivalent to the (derived) category of modules of the algebra $A_\pa = \mathrm{End}_{\mathcal{C}}(\mathcal{B})$ } that are transversal to each other. Here "transversal" means that if we place the two boundary conditions on two sides of an interval respectively, then the theory should be trivial in this configuration. We expect the boundary algebras corresponding to the two transversal boundary conditions to be Koszul dual, at least when the theory is topological along the perpendicular direction.  We refer to \cite{Paquette:2021cij,Zeng:2021zef} for more examples of Koszul duality from boundary algebra. For associative algebra, this perspective of Koszul duality can be mathematically expressed as
	\begin{equation}
		A^! = \mathrm{Ext}_A(\C,\C).
	\end{equation}
	
	Then one can ask the following: what happens when we have both realizations of Koszul duality in a quantum field theory? This can be achieved in a general setup by transforming a bulk-defect system into a bulk-boundary system.
	
	
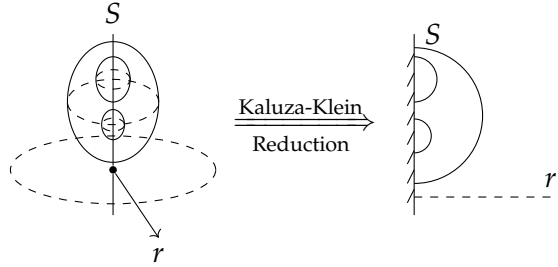
\begin{center}
		\begin{tikzpicture}[scale  = 0.6]
			
			\def\r{3}
			
			\draw (0,0) node[circle, fill, inner sep=1] (orig) {};
			\draw[dashed] (orig) ellipse (3*\r/4 and \r/4);
			\draw[->] (orig) -- ++(\r/3, -\r/2) node[below] (x1) {$r$};
			\draw (orig) -- (0, \r) node[above] (x3) {$S$};
			\draw (orig) -- (0,- \r/3) ;
			
			\draw (0, 2*\r/3) ellipse (\r/8 and \r/6);
			\draw [dashed] (0, 2*\r/3) ellipse (\r/8 and \r/14);
			
			\draw (0, \r/3) ellipse (\r/12 and \r/9);
			\draw [dashed] (0, \r/3) ellipse (\r/12 and \r/21);
			
			\draw (0, \r/2) ellipse (\r/3 and 4*\r/9);
			\draw [dashed] (0, \r/2) ellipse (\r/3 and \r/6);
			
			\node at (14*\r/10, \r/3) {$\xRightarrow[\text{Reduction}]{\text{Kaluza-Klein}}$};
			
			\draw (2.2*\r,\r)  node[above,right] {$S$} -- (2.2*\r, -\r/3);
			
			\foreach \x in {1,2,3,4,5,6,7,8}
			\draw (2.2*\r,\r - \x*\r/7) -- (2.2*\r - \r/20,9*\r/10 - \x*\r/7) ;
			
			\draw[dashed] (2.2*\r,-\r/5)  -- (3.2*\r,-\r/5) node[above] {$r$};
			\draw (2.2*\r,\r/2) arc (-90:90:\r/6);
			\draw (2.2*\r,\r/8) arc (-90:90:\r/8);
			\draw (2.2*\r,-\r/10) arc (-90:90:\r/2);
		\end{tikzpicture}
	\end{center}

	Suppose we are considering a $d$ dimensional defect $\R^{d}\times \{0\}$ in $\R^d\times \R^{n - d}$. For a more general situation, we can consider a small tubular neighborhood $S\times U \approx \R^d\times \R^{n - d}$ of the defect $S$. We are interested in the algebra of operators restricted along $S$. Instead of studying the algebra directly, we first remove the defect locus and perform dimensional reduction on the unit sphere of $\R^{n - d}$ normal to the defect. Namely, we consider the projection
	\begin{equation}
		\pi^{S^{n - d - 1}}: \R^d \times (\R^{n - d}\backslash \{0\}) \to  \R^d\times \R_{ > 0},
	\end{equation}
	where we send $(\mathbf{x},\mathbf{y})$ to $(\mathbf{x},r = |\mathbf{y}|)$. 
	
	After this dimensional reduction, we get a theory on $\R^d\times \R_{ > 0}$. Since we have removed the defect locus, fields of this theory correspond to fields of the original theory with arbitrary poles at $r = 0$. Now, we would like to add back the defect locus. This means that we need to extend the theory to $\R^d\times \R_{ \geq 0} $. To do this, we need to impose a boundary condition at $\R^d\times \{r = 0\} $. There is a natural candidate for the boundary condition, given by requiring that fields of the original theory have no pole at $r\to 0$. 
	
	Though it may not be true that this is always a valid boundary condition, in the examples of our study, we find that this requirement always leads to a boundary condition. Moreover, we expect that the boundary algebra after imposing this boundary condition is the same as the algebra of the original theory restricted along the defect locus.

	Then both pictures of Koszul duality enter this construction. On the one hand, we can consider the universal defect algebra on $S$. On the other hand, we consider the boundary algebra with the boundary condition transversal to the one we mentioned. Both are Koszul dual to the same algebra, so they are the same. 
	
		\begin{center}
		\begin{tikzpicture}
			\node (1) at (0,0) {universal defect algebra along $S$};
			\node (2) at (8,0) {algebra of operators restricted along $S$};
			\node (3) at (0,-2) {boundary algebra at $r = \infty$ };
			\node (4) at (8,-2) {boundary algebra of the KK theory at $r = 0$};
			\draw [double distance=2pt][shorten >=4pt,shorten <=4pt] (2) -- (4);
			\draw [<->][shorten >=5pt,shorten <=5pt] (1) --  (2) node [midway,above] {Koszul};
			\draw [<->][shorten >=5pt,shorten <=5pt] (3) --  (4) node [midway,above] {Koszul};
			\draw [double distance=2pt][shorten >=4pt,shorten <=4pt] (1) -- (3);
		\end{tikzpicture}
	\end{center}
	
	As is often the case, the transversal boundary condition can also be realized as the boundary condition that requires fields of the original theory to have no pole at $r\to \infty$. Then one can see that holography principles manifest in the above reasoning. The boundary system at $r\to \infty$ is "holographic dual" to a different system, given by the universal defect. Koszul duality here plays as a black box in the intermediate.
	
	String theory offers abundant examples of bulk/defect systems. Branes are defect objects in string theory. Moreover, from the philosophy of open-closed coupling \cite{Costello:2015xsa}, we expect that the stack of $N$ branes is a universal defect in the large $N$ limit.
	
	Therefore, we can add one more piece of ingredient to the twisted holography program
	\begin{itemize}
		\item The boundary algebra $\mathcal{A}_\pa$ of the theory obtained by performing KK compactification of the twisted supergravity along the unit sphere of the normal direction to the defect. The boundary condition is chosen to be transversal to the boundary condition which requires fields to have no pole at the defect locus. 
	\end{itemize}
	Then the twisted holography conjecture predicts that
	\begin{equation}
		\lim_{N \to \infty}\mathcal{A}_{N} = \mathcal{A}_\pa.
	\end{equation}

In this paper, we will examine this form of twisted holography using the model proposed in \cite{Costello:2018zrm}. This model can be viewed as a twist of the canonical example of $AdS_5/ CFT_4$ duality introduced in \cite{Maldacena:1997re}.

	Following \cite{Costello:2022wso}, this technique also provides an efficient method to compute the chiral algebra originated from the scattering amplitude and celestial holography program. The boundary chiral algebra at $r\to \infty$ in twistor space produce the celestial chiral algebra living on the celestial sphere. It is a remarkable fact that some of the structure constants of the OPE also appear in the twisted holography model we studied, where they also have a combinatorial explanation in terms of matrix contraction.
	
	\subsection{Cohomological Kaluza Klein reduction}
	In this paper, we pin down the above general ideas to $2$d holomorphic defect in $6$d theories. We consider dimensional reduction of holomorphic theory on $\C^3$ along the projection
	\begin{center}
		\begin{tikzpicture}[scale=0.7]
			\node (0) at (0,0) {$S^3$};
			\node (1) at (2,0) {$\C^3\backslash\C $};
			\node (2) at (2,-2) {$\C\times \R_{ > 0}$};
			\draw [->][shorten >=1pt,shorten <=1pt] (1) -- (2);
			\draw [->][shorten >=1pt,shorten <=1pt] (0) -- (1);
		\end{tikzpicture}.
	\end{center}
	
	A full Kaluza Klein reduction is usually not practical for supersymmetric gauge theory and gravity. As the KK modes go higher, it soon becomes difficult to keep track of their interactions. Twisted theories improve this situation with the help of techniques from homological algebras.
	
	Twisted theory simplify the original physical theory through a cancellation between a part of the bosonic and fermionic fields. The field content is thus simplified. We find a similar phenomenon when we perform KK reduction to the twisted theory. A large number of KK modes are cohomologically trivial and can be integrated out. Then we are left with a KK theory which is more manageable.
	
	However, by integrating out the cohomologically trivial KK modes, classical effective interactions are generated at tree-level through the exchanges of these modes. Fortunately, there is a systematic way to analyze the effective interactions. Here we work with the BV formalism and use the language of homotopy algebra. The structure of the classical field theory is compactly organized into the structure of $L_\infty$ algebra. Then integrating out fields at the tree level can be translated to the mathematical notion of homotopy transfer, which provides a systematic way to compute the higher brackets that build the effective higher order interaction.
	
	We find the KK theory to be the $3d$ holomorphic topological theory introduced in \cite{Aganagic:2017tvx,Costello:2020ndc} with an infinite number of effective interactions. Without the higher order interactions, these theories arise from the holomorphic topological twist of $3d$ $\mathcal{N} =  2$ supersymmetric field theories. The boundary chiral algebra of the $3d$ theory is studied in \cite{Costello:2020ndc}. Applying results from \textit{loc. cit.} to our setup, we can produce various defect chiral algebra originated from twisted and celestial holography programs.
	\begin{center}
		\begin{tikzpicture}
			\node (1) at (0,0) {universal defect chiral algebra};
			\node (2)[text width=5cm] at (8,0) {chiral algebra along $\C$ in $\C^3$  (or $SL_2(\C)$)};
			\node (3)[text width=5cm] at (0,-2) {boundary chiral algebra with Dirichlet b.c. };
			\node (4)[text width=5cm] at (8,-2) {boundary algebra of $3d$ holomorphic/topological theory with Neumann b.c};
			\draw [double distance=2pt][shorten >=4pt,shorten <=4pt] (2) -- (4);
			\draw [<->][shorten >=5pt,shorten <=5pt] (1) --  (2) node [midway,above] {Koszul};
			\draw [<->][shorten >=5pt,shorten <=5pt] (3) --  (4) node [midway,above] {Koszul};
			\draw [double distance=2pt][shorten >=4pt,shorten <=4pt] (1) -- (3);
		\end{tikzpicture}
	\end{center}
	
	For holomorphic Chern-Simons theory and holomorphic BF theory, the structures of the effective interactions for the KK theories can be compactly organized into the $A_\infty$ structure on the tangential Cauchy Riemann cohomology $H_{b}^{0,\sbullet}(S^3)$ on $S^3$
	\begin{equation}
		\{H_{b}^{0,\sbullet}(S^3) ,m_2 ,m_3 \dots \}.
	\end{equation} 
	The structure constants of this $A_\infty$ algebra are also important building blocks for the OPE coefficients of the boundary chiral algebra. For holomorphic BF theory that gives the celestial chiral algebra of self-dual Yang Mills, the tree-level boundary OPE takes the following schematic form
	\begin{equation}
\begin{aligned}
		B(z)B(0) \sim \frac{1}{z} \left( B(0) + (m_3)B\tilde{B}(0) + (m_4)B\tilde{B}\tilde{B}(0) \dots \right),\\
		\tilde{B}(z)\tilde{B}(0) \sim \frac{1}{z} \left( \tilde{B}(0) + (m_3)\tilde{B}\tilde{B}(0) + (m_4)\tilde{B}\tilde{B}\tilde{B}(0) \dots \right).
\end{aligned}
	\end{equation}
	 We provide explicit formula that compute all the higher product $m_n$ in Section \ref{sec:all_mn} and a explicit formula for the above OPE in section \ref{sec:cele_HBF}
	
	We also find an $L_\infty$ deformation of the Lie algebra $T^*[1]\mathrm{Ham}(\C^2)$
		\begin{equation}
		\{T^*[1]\mathrm{Ham}(\C^2) ,\{-,-\} ,\{...\}_3 \dots \} ,
	\end{equation} 
	where $\mathrm{Ham}(\C^2)$ is the Lie algebra of Hamiltonian vector fields on $\C^2$. The structure constants of this $L_\infty$ algebra encode the effective interactions of the KK theory of holomorphic Poisson BF theory. 
	
	KK reduction of Kodaira-Spencer gravity can be analyzed in a similar fashion.
		
	In twisted holography, the spacetime geometry will be deformed due to the presence of (topological) D brane. Theories on the new geometry will also be deformed. This can be understood as a deformation of the $A_\infty$ algebra structure on $\mathcal{O}_{\C}\otimes H_{b}^{0,\sbullet}(S^3) $. These deformed higher product structure is an important building block for the holography chiral algebra in twisted holography.
	
	\subsection{Twisted holography in a non-planar sector}
One novel feature of our KK theory is that it provides a systematic way to compute a certain class of OPE correspond to non-planar contribution. We first explain what planar means in the twisted holography setup. 

The chiral algebra on B branes consists of a $bc$ system valued in $\mathfrak{gl}_N$, a pair of symplectic bosons $(Z_1,Z_2)$ valued in the adjoint representation of $\mathfrak{gl}_N$ and a pair of symplectic bosons $(I,J)$ valued in the (anti)fundamental representation of $\mathfrak{gl}_N$ (tensored with the fundamental representation of $\mathfrak{gl}_{K|K}$, where $K|K$ is the number of space filling (anti)branes). In the large $N$ limit, the BRST cohomology is generated by single trace "closed string" operators and "open string" operators. For a detailed discussion, see \cite{Costello:2018zrm}. For illustrative propose, we consider "open string" operators. They are matrix products of adjoint fields $(Z_1,Z_2)$ sandwiched between an anti-fundamental and a fundamental field $I$-$J$. We denote $IZ^nJ$ an open string operator with $n$ adjoint fields. Operator product expansions of these operators are computed by wick contraction. The operator $IZ^nJ$ has two-point function of order $N^{n+1}$. If we normalize them by $\frac{IZ^nJ}{\sqrt{N^{n}}}$, the OPE takes the following schematic form
\begin{equation}
	\frac{IZ^{n_1}J}{\sqrt{N^{n_1}}}\frac{ IZ^{n_2}J}{\sqrt{N^{n_2}}} \sim N\delta_{n_1,n_2} + \sum_{k\geq 0}\frac{IZ^{n_1+n_2 - 2k}J}{\sqrt{N^{n_1 + n_2 - 2k+1}}} + O(\frac{1}{N}).
\end{equation} 
In the planar limit, we drop all sub-leading terms. There are different types of non-planar contributions to the sub-leading terms. Non-planar contraction of $Z$ fields can lead to OPE of the form $\frac{1}{N^{l}}\frac{IZ^{n_1+n_2 - 2k}J}{\sqrt{N^{n_1 + n_2 - 2k}}}$, $l\geq 2$. We do not analyze these OPE in this paper. There are other types of non-planar contributions, which produce more than one "open string" operators and take the following form 
\begin{equation}
	\frac{1}{N}\sum_{k\geq 1} \sum_{n_3 + n_4 = n_1 + n_2 -2k}\frac{IZ^{n_3}J}{\sqrt{N^{n_3}}}\frac{IZ^{n_4}J}{\sqrt{N^{n_4}}} + \frac{1}{N^2}\sum_{n_3 + n_4 + n_5 = n_1 + n_2 -2k - 2}\frac{IZ^{n_3}J}{\sqrt{N^{n_3}}}\frac{IZ^{n_4}J}{\sqrt{N^{n_4}}} \frac{IZ^{n_5}J}{\sqrt{N^{n_5}}} + \dots.
\end{equation}
We will see in this paper that OPE of the above form corresponds to the deformed $A_\infty$ structure on $\mathcal{O}_{\C}\otimes H_{b}^{0,\sbullet}(S^3)$ we mentioned in the last section. To compactly organize the non-planar sector we can effectively compute, we consider a different normalization of the operators \footnote{We thank K. Costello for a discussion of this point.}. We normalize the "open string" operators by $\frac{IZ^{n_1}J}{\sqrt{N^{n_1 +2}}}$. The OPE takes the following schematic form
\begin{equation}
	\label{sub_non_planar}
\begin{aligned}
		\frac{IZ^{n_1}J}{\sqrt{N^{n_1+2}}}\frac{ IZ^{n_2}J}{\sqrt{N^{n_2+2}}} \sim &\frac{1}{N}\left( \delta_{n_1,n_2} + \sum_{k\geq 0}\frac{IZ^{n_1+n_2 - 2k}J}{\sqrt{N^{n_1 + n_2 - 2k +2}}}+ \sum_{n_3 + n_4 = n_1 + n_2 -2k-2}\frac{IZ^{n_3}J}{\sqrt{N^{n_3+2}}}\frac{IZ^{n_4}J}{\sqrt{N^{n_4+2}}}  + \dots \right)\\
		& + O(\frac{1}{N^2}).
\end{aligned} 
\end{equation}
Then we find that by dropping all sub-leading terms in the above normalization, the OPE is isomorphic to the tree-level boundary OPE of the KK theory in the deformed geometry. The above limit of the chiral algebra contains much more information than the planar limit, but still can be effectively computed and checked. We present some explicit examples in Section \ref{sec:OPE_HCS_4}.

\subsection{Future direction}

\paragraph{KK reduction on supermanifolds}
Our results on the KK reduction can be generalized to supermanifolds. For example, the following KK reduction should be straightforward based on our results
\begin{equation}
	\C^{3|N}\backslash \C^{1|0} \to \C^{1|N}\times \R_{>0}.
\end{equation} 
Holomorphic theories on $\C^{3|N}$ for different value of $N$ and their variants are important in various twisted and celestial holography setups. 
\begin{enumerate}
	\item Holomorphic Chern-Simons on super twistorspace $\mathbb{P}\mathbb{T}^{3|4}$ gives rise to the self-dual limit of $\mathcal{N} = 4$ Yang-Mills on $\R^4$ \cite{Witten:2003nn}.
	\item Holomorphic BF on super twistorspace $\mathbb{P}\mathbb{T}^{3|N}$ gives rise to the self-dual limit of $\mathcal{N} = N$ Yang-Mills on $\R^4$ \cite{Boels:2006ir} for $N = 1,2$.
	\item   Kodaria-Spencer theory on a deformation of $\C^{3|4}\backslash \C^{1|0}$ is holographic dual to a orbifold SCFT on $\mathrm{Sym}^N(T^4)$ \cite{Costello:2020jbh}.

We provide more discussion of KK reduction on superspace in section \ref{sec:other}.

\end{enumerate}
\paragraph{4d holomorphic theories}
As we have sketched, the method of KK reduction can be useful to analyze bulk/defect systems in many different scenarios. A variant of our setup is the KK reduction of $4d$ holomorphic theory along the unit sphere $S^3$. Based on our results, the KK reduction will produce a topological quantum mechanics in the radial direction $r$. This topological quantum mechanical system also has interactions built out of the $A_\infty$ structure on the CR cohomology $H_{b}^{0,\sbullet}(S^3)$. By considering the algebraic structure of the bulk algebra and its action on the boundary algebra at $r = 0$. We expect to reproduce part of the algebraic structure of the bulk algebra of the original $4d$ theory studied in \cite{BGKWWY}. We provide more details of this construction in Section \ref{sec:4dKK}.

\paragraph{Bootstrap method}

In the physical AdS/CFT setup, holographic correlation functions are notoriously hard to compute, and only until recently new methods were invented based on bootstrap ideas (see e.g. \cite{Rastelli:2016nze,Bissi:2022mrs}). It will be interesting to see if these new techniques from bootstrap can be applied to the twisted holography setup. 

In the same spirit, \cite{Costello:2022upu} (see also \cite{{Bittleston:2022jeq}}) showed that a one-loop correction to the celestial OPE is strictly constrained by the associativity of chiral algebra. This illustrates the potential power of this method.

In a different direction \cite{Budzik:2022mpd}, the authors proposed a method to solve a large class of Feynman integrals for twisted theories in various dimensions. The Feynman integrals are constrained by a set of identities that correspond to the associativity of the local operator algebras. So one can compute Feynman integral via bootstrap techniques. It will be interesting to extend their method to $3d$ holomorphic topological theory with boundary. This will be an important ingredient to the loop corrections to the chiral algebras we considered.

	\subsection{Organization of this paper}
	
	This paper is organized as follows. In Section \ref{sec:toy}, we introduce a toy model in $2d$ that illustrates how to perform the KK reduction for holomorphic theory. We show in this example that by properly choosing a boundary condition, the boundary algebra indeed reproduces the algebra of the original theory.
	
	In Section \ref{sec:HCS}, we perform the KK reduction of holomorphic Chern-Simons theory on $\C^3\backslash\C$. We show that techniques from homotopy algebra can help us in understanding the KK theory. An important result in this section is that the classical interaction of the KK theory is encoded in the mathematical structures called the $A_\infty$ structure on the tangential Cauchy-Riemann cohomology $H_b^{0,\sbullet}(S^3)$. Section \ref{sec:A_inf} is devoted to compute the whole $A_\infty$ structure. 
	
	In Section \ref{sec:HBF},\ref{sec:PBF},\ref{sec:KS}, we provide more examples of KK reduction in $6d$. We analyze holomorphic BF theory in Section \ref{sec:HBF}, Poisson BF theory in Section \ref{sec:PBF} and Kodaira-Spencer gravity in Section \ref{sec:KS}. 
	
	In Section \ref{sec:def}, we analyze the KK theory when we deform the $6d$ geometry to $SL_2(\C)$. This deformation can be understood as a B-brane backreaction. Results in this section are applied in Section \ref{sec:twist_hol} in the twisted holography setup. We compute the boundary OPE using the interaction vertices of the KK theory. We also match the boundary OPE with the B-brane OPE computed from wick contraction.
	
	In Section \ref{sec:cele_holo}, we apply our results to compute various celestial chiral OPE. The $A_\infty$ structure $m_n$ of the tangential Cauchy-Riemann cohomology gives us loop corrections to the celestial chiral algebra at each loop level $n$ with maximal numbers of fields in the OPE. Other loop corrections are also discussed.
	
	In Section \ref{sec:4d} and \ref{sec:other}, we discuss some other applications of our results.
	
	
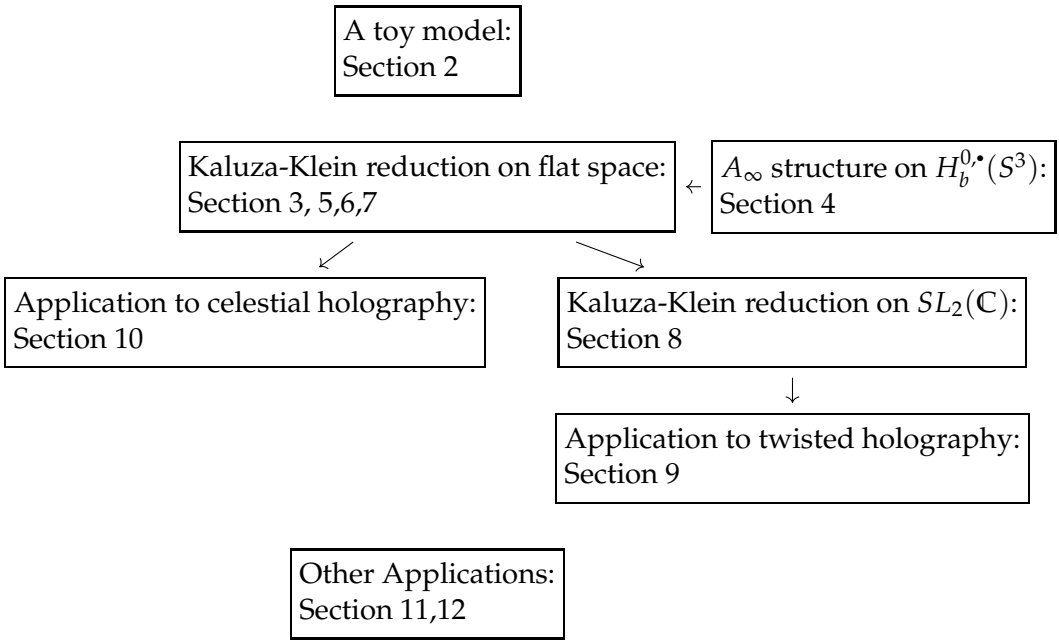
\begin{center}
				\begin{tikzpicture}[scale=0.6]
			\node (0) at (0,0) {\minibox[frame]{A toy model:\\Section \ref{sec:toy}}};
			\node (10) at (0,-3) {\minibox[frame]{Kaluza-Klein reduction on flat space:\\ Section \ref{sec:HCS}, \ref{sec:HBF},\ref{sec:PBF},\ref{sec:KS}}};
			\node (11) at (10,-3) {\minibox[frame]{$A_\infty$ structure on $H_b^{0,\sbullet}(S^3)$:\\ Section \ref{sec:A_inf}}};
			\node (20) at (-4,-6) {\minibox[frame]{Application to celestial holography:\\ Section \ref{sec:cele_holo}}};
			\node (21) at (8,-6)  {\minibox[frame]{Kaluza-Klein reduction on $SL_2(\C)$:\\ Section \ref{sec:def}}};
			\node (31) at (8,-9) {\minibox[frame]{Application to twisted holography:\\ Section \ref{sec:twist_hol}}};
			\node (40) at (0,-12) {\minibox[frame]{Other Applications:\\ Section \ref{sec:4d},\ref{sec:other}}};
			\draw [->]   (11) -- (10);
			\draw [->]    (10) -- (20);
			\draw [->]    (10) -- (21);
			\draw [->]  (21) -- (31);
		\end{tikzpicture}
	\end{center}

\section*{Acknowledgements}
I would like to thank Kevin Costello, Roland Bittleston, Davide Gaiotto, Si Li, Xi Li and Gongwang Yan for illuminating discussion. Research at Perimeter Institute is supported in part by the Government of Canada through the Department of Innovation, Science and Economic Development Canada and by the Province of Ontario through the Ministry of Colleges and Universities.

	\section{A toy model}
	\label{sec:toy}
	Before we dive into the details of our main examples, we study 2d holomorphic theory and dimensional reduction along the unit circle. This toy model illustrates some of the ideas we sketched in the introduction, but without involving too many technical difficulties.
	\subsection{Dolbeault complex under dimensional reduction}
	\label{sec:Dol_2d}
	In the BV formalism, the constructions of many holomorphic theories on a complex curve $X$ are built on the Dolbeault complex $(\Omega^{0,\sbullet}(X),\pab)$. So before we consider any specific theory, we first understand the properties of Dolbeault complex under the dimensional reduction along the unit circle.
	
	The Dolbeault complex on $\C\backslash\{0\}$ is given by:
	\begin{equation}
		(\Omega^{0,\sbullet}(\C\backslash\{0\}),\pab) = (C^{\infty}(\C\backslash\{0\}) \oplus C^{\infty}(\C\backslash\{0\})d\bar{z},d\bar{z}\frac{\pa}{\pa \bar{z}}).
	\end{equation}
	To perform the dimensional reduction, it is natural to work with the polar coordinate $z = re^{i\theta}$. Standard Fourier transform gives us the following decomposition of the Dolbeault complex
	\begin{equation}
\begin{aligned}
			\Omega^{0,0}(\C\backslash\{0\}) &= \{\sum_{n \in \Z} f_n(r)e^{in\theta}\}, \\
			\Omega^{0,1}(\C\backslash\{0\}) &=  \{\sum_{n \in \Z} g_n(r)e^{in\theta}d\bar{z}\}.
\end{aligned}
	\end{equation}
	However, under the above decomposition, the Dolbeault differential does not behave very well. Both $ \pab f_n(r)$ and $\pab e^{in\theta} $ are non-zero and contribute to the Dolbeault differentia.
	
	An easy way to improve this situation is to consider a redefinition of the functions $f_n,g_n$ in the Fourier transform. We consider the following:
	\begin{equation}
	\begin{aligned}\label{dec_Dol_2d}
			\Omega^{0,0}(\C\backslash\{0\}) & = 	\{\sum_{n} \tilde{f}_n(r)r^ne^{in\theta}\},\\
			 \Omega^{0,1}(\C\backslash\{0\}) &= \{\sum_{n}\tilde{g}_n(r)r^ne^{in\theta}(\frac{1}{2}e^{i\theta}d\bar{z})\}.
		 	\end{aligned}
\end{equation}
Then the Dolbeault differential is better behaved, and only acts on the function $\tilde{f}_n(r)$ as a de Rham differential:
	\begin{equation}
		\pab  \tilde{f}_n(r)r^ne^{in\theta} = \pab \tilde{f}_n(r)z^n =  (\pa_r \tilde{f}(r)) r^n e^{in\theta} (\frac{1}{2}e^{i\theta}d\bar{z}).
	\end{equation}
We see that $\pab = \pab r \frac{d}{d r}$ in the new decomposition.

As a result, we have the following identification of the Dolbeault complex and the de Rham complex:
	\begin{equation}
		\begin{aligned}
			\Omega^{0,\sbullet}(\C\backslash\{0\}) & \cong \Omega_{dR}^{\sbullet}(\R_{ > 0})\otimes \C[w,w^{-1}]\\
			& =( C^{\infty}(\R_{> 0})[dt]\otimes \C[w,w^{-1}], dt\frac{\pa}{\pa t}).
		\end{aligned}
	\end{equation}
	under the change of variable
	\begin{equation}
		z \to w,\;\; \bar{z} \to t^{2}w^{-1},\;\; d\bar{z} \to 2tw^{-1} dt.
	\end{equation}
Typically, a theory built on de Rham complex on $\R$ is a topological quantum mechanics. This is consistent with the physical intuition that the dimensional reduction of a $2d$ chiral theory along the circle is a quantum mechanics in the radial direction.
	\subsection{2d Holomorphic Chern-Simons}
	In this section, we briefly analyze the KK reduction of $2d$ holomorphic Chern-Simons theory \cite{Costello:2015xsa} and its boundary algebra. The $2d$ Holomorphic Chern-Simons has BV field content given by
	\begin{equation}
		\pmb{\EuScript{A}}^{2d} = c^{2d} + \EuScript{A}^{2d}_{\bar{z}}d\bar{z} \in \Omega^{0,\sbullet}\otimes \mathfrak{g}
	\end{equation}
	The BV action functional is given  by
	\begin{equation}
	\mathrm{hCS}(\pmb{\EuScript{A}}^{2d})  =	\frac{1}{2}\int \pmb{\EuScript{A}}^{2d}(\pab \pmb{\EuScript{A}}^{2d} + \frac{1}{3}[\pmb{\EuScript{A}}^{2d},\pmb{\EuScript{A}}^{2d}]).
	\end{equation}
	From our analysis of the Dolbeault complex in the last section, we immediately find the following KK tower of fields after dimensional reduction
	\begin{equation}
\begin{aligned}
			\mathbf{A} &\in	\Omega^{\sbullet}(\R_{ > 0})\otimes \mathfrak{g}[w], \\
			 \mathbf{B} &\in \Omega^{\sbullet}(\R_{ > 0})\otimes \mathfrak{g}w^{-1}[w^{-1}].
\end{aligned}
	\end{equation}
The Dolbeault differential then becomes the de Rham differential. Therefore the action functional becomes
	\begin{equation}
		\int_t \Tr(\mathbf{B} (d_t \mathbf{A}  + \frac{1}{2}[\mathbf{A},\mathbf{A}])),
	\end{equation}
	where $\Tr = \Tr_{\mathfrak{g}}\otimes \oint dw$. We can also identify $\mathfrak{g}w^{-1}\C[w^{-1}]$ with the dual of $\mathfrak{g}[w]$ via the pairing. We see that after dimensional reduction, the KK theory becomes a $1d$ BF theory taking values in Lie algebra $\mathfrak{g}[w]$.
	
	There is a natural boundary condition for this theory at $t = 0$. If we look back at the decomposition \ref{dec_Dol_2d}, we find that fields of the original $2d$ theory are related to the fields of the $1d$ theory as follows
	\begin{equation}
\begin{aligned}
	\sum B[n]w^{-1 - n} ,\; \sum A[n]w^n &\to c^{2d} = \sum_{n\geq 0} B[n](r) r^{-n-1} e^{-i(n+1)\theta} + A[n](r) r^{n} e^{in\theta},\\
	\sum B_t[n]w^{-1 - n}dt ,\; \sum A_t[n]w^ndt &\to \EuScript{A}^{2d}_{\bar{z}}d\bar{z} = \sum_{n\geq 0}(B_t[n](r) r^{-n-1} e^{-i(n + 1)\theta} + A_t[n](r) r^{n} e^{in \theta})\pab r.
\end{aligned}
	\end{equation}
	
	 Therefore, if we impose the following boundary condition
	 \begin{equation}
	 	\mathbf{B}|_{t = 0} = 0\;,
	 \end{equation}
	 the field $\pmb{\EuScript{A}}^{2d}$ of the original $2d$ theory has no pole at $r = 0$ and is well defined on the whole complex plane $\C$.
	 For general $1d$ BF theory, this boundary condition is analyzed in \cite{2021arXiv211101757R,Wang:2022bvq}, and is called the A boundary condition. The corresponding space of boundary local operators is given by the \CE algebra of $\mathfrak{g}[w]$:
	\begin{equation}
		 C^{\sbullet}(\mathfrak{g}[[w]]).
	\end{equation}
	This coincides with the space of classical local operators of $2d$ holomorphic Chern-Simons theory at the origin:
	\begin{equation}
		\mathrm{Obs}^{2d}|_{0} = C^{\sbullet}(\mathfrak{g}[[w]]).
	\end{equation}
	We see that by properly choosing the boundary condition, the space of boundary local operators of the KK theory reproduces the local operators of the original theory restricted at the origin.

	\section{Holomorphic Chern-Simons theory}
	\label{sec:HCS}
	\subsection{A brief review}
	
	In this section, we introduce the holomorphic Chern-Simons theory in $6d$. This is the open-string theory of space filling branes in the topological B-model.
	
	Let $\mathfrak{g}$ be a (super) Lie algebra with a non-degenerate Killing pairing $\Tr$. For $X$ a Calabi Yau $3$-fold with the holomorphic volume form $\Omega_X$, we can define holomorphic Chern-Simons theory on $X$ as follows. Fundamental fields of holomorphic Chern-Simons theory are
	\begin{equation}
		\EuScript{A} \in \Omega^{0,1}(X)\otimes \mathfrak{g}.
	\end{equation}
	The action functional is given by
	\begin{equation}
		S[	\EuScript{A}] = \int_X \Omega_X\wedge \Tr(\frac{1}{2}	\EuScript{A}\wedge \pab 	\EuScript{A} + \frac{1}{6} \EuScript{A} \wedge[	\EuScript{A},\EuScript{A}]).
	\end{equation}
	By varying the action functional $S$ we obtain the following equations of motion for $\EuScript{A} $
	\begin{equation}
		\pab 	\EuScript{A} + \frac{1}{2}[	\EuScript{A} ,	\EuScript{A} ] = 0.
	\end{equation}
	This theory has infinitesimal gauge symmetry parameterized by $c \in \Omega^{0,0}(X)\otimes \mathfrak{g}$. Gauge transformation takes the form
	\begin{equation}
		\delta 	\EuScript{A}  = \pab c + [c,\EuScript{A}].
	\end{equation}
	As for any gauge theory, we can reformulate the holomorphic Chern-Simons theory in the BV formalism. By adding ghosts, anti-fields and anti-ghost, the BV field of holomorphic Chern-Simons theory can be organized into 
	\begin{equation}
		\pmb{\EuScript{A}}\in \Omega^{0,\sbullet}(X)\otimes \mathfrak{g}[1].
	\end{equation}
	Here the symbol $[1]$ means a degree shifting such that fields in $\Omega^{0,p}(X)\otimes \mathfrak{g}[1] $ have cohomological degree $p - 1$ and, in particular, fields in $\Omega^{0,1}(X)\otimes \mathfrak{g}[1] $ have cohomological degree $0$. The BV action functional takes the same form as the non-BV version by replacing the field $\EuScript{A}$ with the BV field $\pmb{\EuScript{A}}$
	\begin{equation}
		HCS[\pmb{\EuScript{A}}] =  \int_X \Omega_X\wedge \Tr(\frac{1}{2}\pmb{\EuScript{A}}\wedge \pab\pmb{\EuScript{A} }+ \frac{1}{6} \pmb{\EuScript{A} }\wedge[\pmb{\EuScript{A} },\pmb{\EuScript{A} }]).
	\end{equation}
	
	The advantage of the BV formalism is that, classically, the structure of the theory is completely characterized by the dg Lie algebra:
	\begin{equation}
		(\Omega^{0,\sbullet}(X)\otimes \mathfrak{g}, \bar{\pa},[-,-]),
	\end{equation}
	together with the integration map
	\begin{equation}
		\label{int_CY}
		\int_X \Omega_X : \;\Omega^{0,3}(X) \to \C.
	\end{equation}
	Then physical manipulation such that dimensional reduction can be compactly organized into mathematical operations on this dg Lie algebra $(\Omega^{0,\sbullet}(X)\otimes \mathfrak{g}, \bar{\pa},[-,-])$. The powerful toolbox from homological algebra can be utilized in our study.
	
	\subsection{From Dolbeault to tangential Cauchy Riemann complex }
	We are interested in the holomorphic Chern-Simons theory on $\C^3\backslash \C \cong \C\times\R_{ > 0}\times S^3$ and its Kaluza-Klein reduction to $\C\times \R_{> 0}$. As we have mentioned in the last section, the structure of the field theory is partly encoded in the algebraic structure of the Dolbeault complex $(\Omega^{0,\sbullet}(\C^3\backslash \C), \bar{\pa})$. Therefore, it is important to study the Dolbeault complex under dimensional reduction.

	We parameterize $\C^3$ by $z, z_1, z_2$ so that $\C$ is defined by the curve $z_1 = z_2 = 0$. We have a homeomorphism $\C^3 \backslash \C \cong \C \times \R_{> 0}\times S^3$, which can be realized in coordinates by
	\begin{equation}\label{coor_change}
		(z,z_1,z_2) \mapsto (z,r = (|z_1|^2 + |z_2|^2)^{\frac{1}{2}},w_i = \frac{z_i}{r}),
	\end{equation}
	where $(z,r)$ serve as the coordinates on $\C \times \R_{> 0}$ and $(w_i)$ serve as coordinates on $S^3$.
	
	Using the above homeomorphism we can decompose $\Omega^{0,\bullet}(\C^3 \backslash \C) $ as follows
	\begin{equation}\label{Dol_C3}
		\Omega^{0,\bullet}(\C^3 \backslash \C)  \cong C^{\infty}(\C \times \R_{> 0})[d\bar{z},\bar{\pa}r]\otimes  C^{\infty}(S^3)[\epsilon^{ij}\bar{z}_id\bar{z}_j].
	\end{equation}
	However, the Dolbeault differential does not behave well under the above decomposition, as the Dolbeault differential of a function on $S^3$ generally involves a factor of $\pab r$. To proceed, We need to introduce the tangential Cauchy Riemann (CR) complex $\Omega_b^{0,\bullet}(S^3)$. For our purpose, we adopt the definition using the embedding of $S^3$ into $\C^2$ (defined by the equation $r = 1$)\cite{Folland}. And define $\Omega_b^{0,\bullet}(S^3)$ as the quotient of $\Omega^{0,\sbullet}(\C^2)|_{S^3}$ by the ideal $I(\pab r)$ generated by $\bar{\pa}r$. By choosing a metric $\langle-,-\rangle$ on $\C^2$, the CR complex $\Omega_b^{0,\bullet}(S^3)$  can also be identified with the orthogonal complement of $I(\pab r)$ in $\Omega^{0,\sbullet}(\C^2)|_{S^3}$.
	
	The CR differential $\pacr$ can be defined as follows. Let $f \in C^{\infty}(S^3)$ be a function on $S^3$, and $f'$ be an extension of $f$ to $\C^2$. Then $\pacr f$ is the restriction to $S^3$ of
	\begin{equation}
		\pab f' - \frac{\langle \pab f',\pab r\rangle}{\langle \pab r,\pab r\rangle}  \pab r.
	\end{equation}
	In particular, we have
	\begin{equation}
		\xi_i : = \pacr \bar{w}_i = d\bar{w}_i - 2\bar{w}_i\pab r .
	\end{equation}
	We can identify the CR complex with
	\begin{equation}
		\Omega_b^{0,\bullet}(S^3) = \C[w_1,w_2,\bar{w}_1,\bar{w}_2,\xi_1,\xi_2]/\langle |w_1|^2 + |w_2|^2 = 1,w_1\xi_1 + w_2\xi_2 = 0\rangle,
	\end{equation}
	and the differential $\pacr$ with
	\begin{equation}
		\pacr f = \sum\frac{\pa f}{\pa \bar{w}_i}\xi_i.
	\end{equation}
	
	We also introduce an auxiliary graded algebra $\Omega^{\sbullet,(j)}_{3d}$ in $3d$, that encodes the transverse holomorphic structure on the three manifold. In local coordinates, this graded algebra can be written as
	\begin{equation}
		\Omega^{\sbullet,(j)}_{3d} = \C^{\infty}(\R^3)[dt,d\bar{z}](dz)^j,
	\end{equation}
	with differential
	\begin{equation}
		\hat{d} = dt\frac{\pa }{\pa t} +  d\bar{z}\frac{\pa}{\pa \bar{z}} .
	\end{equation}
	We denote $\Omega^{\sbullet}_{3d} = \Omega^{\sbullet,(0)}_{3d}$. In terms of components, we have
		\begin{equation}
		\Omega^{0}_{3d} = \C^{\infty}(\R^3),\quad \Omega^{1}_{3d} = \C^{\infty}(\R^3)dt\oplus\C^{\infty}(\R^3)d\bar{z},\quad\Omega^{2}_{3d} = \C^{\infty}(\R^3)dtd\bar{z}.
	\end{equation}
	
	We can generalize the construction in Section \ref{sec:Dol_2d} to decompose the Dolbeault complex $\Omega^{0,\sbullet}(\C^3\backslash\C)$ as follows
	\begin{prop}\label{K}
		There is an isomorphism of dg algebras:
		\begin{equation}
			K: (\Omega^{0,\bullet}(\C^3 \backslash \C) ,\pab) \cong ( \Omega_{3d}^{\sbullet}(\C \times \R_{> 0})\otimes 	\Omega_b^{0,\bullet}(S^3) , \hat{d} + \pacr).
		\end{equation}
	\end{prop}
	
	\begin{proof}
		The isomorphism is constructed as follows. On local coordinates functions, $K$ is defined by
		\begin{equation}
			K: (z,\bar{z},z_i,\bar{z}_i) \mapsto (z,\bar{z},w_i,t^2\bar{w}_i).
		\end{equation}
		On one forms, $K$ is defined by
		\begin{equation}
			K: (d\bar{z}, d\bar{z}_i)  \mapsto (d\bar{z}, t^2\xi_i + 2t\bar{w}_i dt).
		\end{equation}
		$K$ intertwine the differential $\pab$ and $\hat{d} + \pacr$ by construction. To prove that $K$ is an isomorphism, we construct its inverse
		\begin{equation}\label{Kinv}
			K^{-1}: (z,\bar{z},t,w_i,\bar{w}_i)\mapsto (z,\bar{z},r,z_i,r^{-2}\bar{z}_i),
		\end{equation}
	where $r^2 = (z_1\bar{z}_1+ z_2\bar{z}_2)$. On one forms, $K^{-1}$ is defined by
		\begin{equation}\label{Kinv1}
			K^{-1}: (d\bar{z},dt,\xi_i) \mapsto ( d\bar{z}, \bar{\pa}r, \frac{d\bar{z}_i}{r^2} - 2\frac{\bar{z}_i}{r^3}\pab r).
		\end{equation}
		
		We can verify that 
		\begin{equation}
			\begin{aligned}
				&1 = K^{-1}(w_1\bar{w}_1 + w_2\bar{w}_2) = r^{-2}(z_1\bar{z}_1+ z_2\bar{z}_2)  = 1,\\
				&0 = K^{-1}(w_1\xi_1 + w_2\xi_2) =  z_1(\frac{d\bar{z}_1}{r^2} - 2\frac{\bar{z}_1}{r^3}\pab r) + z_2(\frac{d\bar{z}_2}{r^2} - 2\frac{\bar{z}_2}{r^3}\pab r) = 0.
			\end{aligned}
		\end{equation}
		Therefore $K^{-1}$ is well defined and gives the inverse of $K$.
	\end{proof}
For a refined version of the above results, see \cite{Gwilliam:2018lpo}.

The isomorphism $K$ will play an important role in the KK reduction of $6d$ holomorphic theories.
	
	\subsection{A direct approach to KK reduction}
	
	Using the identification of Dolbeault complex with the tangential Cauchy-Riemann complex from the last section, we can encode the structure of holomorphic Chern-Simons theory into the following dg Lie algebra.
	\begin{equation}
		( \Omega_{3d}^{\sbullet}(\C \times \R_{> 0})\otimes 	\Omega_b^{0,\bullet}(S^3)\otimes\mathfrak{g}[1] , \hat{d} + \pacr,[-,-]).
	\end{equation}
	
	The tangential Cauchy Riemann complex $\Omega_b^{0,\bullet}(S^3)$ has a nice Harmonic decomposition into $SU(2)$ representations  (We provide more detail of this in Appendix \ref{apx:Har}. See also \cite{Folland}.)
	\begin{equation}
		\begin{aligned}
			\Omega_b^{0,0}(S^3) &= \bigoplus_{j,\bar{j}\in \frac{1}{2} \Z_{\geq 0}} \mathcal{H}_{j,\bar{j}},\\
			\Omega_b^{0,1}(S^3) &= \bigoplus_{j,\bar{j}\in \frac{1}{2} \Z_{\geq 0}} \mathcal{H}_{j,\bar{j}} \epsilon,
		\end{aligned}
	\end{equation}
	where 
	\begin{equation}
		\mathcal{H}_{j,\bar{j}} = \left\lbrace  \begin{array}{l}
			\text{Harmonic polynomials that is homogeneous }\\
			\text{of degree $2j$ in $(w_1, w_2)$ and degree $2\bar{j}$ in $(\bar{w}_1,\bar{w}_2)$ }
		\end{array}\right\rbrace,\; \text{ for }j,\bar{j}\in \frac{1}{2} \Z_{\geq 0},
	\end{equation}
	and $\epsilon = \sum\epsilon_{ij}\bar{w}_i\xi_j$ is the $SU(2)$ invariant one form that generate $\Omega_b^{0,1}(S^3)$. The labeling of the space $\mathcal{H}_{j,\bar{j}}$ using half integers seems unnatural. We will see in later sections that this is to make the notation compatible with the standard quantum mechanical notation of $SU(2)$ spin. 
	
	The above harmonic decomposition already provides us with a Kaluza-Klein reduction of the holomorphic Chern-Simons theory. The space $\oplus\mathcal{H}_{j,\bar{j}} $ consists of harmonic polynomials on $S^3$ and labels the full KK towers of fields.
	
	Equivalently, we can regard the full KK tower of BV fields as differential forms $(\Omega_{3d}^{\sbullet}(\C \times \R_{> 0}), \hat{d})$ valued in the very large dg Lie algebra
	\begin{equation}\label{dg_KK}
		(\bigoplus_{j,\bar{j}\in \frac{1}{2} \Z_{\geq 0}} \mathcal{H}_{j,\bar{j}}\otimes \mathfrak{g} ) \oplus(\bigoplus_{j,\bar{j}\in \frac{1}{2} \Z_{\geq 0}} \mathcal{H}_{j,\bar{j}} \epsilon \otimes \mathfrak{g}).
	\end{equation}
	equipped with the tangential Cauchy Riemann differential $\pacr$. The Lie bracket is given by the product on $\Omega_b^{0,\bullet}(S^3)$ and the Lie bracket of $\mathfrak{g}$. Thus we can write down the BV fields of the KK theory as follows
	\begin{equation}
		\begin{aligned}
			\bm{\mathcal{A}}& \in \Omega_{3d}^{\sbullet}(\C \times \R_{> 0})\otimes \bigoplus_{j,\bar{j}\in \frac{1}{2} \Z_{\geq 0}} \mathcal{H}_{j,\bar{j}}[1],\\
			\bm{\mathcal{B}}& \in \Omega_{3d}^{\sbullet}(\C \times \R_{> 0})\otimes\bigoplus_{j,\bar{j}\in \frac{1}{2} \Z_{\geq 0}} \mathcal{H}_{j,\bar{j}}\epsilon.
		\end{aligned}
	\end{equation}
	\begin{remark}
		By decomposing $\Omega_b^{0,\sbullet}(S^3)$ into $SU(2)$ harmonic, the product of two harmonic polynomials is not necessarily a harmonic polynomial. Instead, the product is decomposed into a sum of harmonic polynomials.
		
		It follows from general fact of spherical harmonics that any homogeneous polynomial $p(w_i,\bar{w}_i)$ of bi-degree $(2j,2\bar{j})$ can be expressed as the following sum
		\begin{equation}
			p(w_i,\bar{w}_i) = h_0 + (w_1\bar{w}_1 + w_2\bar{w}_2)h_1 + \dots (w_1\bar{w}_1 + w_2\bar{w}_2)^kh_k  + \dots,
		\end{equation}
		where $h_k \in \mathcal{H}_{j-\frac{k}{2},\bar{j}-  \frac{k}{2}}$.  (We provide a proof of this fact in Appendix \ref{apx:Har})
		
		As a result, for any $h \in \mathcal{H}_{j,\bar{j}}$ and  $h' \in \mathcal{H}_{j',\bar{j}'}$, $hh'$ is a polynomial of bi-degree $(2(j+j'),2(\bar{j} + \bar{j}'))$. So the product $hh'$ has the following decomposition:
		\begin{equation}
		hh' \in	\mathcal{H}_{j+j',\bar{j}+\bar{j}'} \oplus \mathcal{H}_{j+j' - \frac{1}{2},\bar{j}+\bar{j}' - \frac{1}{2}} \oplus  \cdots \oplus \begin{cases}
				\mathcal{H}_{j+j'-\bar{j} - \bar{j}',0} & j+j' \geq \bar{j}+ \bar{j}'\\
				\mathcal{H}_{0,\bar{j}+ \bar{j}' -j - j'} & j+j' < \bar{j}+ \bar{j}'
			\end{cases}.
		\end{equation}
		Writing the product into the above summation gives us the right product rule on $\Omega_b^{0,\sbullet}(S^3)$ after harmonic decomposition.
	\end{remark}

We can also translate the result using non-BV language. Fields of this theory are simply $\mathcal{A}$ -- the one form component of $\bm{\mathcal{A}}$ and $\mathcal{B}$ -- the zero-form component of $\bm{\mathcal{B}}$. We can write down the equation of motion from the dg Lie algebra structure:
	\begin{equation}\label{Eom_naiveKK}
\begin{aligned}
			&\hat{d}\mathcal{A} + \frac{1}{2}[\mathcal{A},\mathcal{A}] = 0,\\
			&\hat{d}\mathcal{B}  + \pacr\mathcal{A} + [\mathcal{A},\mathcal{B}] = 0.
\end{aligned}
	\end{equation}
Gauge transformation are labeled by $\mathcal{C} \in \Omega_{3d}^{0}(\C\times \R_{>0})\otimes  \bigoplus_{j,\bar{j}\in \frac{1}{2} \Z_{\geq 0}} \mathcal{H}_{j,\bar{j}}$ and takes the following form
\begin{equation}
\begin{aligned}
		\delta \mathcal{A}  &= \hat{d} \mathcal{C} + [\mathcal{C},\mathcal{A}],\\
		\delta \mathcal{B} & = \pacr \mathcal{C} + [\mathcal{C},\mathcal{B}].
\end{aligned}
\end{equation}

	To obtain the action functional of the KK theory, we need to analyze the integration map \ref{int_CY}. We choose the holomorphic volume form to be $\Omega_{\C^3} = dzdz_1dz_2$. Then the integration map on $\Omega^{0,\sbullet}(\C^3 \backslash \C) $ is given by
	\begin{equation}
		\int: \alpha  \in \Omega^{0,3}(\C^3 \backslash \C)  \mapsto \int_{\C^3} \Omega_{\C^3}		\alpha .
	\end{equation}
	It induces an integration map 
	\begin{equation}
		\int\circ \,K^{-1} : \Omega_{3d}^{2}(\C \times \R_{> 0}) \otimes 	\Omega_b^{0,1}(S^3) \to \C.
	\end{equation}
	
	\begin{prop}\label{prop_int}
		The induced integration map $\int\circ \,K^{-1}$ on $\Omega_{3d}^{2}(\C \times \R_{> 0})\otimes 	\Omega_b^{0,1}(S^3)$ can be decomposed as
		\begin{equation}
			\int\circ \,K^{-1}  = \int_{3d} \otimes \Tr_{S^3},
		\end{equation} 
		where $\int_{3d}: \Omega_{3d}^{2}(\C \times \R_{> 0}) \to \C $ is the integration map 
		\begin{equation}
			\int_{3d} f(z,\bar{z},t)dtd\bar{z}  = \int_{t > 0}\int_{\C}f(z,\bar{z},t)dtd^2z,
		\end{equation}
		and $\Tr_{S^3}: \Omega_b^{0,1}(S^3) \to \C$ is the surface integral
		\begin{equation}
			\Tr_{S^3}Y(w_i,\bar{w_i})\epsilon = \int_{S^3}Y(w_i,\bar{w_i})d\sigma_{S^3}.
		\end{equation}
	\end{prop}
	\begin{proof}
		Let 
		\begin{equation}
			K(\alpha) = f(z,\bar{z},t)d\bar{z} \wedge dt \wedge Y(w_i,\bar{w}_i) \epsilon \in \Omega_{3d}^{1}(\C \times \R_{> 0})\otimes \Omega_{b}^{0,1}(S^3)
		\end{equation}
	and suppose that $Y \in \mathcal{H}_{j,\bar{j}}$ is homogeneous of degree $(2j,2\bar{j})$.
		Using the formula \ref{Kinv},\ref{Kinv1} for $K^{-1}$ we have
		\begin{equation}
			\begin{aligned}
				\alpha &= f(z,\bar{z},r)d\bar{z} \wedge \bar{\pa}r\wedge r^{-4\bar{j}}Y(z_i,\bar{z}_i) r^{-4}(\bar{z}_1d\bar{z}_2 - \bar{z_2}d\bar{z}_1)\\
				&  = f(z,\bar{z},r)d\bar{z} \wedge r^{2(j - \bar{j})}Y(\frac{z_i}{r},\frac{\bar{z}_i}{r}) \frac{1}{2}r^{-3} d\bar{z}_1\wedge d\bar{z}_2.
			\end{aligned}
		\end{equation}
		We find that 
		\begin{equation}
			\int_{\C^3} \Omega_{\C^3}\wedge \alpha = \int_{3d} \frac{1}{2} f(z,\bar{z},r) r^{2(j - \bar{j})}  d^2z  dr \int_{S^3} Y(\frac{z_i}{r},\frac{\bar{z}_i}{r}) d\sigma_{S^3},
		\end{equation}
		where we used $d^2z_1d^2z_2 = r^3 dr d\sigma_{S^3}$.
		
		Note that $\int_{S^3} Y(\frac{z_i}{r},\frac{\bar{z}_i}{r}) d\sigma_{S^3}$ is only nonzero when $j = \bar{j}$.
		
		Therefore we have
		\begin{equation}
			\int_{\C^3} \Omega_{\C^3}\wedge \alpha = \int_{3d}f(z,\bar{z},r)dt dz \times \Tr_{S^3} Y(w_i,\bar{w}_i)\epsilon.
		\end{equation}
		This gives us the desired decomposition.
	\end{proof}
	
	The above result tells us that, to write down the action functional of the KK theory, we simply replace the $6d$ integration maps with the $3d$ integral and a $\Tr$ on the Lie algebra of modes \ref{dg_KK}. We obtain the following form of BV action functional for the KK fields:
	\begin{equation}
		\int_{3d} \Tr( \bm{\mathcal{B}} \wedge (\hat{d} \bm{\mathcal{A}} + \frac{1}{2}[\bm{\mathcal{A}} ,\bm{\mathcal{A}} ]) + \bm{\mathcal{A}} \wedge \pacr \bm{\mathcal{A}}  ).
	\end{equation}
	Now $\Tr$ is understood as the killing pairing of $\mathfrak{g}$ combined with $\Tr_{S^3}$. Varying this action functional we can reproduce the equation of motion \ref{Eom_naiveKK}.
	
	As we have mentioned in the introduction, we can further simplify this field content by integrating out massive fields. We will see in the next section that the effective interaction of this theory is encoded in the $A_\infty$ structure of the cohomology $H^{0,\sbullet}_{CR}(S^3)$ of the tangential Cauchy Riemann complex.

	\subsection{A cohomological approach to KK reduction}
	\label{sec:dir_KK}
	In the BV formalism, we can use a dg Lie algebra, or more generally an $L_{\infty}$ algebra to encode the structure of the classical field theory. For instance, the holomorphic Chern-Simons theory is described by the dg Lie algebra $(\Omega^{0,\sbullet}(X),\pab,[-,-])$. More generally, we consider an $L_\infty$ algebra $(\mathcal{E},l_1,l_2,\dots)$ \footnote{More precisely, a field theory on a manifold $M$ should correspond to a sheaf of $L_\infty$ algebra on $M$.}. A symmetric pairing $\langle-,-\rangle$ on $\mathcal{E}$ is called cyclic if, for all $n$, the linear map 
	\begin{equation}
		\alpha_0\otimes \dots \otimes \alpha_n \mapsto \langle\alpha_0,l_n(\alpha_1,\dots,\alpha_n)\rangle,\quad \alpha_i \in \mathcal{E}
	\end{equation}
	is (graded) invariant under cyclic permutation of $\alpha_i$'s. For the holomorphic Chern-Simons theory, this cyclic pairing is provided by the integration map and the symmetric pairing of the Lie algebra $\mathfrak{g}$. Given an $L_\infty$ algebra with such a pairing, one can build a classical field theory, with action functional
	\begin{equation}\label{act_Linf}
		S[\alpha] = \sum_{n = 1}^\infty \frac{1}{(n+1)!}\langle \alpha ,l_n(\alpha,\dots,\alpha)\rangle .
	\end{equation}
	In fact, many field theories can be described in this language. We refer to \cite{Hohm:2017pnh} for a review.
	
	By varying the action functional, we obtain the equation of motion
	\begin{equation}
		\sum_{n = 1}^\infty \frac{1}{n!} l_n(\alpha,\dots,\alpha) = 0 .
	\end{equation}
	This is also known as the \MC equation associated with the $L_\infty$ algebra $\mathcal{E}$. Elements that solve the \MC equation (equation of motion) are called \MC elements. Gauge transformation can be defined by exponentiating the infinitesimal gauge transformation
	\begin{equation}
		\delta_c \alpha = l_1(c) + \sum_{n\geq 2}\frac{1}{(n-1)!}l_n(c,\alpha,\dots,\alpha).
	\end{equation}
 Classical field theory concerns the space of equivalent class of solutions to the equation of motion module gauge:
	\begin{equation}
	 \{ \text{Solution to the equation of motion}\}/\{\text{Gauge transform}\},
	\end{equation}
	while mathematician considers the deformation functor associated with the $L_\infty$ algebra defined in the above manner. One important lesson we learned from deformation theory is that for two $L_{\infty}$ algebra that is $L_\infty$ quasi-isomorphic to each other, their corresponding spaces of equivalence class of \MC elements are the same (see e.g. \cite{Kontsevich:1997vb}). Physically, this means that we can think of the two classical field theories defined by two $L_\infty$ quasi-isomorphic $L_\infty$ algebras to be the same. 
	
	Now, we apply these ideas to our problem. In the previous section, we find an isomorphism of dg Lie algebra
	\begin{equation}
		K : \;\;(\Omega^{0,\sbullet}(\C^3\backslash\C)\otimes\mathfrak{g} ,d, [-,-]) = (\Omega_{3d}^{\sbullet}(\C \times \R_{> 0})\otimes 	\Omega_b^{0,\bullet}(S^3)\otimes\mathfrak{g}, \hat{d}+ \pacr , [-,-] )\;,
	\end{equation}
	The dg Lie algebra $\Omega_{3d}^{\sbullet}(\C \times \R_{> 0})\otimes \Omega_b^{0,\bullet}(S^3)\otimes\mathfrak{g}$ provide us a naive form of KK reduction of holomorphic Chern-Simons theory studied in the last section.
	
	We can proceed further, by taking the $\pacr$ cohomology. This gives us a quasi-isomorphism of complexes:
	\begin{equation}
		(\Omega^{0,\sbullet}(\C^3\backslash\C)\otimes\mathfrak{g} ,d) \rightarrow (\Omega_{3d}^{\sbullet}(\C \times \R_{> 0})\otimes H^{0,\sbullet}_{b}(S^3)\otimes\mathfrak{g},\hat{d}),
	\end{equation}
	where $H^{0,\sbullet}_{b}(S^3)$ is the cohomology of the complex $(\Omega_b^{0,\bullet}(S^3),\pacr)$ called tangential Cauchy Riemann cohomology. We will show in Section \ref{sec:SDR} that  
	\begin{equation}
		\begin{aligned}
			H_{b}^{0,0}(S^3) & = \bigoplus_{j \in \frac{1}{2}\Z_{\geq 0}} \mathcal{H}_{j,0} = \C[w_1,w_2],\\
			H_{b}^{0,1}(S^3) &=\bigoplus_{\bar{j}\in \frac{1}{2}\Z_{\geq 0}} \mathcal{H}_{0,\bar{j}}\epsilon = \C[\bar{w}_1,\bar{w}_2]\epsilon.
		\end{aligned}
	\end{equation}
	The homotopy transfer theorem \cite{Kadeishvili1980ONTH} equipped the space $H^{0,\sbullet}_{b}(S^3)\otimes\mathfrak{g}$ an $L_{\infty}$ structure, that makes the above quasi-isomorphism an $L_\infty$ quasi-isomorphism. 
		\begin{equation}
		L_\infty\text{ - q.iso: }\;(\Omega^{0,\sbullet}(\C^3\backslash\C)\otimes\mathfrak{g} ,d,[-,-])\rightarrow (\Omega_{3d}^{\sbullet}(\C \times \R_{> 0})\otimes H^{0,\sbullet}_{b}(S^3)\otimes\mathfrak{g},\hat{d},l_2,l_3,\dots).
	\end{equation}
	Therefore, we can equivalently work with the theory defined by the $L_\infty$ algebra $(\Omega_{3d}^{\sbullet}(\C \times \R_{> 0})\otimes H^{0,\sbullet}_{b}(S^3)\otimes\mathfrak{g},\hat{d},l_2,l_3,\dots)$. Physically, this corresponds to integrating out all the cohomologically trivial KK modes. The advantage of using this theory is that it has a much smaller space of fields than our naive form of KK reduction. The BV fields for this theory consist of  
\begin{equation}
\begin{aligned}
		\mathbf{A}& \in \Omega_{3d}^{\sbullet}(\C \times \R_{> 0})\otimes\mathfrak{g}[w_1,w_2] [1],\\
		\mathbf{B} &\in \Omega_{3d}^{\sbullet}(\C \times \R_{> 0})\otimes \mathfrak{g}[\bar{w}_1,\bar{w}_2]\epsilon.
\end{aligned}
\end{equation}
	Using the volume form constructed in \ref{sec:dir_KK} and the $L_{\infty}$ structure $(H^{0,\sbullet}_{b}(S^3)\otimes\mathfrak{g},\hat{d},l_2,l_3,\dots)$, we can write down the action functional of this theory using the prescription \ref{act_Linf}
\begin{equation}
	 \int_{3d}\Tr \mathbf{B}\left(\hat{d}\mathbf{A} + \sum_{n = 2}^{\infty} \frac{1}{2(n - 1)!} l_n(\mathbf{A},\mathbf{A},\mathbf{B},\dots,\mathbf{B}) \right) .
	\end{equation}
The classical effective interactions for this KK theory is completely encoded in the $L_{\infty}$ structure maps $\{l_n\}$.

For our later computation, we need to explicitly expand the KK fields. Suppose $\{t_a\}_a$ is a basis of the Lie algebra. Let $K_{ab}$ be the corresponding symmetric pairing, and $K^{ab}$ the inverse. We define $t^a = \sum_{b}K^{ab}t_b$ the dual basis. Then we expand our KK fields as follows
	\begin{equation}
		\begin{aligned}
			\mathbf{A} &= \sum_{p,q \geq 0} \sum_a \mathbf{A}^a[p,q]t_aw_1^pw_2^q,\\
			\mathbf{B} &= \sum_{p,q \geq 0} \sum_a \frac{(p+q+1)!}{p!q!}\mathbf{B}_a[p,q]t^a\bar{w}_1^p\bar{w}_2^q.\\
		\end{aligned}
	\end{equation}
With this choice of normalization, the kinetic terms of the BV action can be written as
\begin{equation}
	\sum\int_{3d} \mathbf{B}_a[p,q]\hat{d}\mathbf{A}^a[p,q].
\end{equation}
	The $2$-bracket $l_2$ on $H^{0,\sbullet}_{b}(S^3)\otimes\mathfrak{g}$ is simply given by the product on $H^{0,\sbullet}_{b}(S^3)$ and the Lie bracket on $\mathfrak{g}$. Therefore we can write the cubic term of the BV action as follows
	\begin{equation}
		\sum\int_{3d} f_{ab}^c\mathbf{B}_c[p+r,q+s]\mathbf{A}^a[p,q]\mathbf{A}^b[r,s].
	\end{equation}
	
	To access the higher order interactions, one need to obtain the higher brackets $l_n$ on $H^{0,\sbullet}_{b}(S^3)\otimes\mathfrak{g}$. There are at least two approaches. The first one is to employ homotopy transfer directly to the dg Lie algebra $(\Omega^{0,\sbullet}_{CR}(S^3)\otimes\mathfrak{g},\pacr, [-,-])$. The higher brackets $l_n$ are obtained by summing over maps constructed from binary rooted trees, which we provide more details in \ref{apx:Hom_trans}.
	
	The other approach is to note that the dg Lie algebra $(\Omega^{0,\sbullet}_{b}(S^3)\otimes\mathfrak{g},\pacr, [-,-])$ comes from the tensor product of the dg commutative algebra $(\Omega^{0,\sbullet}_{b}(S^3),\pacr, \wedge)$ and the Lie algebra $\mathfrak{g}$. Then one can apply the homotopy transfer theorem to the dg algebra $(\Omega^{0,\sbullet}_{b}(S^3),\pacr, \wedge)$. This will gives us a $A_\infty$ algebra (in fact a $C_\infty$ algebra)
	\begin{equation}
		(H^{0,\sbullet}_{b}(S^3),m_2,m_3,\dots).
	\end{equation}
	Then we consider the tensor product of the $C_\infty$ algebra $H^{0,\sbullet}_{b}(S^3)$ and the Lie algebra $\mathfrak{g}$. As is shown in \cite{2017arXiv170202194R} in a more general context, there is a canonical $L_\infty$ structure on the tensor product $H^{0,\sbullet}_{b}(S^3)\otimes \mathfrak{g}$. Moreover, \cite{2017arXiv170202194R} showed that the $L_\infty$ structure defined by the tensor product is equivalent to the $L_\infty$ structure defined by directly applying the homotopy transfer. 
	
		In this paper, we follow the second approach that first compute the $A_\infty$ structure $\{m_n\}_{n\geq 2}$ on $H^{0,\sbullet}_{b}(S^3)$. We will give an explicit expression that computes all the higher products $m_n$ in Section \ref{sec:all_mn}, where we also provide a complete expansion of all the higher order effective interactions.

	\subsection{Relation with holomorphic twist of $3d$ $\mathcal{N} = 2$ theory}
	\label{sec:KK_HT}
	In this section, we review the holomorphic topological twist of $3d$ $\mathcal{N} = 2$ supersymmetric field theory following \cite{Costello:2020ndc}.
	
	For a $3d$ $\mathcal{N} = 2$ vector multiplet, the holomorphic twisted theory has the following field content
	\begin{equation}
		\begin{aligned}
			\mathbf{A}  &\in \Omega_{3d}^{\sbullet}\otimes \mathfrak{g}[1],\\
			\mathbf{B}  &\in \Omega_{3d}^{\sbullet,(1)}\otimes \mathfrak{g}^*.
		\end{aligned}
	\end{equation}
	The action functional is
	\begin{equation}\label{3d_vec_act}
		S = \int_{3d} \Tr \left(\mathbf{B}(\hat{d}\mathbf{A} + \frac{1}{2}[\mathbf{A},\mathbf{A}])\right).
	\end{equation}
	
	For the KK theory we obtained in the last section, we let $\epsilon \to dz$, and note that $\C[\bar{w}_1,\bar{w}_2]$ is identified with the dual of $\C[w_1,w_2]$ using the pairing $\Tr_{S^3}$. After these identifications, the field content of the KK theory can be written as follows
	\begin{equation}
		(\mathbf{A},\mathbf{B}) \in \Omega_{3d}^{\sbullet}\otimes \mathfrak{g}[w_1,w_2][1] \oplus \Omega_{3d}^{\sbullet,(1)}\otimes \mathfrak{g}[w_1,w_2]^*.
	\end{equation}
	 If we consider only the differential $\hat{d}$ and the bracket $l_2$, we get exactly the action functional of the standard twisted $3d$ $\mathcal{N} = 2$ theory \ref{3d_vec_act}. This leads us to the following result, generalizing the statement in \cite{Costello:2022wso}
	
	\begin{prop}
		The three-dimensional theory arising from the KK reduction of the holomorphic Chern-Simons theory on $\C^3\backslash\C$ is the same as the holomorphic topological twist of $3d$ $\mathcal{N} = 2$ theory, associated to a $\mathcal{N} = 2$ vector multiplet valued in $\mathfrak{g}[w_1,w_2]$, together with the following deformation
		\begin{equation}
			\sum_{n = 3}^{\infty} \frac{1}{2(n - 1)!}\int \Tr(\mathbf{B}\wedge l_n(\mathbf{A},\mathbf{A},\mathbf{B},\dots,\mathbf{B})).
		\end{equation}
	\end{prop}
	\begin{remark}
		Given an $L_\infty$ algebra, we can also consider the corresponding twisted $3d$ $\mathcal{N} = 2$ theory. However, this is not the same type of  theory as the KK theory we studied. It might be easier to see the distinction using the AKSZ formalism. The mapping space that corresponds to a twisted vector multiplet with gauge Lie algebra $\mathfrak{g}$ is given by $\mathrm{Map}(\R_{dR}\times \C_{\pab},T^*[1]\mathfrak{g})$. The higher order interactions we studied correspond to the $L_\infty$ deformations on $T^*[1]\mathfrak{g}$ instead of on $\mathfrak{g}$ alone.
	\end{remark}
	We mention that the KK theory before we pass to the CR cohomology can be formulated in a similar fashion. According to our discussion in \ref{sec:dir_KK}, the corresponding $3d$ $\mathcal{N} = 2$ theory is associated to a vector multiplet valued in the Lie algebra $ \bigoplus_{p,q}\mathfrak{g}\otimes \mathcal{H}_{p,q}$, together with the following deformation
	\begin{equation}
		\int_{3d} \Tr \left(\bm{\mathcal{A}}\pacr \bm{\mathcal{A}}\right).
	\end{equation}
	Though this formulation of theory has a much simpler action functional, we still need to work with the theory build on the CR cohomology for computational convenience.
	
	For future propose, we also introduce the HT twist of $\mathcal{N} = 2$ chiral multiplet. For chiral multiplet, the holomorphic twist depends on the choice of $R$-symmetry that we assign to the field. Suppose the chiral fields live on the space $V$, which decompose according to $R$ charges: $V = \oplus_rV^{(r)}$. Then a twisted chiral multiplet consists of the following field content
	\begin{equation}
			\begin{aligned}
		\bm{\Phi}  &\in \oplus_r \Omega_{3d}^{\sbullet,(r)}\otimes V^{(r)},\\
		\bm{\Psi} &\in \oplus_r \Omega_{3d}^{\sbullet,(1-r)}\otimes (V^{(r)})^*.
	\end{aligned}
	\end{equation}
 The BV action functional involving both vector and chiral multiplets and their coupling is given by
	\begin{equation}\label{3d_vec_chi_act}
	S = \int_{3d} \Tr \left(\mathbf{B}(\hat{d}\mathbf{A} + \frac{1}{2}[\mathbf{A},\mathbf{A}])+ \bm{\Psi}(\hat{d}\bm{\Phi} + [\mathbf{A},\bm{\Phi}])\right).
\end{equation}
	Later we will consider the case when the chiral multiplet is valued in the adjoint representation of the lie algebra $\mathfrak{g}$. In this case, the chiral multiplet has field content
	\begin{equation}
	\begin{aligned}
		\bm{\Phi}  &\in \Omega_{3d}^{\sbullet,(1)}\otimes \mathfrak{g},\\
		\bm{\Psi} &\in \Omega_{3d}^{\sbullet}\otimes \mathfrak{g}^*.
	\end{aligned}
\end{equation}
Together with the $\mathcal{N} = 2$ vector multiplet, this theory actually consists of a $3d$ $\mathcal{N} = 4$ vector multiplet. We will consider a theory of this type in later section.

\section{$A_\infty$ structure on the tangential Cauchy-Riemann cohomology}
	\label{sec:A_inf}
In this section, we study more details about the tangential Cauchy Riemann complex on $S^3$. We present one approach to understand the $A_\infty$ structure $\{m_n\}_{n\geq 2}$ on $H_b^{0,\sbullet}(S^3)$ and provide a general formula for all $m_n$ in Section \ref{sec:all_mn}. We also provide a formula for the $L_\infty$ structure on the tensor product $H_b^{0,\sbullet}(S^3)\otimes \mathfrak{g}$. This will lead us to a complete understanding of the classical effective action of the KK theory.
	\subsection{A special deformation retract}
	\label{sec:SDR}
 First, it will be convenient to choose an orthonormal basis for each space $\mathcal{H}_{j,\bar{j}}$. We use the surface integral $\int_{S^3}d\sigma_{S^3}$ to define a Hermitian form on $\mathcal{H} = \oplus \mathcal{H}_{j,\bar{j}}$ by
 \begin{equation}
 \label{sur_pair}
 \langle a,b \rangle : = \int_{S^3}d\sigma_{S^3}\bar{a}b.
 \end{equation}
Under the $SU(2)$ action, there is a canonical choice of basis generated by the highest weight vector. For the space $\mathcal{H}_{j,0}$, the corresponding orthonormal basis is given by
\begin{equation}
	\left\lbrace e_m^{(j)}: =\sqrt{\frac{(2j+1)!}{(j+m)!(j-m)!}} w_1^{j+m}w_2^{j - m}\mid -j \leq m\leq  j\right\rbrace,
\end{equation}
with $e^{(j)}_j = \sqrt{2j + 1} w_1^{2j}$ the highest weight vector. For the space $\mathcal{H}_{0,\bar{j}}$, the corresponding orthonormal basis is given by
\begin{equation}
	\left\lbrace \bar{e}^{(\bar{j})}_{\bar{m}} := \sqrt{\frac{(2\bar{j}+1)!}{(\bar{j} + \bar{m})!(\bar{j} - \bar{m})!}} \bar{w}_2^{\bar{j} + \bar{m}}(-\bar{w}_1)^{\bar{j} - \bar{m}}\mid -\bar{j} \leq \bar{m}\leq  \bar{j} \right\rbrace.
\end{equation}
For $\mathcal{H}_{j,\bar{j}}$ with $j,\bar{j} \neq 0$, we denote the corresponding orthonormal basis by
\begin{equation}
		\left\lbrace e^{(j,\bar{j})}_m\mid -(j+\bar{j})\leq m \leq j+\bar{j}\right\rbrace.
\end{equation}
The normalized highest weight vector can be expressed as
\begin{equation}\label{ortho_basis_jjb}
	e^{(j,\bar{j})}_{j+\bar{j}} = \sqrt{\frac{(2j+ 2\bar{j} +1)!}{(2j)!(2\bar{j})!}} w_1^{2j}\bar{w}_2^{2\bar{j}}.
\end{equation}
Other elements of this basis are uniquely determined by the $SU(2)$ action. We will postpone writing down explicitly this orthonormal basis in terms of harmonic polynomials until the next section. 

We want to understand how the Cauchy Riemann differential $\pacr$ behaves under the harmonic decomposition. By definition, $\xi_i$ transform under the anti-fundamental representation of $SU(2)$, and $\frac{\pa}{\pa \bar{w}_i}$ transform under the fundamental representation of $SU(2)$. Therefore, the operator $\pacr$ commute with $SU(2)$. By Schur's lemma, on each irreducible subspace, $\pacr$ is either zero or a scalar multiple of the identity onto an irreducible subspace of the same $SU(2)$ representation.

Therefore, it suffices to look at the action of $\pacr$ on the highest weight vector on each irreducible subspace. We have
\begin{equation}
	\pacr (w_1^{2j}\bar{w}_2^{2\bar{j}})= 2\bar{j}w_1^{2j}\bar{w}_2^{2\bar{j} - 1}\xi_2 = 2\bar{j}w_1^{2j + 1}\bar{w}_2^{2\bar{j} - 1}\epsilon.
\end{equation}
We find that
\begin{equation}
	\pacr : \begin{array}{l}
		\mathcal{H}_{j,0} \to 0\\
		\mathcal{H}_{j,\bar{j}} \overset{\simeq}{\to} \mathcal{H}_{j+ \frac{1}{2},\bar{j} - \frac{1}{2}}
	\end{array}.
\end{equation}
Under the orthonormal basis, the Cauchy Riemann differential $\pacr$ is given by the following constant
\begin{equation}
	\pacr|_{\mathcal{H}_{j,\bar{j}} \to \mathcal{H}_{j+ \frac{1}{2},\bar{j} - \frac{1}{2}}} = \sqrt{2\bar{j}(2j+1)}.
\end{equation}

Given this knowledge about $\pacr$, it is easy to compute the tangential Cauchy Riemann cohomology of $S^3$
\begin{equation}
	\begin{aligned}
		H_{b}^{0,0}(S^3) & = \bigoplus_{j \in \frac{1}{2}\Z_{\geq 0}} \mathcal{H}_{j,0},\\
		H_{b}^{0,1}(S^3) &=\bigoplus_{\bar{j}\in \frac{1}{2}\Z_{\geq 0}} \mathcal{H}_{0,\bar{j}}\epsilon.
	\end{aligned}
\end{equation}

We define an operator $h: \Omega_b^{0,1}(S^3) \to \Omega_b^{0,0}(S^3) $ by "inverse" of $\pacr$ as follows
\begin{equation}
	h: \begin{array}{l}
		\mathcal{H}_{0,\bar{j}}\epsilon \to 0\\
		\mathcal{H}_{j,\bar{j}}\epsilon \overset{\simeq}{\to} \mathcal{H}_{j- \frac{1}{2},\bar{j} + \frac{1}{2}}
	\end{array},
\end{equation}
so that $h$ acts on the highest weight vector by
\begin{equation}
	h(w_1^{2j}\bar{w}_2^{2\bar{j}}\epsilon) = \frac{1}{2\bar{j}+1}w_1^{2j - 1}\bar{w}_2^{2\bar{j} + 1} .
\end{equation}
Under the orthonormal basis \ref{ortho_basis_jjb}, $h$ is given by the following constant
\begin{equation}
	h_{j,\bar{j}}: = h|_{\mathcal{H}_{j,\bar{j}}\epsilon \to \mathcal{H}_{j-\frac{1}{2},\bar{j}+\frac{1}{2}}}  = \frac{1}{\sqrt{2j(2\bar{j} + 1)}}.
\end{equation}
This operator $h$ will play the role of "propagator" in obtaining effective interactions as we integrate out the massive fields.
We can verify that 
\begin{equation}
	i\circ p- 1  = \pacr \circ h + h\circ \pacr,
\end{equation}
where $i$ and $p$ are the standard inclusion and projection between $\Omega_b^{0,\sbullet}(S^3)$ and $H^{0,\sbullet}_{CR}(S^3)$.

We can also verify that
\begin{equation}
	h\circ i  = 0,\; p\circ h = 0,\; h\circ h = 0.
\end{equation}
In mathematical language, we have constructed a special deformation retract (SDR)
\begin{equation}
	h\curved (\Omega_b^{0,\sbullet}(S^3),\pacr)\overset{p}{\underset{i}\rightleftarrows} (H^{0,\sbullet}_{CR}(S^3) , 0).
\end{equation}
We give more details about the definition of SDR in Appendix \ref{apx:Hom_trans}, where we also provide its connection with homotopy transfer theorem. 

	It is important to note that the SDR we constructed is compatible with the cyclic structure on $\Omega_{b}^{0,\sbullet}(S^3)$. One easily check that $\Tr_{S^3}(h(a)b) = \Tr_{S^3}(ah(b))$. Under this extra condition, it was shown in \cite{Kajiura:2003ax} that the transferred $A_\infty$ structure on the cohomology is also a cyclic $A_\infty$ algebra. The bilinear pairing on the cohomology is simply given by the restriction of the original bilinear pairing.

\subsection{Product rule of $S^3$ harmonics}
\label{sec:pro_Har}

As we have mentioned, the product rule of $S^3$ harmonic polynomials encodes the interaction of the KK theory. In this section, we analyze the product rule. 
\begin{equation}
	M: \mathcal{H}_{j_1,\bar{j}_1}\otimes \mathcal{H}_{j_2,\bar{j}_2} \to \mathcal{H}_{j_1+j_2,\bar{j}_1+\bar{j}_2}\oplus\mathcal{H}_{j_1+j_2 - \frac{1}{2},\bar{j}_1+\bar{j}_2 - \frac{1}{2}}\oplus \dots
\end{equation}

We first analyze some special cases.
\paragraph{Product on the CR cohomology}
By projecting to the CR cohomology, the $2$-product is simplified. We can identify $H_{b}^{0,0}(S^3) = \C[w_1,w_2]$ and the product on this subspace is the commutative product of the polynomial algebra. By degree reason, the only other nonzero product is $H_{b}^{0,0}(S^3)\otimes H_{b}^{0,1}(S^3) \to  H_{b}^{0,1}(S^3) $, which is given by
\begin{equation}
	m_2 = p\circ M:\; 	\mathcal{H}_{j,0}\otimes \mathcal{H}_{0,\bar{j}} \epsilon \to \begin{cases}
		0 & \text{ if } j > \bar{j},\\
		\mathcal{H}_{0,\bar{j} - j}\epsilon & \text{ if } j \leq \bar{j}.
	\end{cases}
\end{equation}

Using orthogonality of harmonic polynomials under the surface measure $\int_{S^3}d\sigma_{S^3}$, we find that
\begin{equation}\label{2-pro-01}
	m_2(w_1^pw_2^q,\frac{(r+s + 1)!}{r!s!}\bar{w}_1^{r}\bar{w}_2^{s}\epsilon) = \frac{(r+s - p - q +1)!}{(r-p)!(s-q)!}\bar{w}_1^{r-p}\bar{w}_2^{s-q}\epsilon.
\end{equation}
\begin{remark}
	Alternatively, we can identify $\C[\bar{w}_1,\bar{w_2}]$ with the dual of $H_{b}^{0,0}(S^3) = \C[w_1,w_2]$ via the pairing \ref{sur_pair}. The commutative product of $\C[w_1,w_2]$ induce a dual map
	\begin{equation}
		m_2' : H_{b}^{0,0}(S^3)  \otimes H_{b}^{0,1}(S^3) \to H_{b}^{0,1}(S^3) ,
	\end{equation}
	which is defined by $m_2'(f,\bar{g}\epsilon)(h) = \Tr_{S^3} (f\cdot h\cdot\bar{g}\epsilon)$ for $f \in H_{b}^{0,0}(S^3)$ and $\bar{g}\epsilon \in H_{b}^{0,1}(S^3)$.
	
	The fact that $m_2'$ and $m_2$ are the same is a consequence of the cyclic property of $m_2$: $\Tr_{S^3} (hm_2(f,\bar{g}\epsilon)) = \Tr_{S^3}(\bar{g}\epsilon m_2(f,h))$. This fact will be particularly useful when we study the higher products $m_n$ on the CR cohomology. 
\end{remark}

\paragraph{Product on $\mathcal{H}_{j,0}\otimes \mathcal{H}_{0,\bar{j}}$}
We want to compute the product of two general $S^3$ harmonics in $\oplus \mathcal {H}_{j,\bar{j}}$. Note that harmonic polynomials are just certain polynomials of the variable $w_i$ and $\bar{w}_i$. The product in the subalgebra $\C[w_i]$ and $\C[\bar{w}_i]$ respectively are easy and is the commutative product of polynomial. The only difficult part is the product between $\C[w_i]$ and $\C[\bar{w}_i]$ after harmonic decomposition. Therefore, we first compute the product $M$ restricted to $\mathcal{H}_{j,0}\otimes \mathcal{H}_{0,\bar{j}}$:
\begin{equation}
	M: \mathcal{H}_{j,0}\otimes \mathcal{H}_{0,\bar{j}} \to \mathcal{H}_{j,\bar{j}}\oplus  \mathcal{H}_{j - \frac{1}{2},\bar{j} - \frac{1}{2}} \oplus \cdots  \oplus \begin{cases}
		\mathcal{H}_{j-\bar{j},0} & j > \bar{j}\\
		\mathcal{H}_{0,\bar{j} -j} & j \leq \bar{j}
	\end{cases}.
\end{equation}
Denote $\pi_j = \mathcal{H}_{j,0}$. Note that $\{\pi_j\}_{j\in \frac{1}{2}\Z_{\geq 0}}$ enumerate all irreducible representation of $SU(2)$. We have an isomorphism
\begin{equation}
	\phi_{j,\bar{j}}: \pi_{j+\bar{j}} \overset{\cong}{\to}\mathcal{H}_{j,\bar{j}}
\end{equation}
as $SU(2)$ representation. Using the orthogonal basis of  $\pi_j$ and $\mathcal{H}_{j,\bar{j}}$ defined in the last section, we have $e^{(j,\bar{j})}_m = \phi_{j,\bar{j}}(e^{(j+\bar{j})}_m)$. 

Note that we have the following tensor product rule of $SU(2)$ representations
\begin{equation}
	CG: \pi_{j} \otimes \pi_{\bar{j}} \cong \pi_{j+\bar{j}} \oplus \pi_{j+\bar{j} - 1} \oplus \cdots \oplus \pi_{|j-\bar{j}|}.
\end{equation}
The matrix elements of the above isomorphism in the orthogonal basis are given by the $SU(2)$ Clebsch-Gordan coefficients $C^{j_1,j_2;j_3}_{m_1,m_2;m_3}$. Then we consider the map
\begin{equation}
	M\circ CG^{-1} : \pi_{j+\bar{j}} \oplus \pi_{j+\bar{j} - 1} \oplus \cdots \oplus \pi_{|j-\bar{j}|} \to \mathcal{H}_{j,\bar{j}}\oplus  \mathcal{H}_{j - \frac{1}{2},\bar{j} - \frac{1}{2}} \oplus \cdots  \oplus \begin{cases}
		\mathcal{H}_{j-\bar{j},0} & j > \bar{j}\\
		\mathcal{H}_{0,\bar{j} -j} & j \leq \bar{j}
	\end{cases}.
\end{equation}
Since both $M$ and $CG$ intertwine the $SU(2)$ action, by Shur's lemma $M\circ CG^{-1}$ must be a constant multiple of identity on each irreducible subspace. Therefore, we have the following diagram

\begin{center}
	\begin{tikzpicture}[scale=0.7]
		\node (0) at (0,0) {$ {H}_{j,0}\otimes \mathcal{H}_{0,\bar{j}}$};
		\node (10) at (-3,-2) {$ \pi_{j+\bar{j}}$};
		\node (11) at (3,-2) {$ \mathcal{H}_{j,\bar{j}}$};
		\node at (-3,-2.75) {$\oplus$};
		\node at (3,-2.75) {$\oplus$};
		\node (20) at (-3,-3.5) {$ \pi_{j+\bar{j} - 1}$};
		\node (21) at (3,-3.5) {$ \mathcal{H}_{j - \frac{1}{2},\bar{j} - \frac{1}{2}}$};
		\node at (-3,-4.25) {$\oplus$};
		\node at (3,-4.25) {$\oplus$};
		\node at (-3,-4.8) {$\vdots$};
		\node at (3,-4.8) {$\vdots$};
		\node (30) at (-3,-5.7) {$ \pi_{j+\bar{j}- k}$};
		\node (31) at (3,-5.7) {$ \mathcal{H}_{j - \frac{k}{2},\bar{j} - \frac{k}{2}}$};
		\node at (-3,-6.3) {$\vdots$};
		\node at (3,-6.3) {$\vdots$};
		\draw [->][shorten >=1pt,shorten <=1pt] (0) -- (10);
		\draw [->][shorten >=1pt,shorten <=1pt]    (0) -- (11);
		\draw [->][shorten >=7pt,shorten <=5pt]    (10) -- (11);
		\draw [->][shorten >=3pt,shorten <=3pt]     (20) -- (21);
		\draw [->][shorten >=3pt,shorten <=3pt]    (30) -- (31);
		\node at (-1.6,-0.7) {\footnotesize$CG$};
		\node at (1.6,-0.7) {\footnotesize$M$};
		\node at (0,-1.7) {\footnotesize$\lambda_{j,\bar{j},0}$};
		\node at (0,-3.2) {\footnotesize$\lambda_{j,\bar{j},1}$};
		\node at (0,-5.4) {\footnotesize$\lambda_{j,\bar{j},k}$};
	\end{tikzpicture}
\end{center}
We see that the map $M$ is completely characterized by the constant $\lambda_{j,\bar{j},k}$. To compute each constant $\lambda_{j,\bar{j},k}$, it suffices to compute the map $M\circ CG^{-1}$ on each highest weight vector. We have
\begin{equation}
	\begin{aligned}
		M\circ CG^{-1} (e^{(j + \bar{j} - k)}_{j + \bar{j} - k})&  = \sum_{m = -j}^j C^{j,\bar{j};j+\bar{j} - k}_{j - m,\bar{j} - k + m;j + \bar{j} - k}e^{(j)}_{j - m}\bar{e}^{(\bar{j})}_{\bar{j} - k + m}\\
		& = (-1)^k\sqrt{\frac{(2j+1)!(2\bar{j} + 1)!}{k!(2j + 2\bar{j} - k + 1)!}} e^{(j - \frac{k}{2},\bar{j} - \frac{k}{2})}_{j + \bar{j} - k}.
	\end{aligned}
\end{equation}
Therefore,
\begin{equation}
	\lambda_{j,\bar{j},k} = (-1)^k\sqrt{\frac{(2j+1)!(2\bar{j} + 1)!}{k!(2j + 2\bar{j} - k + 1)!}}.
\end{equation}
This gives us the following
\begin{prop}
	The  product map $M$ on the subspace $\mathcal{H}_{j,0}\otimes \mathcal{H}_{0,\bar{j}}$ is given by
\begin{equation}\label{Pro_M}
	M(e^{(j)}_m,\bar{e}^{(\bar{j})}_{\bar{m}}) = \sum_{k = 0}^{\min(2j,2\bar{j})} \lambda_{j,\bar{j},k}C^{j,\bar{j};j+\bar{j} - k}_{m,\bar{m};m+\bar{m}} e^{(j - \frac{k}{2},\bar{j} - \frac{k}{2})}_{m+\bar{m}}.
\end{equation}
\end{prop} 

One can also check that, applying $p\circ M$ to ${H}_{j,0}\otimes \mathcal{H}_{0,\bar{j}}\epsilon$ we reproduce the formula \ref{2-pro-01} of $m_2$.

In fact, we should understand the product $M$ as an identity in the ring $\C[w_i,\bar{w}_i]/(w_1\bar{w}_1 + w_2\bar{w}_2  - 1)$, that express a monomial of $w_i,\bar{w}_i$ as a linear combination of harmonic polynomials. Then the inverse of $M$ on $\mathcal{H}_{j,\bar{j}}$:
\begin{equation}
	M^{-1}: \mathcal{H}_{\bar{j},\bar{j}} \to \mathcal{H}_{j,0}\otimes \mathcal{H}_{0,\bar{j}}
\end{equation}
should be understood as an identity that expresses the orthonormal basis of $\mathcal{H}_{j,\bar{j}}$ as a polynomial of $w_i,\bar{w}_i$.
Since the Clebsch-Gordan coefficients are real and form a unitary matrix, we can easily write down the matrix elements of $M^{-1}$. This leads us to the following
\begin{prop}
	The orthonormal basis of the space of harmonic polynomials $\mathcal{H}_{j,\bar{j}}$ can be written as follows
	\begin{equation}\label{Har_to_Pol}
		\begin{aligned}
			e^{(j,\bar{j})}_l &= \sum_{m}\lambda_{j,\bar{j},0}^{-1} C^{j,\bar{j};j+\bar{j}}_{l - m,m;l}e^{(j)}_{l - m}\bar{e}^{(\bar{j})}_{m}\\
			& =  \sum_{m} (-1)^{\bar{j}  - m}\frac{\sqrt{(2j + 2\bar{j} + 1) (2j)!(2\bar{j})! (j+\bar{j} + l)!(j+\bar{j} - l)}}{(j + l - m)!(j - l +m)!(\bar{j} - m)!(\bar{j}+m)!}w_1^{j+l-m}w_2^{j - l + m}\bar{w}_1^{\bar{j}-m}\bar{w}_2^{\bar{j}+m}.
		\end{aligned}
	\end{equation}
\end{prop}

Given the above results, one can compute the product $M$ of any two $S^3$ harmonics. First, we write the harmonic polynomials into polynomials of $w_i,\bar{w}_i$ using the above formula. Then we can perform the product in the polynomial ring $\C[w_i,\bar{w}_i]$. Finally, we use $M|_{\mathcal{H}_{j,0}\otimes \mathcal{H}_{0,\bar{j}}}$ to decompose the polynomial into harmonic polynomials, which gives us the desired product map. Following this idea, we show in Appendix \ref{apx:pro_Har} that the product of two arbitrary $S^3$ harmonics is given by  
\begin{equation}\label{pro_Har_arb}
	\begin{aligned}
		& M(e^{(j_1,\bar{j}_1)}_{m_1},e^{(j_2,\bar{j}_2)}_{m_2})\\
		= &\sum_{k}\lambda_{j_1,\bar{j}_1,0}^{-1}\lambda_{j_2,\bar{j}_2,0}^{-1} \lambda_{j_1+j_2,\bar{j}_1+\bar{j}_2,k}\sqrt{(2j_1+1)(2j_2+1)(2\bar{j}_1+1)(2\bar{j}_2+1)(2j_1 + 2\bar{j}_1+1)(2j_2+2\bar{j}_2 + 1)} \\
		\times &\begin{Bmatrix}
			j_1&j_2&j_1+j_2\\\bar{j}_1&\bar{j}_2&\bar{j}_1+\bar{j}_2\\j_1+\bar{j}_1&j_2+\bar{j}_2&j_1+j_2+\bar{j}_1+\bar{j}_2 - k
		\end{Bmatrix}C^{j_1+\bar{j}_1,j_2+\bar{j}_2;j_1+j_2+\bar{j}_1+\bar{j}_2 - k}_{m_1,m_2;m_1+m_2}e^{(j_1+j_2- \frac{k}{2},\bar{j}_1+\bar{j}_2 - \frac{k}{2})}_{m_1+m_2},
	\end{aligned}
\end{equation}
where $\begin{Bmatrix}
	j_1&j_2&j_3\\j_4&j_5&j_6\\j_7&j_8&j_9
\end{Bmatrix}$ is the Wigner $9j$ symbol.

\subsection{$A_\infty$ structure on the CR cohomology:  $m_3$}
\label{sec:A_trans}

Now we can proceed to consider the $A_\infty$ structure on CR cohomology that encode the effective interaction of the KK theory. 

\begin{prop}
	There exist a nontrivial $A_\infty$ structure (actually a $C_\infty$ structure) $\{m_n\}_{n\geq2}$ on $H_{b}^{0,\sbullet}(S^3)$, such that the $A_\infty$ algebra $(H_{b}^{0,\sbullet}(S^3),\{m_n\}_{n\geq2})$ is $A_\infty$ quasi-isomorphic to the differential graded commutative algebra $(\Omega_b^{0,\sbullet}(S^3),\pacr,\wedge )$
\end{prop}

The existence of this $A_\infty$ structure is a corollary of the homotopy transfer theorem. Since $(\Omega_b^{0,\sbullet}(S^3),\pacr,\wedge )$ is a differential graded commutative algebra, the transferred structure is also a $C_\infty$ algebra \cite{2006math10912Z}. The fact that this $A_\infty$ structure is nontrivial is shown in \cite{Polishchuk2003} in a more general context. 

The $A_\infty$ operation $m_n$ can be construed as follows
\begin{equation}
	m_n = \sum_{T \in PBT_n}(\pm)m_T.
\end{equation}	
Here the summation is taken over all rooted planar binary trees $T$ with $n$ leaves. The map $m_T$ is construed by assigning the product map $M$ on the vertices, $h$ on the internal edges, $i$ on the leaves and $p$ on the root.

In this section, we warm up by computing the product $m_3$ on $H_{b}^{0,\sbullet}(S^3)$. It is given by the following trees
\begin{center}
	\begin{tikzpicture}[grow' = up]
		\tikzstyle{level 1}=[level distance = 2mm]
		\tikzstyle{level 2}=[sibling distance=8mm,level distance = 8mm]
		\tikzstyle{level 3}=[sibling distance=5mm,level distance = 8mm]
		\tikzstyle{level 4}=[sibling distance=8mm,level distance = 8mm]
		\coordinate
		child {
			edge from parent[draw=none] child {
				node{$m_3$}
				child { child {edge from parent[draw=none]}}  child { child {edge from parent[draw=none]}} child { child {edge from parent[draw=none]}}
			}
		};
		\node at (1,0.5) {$=$}
		;	\end{tikzpicture}
	\begin{tikzpicture}[grow' = up]
		\tikzstyle{level 1}=[sibling distance=8mm,level distance = 2mm]
		\tikzstyle{level 2}=[sibling distance=8mm,level distance = 7mm]
		\tikzstyle{level 3}=[sibling distance=8mm,level distance = 7mm]
		\tikzstyle{level 4}=[sibling distance=8mm,level distance = 7mm]
		\tikzstyle{level 5}=[sibling distance=8mm,level distance = 2mm]
		\coordinate
		node {$p$}child {edge from parent[draw=none]
			child { child { node{$i$} }
				child {	
					child {
						child{node{$i$} edge from parent[draw=none]}}  child {child{node{$i$} edge from parent[draw=none]}} edge from parent node[right] {$h$} }
			}
		};
		\node at (1.5,0.4) {$-$};
	\end{tikzpicture}
	\begin{tikzpicture}[grow' = up]
		\tikzstyle{level 1}=[sibling distance=8mm,level distance = 2mm]
		\tikzstyle{level 2}=[sibling distance=8mm,level distance = 7mm]
		\tikzstyle{level 3}=[sibling distance=8mm,level distance = 7mm]
		\tikzstyle{level 4}=[sibling distance=8mm,level distance = 7mm]
		\tikzstyle{level 5}=[sibling distance=8mm,level distance = 2mm]
		\coordinate
		node {$p$}child {edge from parent[draw=none]
			child {
				child {	
					child {
						child{node{$i$} edge from parent[draw=none]}}  child {child{node{$i$} edge from parent[draw=none]}} edge from parent node[left] {$h$} }    child { node{$i$}}
			}
		};
	\end{tikzpicture}
\end{center}
Explicitly, we have
\begin{equation}
	m_3(a,b,c) = pM(a,hM(b,c)) - pM(hM(a,b),c),
\end{equation}
where we omit the inclusion $i$ for simplicity.	

Since $H_{b}^{0,\sbullet}(S^3)$ is concentrated in degree $0$ and $1$, and $m_3$ is of degree $-1$ by definition, we have that $m_3$ is non zero only in the following subspace of $H_{b}^{0,\sbullet}(S^3)^{\otimes 3}$
\begin{equation}
\begin{aligned}
		&\bigoplus_{\text{perm}}H_{b}^{0,0}(S^3)\otimes H_{b}^{0,0}(S^3)\otimes H_{b}^{0,1}(S^3),\\
		&\bigoplus_{\text{perm}}H_{b}^{0,0}(S^3)\otimes H_{b}^{0,1}(S^3)\otimes H_{b}^{0,1}(S^3).
\end{aligned}
\end{equation}
where we sum over all permutations of the tensor factor. Due to the cyclic structure, $m_3$ on the different subspaces are related by
\begin{equation}
	\Tr_{S^3}(a_0 \times m_3(a_1,\bar{a}_0\epsilon,\bar{a}_1\epsilon)) = -\Tr_{S^3}(\bar{a}_1\epsilon \times  m_3(a_0,a_1,\bar{a}_0\epsilon)) 
\end{equation}
Therefore it suffices to only consider $m_3$ on $\bigoplus_{\text{perm}}H_{b}^{0,0}(S^3)\otimes H_{b}^{0,0}(S^3)\otimes H_{b}^{0,1}(S^3)$.

First we consider $m_3$ restricted on $\mathcal{H}_{j_1,0}\otimes \mathcal{H}_{j_2,0} \otimes \mathcal{H}_{0,\bar{j}}\epsilon$. Because $h = 0$ restricted on $H_{b}^{0,0}(S^3)$, $pM(hM(a,b),c) = 0$ for $a,b \in H_{b}^{0,0}(S^3)$. Therefore, $m_3$ is given by $pM(a,hM(b,c))$ in this case. 

If $\bar{j} < j_2$, $pM(-,hM(-,-))$ is given by the following composition of maps
\begin{equation}
	\begin{split}
		\mathcal{H}_{j_1,0}\otimes \mathcal{H}_{j_2,0} \otimes \mathcal{H}_{0,\bar{j}}\epsilon&\overset{1\otimes M}{\cong} \mathcal{H}_{j_1,0}\otimes ( \mathcal{H}_{j_2,\bar{j}}\epsilon\oplus \mathcal{H}_{j_2 -\frac{1}{2},\bar{j} - \frac{1}{2}}\epsilon \oplus \dots  \oplus \mathcal{H}_{j_2 -\bar{j},0}\epsilon)\\
		&	\overset{1\otimes h}{\to} \mathcal{H}_{j_1,0}\otimes ( \mathcal{H}_{j_2 - \frac{1}{2},\bar{j}+ \frac{1}{2}}\epsilon\oplus \mathcal{H}_{j_2 -1,\bar{j} }\epsilon \oplus \dots  \oplus \mathcal{H}_{j_2 -\bar{j} - \frac{1}{2},\frac{1}{2}}\epsilon)\\
		&\overset{p M}{\to}\mathcal{H}_{j_1 + j_2 - \bar{j} - 1,0} .
	\end{split}
\end{equation}
If $\bar{j} \geq j_2$, the formula for computing $m_3$ takes the same form. However, there is a slight difference in that
\begin{equation}
	\mathcal{H}_{j_1,0}\otimes \mathcal{H}_{j_2,0} \otimes \mathcal{H}_{0,\bar{j}}\epsilon\overset{1\otimes M}{\cong} \mathcal{H}_{j_1,0}\otimes ( \mathcal{H}_{j_2,\bar{j}}\epsilon\oplus \mathcal{H}_{j_2 -\frac{1}{2},\bar{j} - \frac{1}{2}}\epsilon \oplus \dots  \oplus \mathcal{H}_{0,\bar{j} - j_2}\epsilon).
\end{equation}
Note that $\mathcal{H}_{0,\bar{j} - j_2}$ is sent to $0$ after we apply $h$.

We compute $m_3(e^{(j_1)}_{m_1}, e^{(j_2)}_{m_2},\bar{e}^{(\bar{j})}_{\bar{m}}\epsilon)$ following these steps. According to the above formula, we first need to compute $h M(e^{(j_2)}_{m_2},\bar{e}^{(\bar{j})}_{\bar{m}}\epsilon)$. We have
\begin{equation}
	h (M(e^{(j_2)}_{m_2},\bar{e}^{(\bar{j})}_{\bar{m}} \epsilon))=  \sum_{i = 0}^{\min(2\bar{j},2j_2 - 1)}h_{j_2-\frac{i}{2},\bar{j} -\frac{i}{2}}\lambda_{j_2,\bar{j},i}C^{j_2,\bar{j};j_2 + \bar{j} - i}_{m_2,\bar{m},m_2 + \bar{m}} e^{(j_2 - \frac{i + 1}{2},\bar{j} - \frac{i-1}{2})}_{m_2 + \bar{m}}.
\end{equation}
To compute the product $pM(e^{(j_1)}_{m_1},e^{(j_2 - \frac{i + 1}{2},\bar{j} - \frac{i-1}{2})}_{m_2 + \bar{m}})$, we can use a variation of the formula \ref{Har_to_Pol} of $M^{-1}$ to write $e^{(j_2 - \frac{i + 1}{2},\bar{j} - \frac{i-1}{2})}_{m_2 + \bar{m}}$ as a polynomial
\begin{equation}
	e^{(j_2 - \frac{i + 1}{2},\bar{j} - \frac{i-1}{2})}_{m_2 + \bar{m}} = \sum_{m'}\lambda_{j_2 - \frac{1}{2},\bar{j} +\frac{1}{2},i}^{-1} C^{j_2 - \frac{1}{2},\bar{j} +\frac{1}{2};j_2 + \bar{j} - i}_{m',m_2 + \bar{m} - m';m_2 + \bar{m}}e^{(j_2 - \frac{1}{2})}_{m'}\bar{e}^{(\bar{j} +\frac{1}{2})}_{m_2 + \bar{m} - m'}.\\
\end{equation}
Then we find that
\begin{equation}
	\begin{aligned}
		pM(e^{(j_1)}_{m_1}&,e^{(j_2 - \frac{i + 1}{2},\bar{j} - \frac{i-1}{2})}_{m_2 + \bar{m}})= \sum_{m'}(-1)^{2\bar{j} + 1}\lambda_{j_2 - \frac{1}{2},\bar{j} +\frac{1}{2},i}^{-1}\sqrt{\frac{2j_2(2j_1 + 1)(2\bar{j}+2)}{2j_1 + 2j_2}}\\
		& \times C^{j_2 - \frac{1}{2},\bar{j} +\frac{1}{2};j_2 + \bar{j} - i}_{m',m_2 + \bar{m} - m';m_2 + \bar{m}} C^{j_1,j_2 - \frac{1}{2};j_1 + j_2 - \frac{1}{2}}_{m_1,m';m_1+m'}C^{j_1+ j_2 - \frac{1}{2},\bar{j} +\frac{1}{2};j_1+ j_2 - \bar{j} - 1}_{m_1+m',m_2 + \bar{m} - m';m_1+m_2+\bar{m}}e^{(j_1+ j_2 - \bar{j} - 1)}_{m_1+m_2+\bar{m}} \\
		= & (-1)^{2j_1 + 2j_2- i + 1}\lambda_{j_2 - \frac{1}{2},\bar{j} +\frac{1}{2},i}^{-1}\sqrt{2j_2(2j_1+1)(2\bar{j} + 2)(2j_2 + 2\bar{j} - 2i + 1)} \\
		& \times  C^{j_2+\bar{j} - i, j_1 ;j_1 + j_2 - \bar{j} - 1}_{m_2+\bar{m},m_1;m_1+m_2+\bar{m}} \begin{Bmatrix}
			j_2  - \frac{1}{2}&\bar{j} +\frac{1}{2}&j_2+\bar{j} - i \\j_1 + j_2 - \bar{j} - 1& j_1& j_1+j_2 - \frac{1}{2}
		\end{Bmatrix}e^{(j_1+ j_2 - \bar{j} - 1)}_{m_1+m_2+\bar{m}},
	\end{aligned}
\end{equation}
where $\begin{Bmatrix}
	j_1&j_2&j_3\\j_4& j_5& j_6
\end{Bmatrix} $ is the Wigner $6j$-Symbol. 

Combining the above results, we find that $m_3(e^{(j_1)}_{m_1}, e^{(j_2)}_{m_2},\bar{e}^{(\bar{j})}_{\bar{m}}\epsilon)$ is given by 
\begin{equation}\label{m3}
	\begin{aligned}
		m_3(e^{(j_1)}_{m_1}, e^{(j_2)}_{m_2},&\bar{e}^{(\bar{j})}_{\bar{m}}\epsilon)  = \sum_{i = 0}^{\min(2\bar{j},2j_2 - 1)}(-1)^{2j_1 + 2j_2- i + 1}\sqrt{\frac{(2j_1 + 1)2j_2(2j_2+1)(2j_2+2\bar{j} - 2i + 1)}{(2j_2 - i)(2\bar{j} - i +1 )}}\\
		& \times C^{j_2,\bar{j};j_2 + \bar{j} - i}_{m_2,\bar{m},m_2 + \bar{m}} C^{j_2+\bar{j} - i, j_1 ;j_1 + j_2 - \bar{j} - 1}_{m_2+\bar{m},m_1;m_1+m_2+\bar{m}} \begin{Bmatrix}
			j_2  - \frac{1}{2}&\bar{j} +\frac{1}{2}&j_2+\bar{j} - i \\j_1 + j_2 - \bar{j} - 1& j_1& j_1+j_2 - \frac{1}{2}
		\end{Bmatrix}e^{(j_1+ j_2 - \bar{j} - 1)}_{m_1+m_2+\bar{m}}
	\end{aligned}.
\end{equation}

In fact, the above result is sufficient to determine all values of $m_3$. We have
	\begin{align}
\label{m_3_com_1}
m_3(\bar{e}^{(\bar{j})}_{\bar{m}}\epsilon,e^{(j_1)}_{m_1}, e^{(j_2)}_{m_2}) & = - m_3(e^{(j_2)}_{m_2},e^{(j_1)}_{m_1},\bar{e}^{(\bar{j})}_{\bar{m}}\epsilon) , \\
		m_3(e^{(j_1)}_{m_1},\bar{e}^{(\bar{j})}_{\bar{m}}\epsilon, e^{(j_2)}_{m_2}) &= m_3(e^{(j_1)}_{m_1}, e^{(j_2)}_{m_2},\bar{e}^{(\bar{j})}_{\bar{m}}\epsilon) - m_3(e^{(j_2)}_{m_2},e^{(j_1)}_{m_1},\bar{e}^{(\bar{j})}_{\bar{m}}\epsilon).
	\end{align}
Moreover, using the cyclic structure we have
\begin{equation}\label{m3_cyc}
	\Tr_{S^3}\left(  \bar{e}^{(\bar{j}_2)}_{\bar{m}_2}\epsilon \times m_3(e^{(j_1)}_{m_1}, e^{(j_2)}_{m_2},\bar{e}^{(\bar{j}_1)}_{\bar{m}_1}\epsilon)\right) =  \Tr_{S^3}\left( e^{(j_1)}_{m_1}  \times m_3(e^{(j_2)}_{m_2},\bar{e}^{(\bar{j}_1)}_{\bar{m}_1}\epsilon,\bar{e}^{(\bar{j}_2)}_{\bar{m}_2}\epsilon )\right).
\end{equation}
This determines the value of $m_3$ on $H_{b}^{0,0}(S^3)\otimes H_{b}^{0,1}(S^3)\otimes  H_{b}^{0,1}(S^3)$.

For our later application, it will be more convenient to use the value of $m_3$ in a different basis. We make the following change of variable
\begin{equation}
	\begin{cases}
		&p = j_1 + m_1\\&q = j_1 - m_1
	\end{cases},\quad
	\begin{cases}
		&r = j_2 + m_2\\&s = j_2 - m_2
	\end{cases},\quad
	\begin{cases}
		&u_1 = \bar{j}_1  - \bar{m}_1\\&v_1 = \bar{j}_1 + \bar{m}_1
	\end{cases},\quad
\begin{cases}
	&u_2 = \bar{j}_2  - \bar{m}_2\\&v_2 = \bar{j}_2 + \bar{m}_2
\end{cases},
\end{equation}
with the constraint that
\begin{equation}\label{const_uv}
	u_1+u_2 = p+r - 1,\quad v_1 + v_2 = q +s - 1.
\end{equation}
This constraint is equivalent to $\bar{j}_1 + \bar{j}_2 = j_1 + j_2 - 1$, $\bar{m}_1 + \bar{m}_2 = -(m_1 + m_2)$.
Then we define the constant $(m_3)^{p,q;r,s}_{u_1,v_1;u_2,v_2}$ by the following
\begin{equation}\label{m3_2}
	(m_3)^{p,q;r,s}_{u_1,v_1;u_2,v_2} := \frac{(-1)^{\bar{j}_1 - \bar{m}_1}N(j_1,m_1)N(j_2,m_2)}{N(\bar{j}_1,\bar{m}_1)N(\bar{j}_2,\bar{m}_2)} \Tr_{S^3}\left(  \bar{e}^{(\bar{j}_2)}_{\bar{m}_2}\epsilon \times m_3(e^{(j_1)}_{m_1}, e^{(j_2)}_{m_2},\bar{e}^{(\bar{j}_1)}_{\bar{m}_1}\epsilon)\right) .
\end{equation}
where $N(j,m) = \sqrt{\frac{(j-m)!(j+m)!}{(2j+1)!}}$ is the square root of the $S^3$ norm of the monomial $w_1^{j+m}w_2^{j - m}$. This expression is non zero given the constrain \ref{const_uv}. 

The constant $(m_3)^{p,q;r,s}_{u_1,v_1;u_2,v_2}$  can be regarded as the value of $m_3$ in a unnormalized basis. We have
\begin{equation}
	m_3(w_1^{p}w_2^{q},w_1^{r}w_2^{s}, \frac{(u_1 + v_1 + 1)!}{u_1!v_1!}\bar{w}_1^{u_1}\bar{w}_2^{v_1}\epsilon) = (m_3)^{p,q;r,s}_{u_1,v_1;u_2,v_2}w_1^{u_2}w_2^{v_2}.
\end{equation}

Using the relation \ref{m_3_com_1} and the cyclic property, we find that the constant $(m_3)^{p,q;r,s}_{u_1,v_1;u_2,v_2}$ satisfy the following relation
\begin{equation}
	(m_3)^{p,q;r,s}_{u_1,v_1;u_2,v_2} = - (m_3)^{r,s;p,q}_{u_2,v_2;u_1,v_1}.
\end{equation}
Later we will build some interesting chiral OPE using this constant.

\subsection{$A_\infty$ structure on the CR cohomology: $m_n$}
\label{sec:all_mn}
In this section, we analyze all the higher products $m_n$ on the CR cohomology. First, by degree reason, the $n$-th product is only non zero on the following subspace of $H_{b}^{0,\sbullet}(S^3)^{\otimes n}$
\begin{equation}
	\begin{aligned}
		&\bigoplus_{\text{perm}}H_{b}^{0,0}(S^3)^{\otimes 2}\otimes H_{b}^{0,1}(S^3)^{\otimes n - 2},\\
		&\bigoplus_{\text{perm}}H_{b}^{0,0}(S^3)\otimes H_{b}^{0,1}(S^3)^{\otimes n - 1}.
	\end{aligned}
\end{equation}
The fact that $H_{b}^{0,\sbullet}(S^3)$ is concentrated in degree $0,1$ and the homotopy operator $h$ decrees the degree by $1$ strongly restrict possible trees that contribute the higher product $m_n$. A tree that contains the following vertex must be zero
\begin{center}
	\begin{tikzpicture}[grow' = up]
	\tikzstyle{level 1}=[sibling distance=8mm,level distance = 2mm]
	\tikzstyle{level 2}=[sibling distance=8mm,level distance = 7mm]
	\tikzstyle{level 3}=[sibling distance=8mm,level distance = 7mm]
	\coordinate
	node {}child {edge from parent[draw=none]
		child { 
			child {	edge from parent node[left] {$h$} }    child { edge from parent node[right] {$h$}} edge from parent node[right] {$h \quad = 0$}
		}
	};
\end{tikzpicture}
\end{center}
As a result, for any tree that gives a non-zero map, all vertices must be directly connected to a leaf or the root. Moreover, the product map is zero on $\Omega_{b}^{0,1}(S^3)\otimes \Omega_{b}^{0,1}(S^3)$. Therefore, for $n \geq 3$, a tree that gives a non-zero map must only consist of the following vertices:
\begin{center}
	\begin{tikzpicture}[grow' = up]
		\tikzstyle{level 1}=[sibling distance=8mm,level distance = 2mm]
		\tikzstyle{level 2}=[sibling distance=8mm,level distance = 7mm]
		\tikzstyle{level 3}=[sibling distance=8mm,level distance = 8mm]
		\tikzstyle{level 4}=[sibling distance=8mm,level distance = 4mm]
		\coordinate
		node {}child {edge from parent[draw=none]
			child { 
				child { node{$i$} 	child {node {$H_{b}^{0,0}\;$} edge from parent[draw=none] }}    child { node{$i$} 	child {node {$\; H_{b}^{0,1}$} edge from parent[draw=none] }} edge from parent node[right] {$h$}
			}
		};
	\end{tikzpicture}\quad
	\begin{tikzpicture}[grow' = up]
	\tikzstyle{level 1}=[sibling distance=8mm,level distance = 2mm]
	\tikzstyle{level 2}=[sibling distance=8mm,level distance = 7mm]
		\tikzstyle{level 3}=[sibling distance=8mm,level distance = 8mm]
	\tikzstyle{level 4}=[sibling distance=8mm,level distance = 4mm]
	\coordinate
	node {}child {edge from parent[draw=none]
		child { 
			child {	edge from parent node[left] {$h$} }    child { node{$i$} 	child {node {$\; H_{b}^{0,1}$} edge from parent[draw=none] }} edge from parent node[right] {$h$}
		}
	};
\end{tikzpicture}\quad 
	\begin{tikzpicture}[grow' = up]
	\tikzstyle{level 1}=[sibling distance=8mm,level distance = 2mm]
	\tikzstyle{level 2}=[sibling distance=8mm,level distance = 7mm]
	\tikzstyle{level 3}=[sibling distance=8mm,level distance = 7mm]
	\coordinate
	node {$p$}child {edge from parent[draw=none]
		child { 
			child {	edge from parent node[left] {$h$} }    child { edge from parent node[right] {$h$}}         
		}
	};
\end{tikzpicture}
\end{center}
Only a few trees survive under this condition.

First we consider the map $m_n$ on $H_{b}^{0,0}(S^3)\otimes H_{b}^{0,1}(S^3)^{\otimes n - 1}$. In this case, only one tree contribute, which gives us
\begin{equation}
	m_n(a_0,\bar{a}_1\epsilon,\dots ,\bar{a}_{n - 1}\epsilon) = pM( hM( \dots hM(hM(a_0,\bar{a}_1\epsilon),\bar{a}_2\epsilon),\dots ,\bar{a}_{n-1}\epsilon)),
\end{equation}
for $a_0 \in H_{b}^{0,0}(S^3)$ and $\bar{a}_1\epsilon,\dots ,\bar{a}_{n-1}\epsilon \in H_{b}^{0,1}(S^3)$. 

We emphasize that the order of input elements does matter in a higher operation $m_n$. Therefore, the above formula does not directly apply to other cases  when we insert $a_0\in H_{b}^{0,0}(S^3)$ in the middle. However, other cases can be computed using $m_n(a_0,\bar{a}_1\epsilon,\dots ,\bar{a}_{n - 1}\epsilon)$ and a combination of permutations of $\bar{a}_1\epsilon,\dots ,\bar{a}_{n - 1}\epsilon$. We have the following
\begin{equation}
	m_n(\bar{a}_{n-1}\epsilon ,\dots ,\bar{a}_{k+1}\epsilon,a_0,\bar{a}_1\epsilon,\dots, \bar{a}_{k}\epsilon) = \sum_{\sigma \in Sh(k,n - 1 - k)}(\pm)	m_n(a_0,\bar{a}_{\sigma^{-1}(1)}\epsilon,\dots ,\bar{a}_{\sigma^{-1}(n-1)}\epsilon) ,
\end{equation}\\
where $Sh(k,n - 1 -k)$ the subset of $(k,n-1-k)$-shuffles in $S_{n-1}$. We see that it suffice to only compute $m_n(a_0,\bar{a}_1\epsilon,\dots ,\bar{a}_{n - 1}\epsilon)$ to determine the value of $m_n$ on $\bigoplus_{\text{perm}}H_{b}^{0,0}(S^3)\otimes H_{b}^{0,1}(S^3)^{\otimes n - 1}$.

For $n \geq 3$, let  us denote 
\begin{equation}
	\mu_n(a_0;\bar{a}_1\epsilon,\dots ,\bar{a}_{n - 1}\epsilon) = M( hM( \dots hM(hM(a_0,\bar{a}_1\epsilon),\bar{a}_2\epsilon),\dots ,\bar{a}_{n-1}\epsilon)).
\end{equation}
Then $	m_n(a_0,\bar{a}_1\epsilon,\dots ,\bar{a}_{n - 1}\epsilon) = p\mu_n(a_0,\bar{a}_1\epsilon,\dots ,\bar{a}_{n - 1}\epsilon) $. The purpose of defining $\mu_n$ is that it can be computed iteratively as follows
\begin{equation}
	\label{mun_rec}
	\mu_n(a_0;\bar{a}_1\epsilon,\dots ,\bar{a}_{n - 1}\epsilon)  = M(h\mu_{n-1}(a_0;\bar{a}_1\epsilon,\dots ,\bar{a}_{n - 2}\epsilon),\bar{a}_{n-1}\epsilon).
\end{equation}
Computation of $\mu_n$ is similar as the computation of $m_3$ in the last section. First we compute $\mu_3(e^{(j_0)}_{m_0},\bar{e}^{(\bar{j}_1)}_{m_1}\epsilon,\bar{e}^{(\bar{j}_2)}_{m_2}\epsilon)$. Using \ref{Prod_har_2}, we have
\begin{equation}
\begin{aligned}
		&\mu_3(e^{(j_0)}_{m_0},\bar{e}^{(\bar{j}_1)}_{m_1}\epsilon,\bar{e}^{(\bar{j}_2)}_{m_2}\epsilon)\\
		&=\sum_{i_1,i_2} h_{j_0 - \frac{i_1}{2},\bar{j}_1 - \frac{i_1}{2}}\lambda_{j_0,\bar{j}_1,i_1}\lambda_{j_0 - \frac{1}{2},\bar{j}_1 + \frac{1}{2};i_1}^{-1} \lambda_{j_0 - \frac{1}{2},\bar{j}_1+ \bar{j}_2 + \frac{1}{2};i_2}\sqrt{(2\bar{j}_1 + 2)(2\bar{j}_2 + 1)(2j_0 + 2\bar{j}_1 - 2i_1 + 1)}\\
		&\times\begin{Bmatrix}
			\bar{j}_1 + \frac{1}{2}&j_0 - \frac{1}{2}&j_0+\bar{j_1} - i_1\\j_0 + \bar{j}_1 + \bar{j}_2 - i_2&\bar{j}_2&\bar{j}_1 + \bar{j}_2 + \frac{1}{2}
		\end{Bmatrix} C^{j_0,\bar{j}_1;j_0+\bar{j}_1 - i_1}_{m_0,m_1;m_0+m_1}C^{j_0+\bar{j}_1 - i_1,\bar{j}_2;j_0+\bar{j}_1 + \bar{j}_2 - i_0}_{m_0+m_1,m_2;m_0+m_1+m_2}e^{(j_0 - \frac{1}{2} - \frac{i_2}{2},\bar{j}_1+\bar{j}_2 + \frac{1}{2} - \frac{i_2}{2})}_{m_0+m_1+m_2}.
\end{aligned}
\end{equation}
In general, $\mu_n(a_0,\bar{a}_1\epsilon,\dots ,\bar{a}_{n - 1}\epsilon)$ takes the following form
\begin{equation}
	\mu_n(e^{(j_0)}_{m_0},\bar{e}^{(\bar{j}_1)}_{m_1}\epsilon,\dots \bar{e}^{(\bar{j}_{n-1})}_{m_{n-1}}\epsilon) = \sum_{i_1,\dots,i_{n-1}} (\mu_n)^{j_0,\bar{j}_1\dots,\bar{j}_{n-1}}_{i_1,\dots,i_{n-1}}\left(  \prod_{l = 0}^{n - 2}C^{J_l - i_l,\bar{j}_{l+1},J_{l+1} - i_{l+1}}_{M_l,m_{l+1},M_{l+1}}\right)  e^{(j_0 - \frac{n-2}{2} - \frac{i_{n-1}}{2},\bar{J}_{n - 1} + \frac{n-2}{2} - \frac{i_{n - 1}}{2})}_{M_{n-1}}\epsilon,
\end{equation}
where $i_0 = 0$ and we define
\begin{equation}
\begin{aligned}
	&J_l = j_0 + \bar{j}_1 + \dots + \bar{j}_l, \quad l\geq 0,\\
	&\bar{J}_l = J_l- j_0, \quad l \geq 1,\\
	& M_l = m_0 + \dots +m_l,\quad l \geq 0.
\end{aligned}
\end{equation}
Using the recursion relation \ref{mun_rec} and the formula \ref{Prod_har_2}, we find that
\begin{equation}
\begin{aligned}
		(\mu_n)^{j_0,\bar{j}_1\dots,\bar{j}_{n-1}}_{i_1,\dots,i_{n-1}} & = \prod_{l = 2}^{n - 1} h_{j_0 - \frac{l-2}{2} - \frac{i_{l-1}}{2},\bar{J}_{l - 1} + \frac{l-2}{2} - \frac{i_{l - 1}}{2}}\sqrt{(2\bar{J}_{l-1} + l)(2\bar{j}_l + 1)(2J_{l - 1} - 2i_{l-1} + 1)}\\
		&\times \left( \prod_{l = 1}^{n - 2}\lambda_{j_0 - \frac{l-1}{2},\bar{J}_{l}+ \frac{l-1}{2};i_l}\lambda_{j_0 - \frac{l}{2},\bar{J}_{l}+ \frac{l}{2};i_l}^{-1}\right) \lambda_{j_0 - \frac{n-2}{2},\bar{J}_{n-1}+ \frac{n-2}{2};i_{n-1}}\\
		&\times \prod_{l=2}^{n-1} \begin{Bmatrix}
			\bar{J}_{l-1} + \frac{l_1}{2}&j_0 - \frac{l-1}{2}& J_{l-1} - i_{l-1}\\ J_l - i_l& \bar{j}_l& \bar{J}_{l} + \frac{l-1}{2}
		\end{Bmatrix}.
\end{aligned}
\end{equation}
The above expression simplifies to 
\begin{equation}
	\begin{aligned}
		&(\mu_n)^{j_0,\bar{j}_1\dots,\bar{j}_{n-1}}_{i_1,\dots,i_{n-1}} =  \\
	 &\lambda_{j_0 - \frac{n-2}{2},\bar{J}_{n-1}+ \frac{n-2}{2};i_{n-1}}\prod_{l = 2}^{n - 1} \sqrt{\frac{(2j_0 - l + 3)(2\bar{j}_l + 1)(2J_{l - 1} - 2i_{l-1} + 1)}{(2j_0 - l + 2 - i_{l-1})(2\bar{J}_{l-1}+ l - 1 - i_{l-1})}}\begin{Bmatrix}
			\bar{J}_{l-1} + \frac{l-1}{2}&j_0 - \frac{l-1}{2}& J_{l-1} - i_{l-1}\\ J_l - i_l& \bar{j}_l& \bar{J}_{l} + \frac{l-1}{2}
		\end{Bmatrix}.
	\end{aligned}
\end{equation}
As a result, we have the following formula for $m_n(e^{(j_0)}_{m_0},\bar{e}^{(\bar{j}_1)}_{m_1}\epsilon,\dots \bar{e}^{(\bar{j}_{n-1})}_{m_{n-1}}\epsilon)$
\begin{equation}\label{all_mn}
	m_n(e^{(j_0)}_{m_0},\bar{e}^{(\bar{j}_1)}_{m_1}\epsilon,\dots \bar{e}^{(\bar{j}_{n-1})}_{m_{n-1}}\epsilon)= \left.\sum_{i_1,\dots,i_{n-2}} (\mu_n)^{j_0,\bar{j}_1\dots,\bar{j}_{n-1}}_{i_1,\dots,i_{n-1}}\left(  \prod_{l = 0}^{n - 2}C^{J_l - i_l,\bar{j}_{l+1},J_{l+1} - i_{l+1}}_{M_l,m_{l+1},M_{l+1}}\right)\right|_{\substack{i_0 =0;\\i_{n - 1} = 2j_0 - n + 2}}  \bar{e}^{(\bar{J}_{n - 1} - j_0 + n-2 )}_{M_{n-1}}\epsilon.
\end{equation}
Here, the range of summation is taken to be
\begin{equation}
	0\leq i_l \leq \min\{2j_0 - l,2\bar{J}_l + l - 1\},\; \text{ for }l = 1,\dots n - 2.
\end{equation}
However, the actual range of summation is much smaller due to the constraint of the Wigner $6j$ symbol and the Clebsch–Gordan coefficients. For example, the summand is nonzero only when
\begin{equation}
	i_l \geq i_{l-1} ,\; \text{ for }l = 2,\dots n - 1.
\end{equation}

Using the cyclic structure, the above result is sufficient to determine the whole $m_n$. For example, we have the following
\begin{equation}
	\Tr_{S^3}\left( e^{(j_1)}_{m_1}\times m_n(e^{(j_2)}_{m_2},\bar{e}^{(\bar{j}_1)}_{m_1}\epsilon,\dots \bar{e}^{(\bar{j}_{n-1})}_{m_{n-1}}\epsilon)\right) = \Tr_{S^3}\left(\bar{e}^{(\bar{j}_{n-1})}_{m_{n-1}}\epsilon \times m_n( e^{(j_1)}_{m_1},e^{(j_2)}_{m_2},\bar{e}^{(\bar{j}_1)}_{m_1}\epsilon,\dots \bar{e}^{(\bar{j}_{n-2})}_{m_{n-2}}\epsilon)\right) 
\end{equation}
This determines the value of $m_n$ on $H_{b}^{0,0}(S^3)^{\otimes 2}\otimes H_{b}^{0,1}(S^3)^{\otimes n - 2}$. Other values of $m_n$ can be determined similarly.

\subsection{$L_\infty$ structure on the tensor product}
\label{sec:flat_full}
In this section, we study the $L_\infty$ structure on the tensor product between the $C_\infty$ algebra $H^{0,\sbullet}_b(S^3)$ and a Lie algebra $\mathfrak{g}$. It was shown in \cite{2017arXiv170202194R} that for an arbitrary $C_\infty$ algebra and a Lie algebra, there exists a canonical $L_\infty$ structure on the tensor product. However, an explicit formula is not provided in \textit{loc. cit.} due to the fact that the operad $\mathcal{L}ie$ does not have a canonical basis. 

In our cases, an explicit formula can be easily obtained by thinking about the homotopy transfer of the Lie algebra $\Omega_{b}^{0,\sbullet}(S^3)\otimes \mathfrak{g}$. The homotopy transfer from a Lie algebra to a $L_\infty$ algebra is similar to the homotopy transfer of associative algebra, except that we consider all (not necessarily planar) binary rooted trees and we replace the product map with the Lie bracket on each vertex. In our cases, the condition on non zero trees discussed in the last section still hold. Therefore, up to a permutation of the leaves, the trees that contribute to the $L_\infty$ operations are the same as in the last section. As a result, we have the following formula for the higher bracket $l_n$ restricted on $(H_{b}^{0,0}(S^3)\otimes \mathfrak{g}) \bigwedge^{n-1}(H_{b}^{0,1}(S^3)\otimes \mathfrak{g})$
\begin{equation}
\begin{aligned}
		&l_n(a_0\otimes x_0,\bar{a}_1\epsilon\otimes x_1,..., \bar{a}_{n-1}\epsilon\otimes x_{n-1}) \\
		= &\sum_{\sigma \in S_{n-1}}	m_n(a_0,\bar{a}_{\sigma(1)}\epsilon,... ,\bar{a}_{\sigma(n-1)}\epsilon)[...[[x_0,x_{\sigma(1)}],x_{\sigma(2)}],...,x_{\sigma(n-1)}],
\end{aligned}
\end{equation}
where $a_0\in (H_{b}^{0,0}(S^3),\bar{a}_i\epsilon \in H_{b}^{0,1}(S^3)$ and $x_i \in \mathfrak{g}$. We have that $|S_{n - 1}| = (n - 1)! = \dim\mathcal{L}ie(n)$. In fact, the elements $[...[[x_0,x_{\sigma(1)}],x_{\sigma(2)}],...,x_{\sigma(n-1)}],]$, $\sigma \in S_{n-1}$ form a basis of the Lie operad $\mathcal{L}ie(n)$ and are called Dynkin elements. 

Using the above $L_\infty$ structure, one can write down the full effective interaction of the KK theory. Recall that we denote $\{t_a\}$ a basis of the Lie algebra. We let $K: \mathfrak{g}^{\otimes 2} \to \C$ the Killing form and denote $K_{ab} = K(t_a,t_b)$. We define $t^a = \sum_{b}K^{ab}t_b$ where $K^{ab}$ is the inverse of $K_{ab}$. The KK fields are expanded as  $\mathbf{A} = \sum_{p,q \geq 0} \sum_a \mathbf{A}^a[p,q]t_aw_1^pw_2^q$ and $\mathbf{B} = \sum_{p,q \geq 0} \sum_a \frac{(p+q+1)!}{p!q!}\mathbf{B}_a[p,q]t^a\bar{w}_1^p\bar{w}_2^q$
Then we have
\begin{equation}\label{act_full}
\begin{aligned}
		&\frac{1}{2(n-1)!}\sum_{n \geq 2}\sum_{p,q,r,s \geq 0}\sum_{\substack{u_1 + \dots u_{n-1} = p+r -  n +2\\v_1+\dots v_{n-1} = q+s - n+2}}\sum_{\substack{a,b\\c_1,\dots,c_{n-1}}}\sum_{\sigma \in S_{n-1}}\\
		&\left( \prod_{i = 1}^{n-1}\frac{(u_i+v_i+1)!}{u_i!v_i!}\right) \Tr_{S^3}(w_1^pw_2^q\times m_n(w_1^rw_2^s,\bar{w}_1^{u_{\sigma(1)}}\bar{w}_2^{v_{\sigma(1)}}\epsilon,\dots ,\bar{w}_1^{u_{\sigma(n-1)}}\bar{w}_2^{v_{\sigma(n-1)}}\epsilon ))\\
		&\times K(t_a,[...[[t_b,t^{c_{\sigma(1)}}],t^{c_{\sigma(2)}}],...,t^{c_{\sigma(n-1)}}]) \int \mathbf{A}^a[p,q]\mathbf{A}^b[r,s]\mathbf{B}_{c_1}[u_1,v_1] \cdots\mathbf{ B}_{c_{n-1}}[u_{n-1},v_{n-1}],
\end{aligned}
\end{equation}
where $m_n$ can be computed by formula \ref{all_mn}. This gives us a complete answer to the classical action functional of the KK theory.

	\section{Holomorphic BF theory}
		\label{sec:HBF}
	\subsection{A brief review}
	In this section, we introduce the holomorphic BF theory. Fields of holomorphic BF theory are
	\begin{equation}
		\begin{aligned}
			&\EuScript{A} \in \Omega^{0,1}(X,\mathfrak{g})[1],\\
			&\EuScript{B} \in \Omega^{3,1}(X,\mathfrak{g})[1].
		\end{aligned}
	\end{equation}
	The action functional is given by
	\begin{equation}
		S = \int_X \EuScript{B} (\bar{\pa}\EuScript{A} + \frac{1}{2}[\EuScript{A},\EuScript{A}]).
	\end{equation}
	By varying the action function we obtain the following equations of motion.
	\begin{equation}
		\begin{aligned}
			&\bar{\pa}\EuScript{A} + \frac{1}{2}[\EuScript{A},\EuScript{A}] = 0,\\
			&\bar{\pa}\EuScript{B} + [\EuScript{A},\EuScript{B}] = 0.
		\end{aligned}
	\end{equation}
	The gauge transformations are generated by $\chi \in \Omega^{0,0}(X,\mathfrak{g}),\nu\in \Omega^{3,0}(X,\mathfrak{g})$ with
	\begin{equation}
		\begin{aligned}
			\delta \EuScript{A}= \bar{\pa} \chi + [\EuScript{A},\chi],\\
			\delta \EuScript{B} = \bar{\pa} \nu + [\EuScript{B},\nu].
		\end{aligned}
	\end{equation}
	
	For the purpose of performing KK reduction, it is more convenient to reformulate the theory into BV formalism. By adding ghosts, anti-fields and anti-ghost, the BV field content can be described by 
	\begin{equation}
		\begin{aligned}
			&\pmb{\EuScript{A}} \in \Omega^{0,\sbullet}(X,\mathfrak{g})[1],\\
			&\pmb{\EuScript{B} }\in \Omega^{3,\sbullet}(X,\mathfrak{g})[1].
		\end{aligned}
	\end{equation}
	The BV action functional is given by
	\begin{equation}
		HBF[\pmb{\EuScript{A}},\pmb{\EuScript{B} }] = \int_X \pmb{\EuScript{B} } (\bar{\pa}\pmb{\EuScript{A}} + \frac{1}{2}[\pmb{\EuScript{A}},\pmb{\EuScript{A}}]).
	\end{equation}
We see that the BV theory is encoded in the following dg Lie algebra
	\begin{equation}
	(\Omega^{0,\sbullet}(X,\mathfrak{g})\oplus \Omega^{3,\sbullet}(X,\mathfrak{g}) ,\pab,[-,-]).
\end{equation}

	\subsection{Cohomological KK reduction}
	\label{sec:KK_HBF}
	In this section, we study KK reduction of holomorphic BF theory on $\C^3/\C$.
	
	Since fields of holomorphic BF theory are also Dolbeault forms, we can use the same results as we analyze the holomorphic Chern-Simons theory to perform the KK reduction. We work in the BV formalism, decompose the field content into $S^3$ harmonics and then pass to the CR cohomology. The field $\pmb{\EuScript{A}}$ leads to the following KK tower of fields\footnote{To avoid introducing too many symbols, we used the same letter $ \mathbf{A}$ and $ \mathbf{B}$ here as in the holomorphic Chern-Simons theory. The correct meaning of the symbol should be clear from the context.}
	\begin{equation}
		\begin{aligned}
			&\mathbf{A} \in	\Omega_{3d}^{\sbullet}(\C\times\R_{>0})\otimes H_{b}^{0,0}(S^3)\otimes \mathfrak{g}[1],\\
			& \tilde{\mathbf{B}} \in	\Omega_{3d}^{\sbullet}(\C\times\R_{>0})\otimes H_{b}^{0,1}(S^3)\otimes \mathfrak{g} .	
		\end{aligned}
		\end{equation}
	We can choose the standard holomorphic volume form $\Omega_{\C^3} = dzdz_1dz_2$.Then the field $\pmb{\EuScript{B}}$ leads to the following KK tower of fields
		\begin{equation}
		\begin{aligned}
			& \mathbf{B} \in \Omega_{3d}^{\sbullet}(\C\times\R_{>0})\otimes H_{b}^{0,1}(S^3)\otimes \mathfrak{g}^*,\\
			& \tilde{\mathbf{A}} \in \Omega_{3d}^{\sbullet}(\C\times\R_{>0})\otimes H_{b}^{0,0}(S^3)\otimes \mathfrak{g}^*[1].
		\end{aligned}
	\end{equation}
\begin{remark}
	More precisely, the KK fields corresponding to $\pmb{\EuScript{B}}$ live in $\Omega_{3d}^{\sbullet,(1)}(\C\times\R_{>0})\otimes H_{b}^{2,1}(S^3)\otimes \mathfrak{g}^*$$\oplus \Omega_{3d}^{\sbullet,(1)}(\C\times\R_{>0})\otimes H_{b}^{2,0}(S^3)\otimes \mathfrak{g}^*[1]$. By choosing the volume form, we can forget the difference in flat space.
\end{remark}
	We can replace $\epsilon$ by $dz$ and identify $\C[\bar{w}_1,\bar{w}_2]$ with the dual of $\C[w_1,w_2]$ as in Section \ref{sec:KK_HT}. According to our discussion in Section \ref{sec:KK_HT}, this field content is exactly the holomorphic twist of $3d$ $\mathcal{N} = 2$ theory with a vector multiplet valued in $\mathfrak{g}[w_1,w_2]$ and an adjoint chiral multiplet.

	Having the field content, the next step is to analyze the action functional. The information about the action functional is encoded in the $L_\infty$ structure on the BV fields and the integration map. With the choice of holomorphic volume form, the integration map is the same as the integration map we studied in holomorphic Chern-Simons theory \ref{prop_int}. The dg Lie algebra of the holomorphic BF theory can be identified with the shifted cotangent structure of $(\Omega^{0,\sbullet}(X,\mathfrak{g}),\pab ,[-,-])$, which encode holomorphic Chern-Simons theory. As a result, we don't need to go through the process of homotopy transfer again. After we pass to the CR cohomology, the $L_\infty$ structure is given by the shifted cotangent structure of the $L_\infty$ algebra $(\Omega_{3d}^{\sbullet}(\C\times\R_{>0})\otimes H_{b}^{0,\sbullet}(S^3)\otimes \mathfrak{g},\{l_k\})$. We identify the dual of the Lie algebra with itself via the Killing pairing. In summary, the $L_\infty$ structure corresponding to the KK reduction of holomorphic BF theory consists of the following data
\begin{enumerate}
	\item the $L_\infty$ algebra $\mathfrak{g}_{CR} = (H_{b}^{0,\sbullet}(S^3)\otimes \mathfrak{g},\{l_k\})$.
	\item A copy $\mathfrak{g}_{CR}'$ of $H_{b}^{0,\sbullet}(S^3)\otimes \mathfrak{g}$ considered as an abelian Lie algebra.
	\item  An $\mathfrak{g}_{CR}$ module structure on $\mathfrak{g}_{CR}'$. This $L_\infty$ module structure is given by maps
	\begin{equation}
		\begin{aligned}
			\mathfrak{g}_{CR}^{\otimes n - 1} \otimes \mathfrak{g}_{CR}' &\to \mathfrak{g}_{CR},\\
			(t_1,\dots, t_{n-1},t'_n) &\mapsto l_n	(t_1,\dots, t_{n-1},t'_n) ,
		\end{aligned}
	\end{equation}
where $t_i\in \mathfrak{g}_{CR} $ and $t_n' \in \mathfrak{g}_{CR}'$ .
\end{enumerate}
	
	As a consequence, we can build up the equation of motion and action functional for the KK theory of holomorphic BF theory by using the same set of data $\{l_k\}$ as in the last section. Using the differential and the $2$-bracket, we get the following action functional
		\begin{equation}\label{act_KKBF}
		\int_{3d} \Tr \left(\mathbf{B}(\hat{d}\mathbf{A} + \frac{1}{2}[\mathbf{A},\mathbf{A}])+ \tilde{\mathbf{A}}(\hat{d}\tilde{\mathbf{B}} + [\mathbf{A},\tilde{\mathbf{B}}])\right),
	\end{equation}
which coincide with the action functional of HT twist of $\mathcal{N} = 2$ vector multiplet with an adjoint chiral multiplet \ref{3d_vec_chi_act}.
Higher brackets lead to the following interactions
\begin{equation}\label{def_BF}
	\sum_{n = 3}^{\infty}\int \frac{1}{2(n-2)!} \Tr(\mathbf{B} \wedge l_n(\mathbf{A},\mathbf{A},\underbrace{\tilde{\mathbf{B}},\dots,\tilde{\mathbf{B}}}_{n - 2}))+ \frac{1}{(n-1)!} \Tr(\tilde{\mathbf{A}}\wedge l_n(\mathbf{A},\underbrace{\tilde{\mathbf{B}},\dots,\tilde{\mathbf{B}}}_{n-1})).
\end{equation}
In summary, we have the following
		\begin{prop}
		The KK reduction of the holomorphic BF theory on $\C^3\backslash\C$ to $\C\times \R_{ > 0}$ is the same as the holomorphic twist of $3d$ $\mathcal{N} = 2$ theory, associated to a $\mathcal{N} = 2$ vector multiplet valued in $\mathfrak{g}[w_1,w_2]$ and an adjoint chiral multiplet. The action functional consist of the standard holomorphic twist of $3d$ $\mathcal{N} = 2$ action, together with the deformation \ref{def_BF}.
	\end{prop}

The full action functional can be written down explicitly using the higher product $m_n$ we computed  in Section \ref{sec:all_mn}, and is similar to the expansion of holomorphic Chern-Simons \ref{act_full}.
	
	\subsection{Residue supersymmetry and its breaking}
	
	According to the discussion of the last section, the zero modes of KK theory of holomorphic BF theory gives rise to the holomorphic twist of $3d$ $\mathcal{N} = 2$ theory with a vector multiplet and an adjoint chiral multiplet. In fact, this field content consists of a $3d$ $\mathcal{N} =4$ vector multiplet. A $3d$ $\mathcal{N} =4$ theory has more supersymmetry than $\mathcal{N} = 2$ theory, and has residue supersymmetry that survives after the 
	holomorphic twist. This residue supersymmetry and be used to further twist the theory, and will lead to the $A$ and $B$ topological twists of $3d$ $\mathcal{N} =4$ theory. Details of the relation between holomorphic twist and the $A$ and $B$ twist are studied in \cite{Garner:2022vds}. Here we briefly review some of the results.
	
	The theory formulated using the CR cohomology has an infinite series of interactions, which is cumbersome for the analysis of symmetry. Therefore, we use the equivalent formulation of the KK theory built on the CR cochain. The BV fields consist of 
		\begin{equation}
		\begin{aligned}
			&\bm{\mathcal{A}} \in	\Omega_{3d}^{\sbullet}(\C\times\R_{>0})\otimes\Omega_{b}^{0,0}(S^3)\otimes \mathfrak{g}[1],\\
			& \bm{\Phi} \in	\Omega_{3d}^{\sbullet}(\C\times\R_{>0})\otimes \Omega_{b}^{0,1}(S^3)\otimes \mathfrak{g},\\	
			& \bm{\mathcal{B}}\in \Omega_{3d}^{\sbullet}(\C\times\R_{>0})\otimes \Omega_{b}^{0,1}(S^3)\otimes \mathfrak{g}^*,\\
			& \bm{\Psi} \in \Omega_{3d}^{\sbullet}(\C\times\R_{>0})\otimes \Omega_{b}^{0,0}(S^3)\otimes \mathfrak{g}^*[1].
		\end{aligned}
	\end{equation}
The action functional of this theory contains the standard action functional \ref{3d_vec_chi_act} for a $3d$ $\mathcal{N} =4$ vector multiplet, together with the following deformation
	\begin{equation}\label{deff_BF}
	\int(\bm{\Psi}\pacr \bm{\mathcal{A}}).
\end{equation}

For the $3d$ $\mathcal{N} =4$ vector multiplet, with the standard holomorphic twisted action $S_{HT}$, \cite{Garner:2022vds} identified the deformation corresponding to the further $A$ and $B$ topological twist as follows
	\begin{equation}
		\delta S_A =  \int \bm{\mathcal{B}}\bm{\Phi},\quad\quad \delta S_B = \int \bm{\Psi} \pa_z \bm{\mathcal{A}} .
	\end{equation}
These two deformations are compatible with the standard holomorphic twisted action $S_{HT}$
\begin{equation}
	\{	\delta S_A,S_{HT}\} = 	\{	\delta S_B,S_{HT}\} = 0.
\end{equation}
They also generate the residue supersymmetry. Let $Q_A = 	\{	\delta S_A,-\} $ and $Q_B = 	\{	\delta S_B,-\} $, we have
\begin{equation}
	\{Q_A,Q_B\} = \pa_z.
\end{equation}
For theory with the standard HT twist of $3d$ $\mathcal{N} = 4$ action, we can add either $Q_A$ or $Q_B$ to the BRST differential, or equivalently add $\delta S_A$ or $\delta S_B$ to the BV action functional. Then $\pa_z$ will become BRST exact, and the new theory will become topological. They correspond to the $A$ and $B$ topological twist of $3d$ $\mathcal{N} = 4$ theory respectively.

	However, adding the deformed term \ref{deff_BF} will break the residue supersymmetry. We have
\begin{equation}
\{ \int(\bm{\Psi}\pacr \bm{\mathcal{A}}),\delta S_A \} = \int\bm{\Psi}\pacr \bm{\Phi} + \int(\bm{\mathcal{B}}\pacr \bm{\mathcal{A}}) \neq 0 .
\end{equation}
	As a consequence, the residue supercharge $Q_A$ does not survive the interaction of the KK modes.
	
	Although the residue supersymmetry is broken, the other supercharge $Q_B$, is still compatible with the full KK theory. 
	\begin{equation}
	\{ \int(\bm{\Psi}\pacr \bm{\mathcal{A}}),\delta S_B \} = 0.
	\end{equation}
	This means that we can still perform a further "$B$ twist" by adding $Q_B$ to the BRST differential. This has the effect of turning the differential $\hat{d}$ into $\hat{d} + \epsilon\pa_z$. However, this "$B$ twist" will not lead to a topological theory as the $\pa_z$ is no longer BRST exact \footnote{This deformation is topological only for the zero mode}.  It is interesting that this further "$B$ twist" can also be accomplished by putting the $6d$ theory on the deformed geometry $SL_2(\C)$. We will see this in Section \ref{sec:def}.

	\section{Poisson BF theory}
	\label{sec:PBF}
\subsection{A brief review}
Holomorphic Poisson BF theory, or Poisson BF theory for short, is the same as an abelian BF theory together with a Poisson bracket. The field content of Poisson BF theory consist
\begin{equation}
\begin{aligned}
		&\EuScript{H} \in \Omega^{0,1}(\C^3),\\
	&\EuScript{W} \in \Omega^{3,1}(\C^3).
\end{aligned}
\end{equation}
We introduce the following holomorphic Poisson tensor
\begin{equation}\label{Poisson}
	\pi = \frac{\pa}{\pa z_1} \wedge \frac{\pa}{\pa z_2} .
\end{equation}
It induces the Poisson bracket
\begin{equation}
	\{f,g\}_{\pi} = \epsilon_{ij} \frac{\pa f}{\pa z_i} \frac{\pa g}{\pa z_j}.
\end{equation}
Using the Poisson bracket, the action functional can be written as
\begin{equation}
	\int \EuScript{W}(\pab\EuScript{H} + \frac{1}{2}\{\EuScript{H},\EuScript{H}\}_{\pi}).
\end{equation}
Like holomorphic BF theory, Poisson BF theory has gauge transformations generated by $\chi \in \Omega^{0,0}(X,\mathfrak{g}),\nu\in \Omega^{3,0}(X,\mathfrak{g})$. It takes the following form
	\begin{equation}
	\begin{aligned}
		\delta \EuScript{H}= \bar{\pa} \chi + \{\EuScript{H},\chi\}_{\pi},\\
		\delta \EuScript{W} = \bar{\pa} \nu +\{\EuScript{W},\nu\}_{\pi}.
	\end{aligned}
\end{equation}

\begin{remark}
	For Poisson BF theory on twistor space $\mathbb{PT}$, the fields should be $\EuScript{H} \in \Omega^{0,1}(\mathbb{PT},\mathcal{O}(2))$, $\EuScript{W} \in \Omega^{3,1}(\mathbb{PT},\mathcal{O}(-2))$. However, we will be working with the theory on a open patch $\C^3 \subset \mathbb{PT}$, so we can forget the difference o the line bundle.
\end{remark}

It will be convenient to reformulate the theory into BV formalism. The BV field content consists of
\begin{equation}
	\begin{aligned}
		&\pmb{\EuScript{H} }\in \Omega^{0,\sbullet}(\C^3),\\
		&\pmb{\EuScript{W}} \in \Omega^{3,\sbullet}(\C^3).
	\end{aligned}
\end{equation}
The BV action functional is given by
\begin{equation}
		\int \pmb{\EuScript{W}}(\pab\pmb{\EuScript{H}} + \frac{1}{2}\{\pmb{\EuScript{H}},\pmb{\EuScript{H}}\}_{\pi}).
\end{equation}
Note that the BV structure of the theory is encoded in the following dg Lie algebra structure 
\begin{equation}
	(\Omega^{0,\sbullet}(\C^3)\oplus\Omega^{3,\sbullet}(\C^3),\pab,\{-,-\}_{\pi}),
\end{equation}
\subsection{Cohomological KK reduction}
\label{sec:KK_Pois}
In this section, we analyze the KK reduction of Poisson BF theory, specifically the KK theory built on the $\pacr$ cohomology. Note that the dg Lie algebra structure encoding Poisson BF theory is the same as the (shifted) cotangent structure on the following dg Lie algebra
\begin{equation}
	(\Omega^{0,\sbullet}(\C^3),\pab,\{-,-\}_{\pi}).
\end{equation}
Therefore, we only need to understand how the Lie bracket $\{-,-\}_{\pi}$ on $\Omega^{0,\sbullet}(\C^3)$ behave under the KK reduction. As a first step, we use the isomorphism $K$ to identify the complex with $(\Omega_{3d}^{\sbullet}\otimes \Omega_b^{0,\sbullet}(S^3),\hat{d} + \pacr)$. The Poisson bracket $\{-,-\}_{\pi}$ (which is a Lie bracket) induces a dg Lie algebra structure on the space $\Omega_{3d}^{\sbullet}\otimes \Omega_b^{0,\sbullet}(S^3)$. Therefore, we have the following isomorphic dg Lie algebra:
\begin{equation}\label{com_PBF_0}
	(\Omega_{3d}^{\sbullet}\otimes \Omega_b^{0,\sbullet}(S^3),\hat{d} + \pacr, \{-,-\}_{\pi}).
\end{equation}
Note that the Lie bracket above is defined as $K(\{K^{-1}(\alpha),K^{-1}(\beta)\}_\pi)$. To simplify the notation we still denote it by $\{-,-\}_{\pi}$.
Unlike the KK reduction of holomorphic Chern-Simons, where the product structure comes separately from the product on CR complex and the product on differential forms $\Omega_{3d}^{\sbullet}$, the induced bracket $\{-,-\}_{\pi}$ contains mixing term between the two complexes. For example, we compute the bracket of $t^2$ and $w_1$. The images of $t^2$ and $w_1$ under the map $K^{-1}$ are $z_1\bar{z}_1+ z_2\bar{z}_2$ and $z_1$ respectively. We have $\{(z_1\bar{z}_1+ z_2\bar{z}_2),z_1\}_\pi = -\bar{z}_2$. This implies 
\begin{equation}
	\{t^2,w_1\}_{\pi} = - t^2\bar{w}_2.
\end{equation}
Terms like this will lead to a complicated form of the Lie brackets  $\{-,-\}_{\pi}$  on the complex $\Omega_{3d}^{\sbullet}\otimes \Omega_b^{0,\sbullet}(S^3)$, and will result in a cumbersome action functional. Passing to the $\pacr$ cohomology will make things even worse. Computing the higher brackets on the $\Omega_{3d}^{\sbullet}\otimes H_b^{0,\sbullet}(S^3)$ will be an elaborate job.

One can actually simplifies the Lie bracket on $\Omega_{3d}^{\sbullet}\otimes \Omega_b^{0,\sbullet}(S^3)$. Instead of passing to the $\pacr$ cohomology, we first pass to the cohomology of $\hat{d}$. The cohomology of $(\Omega_{3d}^{\sbullet}(\C\times\R_{>0}),\hat{d})$ is simple and gives us $\mathcal{O}_\C$, which is the space of holomorphic functions on $\C$. Since the cohomology is concentrated in degree $0$, no higher bracket is generated at this step. Therefore, we find the following complex that is quasi-isomorphic to the complex \ref{com_PBF_0}
\begin{equation}\label{com_PBF_1}
	(\mathcal{O}_\C\otimes \Omega_b^{0,\sbullet}(S^3),\pacr, \{-,-\}_{\bar{\pi}}).
\end{equation}
By restricting the Poisson bracket to $\mathcal{O}_{\C}\otimes \Omega_b^{0,\sbullet}(S^3)$, we find that there is no mixing term in the bracket  $\{-,-\}_{\bar{\pi}}$  between $\mathcal{O}_{\C}$ and $\Omega_b^{0,\sbullet}(S^3)$ 
\begin{equation}
	\{f(z),Y(w_i,\bar{w}_i)\}_{\bar{\pi}} = 0.
\end{equation}	
Therefore, this bracket can be defined solely on $\Omega_b^{0,\sbullet}(S^3)$, and we have the dg Lie algebra $(\Omega_b^{0,\sbullet}(S^3),\pacr, \{-,-\}_{\bar{\pi}})$. In this way, the dg Lie algebra $(\mathcal{O}_\C\otimes \Omega_b^{0,\sbullet}(S^3),\pacr, \{-,-\}_{\bar{\pi}})$ is the tensor product of the commutative algebra $\mathcal{O}_\C$ with the dg Lie algebra $	(\Omega_b^{0,\sbullet}(S^3),\pacr, \{-,-\}_{\bar{\pi}})$. 

Then we can pass to the CR cohomology. Applying homotopy transfer theorem the dg Lie algebra $	(\Omega_b^{0,\sbullet}(S^3),\pacr, \{-,-\}_{\bar{\pi}})$, we obtain an $L_\infty$ algebra
\begin{equation}\label{Linf_PBF}
	(H_{b}^{0,\sbullet}(S^3),0, \{-,-\}_{\bar{\pi}},\{-,-,-\}_{\bar{\pi},3},\dots).
\end{equation}
If we tensor the above $L_\infty$ algebra with the dg commutative algebra $(\Omega_{3d}^{\sbullet},\hat{d})$, we obtain a $L_\infty$ algebra 
\begin{equation}
	(\Omega_{3d}^{\sbullet}\otimes H_{b}^{0,\sbullet}(S^3),\hat{d},\{...\}_{\bar{\pi},n},n\geq 2).
\end{equation}
Passing to the $\hat{d}$ cohomology, we get an $L_\infty$ quasi-isomorphic $L_\infty$ algebra $(\mathcal{O}_\C\otimes H_{b}^{0,\sbullet}(S^3),\{...\}_{\bar{\pi},n},\;n \geq 2)$, which is further $L_\infty$ quasi-isomorphic to \ref{com_PBF_1}. 

To summary, we have the following chain of $L_{\infty}$-quasi isomorphism
\begin{equation}
	(\Omega_{3d}^{\sbullet}\otimes \Omega_b^{0,\sbullet}(S^3),\hat{d} + \pacr, \{-,-\}_{\pi}) \overset{i}{\leftarrow} (\mathcal{O}_\C\otimes H_{b}^{0,\sbullet}(S^3),\{...\}_{\bar{\pi},n}) \overset{i}{\rightarrow} (\Omega_{3d}^{\sbullet}\otimes H_{b}^{0,\sbullet}(S^3),\{...\}_{\bar{\pi},n}).
\end{equation}

As a consequence, we can equivalently work with the $L_{\infty}$ algebra $(\Omega_{3d}^{\sbullet}\otimes H_{b}^{0,\sbullet}(S^3),\{...\}_{\bar{\pi},n})$ for our field theory. The difference between this $L_\infty$ structure and the original structure \ref{com_PBF_0} is that this $L_\infty$ structure is defined completely on the tangential Cauchy Riemann cohomology $ H_{b}^{0,\sbullet}(S^3)$. Analysis of this field theory can be done under the same framework of twisted $3d$ $\mathcal{N} = 2$ theory.

If we omit all the higher brackets $\{...\}_{\bar{\pi},n}$, $n\geq 3$, the Lie algebra $(H_{b}^{0,\sbullet}(S^3),\{-,-\}_{\bar{\pi}})$ can be identified with 
\begin{equation}
	T^*[1]\mathrm{Ham}(\C^{2}),
\end{equation}
where $\mathrm{Ham}(\C^{2})$ is the Lie algebra of Hamiltonian vector fields on $\C^2$. We provide more detail about this fact in the next section.

Using the special deformation retract constructed in \ref{sec:SDR}, the transferred $L_\infty$ structure is a cyclic $L_\infty$ algebra with the symmetric bilinear pairing given by $\Tr_{S^3}$. Then we can formulate the KK theory as follows, generalizing the results in \cite{Costello:2022wso}
\begin{prop}
	The KK reduction of the Poisson BF theory on $\C^3\backslash\C$ to $\C\times \R_{ > 0}$ is the same as the holomorphic twist of $3d$ $\mathcal{N} = 2$ theory, associated to a $\mathcal{N} = 2$ vector multiplet valued in the Lie algebra $\mathrm{Ham}(\C^{2})$ and a chiral multiplet valued in the adjoint representation
		\begin{equation}
		\begin{aligned}
			&\mathbf{H} \in	\Omega_{3d}^{\sbullet}(\C\times\R_{>0})\otimes \mathrm{Ham}(\C^{2}),\\
			& \tilde{\bm{w}} \in	\Omega_{3d}^{\sbullet}(\C\times\R_{>0})\otimes \mathrm{Ham}(\C^{2})^{*}\epsilon,\\
			& \bm{w}\in \Omega_{3d}^{\sbullet}(\C\times\R_{>0})\otimes\mathrm{Ham}(\C^{2})^{*}\epsilon,\\
			& \tilde{\mathbf{H}}\in \Omega_{3d}^{\sbullet}(\C\times\R_{>0})\otimes\mathrm{Ham}(\C^{2}).
		\end{aligned}
	\end{equation}
The BV action functional consist of the standard twisted $3d$ $\mathcal{N} = 2$ action
\begin{equation}
	\int \Tr_{S^3} \bm{w}(\hat{d}\mathbf{H}  + \{\mathbf{H} ,\mathbf{H} \}_{\bar{\pi}} ) + \int \Tr_{S^3} \tilde{\bm{w}}(\hat{d}\tilde{\mathbf{H} } + \{\mathbf{H} ,\tilde{\mathbf{H} } \}_{\bar{\pi}} ) .
\end{equation}
Together with the following deformations
\begin{equation}\label{def_PoiBF}
	\sum_{n = 3}^{\infty}\int  \frac{1}{2(n-2)!}\Tr(\bm{w} \{\mathbf{H} ,\mathbf{H} ,\underbrace{\tilde{\bm{w}},\dots,\tilde{\bm{w}}}_{n - 2}\}_{\bar{\pi},n} ) + \frac{1}{(n - 1)!}\Tr(\tilde{\mathbf{H}}  \{\mathbf{H} ,\underbrace{\tilde{\bm{w}},\dots,\tilde{\bm{w}}}_{n-1}\}_{\bar{\pi},n} ).
\end{equation}
\end{prop}

\subsection{The $L_\infty$ structure on $T^*[1]\mathrm{Ham}(\C^{2})$}
In this section, we analyze the $L_\infty$ structure on the CR cohomology $H_{b}^{0,\sbullet}(S^3)$ induced from the Poisson bracket.

The 2-bracket on $H^{0,0}_{CR}(S^3) = \C[w_1,w_2]$ is easy, because $w_i$ is mapped to $z_i$ under the isomorphism $K^{-1}$. Explicitly, we have
\begin{equation}
\begin{aligned}
		\{w_1^pw_2^q,w_1^rw_1^s\}_{\bar{\pi}} =  \epsilon_{ij}(\pa_{w_i}w_1^pw_2^q)(\pa_{w_j}w_1^rw_1^s)\\
		= (ps-qr)w_1^{p+r - 1}w_2^{q + s - 1}.
\end{aligned}
\end{equation}
Therefore, we can identify $(H^{0,0}_{CR}(S^3),\{-,-\}_{\bar{\pi}})$ with $\mathrm{Ham}(\C^2)$. The 2-bracket is also non zero on $H^{0,0}_{CR}(S^3) \otimes H^{0,1}_{CR}(S^3)$. To compute this, we recall that the isomorphism $K$ maps $f(\bar{w}_i)\epsilon$ to $f(\frac{\bar{z}_i}{r^2})\frac{\bar{z}_1d\bar{z}_2 - \bar{z}_2d\bar{z}_1}{r^4}$. So we first compute
\begin{equation}
	\{z_1^pz_2^q, \frac{\bar{z}_1^r\bar{z}_2^s\epsilon_{ij}\bar{z}_i d \bar{z}_j}{(z_1\bar{z}_1 + z_2\bar{z}_2)^{r+s + 2}}\}_{\pi} = (r+s +2)(q z_1\bar{z}_1 - pz_2\bar{z}_2)\frac{z_1^{p-1}z_2^{q - 1}\bar{z}_1^r\bar{z}_1^s\epsilon_{ij}\bar{z}_i d \bar{z}_j}{(z_1\bar{z}_1 + z_2\bar{z}_2)^{r+s + 3}}.
\end{equation}\label{Poi_jjbar}
Under the isomorphism $K$ this is mapped to the element
\begin{equation}
(r+s +2)(qw_1\bar{w}_1 -  pw_2\bar{w}_2 ) w_1^{p-1}w_2^{q - 1}\bar{w}_1^r\bar{w}_1^s\epsilon.
\end{equation}
This is an element in $\mathcal{H}_{\frac{p+q-1}{2},0}\otimes\mathcal{H}_{0,\frac{r+s + 1}{2}}\epsilon$. We need to compute its projection to $\mathcal{H}_{0,\frac{r+s -p-q}{2} + 1}\epsilon$.  This is analyzed in \ref{2-pro-01}, and we have the following
\begin{equation}
	\{w_1^pw_2^q,\frac{(r+s+1)!}{r!s!}\bar{w}_1^r\bar{w}_1^s\epsilon\}_{\bar{\pi}}  = ((r+1)q - (s+1)p)\frac{(r+s - p - q +3)!}{(r - p + 1)!(s - q +1)!}\bar{w}_1^{r - p+1}\bar{w}_2^{s - q + 1}\epsilon.
\end{equation}

In fact, by identifying $\C[\bar{w}_1,\bar{w}_2]\epsilon$ with the (degree shifted) dual of $\C[w_1,w_2]$, we can check that the above bracket can be identified with the bracket induced from the action of $\mathrm{Ham}(\C^2)$ on its dual. Therefore, the Lie algebra $(H^{0,\sbullet}_{CR}(S^3),\{-,-\}_{\bar{\pi}})$ is naturally identified with the Lie algebra $T^*[1]\mathrm{Ham}(\C^2)$.

Higher brackets on $H^{0,\sbullet}_{CR}(S^3)$ correspond to the $L_\infty$ deformation of $T^*[1]\mathrm{Ham}(\C^2)$. The method to compute the $L_\infty$ structure is similar to the method to compute the $A_\infty$ structure \ref{sec:A_trans}. The difference is that we sum over maps constructed from binary rooted trees instead of binary planar rooted trees. For each rooted trees, we put the brackets on the vertices and $h$ on the internal edges. 
\begin{equation}
	\{...\}_{\bar{\pi},n} = \sum_{T \in BT_n}(\pm)\{...\}_T.
\end{equation}
For example, the induced $3-$bracket is given by
\begin{equation}\label{brac_Poi3}
	\{a,b,c\}^{\bar{\pi}}_3 = p\{a,h\{b,c\}\} - p\{h\{a,b\},c\} + p\{b,h\{c,a\}\}.
\end{equation}
Using the bilinear pairing $\Tr_{S^3}$, the higher brackets satisfy
\begin{equation}
	\Tr_{S^3}(a_0	\{a_1,a_2,\dots,a_{n}\}_{\bar{\pi},n} ) = (-1)^{n + |a_n|(|a_0| + \dots + |a_{n-1}|)}	\Tr_{S^3}(a_n	\{a_0,a_1,\dots,a_{n_1}\}_{\bar{\pi},n} ).
\end{equation}
Therefore, for the $n$-th bracket $	\{...\}_{\bar{\pi},n}$, it suffice to compute it on the subspace $H_{b}^{0,0}(S^3)\otimes H_{b}^{0,1}(S^3)^{\otimes n - 1}$. For $a_0 \in H_{b}^{0,0}(S^3)$ and $\bar{a}_{i}\epsilon \in  H_{b}^{0,1}(S^3)$, we have
\begin{equation}
	\{a_0,\bar{a}_1\epsilon,\dots ,\bar{a}_{n - 1}\epsilon\}_{\bar{\pi},n} = \sum_{\sigma \in S_{n-1}}p\{ \dots h\{h\{a_0,\bar{a}_{\sigma(1)}\epsilon\},\bar{a}_{\sigma(2)}\epsilon\},\dots ,\bar{a}_{\sigma(n-1)}\epsilon\}.
\end{equation}
We define the following constant
\begin{equation}
\begin{aligned}
		&(\pi_n)^{p,q;r,s}_{u_1,v_1,\dots,u_{n-1},v_{n-1}}: =   \\
		&\prod_i\frac{(u_i+v_i+1)!}{u_i!v_i!}\times\Tr_{S^3}(w_1^pw_2^q p\{ \dots h\{h\{w_1^rw_2^s,\bar{w}_1^{u_1}\bar{w}_2^{v_1}\epsilon\},\bar{w}_1^{u_2}\bar{w}_2^{v_2}\epsilon\},\dots ,\bar{w}_1^{u_{n-1}}\bar{w}_2^{v_{n-1}}\epsilon\} ).
\end{aligned}
\end{equation}
We have that $(\pi_n)^{p,q;r,s}_{u_1,v_1,\dots,u_{n-1},v_{n-1}}$ is non zero only when the following condition is hold
\begin{equation}\label{rel_higher_poi}
	u_1 + \dots +u_{n-1} = p+r- 2n + 3,\quad v_1 + \dots +v_{n-1} = q+s - 2n + 3.
\end{equation}
Then the $n$-th bracket  $	\{...\}_{\bar{\pi},n}$ on the subspace $H_{b}^{0,0}(S^3)\otimes H_{b}^{0,1}(S^3)^{\otimes n - 1}$ can be expressed as follows
\begin{equation}
\begin{aligned}
	&\prod_i\frac{(u_i+v_i+1)!}{u_i!v_i!}	\{w_1^rw_2^s,\bar{w}_1^{u_1}\bar{w}_2^{v_1}\epsilon,\dots ,\bar{w}_1^{u_{n-1}}\bar{w}_2^{v_{n-1}}\epsilon\}_{\bar{\pi},n} = \\
	&\left( \sum_{\sigma \in S_{n-1}}(\pi_n)^{p,q;r,s}_{u_{\sigma(1)},v_{\sigma(1)},\dots,u_{\sigma(n-1)},v_{\sigma(n-1)}}\right)   \frac{(p+q+1)}{p!q!}\bar{w}_1^p\bar{w}_2^q\epsilon,
\end{aligned}
\end{equation}
with $p$ and $q$ satisfying \ref{rel_higher_poi}.

This expression of the higher brackets also leads us to an expansion of the effective interaction of the KK theory of Poisson BF theory 
\begin{equation}\label{KK_Pois}
	\begin{aligned}
		&\sum_{n \geq 3} \sum_{p,q,r,s \geq 0}\sum_{\substack{u_1 + \dots u_{n-1} = p+r -  2n +3\\v_1+\dots v_{n-1} = q+s - 2n+3}}\left( \sum_{\sigma \in S_{n-1}}(\pi_n)^{p,q;r,s}_{u_{\sigma(1)},v_{\sigma(1)},\dots,u_{\sigma(n-1)},v_{\sigma(n-1)}}\right)\\
		&\left(  \frac{1}{2(n-2)!} \int\mathbf{H}[p,q] \mathbf{H}[r,s]\mathbf{w}[u_1,v_1]\tilde{\mathbf{w}}[u_2,v_2]\dots \tilde{\mathbf{w}}[u_{n-1},v_{n-1}] \right.\\
		&\left.+ \frac{1}{(n-1)!} \int\tilde{\mathbf{H}}[p,q] \mathbf{H}[r,s]\tilde{\mathbf{w}}[u_1,v_1]\dots \tilde{\mathbf{w}}[u_{n-1},v_{n-1}]\right) .
	\end{aligned}
\end{equation}

We provide an explicit expression for the constant $(\pi_3)^{p,q;r,s}_{u_1,v_1,u_2,v_2}$ in \ref{apx:Poi_3}.

\section{Kodaira Spencer gravity}
	\label{sec:KS}

	\subsection{A brief review}
	\label{sec:rev_KS}
	\paragraph{algebraic structure of Polyvector fields}
	The closed-string sector of the topological B-model can be described by a quantum field theory on space time called the Kodaira-Spencer theory of gravity \cite{Bershadsky:1993cx}. This theory is generalized in \cite{Costello:2012cy} by turning on gravitational descendants and is called BCOV theory. To introduce this theory, we first introduce the space of polyvector fields on a Calabi-Yau three-fold $X$. Let
	\begin{equation}
		\mathrm{PV}^{j,i}(X) = \Omega^{0,i}(X,\wedge^{j}TX)
	\end{equation}
	be the bundle of $(0,i)$ forms with coefficients in the $j$-th exterior power of the holomorphic tangent bundle of $X$. We call them poly-vector fields on $X$. The anti-holomorphic Dolbeault differential $\pab$ on forms extended to polyvector fields and gives us a differential
	\begin{equation}
		\pab: 	\mathrm{PV}^{j,i}(X)  \to 	\mathrm{PV}^{j,i+1}(X).
	\end{equation}
	Wedge product together with the differential $\pab$ makes $\mathrm{PV}^{\sbullet,\sbullet}(X)$ a differential graded commutative algebra.
	
	Next we can identify $\mathrm{PV}^{j,i}(X) $ with $\Omega^{3-j,i}$ via contracting poly-vector fields with the holomorphic volume form $\Omega_X$ on $X$.
	\begin{equation}
		\begin{aligned}
			\mathrm{PV}^{j,i}(X) & \cong \Omega^{3-j,i},\\
			\alpha &\mapsto  \alpha \vee \Omega_X.
		\end{aligned}
	\end{equation}
	The differential $\pa$ on forms $ \Omega^{\sbullet,\sbullet}$ defines an operator on $\mathrm{PV}^{\sbullet,\sbullet}(X) $ via the above isomorphism
	\begin{equation}
		\pa: \mathrm{PV}^{j,i}(X)  \to \mathrm{PV}^{j - 1,i}(X) .
	\end{equation}
	The $\pa$ operator on polyvector fields is not a derivation with respect to the wedge product. The failure of it being a derivation is measured by the following bracket
	\begin{equation}
		\{ \alpha, \beta \} := \pa(\alpha\wedge \beta) - ( \pa\alpha)\wedge \beta -(-1)^{|\alpha|} \alpha \wedge (\pa \beta),
	\end{equation}
	which coincides with the Schouten-Nijenhuis bracket on polyvector fields (up to a sign). The fundamental algebraic structures of polyvector fields on Calabi-Yau geometry can be
	summarized by saying that the tuple  $\{\mathrm{PV}^{\sbullet,\sbullet}(X), \pab, \wedge , \pa, \{-,-\}\}$ defines a differential graded Gerstenhaber-Batalin-Vilkovisky algebra. The BV theory of Kodaira-Spencer gravity, or BCOV theory, is built on this algebraic structure.
	
	In a local coordinate, the algebra of polyvector fields can be written as
	\begin{equation}
		\mathrm{PV}^{\sbullet,\sbullet}(\C^3) = C^{\infty}(\C^3)[\eta^i,d\bar{z}_i].
	\end{equation}
	The differential $\pa$ can be given by the expressions
	\begin{equation}
		\pa = \sum_{i = 1}^3\frac{\pa}{\pa \eta^i} \frac{\pa}{\pa z_i}.
	\end{equation}
	The bracket can be given by the expression
	\begin{equation}
		\{\alpha,\beta\} =  \sum_{i = 1}^3\frac{\pa \alpha}{\pa \eta^i} \frac{\pa \beta }{\pa z_i} + (-1)^{|\alpha|}  \frac{\pa \alpha}{\pa z_i}  \frac{\pa \beta}{\pa \eta^i}.
	\end{equation}
	
	\paragraph{Kodaira-Spencer theory}
	In the original formulation, Kodaira-Spencer theory has BV fields described by divergence free polyvector fields
	\begin{equation}
	\mu \in	\ker \pa \subset \mathrm{PV}^{\sbullet,\sbullet}(X)[2].
	\end{equation}
	The Lagrangian takes the form
	\begin{equation}\label{act_KS}
		\frac{1}{2} \int_X (\pa^{-1}\mu) \wedge \pab \mu + \frac{1}{6} \int_X \mu\wedge \mu \wedge \mu.
	\end{equation}
	The equation of motion can be obtained by varying the action functional, which is
	\begin{equation}
		\pab \mu + \frac{1}{2}\{\mu,\mu\} = 0.
	\end{equation}

One can see that the BV formalism of Kodaira-Spencer theory is encoded in the dg Lie algebra $(\ker \pa ,\pab,\{-,-\})$, together with an ill-defined pairing $\int_X \pa^{-1}\otimes\mathrm{Id}$. Though the kinetic term for Kodaira-Spencer theory is ill-defined, this theory has a well-defined propagator which allows one to analyze the theory in perturbation theory.

There is a consistent truncation of Kodaira-Spencer theory by simply dropping the fields in $\mathrm{PV}^{0,\sbullet}$ and $\mathrm{PV}^{2,\sbullet}$. We are left with fields in $\ker \pa \subset \mathrm{PV}^{1,\sbullet}(X)[2]$ and action functional of the same form as \ref{act_KS}. This theory is called the $(1,0)$ Kodaira-Spencer theory, or type I Kodaira Spencer gravity because its relationship to the ordinary B-model is similar to the relationship of physical type I string theory to the IIB string theory \cite{Costello:2019jsy}. It is a pleasant feature shown in \cite{Costello:2018zrm} that the full Kodaira-Spencer theory is equivalent to an Abelian Chern-Simons theory valued in $\Pi\C^2$ coupled with the $(1,0)$ Kodaira-Spencer theory. Therefore, for the analysis of KK reduction of Kodaira-Spencer theory, we only need to focus on the fields in $\mathrm{PV}^{1,\sbullet}(X)$.

In \cite{Costello:2012cy}, a different formulation of Kodaira-Spencer theory, called BCOV theory, is proposed, where the constraint $\pa \mu = 0 $ is imposed homologically. This theory produces a better behaved BV theory and admits rigorous perturbation quantization \cite{Costello:2012cy}. However, in this paper, we still proceed with the old formulation of Kodaira-Spencer gravity for simplicity.

	\subsection{Some useful definitions}
	The space of Polyvector fields on $\C^3\backslash \C$ is given by 
	\begin{equation}
		\mathrm{PV}^{\sbullet,\sbullet}(\C^3\backslash \C) = \C^{\infty}(\C^3\backslash\C)[\pa_z,\pa_{z_1},\pa_{z_2},d\bar{z},d\bar{z}_1,d\bar{z}_2].
	\end{equation}
	We extend the isomorphism \ref{K} to an isomorphism
	\begin{equation}
		K: \mathrm{PV}^{\sbullet,\sbullet}(\C^3\backslash \C) \to \Omega_{3d}^{\sbullet}(\C\times \R_{ > 0})\otimes\Omega_b^{0,\sbullet}(S^3)[\pa_z,\eta_1,\eta_2],
	\end{equation}
	where we simply maps $\pa_{z_i}$ to $\eta_i$. Using this isomorphism, the BV algebra structure on polyvector fields induces a BV algebra structure on $\Omega_{3d}^{\sbullet}\otimes\Omega_b^{0,\sbullet}(S^3)[\pa_z,\eta_1,\eta_2]$. For example, the BV operator is defined by $K\circ\pa\circ K^{-1}$.
	
	 By identifying the space of polyvector fields $\mathrm{PV}^{\sbullet,\sbullet}$ with $\Omega^{\sbullet,\sbullet}_{Dol}$ using the volume form, we can further identify the above complex with $\Omega^{(\sbullet),\sbullet}_{3d}(\C\times \R_{ > 0})\otimes \Omega_{b}^{\sbullet,\sbullet}(S^{3})$. However, it will be more convenient to use the space $\Omega_{3d}^{\sbullet}(\C\times \R_{ > 0})\otimes\Omega_b^{0,\sbullet}(S^3)[\pa_z,\eta_1,\eta_2]$ when we deal with the open-closed coupling in later sections.
	
	We are interested in the $SU(2)$ decomposition of the above complex. For this purpose, we construct the following two $SU(2)$ invariant vector fields
	\begin{equation}
		\begin{aligned}
			\eta & = \bar{w}_2 \eta_1- \bar{w}_1\eta_2,\\
			\hat{\eta} & = w_1 \eta_1 + w_2 \eta_2.
		\end{aligned}
	\end{equation}
	Note that $	\eta$ is the basis of the bundle $T^{\C}S^3 \cap T^{(1,0)}\C^2$ where $\C^2$ is the ambient space of $S^3$. The vector fields $\pa_z$,  $\eta$ and $\hat{\eta}$ generate the whole polyvector fields because
		\begin{equation}
		\begin{aligned}
			\eta_1 & = w_2 \eta +  \bar{w}_1\hat{\eta},\\
			\eta_2 & = -w_1 \eta - \bar{w}_2 \hat{\eta},
		\end{aligned}
	\end{equation}
	and 
	\begin{equation}
		\eta\wedge\hat{\eta} = \eta_1\wedge\eta_2.
	\end{equation}

	The vector fields $\eta$ and $\hat{\eta}$ act on the complex $\Omega_{3d}^{\sbullet}\otimes\Omega_b^{0,\sbullet}(S^3)$ via the Schouten bracket. These actions are important because they determine the interactions of the Kodaira-Spencer theory and its coupling with the holomorphic Chern-Simons theory. 
	
First, we analyze the action of the two vector fields $\eta$ and $\hat{\eta}$ on $\Omega_b^{0,\sbullet}(S^3)$. This can be obtained from the $SU(2)$ decomposition of $\Omega_b^{0,\sbullet}(S^3)$. Since both vector fields are $SU(2)$ invariant, on each irreducible subspace they act as either zero or scalar multiples of the identity onto an irreducible subspace of the same $SU(2)$ representation.
 For $\eta$, we have
	\begin{equation}
		\eta: \begin{array}{l}
		\mathcal{H}_{0,\bar{j}} \to 0\\
		\mathcal{H}_{j,\bar{j}} \overset{\simeq}{\to} \mathcal{H}_{j- \frac{1}{2},\bar{j} + \frac{1}{2}}
	\end{array}.
\end{equation}
The action of $\eta$ on $\Omega_b^{0,\sbullet}(S^3)$ is completely characterized by its action on the highest weight vector of each space $\mathcal{H}_{j,\bar{j}}$
\begin{equation}
	\eta(w_1^{2j} \bar{w}_2^{2\bar{j}}) = 2j w_1^{2j  - 1}\bar{w}_2^{2\bar{j} + 1}.
\end{equation}
Using this vector $\eta$ we can define another action on $\Omega_b^{0,\sbullet}(S^3)$ by $\pa(-\wedge \eta)$. We can check that $\pa(-\wedge \eta) = \eta(-)$.

We see that the operator $\eta$ has a very similar behavior as the homotopy operator $h$ defined in Section \ref{sec:SDR} except that $\eta$ is of degree $0$. It is useful to define a degree zero operator $\tilde{h} : \oplus\mathcal{H}_{j,\bar{j}} \to \oplus\mathcal{H}_{j,\bar{j}} $ as follows 
\begin{equation}\label{def_tildh}
	\tilde{h}(Y)  = h(Y\epsilon),\quad \text{for any } Y \in \oplus\mathcal{H}_{j,\bar{j}}.
\end{equation}
Then the operator $\eta$ and $\tilde{h}$ are related as follows
\begin{equation}
	\eta|_{\mathcal{H}_{j,\bar{j}}} = (2j)(2\bar{j}+1) \tilde{h}|_{\mathcal{H}_{j,\bar{j}}}.
\end{equation} 
This relation will be useful later. We can use it to reduce computation in Kodaira-Spencer theory into computation in holomorphic Chern-Simons theory.

Now we consider the other $SU(2)$ invariant vector $\hat{\eta}$. The action of $\hat{\eta}$ on $\Omega_{b}^{0,0}(S^3)$ is given by
\begin{equation}
	\hat{\eta}: \begin{array}{l}
		\mathcal{H}_{j,j} \to 0\\
		\mathcal{H}_{j,\bar{j}} \overset{\simeq}{\to} \mathcal{H}_{j,\bar{j}}\quad \text{ for }j \neq \bar{j}
	\end{array}.
\end{equation}
It act on the highest weight vector of $\mathcal{H}_{j,\bar{j}}$ by the constant $ \hat{\eta}(w_1^{2j} \bar{w}_2^{2\bar{j}}) = (2j - 2\bar{j}) w_1^{2j}\bar{w}_2^{2\bar{j}}$. Therefore we have
\begin{equation}
	\hat{\eta}|_{\mathcal{H}_{j,\bar{j}}} = (2j - 2\bar{j})\mathrm{Id}_{\mathcal{H}_{j,\bar{j}}}.
\end{equation}
We can also consider the operation defined by $\pa(-\wedge \hat{\eta}) $. Unlike $\eta$, the operations $\pa(-\wedge \hat{\eta}) \neq \hat{\eta}(-)$. We find that 
\begin{equation}
	\begin{aligned}
		&\pa(-\wedge \hat{\eta})|_{\mathcal{H}_{j,\bar{j}}} = (2j - 2\bar{j} + 2)\mathrm{Id}_{\mathcal{H}_{j,\bar{j}}},\\
		&\pa(-\wedge \hat{\eta})|_{\mathcal{H}_{j,\bar{j}}\epsilon} = (2j - 2\bar{j})\mathrm{Id}_{\mathcal{H}_{j,\bar{j}}}.
	\end{aligned}
\end{equation}

In fact, the action of $\hat{\eta}$ on $\Omega_{3d}^{\sbullet}$ is also nontrivial. For example, $\hat{\eta}(t)$ is computed by $(z_1\pa_{z_1} + z_2\pa_{z_2})(r) = \frac{1}{2}r$. We have
\begin{equation}
	\hat{\eta}(f(t)) = \frac{1}{2}t\pa_tf(t).
\end{equation}
This means that $\hat{\eta}$ will have a nontrivial action on the topological $t$ direction. Terms like this will be annoying in our later formulation of the KK theory. However, we can proceed as we deal with the Poisson BF theory. We first pass to the $\hat{d}$ cohomology without introducing any higher product. By passing to the $\hat{d}$ cohomology, the action of $\eta$ and $\hat{\eta}$ on $\mathcal{O}_\C$ is trivial. As a consequence, we can forget the action of $\eta$ and $\hat{\eta}$ on $\Omega_{3d}^{\sbullet}$ and only consider their action on $\Omega_b^{0,\sbullet}(S^3)$. This greatly simplifies our later analysis of the KK theory.

	\subsection{KK reduction}
	\label{sec:KK_KS}
	It was shown in \cite{Costello:2018zrm} that the $PV^{0,\sbullet}$ and $PV^{2,\sbullet}$ part of Kodaira-Spencer fields can be reformulated as a (super) abelian holomorphic Chern-Simons theory. So it suffices to focus on the $PV^{1,\sbullet}$ fields. Kodaira-Spencer theory is more subtle than our previous examples because fields are subject to the constraint $\pa\bm{\mu} = 0$. This equation cut out a subspace of $\Omega_b^{0,\sbullet}(S^3)\pa_z \oplus \Omega_b^{0,\sbullet}(S^3)\eta\oplus\Omega_b^{0,\sbullet}(S^3)\hat{\eta}$ that forms the tower of KK fields. One can proceed as in our previous examples. The Dolbeault differential and the Schouten bracket of polyvector fields defined a $L_\infty$ structure on the cohomology of the subspace subject to the $\pa$ constraint. However, this approach will end up with cumbersome notation. To simplify the discussion, we directly write down the KK fields that solve the $\pa$ constraint. We find the following fields
	\begin{equation}\label{KK_KS_1}
\begin{aligned}
		\bm{\beta} = & \sum_{p+q\geq 1}\frac{(p+q + 1)!}{p!q!}(\bm{\beta}[p,q]\bar{w}_1^p\bar{w}_2^q\epsilon \pa_z  + \frac{1}{p+q}\pa_z\bm{\beta}[p,q]\bar{w}_1^p\bar{w}_2^q\epsilon \hat{\eta}) ,\\
		\bm{\alpha} = &(p+q)\bm{\alpha}[p,q]w_1^pw_2^q\eta - \sum_{p+q\geq 1}\bm{\alpha}[p,q]\eta(w_1^{p}w_2^q)\hat{\eta}.
\end{aligned}
	\end{equation}
Note that the KK tower start from the modes $p + q = 1$.

We can check that both of the fields are in the $\pacr$ cohomology and the kernel of $\pa$. The corresponding action functional can be obtained from the  Kodaira-Spencer action functional as follows
\begin{equation}
	\int \Tr_{S^3}\left( (\pa^{-1}\bm{\beta}) \hat{d} \bm{\alpha} + \bm{\beta} \bm{\alpha} \bm{\alpha}\right) .
\end{equation}
With the help of the decomposition into KK towers, we can explicit compute the $\pa^{-1}$ operation on the fields:
\begin{equation}
	\begin{aligned}
		\pa^{-1}\bm{\beta} = &\frac{(p+q + 1)!}{p!q!}(\frac{1}{p+q}\bm{\beta}[p,q]\bar{w}_1^p\bar{w}_2^q\epsilon  \pa_z \wedge \hat{\eta}) ,\\
		\pa^{-1}\bm{\alpha} = &\bm{\alpha}[p,q]w_1^{p}w_2^q\eta\wedge \hat{\eta} .
	\end{aligned}
\end{equation}
This feature displays the major advantage of working with the KK theory. The KK theory utilizes the $SU(2)$ action on fields that allows us to solve the $\pa^{-1}$ operation, which is annoying in the analysis of the original $6d$ theory. We can expand the action functional explicitly in terms of the KK modes as follows
\begin{equation}
	\sum \int \bm{\beta}[p,q]\hat{d}\bm{\alpha}[p,q] + \sum(ps - qr)\int \bm{\alpha}[p,q] \bm{\alpha}[r,s] \bm{\beta}[p+r - 1,q+s - 1].
\end{equation}
The coefficient in front of the quadratic term is computed by
\begin{equation}
\frac{(p+q + 1)!}{p!q!}\Tr_{S^3}(\bar{w}_1^p\bar{w}_2^qw_1^{p}w_2^q)  = 1.
\end{equation}
The coefficient in front of the cubic term is computed by
\begin{equation}
\begin{aligned}
	&\frac{(p+q+s+r - 1)!}{(p+r - 1)!(s+m - 1)!}\Tr_{S^3}(\eta(w_1^{p}w_2^q)\hat{\eta}(w_1^rw_2^s)\bar{w}_1^{p+r - 1}\bar{w}_2^{q+s - 1}) - (p\leftrightarrow r q\leftrightarrow s) \\
	&= (ps - qr).
\end{aligned}
\end{equation}
We see that the action functional takes the same form as the $3d$ holomorphic topological theory. As we will see later, this feature allows us to easily compute the boundary chiral algebra.

The coefficient in front of the cubic term is the same as the structure constant of the Lie algebra $\mathrm{Ham}(\C^2)$. Therefore, this part of Kodaira-Spencer theory is the same as the holomorphic twist of $3d$ $\mathcal{N} = 2$ theory associated to a vector multiplet valued in $\mathrm{Ham}(\C^{2})$. Of course, higher structures exist and should give rise to higher order interaction terms. We briefly discuss them in the next section.

Similarly, we have the following fields

\begin{equation}
	\begin{aligned}
	\tilde{\bm{\beta}} =	& \sum_{p,q\geq0} \frac{(p+q+1)!}{p!q!}\tilde{\bm{\beta}}[p,q]\bar{w}_1^p\bar{w}_2^q\epsilon \eta, \\
	\tilde{\bm{\alpha}} =	& \sum_{p,q\geq0}\pa_z\tilde{\bm{\alpha}}[p,q]w_1^pw_2^q\hat{\eta} - (p+q+2)\tilde{\bm{\alpha}}[p,q]w_1^{p}w_2^q\pa_z.
	\end{aligned}
\end{equation}
This KK tower starts from the modes $p,q = 0$. The corresponding action functional is given by
\begin{equation}
	\int (\pa^{-1}\tilde{\bm{\beta}}) \hat{d} \tilde{\bm{\alpha} }+ \tilde{\bm{\beta}} \tilde{\bm{\alpha}} \tilde{\bm{\alpha}}.
\end{equation}

We can compute the $\pa^{-1}$ operation on the fields:
\begin{equation}
	\begin{aligned}
		\pa^{-1}\tilde{\bm{\beta}} = &\frac{(p+q + 1)!}{p!q!}(\frac{1}{p+q + 2}\tilde{\bm{\beta}}[p,q]\bar{w}_1^p\bar{w}_2^q\epsilon  \eta\wedge \hat{\eta}), \\
		\pa^{-1}\tilde{\bm{\alpha}} = &\tilde{\bm{\alpha}}[p,q]w_1^{p}w_2^q\pa_z\wedge \hat{\eta} .
	\end{aligned}
\end{equation}

This allows us to expand the action functional explicitly as follows
\begin{equation}
\begin{aligned}
		\sum \int \tilde{\bm{\beta}}[p,q]\hat{d}\tilde{\bm{\alpha}}[p,q]  + (p+q+2)\int\bm{\beta}[p+r ,q+s] \bm{\alpha}[p,q] \pa_z\bm{\alpha}[r,s] \\
		 -(r+s + 2)\int \bm{\beta}[p+r ,q+s]\pa_z\bm{\alpha}[p,q] \bm{\alpha}[r,s] .
\end{aligned}
\end{equation}
Due to the $\pa_z$ that is present in the action functional, this theory does not come from any $3d$ $\mathcal{N} = 2$ theory. However, it still shares many similar features with the holomorphic twisted theory and can be analyzed in a similar fashion. 

There are also coupling between the two part of Kodaira-Spencer fields, given as follows
\begin{equation}
	\int \Tr_{S^3} \bm{\beta} \bm{\alpha} \tilde{\bm{\alpha}} + \tilde{\bm{\beta}}\bm{\alpha} \tilde{\bm{\alpha}}.
\end{equation}

We emphasize that the subspace of $PV^{1,\sbullet}$ cut out by the constraint $\pa\bm{\mu} = 0$ does not have a canonical choice of basis. In other words, the KK theory of Kodaria-Spencer gravity can be formulated by different representatives of KK fields that have different action functional. The different choices of basis should be related by a recombination of the KK modes. We made our above choice to make the boundary chiral algebra compatible with the celestial chiral algebra studied in \cite{Costello:2022wso}.

\subsection{Higher order interaction}
Like our previous examples, higher order interactions for the KK theory of Kodaira-Spencer gravity exist and can be obtained from homotopy transfer. However, we need to be careful here because the homotopy operator in this case is not the homotopy operator $h$ defined in Section \ref{sec:SDR} anymore. The reason for this is that the CR differential takes a slightly different form on $\mathrm{PV}^{\sbullet,\sbullet}$ under the basis $\eta,\hat{\eta}$. We can compute the following
\begin{equation}
	\pacr(\pa_z) = \pacr(\hat{\eta}) = 0,\;\; \pacr(\eta) = \epsilon\hat{\eta}.
\end{equation}
This implies that the CR differential $\pacr$ act on $\Omega_{b}^{0,\sbullet}(S^3)\eta$ as follows
\begin{equation}
	\begin{aligned}
		\pacr: \mathcal{H}_{j,\bar{j}}\eta &\to \begin{array}{l  l}
			\mathcal{H}_{j,\bar{j}}\epsilon \hat{\eta}& \bar{j} = 0,\\
			\mathcal{H}_{j+ \frac{1}{2},\bar{j} - \frac{1}{2}}\epsilon\eta \oplus \mathcal{H}_{j,\bar{j}}\epsilon \hat{\eta}& \bar{j} \geq \frac{1}{2},
		\end{array}\\
	Y\eta &\mapsto \pacr(Y)\eta + Y\epsilon\hat{\eta}.
	\end{aligned}
\end{equation}

Accordingly, we get a new homotopy operator $\mathds{h}$. Since $\pacr$ act on $\Omega_{b}^{0,0}(S^3)\pa_z$ and $\Omega_{b}^{0,0}(S^3)\hat{\eta}$ in the same way as it act on  $\Omega_{b}^{0,0}(S^3)$, the new homotopy operator $\mathds{h}$ acts on $\Omega_{b}^{0,1}(S^3)\pa_z$ and $\Omega_{b}^{0,1}(S^3)\hat{\eta}$ in the same way as $h$ act on them. The only difference is the action of $\mathds{h}$ on $\Omega_{b}^{0,1}(S^3)\eta$, given as follows:
\begin{equation}\label{homotopy_KS}
	\begin{aligned}
		\mathds{h}: \mathcal{H}_{j,\bar{j}}\epsilon\eta &\to \begin{array}{l  l}
			0 & j =  0,\\
			\mathcal{H}_{j -  \frac{1}{2},\bar{j} + \frac{1}{2}}\epsilon\eta & j  = \frac{1}{2},\\
			\mathcal{H}_{j -  \frac{1}{2},\bar{j} + \frac{1}{2}}\epsilon\eta\oplus \mathcal{H}_{j -1 ,\bar{j} + 1}\epsilon \hat{\eta}& j \geq 1,
		\end{array}\\
		Y\epsilon\eta&\mapsto \tilde{h}(Y)\eta - \tilde{h}^2(Y)\hat{\eta}.
	\end{aligned}
\end{equation}
Then higher order interaction of the KK theory is computed by terms like
\begin{equation}
	\int \pa^{-1}\bm{\beta} \{\bm{\alpha},\mathds{h}\{\bm{\alpha},\bm{\beta}\}\} + \dots
\end{equation}
We leave it to future work to study these terms in detail.

\subsection{Coupling to holomorphic Chern-Simons theory}
\label{sec:coup_OC}
BCOV theory can be coupled to the holomorphic Chern-Simons theory. The coupled theory has a unique quantization \cite{Costello:2015xsa} and gives us a concrete realization of open-closed string field theory. In this paper, we are interested in the leading order coupling between the Kodaria spencer fields and the holomorphic Chern-Simons theory. In particular, the coupling between an elements $\bm{\mu} = \bm{\mu}^i\pa_{z_i} \in \mathrm{PV}^{1,\sbullet}(X)\cap \ker \pa$ and the holomorphic Chern-Simons is given by
	\begin{equation}\label{coup-KS-HCS}
\frac{1}{2}	\int \Omega_X \Tr( \pmb{\EuScript{A} } \wedge \bm{\mu}^i(\pa_{z_i} \pmb{\EuScript{A} })) .
\end{equation}

We are interested in the coupling after the KK reduction. If we substitute $\bm{\mu}$ with $\bm{\alpha}$ and $\bm{\beta}$ in the above coupling, we get the following terms
\begin{equation}\label{coup_OC_0}
\begin{aligned}
	&\frac{1}{2} \int \Tr( \bm{\beta}[p + r,q+ s]\mathbf{A}[p,q]\pa_{z}\mathbf{A}[r,s]  )  + \frac{1}{2}\frac{r+s}{p+q+r+s} \int \Tr( \pa_z\bm{\beta}[p + r,q+ s]\mathbf{A}[p,q]\mathbf{A}[r,s]  ) \\
	&+ \frac{1}{2} (ps - qr)\int \Tr( \bm{\alpha}[r,s]\mathbf{B}[p+ r - 1,q+s - 1]\mathbf{A}[p,q])
\end{aligned}
\end{equation}

 If we substitute $\bm{\mu}$ with $\tilde{\bm{\alpha}}$ and $\tilde{\bm{\beta}}$, we get the following terms
 \begin{equation}
 	\begin{aligned}
 		& - \frac{1}{2}\frac{(ps - qr)}{p+q +r+s}\int \Tr(\tilde{\bm{\beta}}[p+r-1,q+s - 1]\mathbf{A}[p,q]\mathbf{A}[r,s]) \\
 		& \frac{1}{2}(p+q)\int \Tr( \pa_z\tilde{\bm{\alpha}}[r,s]\mathbf{B}[p+ r ,q+s]\mathbf{A}[p,q]) - \frac{1}{2}(r+s +2) \int \Tr( \tilde{\bm{\alpha}}[r,s]\mathbf{B}[p+ r ,q+s]\pa_z\mathbf{A}[p,q]  ).
 	\end{aligned}
 \end{equation} 

Just like our previous examples, there are higher order interactions coming from homotopy transfer. \footnote{There are also higher order couplings between BCOV theory and holomorphic Chern-Simons theory. They should also appear in the KK theory but should be distinguished from the higher order interaction discussed in this section that comes from the homotopy transfer of leading order coupling.} They can be computed by the same method of summing over trees as before, replacing the $2$-brackets by the action of Kodaira-Spencer fields on the Chern-Simons form. The simplest one is the following
\begin{equation}
	I_{OC} = \int \mathbf{B} [\mathbf{A}, h(\bm{\beta} \mathbf{A})].
\end{equation}
Using the expansion \ref{KK_KS_1} of the Kodaira-Spencer fields, we have
\begin{equation}\label{coup_OC_3}
\begin{aligned}
	 I  =& \sum_{\substack{u + \tilde{u} = p+r - 1\\v+\tilde{v} = q+s - 1}}	(m_3)^{p,q;r,s}_{u,v} \int \Tr \mathbf{B}[\tilde{u},\tilde{v}]\left[\mathbf{A}[p,q],\left(\bm{\beta}[u,v] \pa_z\mathbf{A}[r,s] +\pa_z\bm{\beta}[u,v] \frac{r+s}{u+v}\mathbf{A}[r,s] \right)\right] \\
	  &+ (m_3)^{r,s;p,q}_{u,v} \int \Tr \mathbf{B}[\tilde{u},\tilde{v}]\left[\left( \pa_z\mathbf{A}[p,q]\bm{\beta}[u,v] + \frac{p+q}{u+v}\mathbf{A}[p,q]\pa_z\bm{\beta}[u,v]\right), \mathbf{A}[r,s] \right].\\
\end{aligned}
\end{equation}
There are also coupling of the form $\int \mathbf{B} [\mathbf{A}, h(\bm{\alpha} \mathbf{B})] +\dots$. Finding their explicit expansion requires us to compute expression of the form $pM(\eta(w_1^rw_2^s) ,\hat{\eta}hM(w_1^pw_2^q , \bar{w}_1^u\bar{w}_2^v\epsilon)) $. We leave them to future work.

\subsection{Coupling to holomorphic BF theory}
\label{sec:coup_HBF}
Kodaria-Spencer gravity can also be coupled to the holomorphic BF theory. The coupling takes the same form as the coupling with holomorphic Chern-Simons theory
\begin{equation}
\frac{1}{2}	\int \Omega_X \Tr(\pmb{\EuScript{A}} \wedge \bm{\mu}^i(\pa_{z_i} \pmb{\EuScript{A}})) .
\end{equation}
The only difference is that $\pmb{\EuScript{A}}$ is the field of holomorphic BF theory now. According to \ref{sec:KK_HBF}, after KK reduction, the field $\pmb{\EuScript{A}}$ decompose into KK fields $\mathbf{A}[p,q] $ and $\tilde{\mathbf{B}}[p,q]$. Since the coupling takes the same form as in the last section, it is straightforward to write down the first-order coupling of the KK fields. We have
\begin{equation}
	\begin{aligned}
		& \frac{1}{2}\int \Tr( \bm{\beta}[p + r,q + s]\mathbf{A}[p,q]\pa_{z}\mathbf{A}[r,s]  )  +  \frac{1}{2} \frac{r+s}{p+q+r+s} \int \Tr( \pa_z\bm{\beta}[p + r,q+ s]\mathbf{A}[p,q]\mathbf{A}[r,s]  ) \\
		&+ \frac{1}{2} (ps - qr)\int \Tr( \bm{\alpha}[r,s]\tilde{\mathbf{B}}[p+ r - 1,q+s - 1]\mathbf{A}[p,q]  ).
	\end{aligned}
\end{equation}
Similarly, we have the coupling with the $\tilde{\bm{\alpha}},\tilde{\bm{\beta}}$ fields
\begin{equation}
	\begin{aligned}
	& -\frac{1}{2}\frac{(ps - qr)}{p+q +r+s}\int \Tr(\tilde{\bm{\beta}}[p+r-1,q+s - 1]\mathbf{A}[r,s]\mathbf{A}[p,q]) \\
	&\frac{1}{2}(p+q)\int \Tr( \pa_z\tilde{\bm{\alpha}}[r,s]\tilde{\mathbf{B}}[p+ r ,q+s]\mathbf{A}[p,q]  ) - \frac{1}{2}(r+s +2) \int \Tr( \tilde{\bm{\alpha}}[r,s]\tilde{\mathbf{B}}[p+ r ,q+s]\pa_z\mathbf{A}[p,q]  ) .
\end{aligned}
\end{equation} 

Higher order couplings also exists and we leave them to future work.

	\section{Theories in deformed geometry}
	\label{sec:def}
	
	\subsection{The deformed geoemtry}
	In the twisted holography setup, we place a stack of $N$ topological B-branes wrapping the complex line $\C = \{z_i = 0\}$ in $\C^3$. This brane configuration sources non-trivial flux and induces gravitational backreaction. The backreaction of such a stack of branes is analyzed in \cite{Costello:2018zrm}. We briefly review the results here.
	
	Solving the equation of motion of Kodaira-Spencer gravity with a localized source on the brane gives us the following Beltrami differential 
	\begin{equation}\label{bel_bk}
		\beta_0 = \frac{2N}{8\pi^2} \frac{\bar{z}_1d\bar{z}_2 - \bar{z}_2d\bar{z}_1}{(|z_1|^2 + |z_2|^2)^2} \pa_z.
	\end{equation}

This Beltrami differential defines a new complex structure on $\C^3\backslash\C$. Holomorphic functions of this new complex structure satisfy 
\begin{equation}
	(\pab + \beta) F= 0,\;
\end{equation}
where $\beta$ act on $F$ as a vector field. Solving this equation, we obtained the new holomorphic coordinates, which are $z_1,z_2$ and
\begin{equation}
\begin{aligned}
	u_1 &= z_1z - \frac{N\bar{z}_2}{|z_1|^2 + |z_2|^2} ,\\
		u_2 &= z_2z + \frac{N\bar{z}_1}{|z_1|^2 + |z_2|^2} .\\
\end{aligned}
\end{equation}
These coordinates are not independent. They satisfy
\begin{equation}
	u_1z_1 - u_1z_2 = N.
\end{equation}
This is the equation defining the deformed conifold $X_N$ as a submanifold of $\C^4$. If we collect the coordinates into a single matrix
\begin{equation}
g = \begin{pmatrix}
	z_1&z_2\\u_1&u_2.
\end{pmatrix}
\end{equation}
Then the defining equation of the deformed conifold becomes
\begin{equation}
	\det g = N.
\end{equation}
Therefore the deformed geometry is also equivalent to $SL_2(\C)$.
	\subsection{Holomorphic Chern-Simons in deformed geometry, I}
	\label{sec:def_diff_HCS}
	In this section, we analyze the KK theory of holomorphic Chern-Simons in the deformed geometry. One way to understand the effect of the new geometry on the holomorphic Chern-Simons theory is through the open-closed coupling. Note that new complex structure of $SL_2(\C)$ is completely encoded in the Beltrami differential $\beta_0 = \frac{2N}{8\pi^2} \frac{\bar{z}_1d\bar{z}_2 - \bar{z}_2d\bar{z}_1}{(|z_1|^2 + |z_2|^2)^2} \pa_z$, which can be regarded as a background value of the dynamical closed string fields. The interaction of complex structure moduli with the open string field is described by the coupling introduced in Section \ref{sec:coup_OC}
	\begin{equation}
			\int \Omega_X \Tr(\frac{1}{2} \pmb{\EuScript{A} } \wedge \bm{\mu}^i(\pa_{z_i} \pmb{\EuScript{A} })).
	\end{equation}
	Therefore, working in the new geometry is equivalent to turning on a background Beltrami differential $\beta_0$ in the above open-closed coupling.
	
	Via the coordinate transform \ref{K} we rewrite the Beltrami differential $\beta_0$ as $N\epsilon\pa_z$, where we absorbed a constant factor which can be restored from a rescaling. We see that by turning on this background Beltrami differential, the open-closed coupling term gives us the following:
	\begin{equation}\label{back_action_1}
		N\int \Tr(  \bm{\mathcal{A}}\wedge \epsilon \pa_z \bm{\mathcal{A}}),
	\end{equation}
where $\bm{\mathcal{A}} \in \Omega_{3d}^{\sbullet}\otimes\Omega_{b}^{0,0}(S^3)\otimes \mathfrak{g}$ is the KK fields of holomorphic Chern-Simons theory before we pass to the $\pacr$ cohomology. This term describes the deformed action functional for holomorphic Chern-Simons theory in the $SL_2(\C)$ geometry.

	We would like to understand the corresponding deformation for the KK theory after we pass to the $\pacr$ cohomology. A naive answer is that we simply restrict the above term to the fields $\mathbf{A}$. This will introduce a term $N\int\Tr(\mathbf{A}[0,0]\epsilon \pa_z\mathbf{A}[0,0])$, which is a Chern-Simons deformation to the zero KK modes. This is compatible with the observation in \cite{Costello:2018zrm}. However, by integrating out massive fields, there are higher order corrections to this zero mode answer.  
	
	 We can see this by analyzing the complex of KK fields. Note that adding the deformation \ref{back_action_1} is equivalent to add the differential $N\epsilon\pa_z$ to the complex $(\Omega_{3d}^{\sbullet}\otimes \Omega_b^{0,\sbullet}(S^3), \hat{d}+ \pacr)$. Therefore, the new theory on deformed geometry is described by the following complex
	\begin{equation}\label{Dol_SL2_K}
		(\Omega_{3d}\otimes \Omega_b^{0,\sbullet}(S^3), \hat{d} + \pacr + N\epsilon\pa_z ).
	\end{equation}
	Due to the presence of the term $N\epsilon\pa_z$, new differential will be generated after we take the $\pacr$ cohomology of the above complex. It will be a laborious work to follow the standard procedure of spectral sequence to obtain the right answer. Here, we use the homological perturbation lemma (we provide more details of homological perturbation lemma in \ref{apx:Hom_trans}), which gives us the correct differential on $\Omega_{3d}^{\sbullet}\otimes H_{b}^{0,\sbullet}(S^3)$ to make it quasi-isomorphic to the above complex \ref{Dol_SL2_K}. We have 
	\begin{equation}
		(\Omega_{3d}^{\sbullet}\otimes H_{b}^{0,\sbullet}(S^3),\hat{d} + D),
	\end{equation}
	where the differential $D$ is obtained by homological perturbation lemma, given by the following formula
	\begin{equation}
		D = (1 -  N \epsilon \pa_z h)^{-1}\epsilon \pa_z =  N \epsilon \pa_z +  N^2 \epsilon \pa_z h \epsilon \pa_z + \dots
	\end{equation}
	Using the degree $0$ operator $\tilde{h}$ defined in \ref{def_tildh}, the above expression can be simplified to 
	\begin{equation}
		D = \sum_{k \geq 1} N^k\epsilon \tilde{h}^{k - 1} \pa_z^{k}.
	\end{equation}
	Given this new differential, we find the following deformed action functional
	\begin{equation}
		I^{(1) }_O = \int \frac{1}{2}\Tr(\mathbf{A}D\mathbf{A}) = \sum_{k\geq 1} I^{(1,k)}_O,
	\end{equation}
	where
	\begin{equation}
		I^{(1,k)}_O = \frac{1}{2}  N^{k}\int \Tr(\mathbf{A}\epsilon \tilde{h}^{k- 1} \pa_z^{k}  \mathbf{A}). 
	\end{equation}
	
	We can simplify the above expression by analyzing the operator $\tilde{h}^k$. Note that the action of $\tilde{h}^k$ takes the following form
	\begin{equation}
		\tilde{h}^k: \begin{array}{l}
			\mathcal{H}_{j,\bar{j}} \to 0,\\
			\mathcal{H}_{j,\bar{j}} \overset{\simeq}{\to} \mathcal{H}_{j- \frac{k}{2},\bar{j} + \frac{k}{2}},
		\end{array}
		\begin{array}{l}
			\text{ for } j < \frac{k}{2},\\
			\text{ for } j \geq \frac{k}{2}.
		\end{array}
	\end{equation}
	When non-zero, $\tilde{h}^k$ acts as the following constant
	\begin{equation}
		\label{const_hk}
		\tilde{h}^k|_{\mathcal{H}_{j,\bar{j}}}  = \sqrt{\frac{(2j - k)!(2\bar{j})!}{(2j)!(2\bar{j} + k)!}}.
	\end{equation}
	We are interested in the action of $\tilde{h}^k$ on $\mathcal{H}_{k, 0}$. It is easy to find that
	\begin{equation}\label{hk_k}
		\tilde{h}^k(e^{(k)}_i) = \frac{1}{k!}\bar{e}^{(k)}_{i}.
	\end{equation}
	Transform this result into the unnormalized basis $\{w_1^pw_2^q\}$, we obtain
	\begin{equation}\label{O1k}
		I^{(1,k)}_O =  \sum_{p+q = k -1}N^k \frac{(-1)^p}{(k - 1)!} \frac{p!q!}{k!}\frac{1}{2} \int \Tr(\mathbf{A}[p,q] dz\pa_z^{k} \mathbf{A}[q,p]) .
	\end{equation}
These interaction terms describe part of the deformation of the KK theory of holomorphic Chern-Simons in the deformed geometry.

\subsection{An alternative approach}

In this section, we briefly comment on performing the KK reduction directly on the deformed conifold. The deformed conifold has a $SU(2)_R$ symmetry by matrix multiplication from the right. The quotient 
\begin{equation}
	SL_2(\C) \to SL_2(\C)/SU(2)_R
\end{equation}
is precisely the hyperbolic space -- the Euclidean version of $AdS_3$.

The quotient space is parameterized by coordinates $\rho = gg^{\dagger}$. We have
\begin{equation}
	\rho = \begin{pmatrix}
		r^2 & \bar{z}r^2\\
		zr^2&\frac{N^2}{16\pi^4}\frac{1}{r^2} + |z|^2r^2
	\end{pmatrix},
\end{equation}
where $r = (z_1\bar{z}_1 + z_2\bar{z}_2)^{\frac{1}{2}}$. Therefore, we can also use the coordinates $(z,\bar{z},r)$ for the quotient. 

KK reduction of field theory on $SL_2(\C)$ amounts to analyzing the Dolbeault complex
\begin{equation}\label{Dol_SL2}
	(\Omega^{0,\sbullet}(SL_2(\C)) ,\pab).
\end{equation}
	Recall that $SL_2(\C)$ is obtained as a deformation of $\C^{3}\backslash\C$. The spaces of functions (not necessarily holomorphic) on them and the decomposition of the Dolbeault complexes under the $SU(2)_R$ action are the same. The only difference is their complex structure, encoded in the Dolbeault differential $\pab$. Note that the new complex structure of $SL_2(\C)$ can be obtained by adding the Beltrami differential \ref{bel_bk} to the Dolbeault differential $\pab$ of $\C^3\backslash\C$. Therefore, the Dolbeault complex on $SL_2(\C)$ is equivalent to the complex obtained by adding the Beltrami differential \ref{bel_bk} to the Dolbeault complex \ref{Dol_C3} of $\C^3\backslash\C$. This is precisely the complex \ref{Dol_SL2_K} we find in the last section.
	
	In summary, by performing KK reduction directly on $SL_2(\C)$, we will end up with the same answer as in the last section using open-closed coupling. This is expected given the self consistency of the open-closed string theory.
	\subsection{Holomorphic Chern-Simons in deformed geometry, II}
	\label{sec:def_pro_HCS}
	In the last section, we discussed how the quadratic terms of the KK theory are modified in the deformed geometry. These terms correspond to a deformed differential of the complex of KK fields. In fact, all the products and higher products constructed in Section \ref{sec:all_mn} that build the action functional of KK theory are deformed in the new geometry. For example, we commented on the higher order open-closed coupling of the KK theory in Section \ref{sec:coup_OC}. By letting $\bm{\beta}[0,0] = N \epsilon \pa_z$ in the formula \ref{coup_OC_3}, we find a deformed action of the form $\int BA\pa_zA$.
	
	To correctly obtain all deformation, we need to use all the data in the special deformation retract (SDR). Recall from Section \ref{sec:SDR} that the SDR  we used to build the product and higher product on CR cohomology is given by the following 
		\begin{equation}
		h\curved (\Omega^{\sbullet}_{3d} \otimes \Omega_b^{0,\sbullet}(S^3), \hat{d} + \pacr )\overset{p}{\underset{i}\rightleftarrows} (\Omega_{3d}\otimes H_{b}^{0,\sbullet}(S^3),\hat{d}).
	\end{equation}
	Using homological perturbation theory, we obtain the new SDR after adding $N\epsilon\pa_z$ to the differential
	\begin{equation}
		h'\curved (\Omega^{\sbullet}_{3d} \otimes \Omega_b^{0,\sbullet}(S^3), \hat{d} + \pacr + N \epsilon\pa_z )\overset{p'}{\underset{i'}\rightleftarrows} (\Omega_{3d}\otimes H_{b}^{0,\sbullet}(S^3),\hat{d} + D).
	\end{equation}
The differential $D$ is given in the previous section. The other operations $h',p',i'$ are given by
	\begin{equation}
		\begin{aligned}
			h' & = \sum_{k \geq 0}(N^k \tilde{h}^k \pa_z^k)h,\\
			p' & = p(\sum_{k \geq 0}N^k \tilde{h}^k \pa_z^k),\\
			i' & =  \sum_{k \geq 0}(N^k \tilde{h}^k \pa_z^k)i.
		\end{aligned}
	\end{equation}
	The new $A_\infty$ structure on $\Omega_{3d}\otimes H_{b}^{0,\sbullet}(S^3)$ is given by the homotopy transfer theorem. They can be computed by the same method as in Section \ref{sec:A_trans}, but now with the new maps $i',p',h'$.
	
	For instance, the deformed product $m_2'$ is defined by $m_2'(-,-) = p'\circ M(i'(-),i'(-))$. For the purpose of obtaining the corresponding deformed action, it suffices to compute it on $(\Omega^{\sbullet}_{3d}\otimes H^{0,0}_{CR}(S^3))^{\otimes 2} $. In this case, $m_2'$ is simplified to
	\begin{equation}
		m_2'(f_1Y_1,f_2Y_2) = \sum_{k_1,k_2\geq 0} N^{k_1+k_2}(\pa_z^{k_1}f_1)(\pa_z^{k_2}f_2)p M(\tilde{h}^{k_1}Y_1,\tilde{h}^{k_2}Y_2),\quad f_i \in \Omega^{\sbullet}_{3d},Y_i \in H^{0,0}_{CR}(S^3).
	\end{equation}
	We denote the map $m_2^{(k_1,k_2)}(Y_1,Y_2)= p M(\tilde{h}^{k_1}Y_1,\tilde{h}^{k_2}Y_2)$, then 
	\begin{equation}
		m_2'(f_1Y_1,f_2Y_2) = \sum_{k_1,k_2\geq 0} N^{k_1+k_2}(\pa_z^{k_1}f_1)(\pa_z^{k_2}f_2)m_2^{(k_1,k_2)}(Y_1,Y_2).
	\end{equation}
	We have that $m_2^{(0,0)} = m_2$ by definition. We illustrate the computation of $m_2^{(1,0)}$ and  $m_2^{(0,1)}$. To compute $m_2^{(1,0)}(w_1^pw_2^{q},w_1^rw_2^{s})$, we first compute $\tilde{h}(w_1^pw_2^{q})$. Using the definition \ref{def_tildh} of $\tilde{h}$ and expand the harmonic polynomial into polynomials of $w_i,\bar{w}_i$, we find
	\begin{equation}
		\tilde{h}(w_1^pw_2^{q}) = \frac{1}{p+q}(pw_1^{p-1}w_2^{q}\bar{w}_2 - qw_1^pw_2^{q- 1}\bar{w}_1).
	\end{equation}
	Therefore,
	\begin{equation}
		m_2^{(1,0)}(w_1^pw_2^{q},w_1^rw_2^{s})  = \frac{ps - qr}{(p+q)(p+q+r+s )} w_1^{p+r - 1}w_2^{q+s - 1}.
	\end{equation}
Similarly, $m_2^{(0,1)}(w_1^pw_2^{q},w_1^rw_2^{s}) = -\frac{ps - qr}{(r+s)(p+q+r+s )} w_1^{p+r - 1}w_2^{q+s - 1}$. Using these deformed products, we obtain the first-order deformed action functional
\begin{equation}\label{defO_2_1}
\begin{aligned}
		I^{(2,1)}_O = &N\sum_{p,q,r,s} \frac{ps - qr}{(p+q+r+s )} \left( \frac{1}{p+q}\int \Tr \mathbf{B}[p+r - 1,q+s - 1][\pa_z\mathbf{A}[p,q], \mathbf{A}[r,s]] \right.\\
		&- \left.\frac{1}{r+s}\int \Tr \mathbf{B}[p+r - 1,q+s - 1][\mathbf{A}[p,q], \pa_z\mathbf{A}[r,s]]\right). \\
\end{aligned}
\end{equation} 

In order to compute all other $m_2^{(k_1,k_2)}$, one can use our formula \ref{pro_Har_arb} for the product of two arbitrary $S^3$ harmonics. Here, we provide a different method. Note that $m_2^{(k_1,k_2)}$ is compatible with the $SU(2)$ action and send $\mathcal{H}_{j_1,0}\otimes \mathcal{H}_{j_2,0}$ to $\mathcal{H}_{j_1+j_2-k_1-k_2,0}$. Therefore the corresponding matrix elements must be proportional to the Clebsch-Gordan coefficients in the orthonormal basis. To determine the constant of proportionality it suffices to compute one non-zero value of $m_2^{(k_1,k_2)}$. First, using \ref{const_hk} we find that $\tilde{h}^k(e^{(j)}_m) = \sqrt{\frac{(2j - k)!}{(2j)!k!}}e^{(j - \frac{k}{2},\frac{k}{2})}_m$. The easiest case of a non-zero product is the product between the highest weight vector $e^{(j - \frac{k}{2},\frac{k}{2})}_j$ and the lowest weight vector. We have
\begin{equation}
\begin{aligned}
		&m_2^{(k_1,k_2)}(e^{(j_1)}_{j_1},e^{(j_2)}_{-j_2}) = \sqrt{\frac{(2j_1 - k_1)!(2j_2 - k_2)!}{k_1!k_2!(2j_1)!(2j_2)!}}m_2(e^{(j_1 - \frac{k_1}{2},\frac{k_1}{2})}_{j_1},e^{(j_2 - \frac{k_2}{2},\frac{k_2}{2})}_{-j_2})\\
		& =(-1)^{k_2}\frac{\sqrt{(2j_1 + 1)(2j_2 + 1)}(2j_1 - k_1)!(2j_2 - k_2)!}{k_1!k_2!(2j_1+2j_2 - k_1-k_2+1)!}\sqrt{\frac{(2j_1 + 2j_2 - 2k_1 - 2k_2 +1)!}{(2j_1 - k_1 - k_2)!(2j_2 - k_1 - k_2)!}}e^{(j_1 + j_2 - k_1 - k_2)}_{j_1 - j_2}.
\end{aligned}
\end{equation}
This implies
 \begin{equation}
\begin{aligned}
	 	m_2^{(k_1,k_2)}(e^{(j_1)}_{m_1},e^{(j_2)}_{m_2}) = &(-1)^{k_2}\sqrt{\frac{(2j_1 + 1)(2j_2 + 1)(2j_1 + 2j_2 - 2k_1 - 2k_2)!}{(2j_1+2j_2 - k_1-k_2+1)!(2j_1)!(2j_2)!(2j_1 - k_1 - k_2)!(2j_2 - k_1 - k_2)!}}\\
	 	&\times \frac{(2j_1 - k_1)!(2j_2 - k_2)!}{k_1!k_2!}C^{j_1,j_2,j_1+j_2 - k_1 - k_2}_{m_1,m_2,m_1+m_2}e^{(j_1 + j_2 - k_1 - k_2)}_{m_1+m_2}.
\end{aligned}
 \end{equation}
Transforming this result into the unnormalized basis $\{w_1^pw_2^q\}$ (see formula (2.18) in \cite{Pope:1989sr}), we have
\begin{equation}
	m_2^{(k_1,k_2)}(w_1^pw_2^{q},w_1^{r}w_2^{s}) = \frac{(-1)^{k_2}R_k(p,q,r,s)w_1^{p+r - k}w_2^{q+s - k}}{k_1!k_2![p+q]_{k_1}[r+s]_{k_2}[p+q+r+s-k+1]_{k}}, \quad k = k_1+k_2,
\end{equation}
where we used the descending Pochhammer symbol $[a]_n = \frac{a!}{(a - n)!}$. The constant $R_k(p,q,r,s)$ \footnote{This constant also appears in the definition of $W_\infty$ algebra \cite{Pope:1989sr}, and we are following the notation of \cite{Bittleston:2022jeq} } is defined by
\begin{equation}\label{const_R}
	R_k(p,q,r,s) : = \sum_{i}(-1)^i\binom{k}{i}[p]_{k-i}[q]_i[r]_i[s]_{k-i}.
\end{equation}

Using this deformed product, we can write down the deformed cubic action of all orders:
	\begin{equation}
\begin{aligned}
	I^{(2,k)}_O = &N^k\sum_{p,q,r,s}\sum_{k_1+k_2 = k}\frac{(-1)^{k_2}R_k(p,q,r,s)}{k_1!k_2![p+q]_{k_1}[r+s]_{k_2}[p+q+r+s-k+1]_{k}} \\
	&\times\int \Tr \mathbf{B}[p+r - k,q+s - k][\pa_z^{k_1}\mathbf{A}[p,q], \pa_z^{k_2}\mathbf{A}[r,s]].\\
\end{aligned}
	\end{equation}
We can check that by taking $k = 1$ in the above formula, we reproduce the first order deformed action $I^{(2,1)}_O$ in \ref{defO_2_1}.

As we have mentioned, the whole $A_\infty$ algebra structure on $\Omega_{3d}^{\sbullet}\otimes H_b^{0,\sbullet}(S^3)$ is deformed in the new geometry. They can be computed by the same method presented in this paper. Understanding the deformed $A_\infty$ structure allows us to have a complete expansion of the action functional in the deformed $SL_2(\C)$ geometry. This can, in principle, give us the full tree level boundary chiral algebra.

\subsection{Kodaira-Spencer gravity in deformed geometry}
	Analogous to the holomorphic Chern-Simons theory, we can also understand Kodaira-Spencer gravity in deformed geometry as a deformation of the theory on flat space. There are two differences in the analysis. Firstly, we add the following differential to the complex:
	\begin{equation}
		\{N\epsilon \pa_z,-\}.
	\end{equation}
Secondly, we use the homotopy operator $\mathds{h}$ defined in \ref{homotopy_KS} in the homological perturbation. After passing to the CR cohomology, the deformed differential is given by
\begin{equation}
\begin{aligned}
		\mathds{D} &= (1 - \{N\epsilon \pa_z,-\}\circ \mathds{h})^{-1}\circ \{N\epsilon \pa_z,-\}\\
		&= \sum_{k \geq 1} \mathds{D}^{(k)},
\end{aligned}
\end{equation}
where
\begin{equation}
	  \mathds{D}^{(k)} = (\{N\epsilon \pa_z,-\} \circ \mathds{h} )^{k - 1}\circ \{N\epsilon \pa_z,-\} .
\end{equation}
This deformed differential corresponds to the following deformed quadratic action
\begin{equation}
 	\int \pa^{-1}\bm{\alpha}\mathds{D}\bm{\alpha}.
\end{equation}
To compute this deformed action, we need to compute the action of $\mathds{D}$ on the fields. We have
\begin{equation}
	\begin{aligned}
		\mathds{D}^{(k)} (f(z)Y(w_i,\bar{w}_i)\pa_z ) &= N^k\pa_z^k(f)\tilde{h}^{k-1}(Y)\epsilon\pa_z , \\
		\mathds{D}^{(k)} (f(z)Y(w_i,\bar{w}_i)\hat{\eta} )&= N^k\pa_z^k(f)\tilde{h}^{k-1}(Y)\epsilon\hat{\eta} - 2kN^k\pa_z^{k-1}(f)\tilde{h}^{k-1}(Y)\epsilon\pa_z,\\ 
		 \mathds{D}^{(k)} (f(z)Y(w_i,\bar{w}_i)\eta)&= N^k\pa_z^k(f)\tilde{h}^{k-1}(Y)\epsilon\eta - (k - 1)N^k\pa^k_z(f)\tilde{h}^k(Y)\epsilon\hat{\eta}\\
		 & + k(k - 1)N^k\pa_z^{k-1}(f)\tilde{h}^k(Y)\pa_z.
	\end{aligned}
\end{equation}
Thus we can compute
\begin{equation}\label{Dk_alpha}
\begin{aligned}
		\mathds{D}^{(k)}(\bm{\alpha}) = \sum_{p,q}N^k(p+q)& \left( k\pa^k_z\bm{\alpha}[p,q]\tilde{h}^k(w_1^pw_2^q)\epsilon\hat{\eta} -  \pa^k_z\bm{\alpha}[p,q]\tilde{h}^{k-1}(w_1^pw_2^q)\epsilon\eta \right. \\
		&\left. -  k(k+1)\pa^{k-1}_z\bm{\alpha}[p,q]\tilde{h}^k(w_1^pw_2^q)\epsilon\pa_z\right) .
\end{aligned}
\end{equation}
Using the above formula, we can expand the deformed action $\int \Tr\pa^{-1}\bm{\alpha}\mathds{D}\bm{\alpha}$. The constant coefficient in front of the term $\int \bm{\alpha}[q,p] dz\pa_z^{k - 1} \bm{\alpha}[p,q]$ is computed by $(p+q)k(k+1)\Tr_{S^3}(w_1^qw_2^p\tilde{h}^k(w_1^pw_2^q))$. In fact, this is already computed in \ref{hk_k},\ref{O1k}. Therefore, we have the following deformed quadratic term
\begin{equation}
	I^{(1,k)}_C =  \sum_{p+q = k}N^k (-1)^q \frac{p!q!}{((k - 1)!)^2}\frac{1}{2} \int \bm{\alpha}[q,p] dz\pa_z^{k - 1} \bm{\alpha}[p,q].
\end{equation}

Using the same method as we analyze the deformed cubic action for holomorphic Chern-Simons, we obtain the following deformed cubic action for Kodaira-Spencer gravity
\begin{equation}
		I^{(2,k)}_C = \sum_{k_1+k_2 = k}\int \Tr \bm{\beta}(\mathds{h}\circ\mathds{D}^{(k_1)}  (\bm{\alpha}))(\mathds{h}\circ\mathds{D}^{(k_2)}  (\bm{\alpha})),
\end{equation}
where we assumed that $\mathds{h}\circ\mathds{D}^{(k)} = \mathrm{Id}$ for $k = 0$. Using the expansion \ref{Dk_alpha} of $ \mathds{D}^{(k)}(\bm{\alpha}) $, we have
\begin{equation}\label{expan_hda}
	\begin{aligned}
		\mathds{h}\circ \mathds{D}^{(k)}(\bm{\alpha}) = \sum_{p,q}N^k(p+q)& \left( (k+1)\pa^k_z\bm{\alpha}[p,q]\tilde{h}^{k+1}(w_1^pw_2^q)\hat{\eta} -  \pa^k_z\bm{\alpha}[p,q]\tilde{h}^{k}(w_1^pw_2^q)\eta \right. \\
		&\left. -  k(k+1)\pa^{k-1}_z\bm{\alpha}[p,q]\tilde{h}^{k+1}(w_1^pw_2^q)\pa_z\right) .
	\end{aligned}
\end{equation}
We can expand the deformed cubic action using the above formula. The constant coefficient in front of the term $\int\bm{\beta}[p+r - k -1,q+s - k - 1]\bm{\alpha}[p,q]\bm{\alpha}[r,s]$ can be computed using the deformed product $m_2^{(k_1+1,k_2)}(w_1^pw_2^q,w_1^rw_2^s)$ and $m_2^{(k_1,k_2+1)}(w_1^pw_2^q,w_1^rw_2^s)$ obtained in the last section. We have the following
\begin{equation}\label{coup_C_2k}
\begin{aligned}
		I^{(2,k)}_C = &\sum_{k_1+k_2 = k}\frac{(-1)^{k_2}(p+q)(r+s)R_{k+1}(p,q,r,s)}{k_1!k_2![p+q+r+s-k]_{k+1}}\left( \frac{1}{[p+q]_{k_1 + 1}[r+s]_{k_2}}+ \frac{1}{[p+q]_{k_1 }[r+s]_{k_2+1}}\right) \\
		&\times \int\bm{\beta}[p+r - k -1,q+s - k - 1]\pa_z^{k_1}\bm{\alpha}[p,q]\pa_z^{k_2}\bm{\alpha}[r,s].
\end{aligned}
\end{equation}

\subsection{Open-closed coupling in deformed geometry}
The open-closed coupling terms studied in Section \ref{sec:coup_OC} are also deformed in the $SL_2(\C)$ background. To understand the deformation, we first understand the flat space coupling as coming from a bracket of the complex $\ker \pa \cap  \mathrm{PV}^{1,\sbullet}(X) \oplus \Omega^{0,\sbullet}(X)$
\begin{equation}
\begin{aligned}
	\{-,-\}: \mathrm{PV}^{1,\sbullet}(X) \otimes \Omega^{0,\sbullet}(X) &\to  \Omega^{0,\sbullet}(X), \\
	 \bm{\mu}^i\pa_{z_i} \otimes \bm{\gamma} &\mapsto \bm{\mu}^i\pa_{z_i}(\bm{\gamma}) .
\end{aligned}
\end{equation}

\begin{remark}
	Only the leading order coupling of the $\mathrm{PV}^{1,\sbullet}(X) $ part of Kodaira-Spencer fields and the holomorphic Chern-Simons can be understood as a $2$-bracket. Other parts of the coupling correspond to $L_\infty$ maps on the space of fields \cite{Costello:2015xsa}. 
\end{remark}

Given this bracket that encodes (part of) the open-closed coupling, we can use the same techniques as before to compute the deformed action. For example, we have the following deformed coupling terms with fields $\mathbf{B},\mathbf{A},\bm{\alpha}$
\begin{equation}
	I_{OC}^{(2,k)}[\mathbf{B},\mathbf{A},\bm{\alpha}] = \sum_{k_1+k_2 = k} \int \Tr \mathbf{B}\left( \mathds{h}\circ\mathds{D}^{k_1}(\bm{\alpha}) (N^{k_2}\tilde{h}^{k_2}\pa_z^{k_2} \mathbf{A})\right) .
\end{equation}
Using the expansion \ref{expan_hda} for $\mathds{h}\circ\mathds{D}^{(k_1)}(\bm{\alpha}) $ and the deformed product  $m_2^{(k_1+1,k_2)}(w_1^pw_2^q,w_1^rw_2^s)$ and $m_2^{(k_1,k_2+1)}(w_1^pw_2^q,w_1^rw_2^s)$ obtained in the previous section, we can easily expand the above deformed action as follows
\begin{equation}
	\begin{aligned}
		I_{OC}^{(2,k)}[\mathbf{B},\mathbf{A},\bm{\alpha}] = &\sum_{k_1+k_2 = k}\left(\frac{r+s - 2k_2}{p+q - k_1}+1\right)\frac{(p+q)(-1)^{k_2}R_{k+1}(p,q,r,s)}{k_1!k_2![p+q]_{k_1}[r+s]_{k_2}[p+q+r+s - k]_{k+1}} \\
		&\times \int\Tr \mathbf{B}[p+r - k-1,q+s-k-1]\pa_z^{k_1}\bm{\alpha}[p,q]\pa_z^{k_2}\mathbf{A}[r,s]\\
		&+ \sum_{\substack{k_1+k_2 = k\\k_1\geq1,k_2 \leq k -1}}\frac{(p+q)(-1)^{k_2}R_{k+1}(p,q,r,s)}{(k_1-1)!k_2![p+q]_{k_1+1}[r+s]_{k_2}[p+q+r+s - k]_{k+1}} \\
		&\times \int\Tr \mathbf{B}[p+r - k-1,q+s-k-1]\pa_z^{k_1 - 1}\bm{\alpha}[p,q]\pa_z^{k_2+1}\mathbf{A}[r,s].
	\end{aligned}
\end{equation}
By a change of summation variable: $k_1 \rightarrow k_1 + 1,k_2\rightarrow k_2- 1$, the second term in the above formula can be rewritten as 
\begin{equation}
	\frac{-k_2(r+s - k_2 + 1)}{(p+q-k_1)(p+q - k_1 - 1)}\frac{(p+q)(-1)^{k_2}R_{k+1}(p,q,r,s)}{k_1!k_2![p+q]_{k_1}[r+s]_{k_2}[p+q+r+s - k]_{k+1}} \int\Tr \mathbf{B}[...]\pa_z^{k_1 }\bm{\alpha}[p,q]\pa_z^{k_2}\mathbf{A}[r,s].
\end{equation}
Since the coefficient in the above formula vanishes for $k_2 = 0$, the summation range can be taken to be all $k_1+k_2 = k$. Thus we can rewrite the deformed coupling as follows
\begin{equation}\label{coup_OC_2k1}
	\begin{aligned}
		I_{OC}^{(2,k)}[\mathbf{B},\mathbf{A},\bm{\alpha}] = &\sum_{k_1+k_2 = k}\frac{(p+q-k-1)}{(p+q - 1)}\frac{(p+q)(-1)^{k_2}R_{k+1}(p,q,r,s)}{k_1!k_2![p+q - 2]_{k_1}[r+s]_{k_2}[p+q+r+s - k - 1]_{k}} \\
		&\times \int\Tr \mathbf{B}[p+r - k-1,q+s-k-1]\pa_z^{k_1}\bm{\alpha}[p,q]\pa_z^{k_2}\mathbf{A}[r,s].\\
	\end{aligned}
\end{equation}
Similarly, we can write down the deformed coupling between the fields $\bm{\beta},\mathbf{A}$
\begin{equation}\label{coup_OC_2k2}
\begin{aligned}
		I_{OC}^{(2,k)}[\mathbf{A},\mathbf{A},\bm{\beta}] = &\sum_{k_1+k_2 = k} \frac{(-1)^{k_2}R_k(p,q,r,s)}{k_!!k_2![p+q]_{k_1}[r+s]_{k_2}[p+q+r+s-k+1]_k}\\
		&\times \left(  \int \bm{\beta}[p+r-k,q+s-k]\pa_z^{k_1}\mathbf{A}[p,q]\pa_z^{k_2+1}\mathbf{A}[r,s] \right.\\
		& \left.+ \frac{r+s - 2k_2}{p+q+r+s - 2k}\int \pa_z\bm{\beta}[p+r-k,q+s-k]\pa_z^{k_1}\mathbf{A}[p,q]\pa_z^{k_2}\mathbf{A}[r,s] \right).
\end{aligned}
\end{equation}
There are also coupling between the fields $\mathbf{A},\mathbf{B}$ and the other parts of the Kodaira-Spencer fields. We leave them to future analysis.

\subsection{Moyal-Weyl product and higher spin}
In our previous analysis, we observe that the constant $R_k(p,q,r,s)$ \ref{const_R} plays an important role in the construction of the cubic action for various KK theories in the deformed geometry. In this section, we briefly recall its connection with higher spin theory.

First, we recall the definition of the Moyal-Weyl product, which is a non-commutative deformation of the commutative product of functions on $\R^{2n}$ given by a constant Poisson tensor. In our case, we focus on the Moyal-Weyl deformation of $\C[w_1,w_2]$ equipped with the Poisson tensor $\epsilon_{ij}\pa_i\pa_j$. The Moyal-Weyl product $\star$ is given by 
\begin{equation}
	f\star g : = m(e^{\frac{\hbar}{2}\epsilon_{ij}\pa_i\otimes\pa_j}f\otimes g),\quad f, g \in \C[w_1,w_2],
\end{equation}
where $m$ is the commutative product map $m(a\otimes b) = ab$, and the exponential is treated as a power series. Expanding the exponential to the first order, one can find
\begin{equation}
	(w_1^pw_2^q)\star (w_1^rw_2^s) = w_1^{p+r}w_2^{q+s} + \frac{\hbar}{2}(ps - qr)w_1^{p+r - 1}w_2^{q+s - 1} + \mathcal{O}(\hbar^2).
\end{equation}
It is straightforward to expand the exponential to all orders, and we find that the Moyal-Weyl product is given by 
\begin{equation}
		(w_1^lw_2^m)\star (w_1^rw_2^s) = \sum_{k}(\frac{\hbar}{2})^k\frac{1}{k!}R_k(p,q,r,s)w_1^{p+r - k}w_2^{q+s - k}.
\end{equation}

In fact, The Moyal-Weyl product and the constant $R_k(p,q,r,s)$ appear in many higher spin related contexts. The global symmetry algebra of higher spin theory- the higher spin algebra $\mathfrak{hs}$ can be defined using the Moyal-Weyl algebra. Vasiliev's formulation of higher spin equation is constructed based on the Moyal-Weyl algebra, which controls its first-order interaction. The constant $R_k(p,q,r,s)$ also appears in the definition of $\mathcal{W}_\infty$ algebra, which is the chiral algebra of higher spin current \cite{Pope:1989ew}. 

Indeed, our KK theory obtained by compactification is naturally a higher spin theory in three dimension. When we perform the compactification from $SL_2(\C)$, the KK theory can be understood as a higher spin theory on $SL_2(\C)/SU(2) $, which is precisely the Euclidean $AdS_3$. Therefore, it is not a coincidence that the same constant $R_k(p,q,r,s)$ appears in the cubic interactions of the KK theory.

The KK theory we obtained contains an infinite series of effective interactions beyond cubic interactions. This is the same for the Vasiliev's theory of higher spin. Though Vasiliev's higher spin theory does not have an action functional, the $L_\infty$ algebra controlling the equation of motion is nontrivial and contains an infinite tower of higher brackets. The $2$-bracket is, of course, defined by the Moyal bracket. The $3$-bracket is first constructed in \cite{Vasiliev:1988sa} and is observed in \cite{Sharapov:2017yde} to correspond to the Feigin–Felder–Shoikhet cocycle \cite{feigin2005hochschild} of Weyl algebra. This $3$-bracket can also be represented as a specific correlator in a topological quantum mechanics of a free particle on the circle \cite{Li:2018rnc}. It will be an interesting problem to see if the quadratic vertex of our KK theory given by a $3$-bracket on the fields can have a similar interpretation.

	\section{Twisted holography from boundary chiral algebra}
	\label{sec:twist_hol}
	In this section, we revisit a model of twisted holography proposed in \cite{Costello:2018zrm}. This model captures a protected subsector of the canonical example of $AdS_5/CFT_4$ duality \cite{Maldacena:1997re,Witten:1998qj}. The “gravity” side of the duality is the string field theory of $B$-model topological string on a $SL_2(\C)$ background. The “gauge theory” side of the duality involves a large $N$ chiral algebra of gauged $\beta\gamma$ systems. In this paper, we present nontrivial results in matching the chiral algebras from both sides of the duality. 
	
	\subsection{Chiral algebra on topological D-branes}
	In this section, we introduce the large $N$ chiral algebra studied in \cite{Costello:2018zrm}. This is the world volume theory on a stack of $N$ topological B-branes wrapping $\C \subset \C^3$ with the inclusion of $K|K$ space-filling branes. This chiral algebra can also be regarded as coming from the $4d$ $\mathcal{N} = 2$ superconformal field theory through the procedure of \cite{Beem:2013sza}. In one word, this is the "twisted" version of the gauge theory side of the AdS/CFT correspondence.
	
	The chiral algebra involves of the following system
	\begin{itemize}
		\item A $bc$ system $(b(z),c(z))$ valued in $\mathfrak{gl}_N$: 
		\begin{equation}
			b(z)c(w) \sim \frac{1}{z - w}.
		\end{equation}
		\item  a $\beta\gamma$ system (or sometime called symplectic boson) $(Z^1(z),Z^2(z))$ valued in the adjoint $\mathfrak{gl}_N$:
		\begin{equation}
			Z^1(z)Z^2(w) \sim \frac{1}{z - w}.
		\end{equation}
		\item a $\beta\gamma$ system $(I(z),J(z))$ valued in the fundamental representation $\C^{K|K}\otimes \C^N$:
		\begin{equation}
			I(z)J(w) \sim \frac{1}{z - w}
		\end{equation}.
	\end{itemize}
The BRST operator takes the form
\begin{equation}
	Q_{BRST} = \oint \Tr cZ^1Z^2  + \Tr IcJ + \frac{1}{2}\Tr bc^2.
\end{equation}

	In the large $N$ limit, the BRST reduction is exactly computed in \cite{Costello:2018zrm}. The BRST cohomology comprises polynomials of the following single-trace operators. 
	\begin{equation}
		\begin{aligned}
			A^{(n)} &= \Tr(Z^{(i_1}Z^{i_1}\dots Z^{i_n )}),\\
			B^{(n)} &= \Tr(bZ^{(i_1}Z^{i_1}\dots Z^{i_n)}),\\
			C^{(n)} & = \Tr(\pa c Z^{(i_1}Z^{i_1}\dots Z^{i_n)}),\\
			D^{(n)} & = \epsilon_{ij}\Tr(\pa Z^{i} Z^{(j}Z^{i_1}\dots Z^{i_n)}) + \dots ,\\
			E^{(n)}_{\mathfrak{t}} & = \Tr (\mathfrak{t}IZ^{(i_1}Z^{i_1}\dots Z^{i_n )}J) .
		\end{aligned}
	\end{equation}
	
	In the twisted holography setup, the $A,B,C,D$ towers of operators correspond to the boundary operators of the Kodaira-Spencer theory, and the $E$ tower of operators correspond to the boundary operators of the holomorphic Chern-Simons. Following our form of the twisted holography conjecture, we will see that this large $N$ chiral algebra can be matched with the boundary chiral algebra of the corresponding KK theory with a proper choice of boundary condition.
	
	\subsection{Boundary algebras from holomorphic Chern-Simons}
	The “gravity” side of the duality is given by the Kodaira-Spencer gravity coupled with holomorphic Chern-Simons theory in the deformed geometry $SL_2(\C)$. First, we focus on the holomorphic Chern-Simons part.
	
	Recall that our $3d$ fields have mode expansion
	\begin{equation}
\begin{aligned}
			\mathbf{A}(t,z,\bar{z}) & = \sum_{p,q\geq 0}\mathbf{A}[p,q](t,z,\bar{z})w_1^pw_2^q,\\
	\mathbf{B}(t,z,\bar{z}) &= \sum_{p,q\geq 0}\frac{(p+q+1)!}{p!q!}\mathbf{B}[p,q](t,z,\bar{z})\bar{w}_1^p\bar{w}_2^q\epsilon.
\end{aligned}
	\end{equation}
	
They are related to the original $6d$ fields as follows
	\begin{equation}
			\pmb{\EuScript{A}} =  \sum_{p,q\geq0} \mathbf{A}[p,q](r, z,\bar{z})z_1^pz_2^q + \frac{(p+q+1)!}{p!q!}\mathbf{B}[p,q](r,z,\bar{z})\frac{\bar{z}_1^p\bar{z}_2^q}{r^{2(p+q)}} \frac{\epsilon_{ij}\bar{z}_id\bar{z}_j}{r^4}.
	\end{equation}
where $r = (z_1\bar{z}_1+ z_2\bar{z}_2)^{\frac{1}{2}}$. We can see that the natural boundary condition we get by requiring the original $6d$ fields to extend to $r = 0$ is to impose the following condition
\begin{equation}
\mathbf{B}[p,q]|_{t = 0} = 0 \text{ for all }p,q.
\end{equation}
This is the Neumann boundary condition for the $3d$ KK theory. The corresponding boundary algebra should reproduce the chiral algebra of the original $6d$ theory restricted along $\C$. We are more interested in the transverse boundary condition. This is the Dirichlet boundary condition given by
\begin{equation}
	\mathbf{A}[p,q]|_{t = 0} = 0 \text{ for all }p,q
\end{equation} 
According to the argument we presented in the introduction, the boundary algebra for this boundary condition should coincide with the universal defect algebra on $\C$. Then one of the prediction of twisted holography implies that this boundary chiral algebra is the same as the large $N$ chiral algebra on B-brane, generated by the $E$ towers of fields.
\begin{remark}
	The above statement is actually not precise. Because the OPE of open string operators will eventually involve closed string operators, and the holomorphic Chern-Simons theory alone suffers from anomaly which requires the introduction of Kodaira-Spencer theory to cancel. A precise matching should involve the inclusion of Kodaira-Spencer gravity and their coupling. We provide some examples of matching the coupling terms and the corresponding B-brane OPE in \ref{sec:OPE_OC}.
\end{remark}

For the Dirichlet boundary condition, the (perturbative) boundary local operators are generated by the lowest component $B$ of $\mathbf{B}$ and its derivatives:
\begin{equation}
	\pa_z^n B[p,q](z).
\end{equation}
  It is easy to see that there is a one-to-one correspondence of boundary local operators and the large $N$ single trace operators from the $E$ tower operators. Let
  \begin{equation}
  	E^{(p,q)} = I Z^{(i_1}Z^{i_2}\dots Z^{i_{n})} J ,\quad \text{where } \{i_1,\dots i_n\}  = \{\overbrace{1,\dots,1}^{p},\overbrace{2,\dots,2}^{q} \}.
  \end{equation}
Then we expect the following correspondence
  \begin{equation}
\pa_z^n B[p,q](z) \Leftrightarrow \pa_z^n E^{(p,q)}(z).
  \end{equation}
In the next section, we provide some nontrivial results showing that OPE's of the corresponding operators are also the same.

In fact, our $3d$ KK theories are holomorphic topological cousins of the $3d$ Chern-Simons theory. In this formulation, the twisted holography duality are reminiscent of the familiar holography duality between the Chern-Simons theory and the WZW model \cite{Witten:1988hf}.

	\subsection{Matching the OPE}
	In this section, we compare the OPE of the boundary chiral algebra and the large $N$ chiral algebra. On the one side, the OPE for the chiral algebra on B-brane is computed by wick contraction of the matrix fields. On the other side, the boundary OPE is computed by Feynman diagrams through bulk interaction vertices. We will show that different wick contraction patterns have gravity interpretations as different bulk interaction vertices.

			It is convenient to organize the $E$ tower of fields into a generating function
		\begin{equation}
			E^{(n)}(\lambda;z) : = I(z)Z(\lambda;z)^nJ(z),
		\end{equation}
		where 
		\begin{equation}
			Z(\lambda;z) := Z^1(z)\lambda_1+ Z^2(z)\lambda_2.
		\end{equation}
	In this notation, the OPE between the $Z$ fields can be written as
	\begin{equation}
		Z(\lambda;z)Z(\lambda';0) \sim \frac{[\lambda\lambda']}{z},
	\end{equation}
where $[\lambda\lambda'] = \lambda_1\lambda'_2 - \lambda_2\lambda'_1$.
		The generating field $E^{(n)}(x;z)$ is related to the symmetrized operator $E^{(p,q)}(z)$ as follows 
		\begin{equation}
			E^{(n)}(\lambda;z) = \sum_{p+q = n}\binom{p+q}{p}E^{(p,q)}(z)\lambda_1^p\lambda_2^q.
		\end{equation}

	\subsubsection{Quadratic vertex and two-point function}
	\label{sec:OPE_HCS_2}
	In this section, we analyze the tree-level boundary OPE from the quadratic vertices of the KK theory. First, we briefly review how to perform such kind of computations in the $3d$ holomorphic topological theories. We refer to \cite{Costello:2020ndc} for a detailed analysis.
	
	The kinetic term $\int \Tr \mathbf{B} \hat{d} \mathbf{A}$ of the $3d$ theory gives us a propagator connecting the superfields $\mathbf{B}$ and $\mathbf{A}$
	\begin{center}
		\begin{tikzpicture}[>=stealth]
			\draw[->] (0,0) -- (0.8,0);
			\draw (0.8,0) -- (1.5,0);
			\draw (0,0) node[left] (B) {$\mathbf{B}$};
			\draw (1.5,0) node[right] (A) {$\mathbf{A}$};
		\end{tikzpicture}
	\end{center}

In the presence of boundary, the corresponding propagator is obtained using the method of reflection \cite{Costello:2020ndc}. Omitting the color factor, we have
\begin{equation}
	P_{\pa}(z,t;z',t') = \frac{1}{2}(P(z,t;z',t') - P(z,t;z',-t')),
\end{equation}
	where $P(z,t;z',t')$ is the pure bulk propagator
	\begin{equation}
		P(z,t;z',t') = \frac{1}{8\pi i}\frac{(\bar{z} - \bar{z}')d(t - t') - \frac{1}{2}(t - t')d(\bar{z} - \bar{z}')}{(|z - z'|^2 + (t - t')^2 )^{\frac{3}{2}}},
	\end{equation}
	The action functional of the KK theory of holomorphic Chern-Simons on $SL_2(\C)$ consists of the flat space action functional and deformation terms coming from the deformed geometry. For the quadratic vertices, the flat space part $\int \Tr \mathbf{B}[p,q] \hat{d} \mathbf{A}[p,q]$ is the kinetic term which gives us the propagator connecting the fields $\mathbf{B}[p,q]$ and $\mathbf{A}[p,q]$. The remaining part coming from $SL_2(\C)$ deformation is analyzed in \ref{sec:def_diff_HCS}. We have
	\begin{equation}
	I^{(1,k)}_O =  \sum_{p+q = k -1}N^k \frac{(-1)^p}{(k - 1)!} \frac{p!q!}{k!}\frac{1}{2} \int \Tr(\mathbf{A}[p,q] dz\pa_z^{k} \mathbf{A}[q,p]) .
\end{equation}
	We have the following tree-level Feynman diagram associated with the vertex $I^{(1,k)}$ 

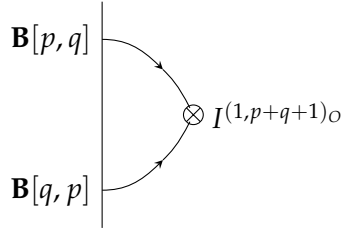
\begin{center}
		    \begin{tikzpicture}[>=stealth]
		\draw (0,-1.5) -- (0,1.5);
		\draw (1.2,0) node[cross,label=right:$I^{(1,p+q+1)_O}$] (O) {} ;
		\draw (O) circle (0.13);
		\draw (0,1) node[left] (B1) {$\mathbf{B}[p,q]$};
		\draw (0,-1) node[left] (B2) {$\mathbf{B}[q,p]$};
		\draw[->] (B2) .. controls (0.3,-1) and (0.5,-0.9) ..  (0.8,-0.6);
		\draw (0.8,-0.6).. controls (0.9,-0.5) and (1,-0.4) ..  (O);
		\draw[->] (B1) .. controls (0.3,1) and (0.5,0.9) ..  (0.8,0.6);
		\draw (0.8,0.6).. controls (0.9,0.5) and (1,0.4) ..  (O);
	\end{tikzpicture}
\end{center}
The corresponding Feynman integral is computed by
	\begin{equation}
	\int_{z' \in \C,s\geq 0}\frac{1}{(|z - z'|^2 + s^2)^{\frac{3}{2}}}\pa_{z'}^{p+q + 1}\frac{1}{(|z'|^2 + s^2)^{\frac{3}{2}}} \bar{z}s ds dz'd\bar{z}'.
\end{equation}
	The evaluation of this Feynman integral is given in Appendix \ref{apx:Feyn} and we find
	\begin{equation}
\frac{(p+q + 1)!}{z^{p+q+2}}.
	\end{equation}
Therefore, the corresponding OPE is given by
	\begin{equation}
		B[p,q](z) B[q,p](0) \sim  (-1)^qN^{p+q +1}\frac{p!q!}{(p+q)!} \frac{1}{z^{p+q+2}} ,
	\end{equation}
where we dropped the Lie algebra indices to simplify the notation. It is easy to restore them by adding the Killing form into the above formula.

On the B-brane side, this OPE should be matched with the full contraction of the $E$ tower operators, which is also the two-point function
	\begin{equation}
	\langle E^{(n)}(\lambda,z)E^{(n)}(\lambda',0)\rangle =  \frac{N^{n +1}[\lambda\lambda']^n}{z^{n+2}}.
\end{equation}
	Expanding $E^{(n)}(\lambda,z)$ into the symmetrized operator $E^{(p,q)}$, we find
	\begin{equation}
	\langle E^{(p,q)}(z)E^{(q,p)}(0)\rangle  =   (-1)^qN^{p+q +1}\frac{p!q!}{(p+q)!} \frac{1}{z^{p+q+2}}.
	\end{equation}
This gives the same result as in the boundary chiral algebra. In fact, one should understand this matching of two-point functions as fixing the normalization of the operators on the two sides. The real non-trivial matching start in the next section, when we analyze OPE from higher order vertices.
	
	\subsubsection{Cubic vertex and three-point function}
	In this section, we analyze the (deformed) cubic vertices and the corresponding chiral OPE. We proceed from the easiest cases to the most general case. As in the last section, we omit the Lie algebra indices. They can be easily restored by adding the structure constant $f_{ab}^c$ to our formula.
	
	\paragraph{flat space action}
	First, we look at the flat space cubic vertex $I^{(2,0)}_O$.
	\begin{equation}
		I^{(2,0)}_O = \int \Tr \mathbf{B}[p+r,q+s] [\mathbf{A}[p,q],\mathbf{A}[r,s]].
	\end{equation}
This vertex gives us the following tree-level Feynman diagram
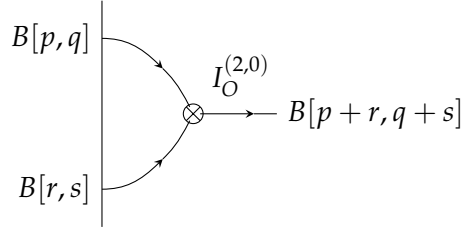
\begin{center}
	\begin{tikzpicture}[>=stealth]
		\draw (0,-1.5) -- (0,1.5);
		\draw (1.2,0) node[cross,label=above right :$I^{(2,0)}_O$] (O) {} ;
		\draw (O) circle (0.13);
		\draw (0,1) node[left] (B1) {$B[p,q]$};
		\draw (0,-1) node[left] (B2) {$B[r,s]$};
		\draw[->] (B2) .. controls (0.3,-1) and (0.5,-0.9) ..  (0.8,-0.6);
		\draw (0.8,-0.6).. controls (0.9,-0.5) and (1,-0.4) ..  (O);
		\draw[->] (B1) .. controls (0.3,1) and (0.5,0.9) ..  (0.8,0.6);
		\draw (0.8,0.6).. controls (0.9,0.5) and (1,0.4) ..  (O);
		\draw (2.3,0) node[right] (B3) {$B[p+r,q+s]$};
		\draw[->] (O) -- (2,0);
		\draw(2,0) -- (B3);
	\end{tikzpicture}
\end{center}
The corresponding Feynman integral is the same as the integral in the last section but without $\pa_{z'}$ derivatives, which gives us the following OPE
\begin{equation}
B[p,q](z)B[r,s](0)\sim \frac{1}{z}B[p+r,q+s](0).
\end{equation}

On the B-brane side, this should be matched with the OPE of a single $I-J$ contraction. For example,
	\begin{equation}
\wick{(I Z^{i_1}\dots \c J )(z)  ( \c I Z^{j_1} \dots J)(0)} =  \frac{1}{z}I Z^{i_1}\dots  Z^{j_1} \dots J.
	\end{equation}
\begin{remark}
	We also have the contraction $\wick{( \c  I Z^{i_1}\dots J ) (  I Z^{j_1} \dots \c J)}$. However, this term gives us the same combinatorial factor and pole and should be attributed to the Lie algebra factor that we omitted. 
\end{remark}
For the symmetrized operators, we have
\begin{equation}
	\text{contraction $I-J$  : } \; E^{(p,q)}(z) E^{(r,s)}(0) = \frac{1}{z }E^{(p+r,q+s)}(0) + \dots.
\end{equation}
This OPE can also be written as
\begin{equation}
	E^{(n)}(\lambda;z)E^{(n')}(\lambda';0) \sim \frac{1}{z}  \frac{n'!}{(n+n')!}(\lambda\cdot\pa_{\lambda'})^nE^{(n+n')}(\lambda';0) + \dots
\end{equation}
where $\lambda\cdot\pa_{\lambda'} = \lambda_1\pa_{\lambda_1'} + \lambda_2\pa_{\lambda_2'}  $.

	We emphasize that there are secretly other terms in this OPE. In fact, after the $I,J$ contraction of the symmetrized operators, we obtain $IZ^{(i_1}\dots Z^{i_n)}Z^{(j_1}\dots Z^{j_n)}J$. This is, in general, not the symmetrized operator $IZ^{(i_1}\dots Z^{i_n}Z^{j_1}\dots Z^{j_n)}J$. However, we can always manipulate the final expression into a sum of the symmetrized operator $IZ^{(i_1}\dots Z^{i_n}Z^{j_1}\dots Z^{j_n)}J$ and other operators in the BRST representative.
	
	For example, an $I,J$ contraction of $IZ^1J$ and $IZ^2J$ is $IZ^1Z^2J$, which is not symmetrized. Using the BRST relation, we find that $I[Z^1,Z^2]J$ is cohomologous to $IJIJ$. Therefore the term $IZ^1Z^2J$ is cohomologous to $IZ^{(1}Z^{2)}J + \frac{1}{2}IJIJ$. The remaining terms like $IJIJ$ are also important, and we will come back to them in the next section. 

\paragraph{first order deformed action}
	
We can also consider the deformed cubic vertices. As a simple example, we consider the first order deformation $I_O^{(2,1)}$ in Section \ref{sec:def_pro_HCS}. Up to a total derivative, this interaction vertex is given by
	\begin{equation}
		\begin{aligned}
				I_O^{(2,1)}	= &N\sum \frac{ps - qr}{(p+q)(r+s)} \int \Tr \mathbf{B}[p+r - 1,q+s- 1][\pa_z\mathbf{A}[p,q], \mathbf{A}[r,s]]\\
			+ &N\sum \frac{ps - qr}{(p+q+r+s)}\frac{1}{r+s} \int \Tr \pa_z\mathbf{B}[p+r - 1,q+s- 1][\mathbf{A}[p,q], \mathbf{A}[r,s]].\\
		\end{aligned}
	\end{equation}
	
	The corresponding tree-level diagram gives us the following OPE
	\begin{equation}
\begin{aligned}
		B[p,q](z)B[r,s](0)\sim &\frac{N}{z^2} \frac{ps - qr}{(p+q)(r+s)}B[p+r - 1,q+s - 1](0)\\
		 &+ \frac{N}{z} \frac{ps - qr}{(p+q+r+s)}\frac{1}{r+s}\pa_zB[p+r - 1,q+s - 1](0).
\end{aligned}
	\end{equation}
	
	This should be matched with the OPE of a single $I-J$ contraction together with a $Z-Z$ contraction next to the $I-J$ contraction. We have
	\begin{equation}
			\wick{(I(z) Z(\lambda;z)\dots  \c2 Z(\lambda;z) \c1 J(z) )  ( \c1 I(0) \c2 Z(\lambda',0)  \dots J(0))} = \frac{N[\lambda\lambda']}{z^2} I(z)Z(\lambda,z)^{n-1}Z(\lambda',0)^{n' - 1}J(0).
	\end{equation}
Using BRST relation to symmetrize the operators and drop operators in other BRST representatives, we get
\begin{equation}
\begin{aligned}
		E^{(n)}(\lambda;z)E^{(n')}(\lambda';0) &\sim \frac{N[\lambda\lambda']}{z^2}  \frac{(n' - 1)!}{(n+n' - 2)!}(\lambda\cdot\pa_{\lambda'})^{n - 1}E^{(n+n' - 2)}(\lambda';0)\\
		&+ \frac{N[\lambda\lambda']}{z}  \frac{n}{n+n'}\frac{(n' - 1)!}{(n+n' - 2)!}(\lambda\cdot\pa_{\lambda'})^{n - 1}\pa_zE^{(n+n' - 2)}(\lambda';0) + \dots
\end{aligned}
\end{equation}
We can rewrite the above OPE using the symmetrized operator $E^{(p,q)}(z)$. The coefficient for the $\frac{1}{z^2}E$ term can be computed by
\begin{equation}
	\frac{\binom{p+q - 1}{q}\binom{r+s - 1}{r} - \binom{p+q - 1}{p}\binom{r+s - 1}{s}}{\binom{p+q}{p}\binom{r+s}{s}} = \frac{ps - qr}{(p+q)(r+s)}.
\end{equation}
The coefficient for the $\frac{1}{z}\pa_z E$ term have an extra $\frac{n}{n+n'} = \frac{p+q}{p+q+r+s}$ factor. Therefore, the corresponding OPE is given by
\begin{equation}
	\begin{aligned}
		E^{(p,q)}(z)E^{(r,s)}(0)\sim &\frac{N}{z^2} \frac{ps - qr}{(p+q)(r+s)}E^{(p+r - 1,q+s - 1)}(0)\\
		&+ \frac{N}{z} \frac{ps - qr}{(p+q+r+s)}\frac{1}{r+s}\pa_zE^{(p+r - 1,q+s - 1)}(0).
	\end{aligned}
\end{equation}

\paragraph{all orders deformed action}
Computation of the OPE corresponding to the higher order deformed action is similar. Recall from Section \ref{sec:def_pro_HCS} that the $k-$th order deformed cubic action $I^{(2,k)}_O$ takes the following form
	\begin{equation}
	\begin{aligned}
		I^{(2,k)}_O = &N^k\sum_{p,q,r,s}\sum_{k_1+k_2 = k}\frac{(-1)^{k_2}R_k(p,q,r,s)}{k_1!k_2![p+q]_{k_1}[r+s]_{k_2}[p+q+r+s-k+1]_{k}} \\
		&\times\int \Tr \mathbf{B}[p+r - k,q+s - k][\pa_z^{k_1}\mathbf{A}[p,q], \pa_z^{k_2}\mathbf{A}[r,s]].\\
	\end{aligned}
\end{equation}
In this case, we have $\pa_z$ derivatives of total order $k$. The Feynman integral contribute a constant $k!$ factor and a pole $\frac{1}{z^{k+1}}$. Therefore, the constant coefficient of the leading term $\frac{N^k}{z^{k+1}}B$ is given by
\begin{equation}
k!\sum_{k_1+k_2 = k} \frac{R_{k}(p,q,r,s)}{k_1!k_2![p+q]_{k_1}[r+s]_{k_2}[p+q+r+s - k + 1]_{k}} = \frac{R_k(p,q,r,s)}{[p+q]_k[r+s]_k},
\end{equation}
where we used the formula \ref{Chu_Vand_var} derived in the Appendix. This give the OPE
\begin{equation}
	B[p,q](z)B[r,s](0)\sim \frac{N^k}{z^{k+1}}\frac{R_k(p,q,r,s)}{[p+q]_k[r+s]_k}B[p+r - k,q+s - k](0) + \dots
\end{equation}
To compute other terms with derivatives of $B$, we move the $\pa_z$ derivatives from one $\mathbf{A}$ field to the other $\mathbf{A}$ field and the $\mathbf{B}$ field in the above cubic action. For example, $\mathbf{B}[\pa_z^{k_1}\mathbf{A},\pa_z^{k_2}\mathbf{A}]$ can be rewritten as $\sum_l(-1)^{k_2}\binom{k_2}{l}\pa_z^l\mathbf{B}[\pa_z^{k_1 + k_2 - l}\mathbf{A},\mathbf{A}]$ up to a total $\pa_z$ derivatives. We see that the term with $\pa_z^lB$ has pole $\frac{1}{z^{k - l+1}}$. The corresponding coefficient is computed by 
\begin{equation}
\begin{aligned}
		&(k - l)!\sum_{k_1+k_2 = k} \binom{k_2}{l}\frac{R_{k}(p,q,r,s)}{k_1!k_2![p+q]_{k_1}[r+s]_{k_2}[p+q+r+s - k + 1]_{k}} \\
		&= \frac{R_k(p,q,r,s)}{l![p+q]_{k - l}[r+s]_k[p+q+r+s - 2k + l + 1]_l},
\end{aligned}
\end{equation}
where we used the formula \ref{Chu_Vand_var_1}. The full OPE corresponding to the deformed cubic action $I^{(2,k)}_O$ is thus given by 
\begin{equation}\label{OPE_HCS_21k}
		B[p,q](z)B[r,s](0)\sim \sum_{l= 0}^k\frac{N^k}{z^{k - l+1}}\frac{R_k(p,q,r,s)\times \pa_z^lB[p+r - k,q+s - k](0)}{l![p+q]_{k - l}[r+s]_k[p+q+r+s - 2k + l + 1]_l}.
\end{equation}

	On the B brane side, this OPE should correspond to a single $I-J$ contraction together with $k$ adjacent $Z-Z$ contractions. Such contraction produce a $N^k$ factor and a pole $\frac{1}{z^{k+1}}$. We have
\begin{equation}
	\wick{(I(z) ... \c3 Z(\lambda;z)...\c2 Z(\lambda;z) \c1 J(z) ) ( \c1 I(0) \c2 Z(\lambda';0) ...\c3 Z(\lambda';0) ... J(0))} =  \frac{N^k}{z^{k+1}}I(z)Z(\lambda;z)^{n - k}Z(\lambda';0)^{n' - k}J(0).
\end{equation}
By symmetrizing the $Z$ operators, we can rewrite this OPE as follows
\begin{equation}
		E^{(n)}(\lambda;z)E^{(n')}(\lambda';0) \sim \frac{N^k[\lambda\lambda']^k}{z^{k+1}}  \frac{(n' - k)!}{(n+n' - 2k)!}(\lambda\cdot\pa_{\lambda'})^{n - k}E^{(n+n' - 2k)}(\lambda';0) \dots,
\end{equation}
where we omitted terms with derivatives of $E$. Expanding the operators into $E^{(p,q)}$, we can compute the constant coefficient by
\begin{equation}
\sum_{i = 0}^k	\frac{(-1)^i\binom{k}{i}\binom{p+q - k}{q - i}\binom{r+s - k}{r - i}}{\binom{p+q}{p}\binom{r+s}{r}} = \frac{R_k(p,q,r,s)}{[p+q]_k[r+s]_k},
\end{equation}
which follows from the definition \ref{const_R} of the constant $R_k(p,q,r,s)$. This matches the leading term of \ref{OPE_HCS_21k}. To obtain the coefficients of the derivative terms, we use the planar three point function
\begin{equation}
	\langle E^{(n)}(\lambda,z)E^{(n')}(\lambda',z')E^{(n'')}(\lambda'',z'')\rangle = \frac{N^{\frac{n+n' + n''}{2}+1}[\lambda\lambda']^{\frac{n+n' - n''}{2}}[\lambda'\lambda'']^{\frac{n'+n'' - n}{2}}[\lambda\lambda'']^{\frac{n+n'' - n'}{2}}}{(z - z')^{\frac{n+n' - n''}{2} +1}(z' - z'')^{\frac{n'+n'' - n'}{2} +1}(z - z'')^{\frac{n+n'' - n'}{2} +1}},
\end{equation}
which is computed via a direct computation of wick contraction. Expanding this three point function we find that
\begin{equation}
	\begin{aligned}
		&\langle E^{(n)}(\lambda,z)E^{(n')}(\lambda',z')E^{(n'')}(\lambda'',z'')\rangle\\
				= & \frac{N^{\frac{n+n' + n''}{2}+1}[\lambda\lambda']^{\frac{n+n' - n''}{2}}[\lambda'\lambda'']^{\frac{n'+n'' - n}{2}}[\lambda\lambda'']^{\frac{n+n'' - n'}{2}}}{(z - z')^{\frac{n+n' - n''}{2} +1}}\sum_{l \geq 0}\frac{(\frac{n+n'' - n'}{2} + l)!}{(\frac{n+n'' - n'}{2})!l!}(-1)^l\left(\frac{z - z'}{z'-z''}\right)^l \frac{1}{(z' - z'')^{n'' + 2}} \\
		= & \sum_{l \geq 0}\frac{N^{\frac{n+n' + n''}{2}+1}[\lambda\lambda']^{\frac{n+n' - n''}{2}}}{(z - z')^{\frac{n+n' - n''}{2} - l +1}}  \frac{[\frac{n+n'' - n'}{2} + l]_l(\frac{n' + n'' - n}{2})!}{[n'' + l+1]_ln''!} (\lambda\cdot\pa_{\lambda'})^{\frac{n+n'' - n'}{2}}[\lambda'\lambda'']^{n''}\frac{\pa_{z'}^l}{l!}\frac{1}{(z' - z'')^{n''+2}}
	\end{aligned}
\end{equation}
Using the two-point function, we can extract from the above expression the following OPE
\begin{equation}
	E^{(n)}(\lambda;z)E^{(n')}(\lambda';0) \sim \sum_{l = 0}^k\frac{N^k[\lambda\lambda']^k}{z^{k - l +1}}  \frac{[n - k + l]_l(n' - k)!}{[n+n' - 2k + l+1]_l(n+n' - 2k)!}(\lambda\cdot\pa_{\lambda'})^{n - k}\frac{\pa_z^l}{l!}E^{(n+n' - 2k)}(\lambda';0).
\end{equation}
We can check that we get exactly the OPE of the boundary chiral algebra \ref{OPE_HCS_21k}.

	\subsubsection{Quartic vertex}
	\label{sec:OPE_HCS_4}
	The quartic vertex is constructed using the higher product $m_3$ that we studied in Section \ref{sec:A_trans}. We have the following Feynman diagram.
	
	
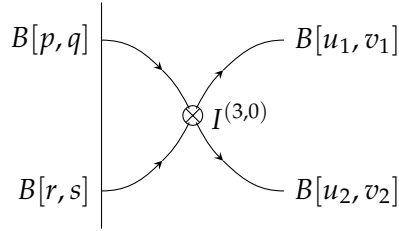
\begin{center}
		\begin{tikzpicture}[>=stealth]
			\draw (0,-1.5) -- (0,1.5);
			\draw (1.2,0) node[cross,label=right :$I^{(3,0)}$] (O) {} ;
			\draw (O) circle (0.13);
			\draw (0,1) node[left] (B1) {$B[p,q]$};
			\draw (0,-1) node[left] (B2) {$B[r,s]$};
			\draw[->] (B2) .. controls (0.3,-1) and (0.5,-0.9) ..  (0.8,-0.6);
			\draw (0.8,-0.6).. controls (0.9,-0.5) and (1,-0.4) ..  (O);
			\draw[->] (B1) .. controls (0.3,1) and (0.5,0.9) ..  (0.8,0.6);
			\draw (0.8,0.6).. controls (0.9,0.5) and (1,0.4) ..  (O);
			\draw (2.4,1) node[right] (B3) {$B[u_1,v_1]$};
			\draw (2.4,-1) node[right] (B4) {$B[u_2,v_2]$};
			\draw[->] (O) .. controls (1.4,0.4) and (1.5,0.5) .. (1.6,0.6);
			\draw (1.6,0.6) .. controls (1.9,0.9) and (2.1,1) .. (B3);
			\draw[->] (O) .. controls (1.4,-0.4) and (1.5,-0.5) .. (1.6,-0.6);
			\draw (1.6,-0.6) .. controls (1.9,-0.9) and (2.1,-1) .. (B4);
		\end{tikzpicture}
	\end{center}
	The corresponding OPE is given by
	\begin{equation}\label{OPE_2to2_1}
			B_a[p,q](z)B_b[r,s](0) \sim \sum_{\substack{u_1 + u_2 = p+r - 1\\v_1+v_2 = q+s - 1}}\sum_{c,d,e,f}\frac{K^{fc}}{z} \left( (m_3)^{p,q;r,s}_{u_1,v_1}f_{bf}^ef_{ae}^d - (m_3)^{r,s;p,q}_{u_1,v_1}f_{af}^ef_{be}^d\right)  B_c[u_1,v_1]B_d[u_2,v_2](0).
	\end{equation}
where we used the shorthand $(m_3)^{p,q;r,s}_{u,v} : = (m_3)^{p,q;r,s}_{u,v;p+r-1-u,q+s-1-v}$. We included the Lie algebra factor here since it becomes important in the analysis. Though the above formula is for general Lie algebra, in the holographic setting we need to use the Lie algebra $\mathfrak{gl}(K|K)$ or $\mathfrak{gl}(K)$. Strictly speaking, the bulk anomaly of the open-closed coupled theory cancels only for the super Lie algebra $\mathfrak{gl}(K|K)$. However, the OPE matching at this level also works for $\mathfrak{gl}(K)$. For simplicity, we use $\mathfrak{gl}(K)$ here. The result can be easily generalized to $\mathfrak{gl}(K|K)$ by taking care of the $\pm$ sign. Note that $\mathfrak{gl}(K)$ has a canonical basis given by the elementary matrices $\{E_{a_1a_2}\}_{1\leq a_1,a_2 \leq K}$. Therefore we replace the indices $a$ by $a_1a_2$ in the above formula. The Killing form $K^{a_1a_2;b_1b_2}$ is given by 
\begin{equation}
K^{a_1a_2;b_1b_2} = \delta_{a_1b_2}\delta_{a_2b_1}.
\end{equation}
The structure constant can be extracted from the commutation relation
\begin{equation}
	[E_{a_1a_2},E_{b_1b_2}] = \delta_{a_2b_1}E_{a_1b_2} - \delta_{a_1b_2}E_{b_1a_2}.
\end{equation}
Then the formula \ref{OPE_2to2_1} can be expanded as
	\begin{equation}
	\begin{aligned}
		B_{a_1a_2}[p,q](z)&B_{b_1b_2}[r,s](0) \sim\frac{1}{z}\sum_{\substack{u_1 + u_2 = p+r - 1\\v_1+v_2 = q+s - 1}}((m_3)^{r,s;p,q}_{u_1,v_1} - (m_3)^{p,q;r,s}_{u_1,v_1}) B_{a_1b_2}[u_1,v_1]B_{b_1a_2}[u_2,v_2](0) \\
		&+ \sum_{c} \left((m_3)^{p,q;r,s}_{u_1,v_1} \delta_{a_1b_2} B_{b_1c}[u_1,v_1]B_{ca_2}[u_2,v_2] -  (m_3)^{r,s;p,q}_{u_1,v_1}\delta_{a_2b_1}B_{a_1c}[u_1,v_1]B_{cb_2}[u_2,v_2] \right)   .
	\end{aligned}
\end{equation}

Naively, we expect that the corresponding wick contraction on the B-brane side is given by a single $Z-Z$ contraction
 	\begin{equation}
 	\wick{(I ... \c Z^{i_k}... J )(z)  (  I ...  \c Z^{j_l} ... J)(0)} =  \frac{1}{z}(I Z^{j_1}... Z^{j_{l - 1}}Z^{i_k+1}... J)(0) (I Z^{i_1}... Z^{i_{k - 1}}Z^{j_{l+1}}  Z^{j_1} ... J)(0).
 \end{equation}
However, as we have mentioned in the last section, after a single $I-J$ contraction, we need to manipulate the final expression into BRST representative. In this process, we can trade the commutator $[Z^1,Z^1]$ with $\sum J_aI_a$. Therefore, OPE of the form $(I\dots J)(z) (I\dots J)(0) \sim \frac{1}{z} (I\dots J) (I\dots J)$ can also be generated in the single $I-J$ wick contraction studied in the last section.

It is a tedious work to analyze the combinatorial factor of the wick contractions in the general situation. To simplify the discussion, we study some special cases.

The simplest example we can look at is the OPE of $I_{a_1}Z^1J_{a_2}$ and $I_{b_1}Z^2J_{b_2}$. A single $Z-Z$ contraction gives us
\begin{equation}
	\wick{(I_{a_1} \c Z^1 J_{a_2} )(z) (I_{b_1} \c Z^2 J_{b_2} )(0)} \sim  \frac{1}{z} (I_{a_1}J_{b_2}) (I_{b_1}J_{a_2})(0).
\end{equation}
A single $I-J$ contraction gives us
\begin{equation}
\begin{aligned}
		\wick{(I_{a_1}  Z^1 \c J_{a_2} )(z) (\c I_{b_1} Z^2 J_{b_2} )(0)}& + 	\wick{(\c I_{a_1}  Z^1 J_{a_2} )(z) ( I_{b_1} Z^2 \c J_{b_2} )(0)} \\
		& \sim  -\frac{1}{2z}\sum_{c}(\delta_{a_2b_1} (I_{a_1}J_c) (I_cJ_{b_2})(0) + \delta_{a_1b_2}( I_{b_1}J_c) (I_cJ_{a_2})),
\end{aligned}
\end{equation}
where we omitted the operator $IZ^{(1}Z^{2)}J$ that analyzed in the last section.

By computing $m_3$ using \ref{m3}, we find that $(m_3)^{1,0;0,1}_{0,0} = -\frac{1}{2}$ and  $(m_3)^{0,1;1,0}_{0,0} = \frac{1}{2}$ . Therefore we have the following OPE
	\begin{equation}
	\begin{aligned}
		B_{a_1a_2}[1,0](z)B_{b_1b_2}[0,1](0) &\sim\frac{1}{z} \Bigl(  B_{a_1b_2}[0,0]B_{b_1a_2}[0,0](0) \\
		&- \frac{1}{2}\sum_{c}\delta_{a_2b_1}B_{a_1c}[0,0]B_{cb_2}[0,0](0) - \delta_{a_1b_2} B_{b_1c}[0,0]B_{ca_2}[0,0](0) \Bigr)  .
	\end{aligned}
\end{equation}
This matches exactly the corresponding OPE on B-brane.

As a more nontrivial example, we use the formula \ref{m3} to compute the following value of $m_3$
\begin{equation} \label{m3_10}
\begin{aligned}
		(m_3)^{1,0;r,s}_{u,v} &= -\frac{r!s!(u+v+1)!(r+s - u-v-1)!}{(r+s+1)!u!v!(r-u)!(s-v-1)!},\\
		(m_3)^{r,s;1,0}_{u,v} &= \frac{r!s!(u+v)!(r+s - u-v)!}{(r+s+1)!u!v!(r-u)!(s-v-1)!}.
\end{aligned}
\end{equation}
Therefore, we have the following OPE 
	\begin{equation}\label{OPE_hCS_2to2_1rs}
	\begin{aligned}
		&B_{a_1a_2}[1,0](z)B_{b_1b_2}[r,s](0)\\
		\sim&   \frac{1}{z}\sum_{u,v}\binom{r+s}{r}^{-1}\binom{u+v}{u}\binom{r+s - u-v-1}{r-u} \Bigl(  B_{a_1b_2}[u,v]B_{b_1a_2}[r-u,s-v-1](0)\\
	-	 & \frac{1}{2}\sum_{c}\frac{r+s - u-v}{r+s+1}\delta_{a_2b_1}B_{a_1c}[u,v]B_{cb_2}[r-u,s-v-1](0)\\
		 -& \frac{1}{2}\sum_{c}\frac{u+v}{r+s+1} \delta_{a_1b_2} B_{b_1c}[u,v]B_{ca_2}[r-u,s-v-1](0) \Bigr)  .
	\end{aligned}
\end{equation}

On the B-brane side, we can compute a single $Z-Z$ contraction as follows
\begin{equation}\label{OPE_ZZ}
	E_{a_1a_2}^{(1)}(\lambda;z)E_{b_1b_2}^{(n')}(\lambda',0) \sim \sum_{n''}\frac{[\lambda\lambda']}{z}E^{(n'')}_{a_1b_2}(\lambda',0)E^{(n' - n'' - 1)}_{b_1a_2}(\lambda',0).
\end{equation}
Expanding the above formula with $E_{b_1b_2}^{(n')}(\lambda',0) = \sum_{r+s = n'}\binom{r+s}{r}E_{b_1b_2}^{(r,s)}(0)\lambda_1'^r\lambda_2'^s$, we find 
	\begin{equation}
		E_{a_1a_2}^{(1,0)}(z)B_{b_1b_2}^{(r,s)}(0)  \sim \frac{1}{z}\sum_{u,v}\binom{r+s}{r}^{-1}\binom{u+v}{u}\binom{r+s - u-v-1}{r-u} E_{a_1b_2}^{(u,v)}B_{b_1a_2}^{(r-u,s-v-1)}(0)  .
\end{equation}
This matches with the first term of the boundary OPE of \ref{OPE_hCS_2to2_1rs}. To match the remaining terms, we consider the single $I-J$ contraction. As in previous discussion, we use the BRST relation and find that $I\dots [Z(\lambda),Z(\lambda')] \dots J $ is cohomologous to $[\lambda\lambda']I\dots J I\dots J$. Therefore, after dropping the symmetrized operator that is already studied in the last section, we have the following OPE
\begin{equation}
\begin{aligned}
			E_{a_1a_2}^{(1)}(\lambda;z)E_{b_1b_2}^{(n')}(\lambda',0) &\sim -\frac{[\lambda\lambda']}{z}\sum_{n''}\sum_c\left(\frac{n'-n''}{n'+1 }\delta_{a_2b_1} E^{(n'')}_{a_1c}(\lambda',0)E^{(n' - n'' - 1)}_{cb_2}(\lambda',0) \right.\\
	 &\left. + \frac{n''+1}{n'+1 }\delta_{a_1b_2} E^{(n'')}_{b_1c}(\lambda',0)E^{(n' - n'' - 1)}_{ca_2}(\lambda',0)\right) .
\end{aligned}
\end{equation}
Expanding this formula, we find 
	\begin{equation}
\begin{aligned}
		&E_{a_1a_2}^{(1,0)}(z)B_{b_1b_2}^{(r,s)}(0)  \sim - \frac{1}{z}\sum_{u,v}\sum_c\binom{r+s}{r}^{-1}\binom{u+v}{u}\binom{r+s - u-v-1}{r-u} \\
		&\times \left(\frac{r+s - u-v}{r+s+1}\delta_{a_2b_1}E_{a_1c}^{(u,v)}B_{cb_2}^{(r-u,s-v-1)}(0) + \frac{u+v+1}{r+s+1}\delta_{a_1b_2}E_{b_1c}^{(u,v)}B_{ca_2}^{(r-u,s-v-1)}(0) \right) .
\end{aligned}
\end{equation}
This gives the remaining terms of the boundary OPE of \ref{OPE_hCS_2to2_1rs}.

As we have discussed, the deformed geometry will introduce deformation to the higher operations and gives us deformed interaction vertices. We expect that the deformed quartic action $I^{(3,k)}_O$ corresponds to the wick contraction of $k + 1$ adjacent $Z-Z$ fields, and also the wick contraction of a single $I-J$ together with $k$ adjacent $Z-Z$ fields. It will be interesting to explore these OPE.

\subsection{Comments on higher order interaction}
\label{sec:twist_higer}
In this section we briefly analyze OPE coming from higher order vertices. First, we study the quintic interaction vertices. By looking at the Lie algebra factor, we see that this interaction term produce the following terms in the OPE of $B_{a_1a_2}(z) B_{b_1b_2}(0)$
\begin{equation}\label{OPE_quin}
\begin{aligned}
		\sum_{c,c'}\delta_{a_2b_1} B_{a_1c} B_{cc'}B_{c'b_2}&,\quad \sum_{c,c'}\delta_{a_1b_2}B_{b_1 c}B_{cc'}B_{c'a_2},\\
		\sum_{c}B_{a_1c}B_{cb_2}B_{b_1a_2}&,\quad \sum_{c}B_{a_1b_2}B_{b_1c}B_{ca_2}.
\end{aligned}
\end{equation}

Naively, wick contraction of two $IZ\dots J$ operators cannot produce operators with three or more pairs of $I,J$ fields. However, as we have analyzed in the last section, we trade the commutator $[Z^1,Z^2]$ with $\sum_{c}J_cI_c$ in this process of symmetrize the adjoint fields. This will gives OPE that produce three or more pairs of $I,J$ fields as desired. For example, a single $Z-Z$ contraction of $I_{a_1}\dots J_{a_2}$ and $I_{b_1}\dots J_{b_2}$ gives $I_{a_1}\dots J_{b_2} I_{b_1}\dots J_{a_2}$. Symmetrizing this operator gives the following two possibilities with three pairs of $I,J$ fields.
\begin{equation}
	\sum_{c}(I_{a_1}\dots J_c) (I_c\dots J_{b_2}) (I_{b_1}\dots J_{a_2}),\quad \sum_{c}(I_{a_1}\dots J_{b_2})( I_{b_1}\dots J_c) (I_c\dots J_{a_2}).
\end{equation}
They match the operators in the second line of \ref{OPE_quin}. We can also consider a single $I-J$ contraction. This will gives us the operators $\delta_{a_2b_1}  I_{a_1}\dots J_{b_2}$ and $\delta_{a_1b_2}  I_{b_1}\dots J_{a_2}$. To symmetrize this operator we need to use $[Z^1,Z^2]\sim \sum_{c}J_cI_c$ more than once in general. This gives the following two possibilities with three pairs of $I,J$ fields. 
\begin{equation}
	\sum_{c,c'}\delta_{a_2b_1} ( I_{a_1}\dots J_c) (I_{c}\dots J_{c'}) (I_{c'}\dots J_{b_2}),\quad \sum_{c,c'}\delta_{a_1b_2}  (I_{b_1}\dots J_c) (I_{c}\dots J_{c'}) (I_{c'}\dots J_{a_2}).
\end{equation}
They match the operators in the first line of \ref{OPE_quin}. So far, we have only considered the matching of Lie algebra indices in the OPE. It will be a nontrivial and interesting prediction of twisted holography that the $A_\infty$ structure on the Cauchy-Riemann cohomology can reproduce the combinatorial factor coming from matrix contraction and symmetrization of the operators.

More generally, we expect that the full tree-level OPE for the KK theory on deformed geometry takes the following schematic form
\begin{equation}
\begin{aligned}
		&B_{a_1a_2}[p,q](z)B_{b_1b_2}[r,s](0) \sim \sum_{n\geq 1}\sum_{k\geq 0}\sum_{\substack{u_1 + \dots u_{n} = p+r -k -n+1\\v_1+\dots v_{n} = q+s - k - n+1}}	\frac{N^k}{z^{k+1}}\times (\text{some constants})\times \\
		&\sum_{c_1,\dots,c_{n-1}}\delta_{a_2b_1}B_{a_1c_1}[u_1,v_1]B_{c_1c_2}[u_2,v_2]\dots B_{c_{n-1}b_2}[u_{n},v_{n}] + \delta_{a_1b_2}B_{b_1c_1}[u_1,v_1]B_{c_1c_2}[u_2,v_2]\dots B_{c_{n-1}a_2}[u_{n},v_{n}]\\
		& \sum_{c_1,\dots,c_{n-2}}\sum_{l = 1}^{n-1} B_{a_1c_1}[u_1,v_1]B_{c_1c_2}[u_2,v_2]\dots B_{c_{l-1}b_2}[u_{l},v_{l}]B_{b_1c_l}[u_{l+1},v_{l+1}]B_{c_1c_{l+1}}[u_{l+2},v_{l+2}]\dots B_{c_{n-2}a_2}[u_{n},v_{n}]\\
		& + \text{derivatives}.
\end{aligned}
\end{equation}

Under the normalization $B[p,q] \to N^{-\frac{p+q + 2}{2}}B[p,q]$, we find that the above OPE's are all of order $\frac{1}{N}$. Dropping the $\frac{1}{N^2}$ and higher terms in the full OPE provides us with an interesting limit of the holography chiral algebra beyond the planar limit. In fact, this limit gives us a Poisson vertex algebra. Matching this Poisson vertex algebra structure is a very nontrivial but also approachable test of the twisted holography.

	\subsection{Chiral algebra from abelian Chern-Simons}
	\label{sec:OPE_abe}
	As explained in \cite{Costello:2018zrm}, the $\mathrm{PV}^{2,0}$ and $\mathrm{PV}^{0,2}$ part of Kodaira-Spencer theory can be reformulated as a super abelian Chern-Simons theory valued in $\Pi\C^2$ \footnote{This notation simply means the abelian Lie super algebra of $2$-dimension in degree $1$.}. This abelian Lie algebra has bilinear pairing given by the symplectic form on $\C^2$. Therefore, we can directly apply our previous results about the holomorphic Chern-Simons theory. After the KK reduction, we have the following field content
	\begin{equation}
	\begin{aligned}
		\mathbf{A}'  &\in \Omega_{3d}^{\sbullet}\otimes \Pi\C^2[w_1,w_2][1],\\
		\mathbf{B}'  &\in \Omega_{3d}^{\sbullet,(1)}\otimes \Pi\C^2[\bar{w}_,\bar{w}_2].
	\end{aligned}
\end{equation}
Since $\Pi \C^{2}$ is an abelian algebra, all brackets and higher brackets $l_k$ vanish for this theory. The only nontrivial terms are the quadratic action and its deformation discussed in \ref{sec:def_diff_HCS}. We immediately obtain the action functional for this theory
\begin{equation}
 \sum_{p,q}	\int \epsilon_{ab}\mathbf{B}'_a[p,q]\hat{d}\mathbf{A}'_b[p,q] +  \frac{(-1)^qN^{p+q+1}}{(p+q+1)!} \frac{p!q!}{(p+q)!}\int \epsilon_{ab}\mathbf{A}'_a[p,q] \pa_z^{p+q +1} \mathbf{A}'_b[q,p].
\end{equation}

Taking the Dirichlet boundary condition, the boundary chiral algebra is generated by 
\begin{equation}
	\pa_z^kB'_{a}[p,q], \;\; a = 1,2.
\end{equation}

As we have discussed, the quadratic term gives the following OPE
\begin{equation}
	B'_{1}[p,q](z)B'_{2}[q,p](0) \sim (-1)^qN^{p+q +1}\frac{p!q!}{(p+q)!} \frac{1}{z^{p+q+2}}.
\end{equation}

On the B brane side, this part of boundary algebra should be matched with the $B^{(n)}$ and $C^{(n)}$ towers of fields.
\begin{equation}
	\begin{aligned}
		\Tr(bZ^{(i_1}Z^{i_1}\dots Z^{i_n)}),\\
		\Tr(\pa c Z^{(i_1}Z^{i_1}\dots Z^{i_n)}).\\
	\end{aligned}
\end{equation}
The corresponding OPE is given by the full wick contraction as follows
	\begin{equation}
	\wick{\Tr(\c4 b \c3 Z^{i_1} \c2 Z^{i_2}\dots \c1 Z^{i_n} )(z)  \Tr(\c4 \pa c \c1 Z^{\bar{i}_n} \dots \c2 Z^{\bar{i}_2}\c3 Z^{\bar{i}_1})(0)} =  (-1)^{\# \text{ of }2 \text{ in }\{i_n\}}\frac{N^{n +1}}{z^{n+2}}
\end{equation}
where $\bar{i} = 2$ if $i = 1$ and $\bar{i} = 1$ if $i = 2$.
Taking into account of the combinatorial factor of symmetrized trace, we have
	\begin{equation}
	\wick{\Tr(\c4 b \c3 Z^{(i_1}\dots \c2 Z^{i_n)} )(z)  \Tr(\c4 \pa c \c2 Z^{(\bar{i}_n} \dots \c3 Z^{\bar{i}_1)}))(0)} =   (-1)^qN^{p+q +1}\frac{p!q!}{(p+q)!} \frac{1}{z^{p+q+2}},
\end{equation}
where $p$ is the number of $1$ in $\{i_n\}$ and $q$ is the number of $2$ in $\{i_n\}$. This is the same as the corresponding boundary OPE.

Since there is no other interaction vertex in the abelian theory itself, we don't have any other OPE in this sector of fields. This is compatible with the Large N chiral algebra. In fact, any contractions of $ \Tr(b Z^{i_1}\dots )$ and $ \Tr(\pa c \dots Z^{j_1})$ will not produce operators in the same towers of fields. The nontrivial terms in the OPE thus only involve other towers of fields. On the gravity side, these OPE come from BCOV interactions and open-closed couplings. 

	\subsection{Chiral algebra from Kodaira-Spencer gravity}
	In this section, we analyze the boundary chiral algebra for the $(\bm{\alpha},\bm{\beta})$ part of Kodaira Spencer fields defined in Section \ref{sec:KK_KS}. Similar to the analysis of holomorphic Chern-Simons, we choose the boundary condition such that 
	\begin{equation}
		\bm{\alpha}|_{t = 0} = 0.
	\end{equation}
	 The corresponding boundary chiral algebra is generated by the boundary fields $\beta[p,q](z)$ and all their derivatives $\pa_{z}^n\beta[p,q]$. 
	
On the B-brane side, the corresponding single trace operator is given by the $A$ towers of fields. Let
\begin{equation}
	A^{(p,q)} = \Tr Z^{(i_1}Z^{i_2}\dots Z^{i_{n})},\quad \text{where } \{i_1,\dots i_n\}  = \{\overbrace{1,\dots,1}^{p},\overbrace{2,\dots,2}^{q} \}.
\end{equation}
We show in this section that OPE's of $A^{(p,q)}$ operators are matched with that of the boundary $\beta[p,q]$ operators. It is convenient to define the generating fields as follows
\begin{equation}
	A^{(n)}(\lambda;z) : = \Tr Z(\lambda;z)^n.
\end{equation}
They are related to $A^{(p,q)}$ by 
\begin{equation}
	A^{(n)}(\lambda;z) = \sum_{p+q = n}\binom{p+q}{l}A^{(p,q)}(z)\lambda_1^l\lambda_2^m.
\end{equation}
Now we consider the boundary chiral algebra for the $\beta$ fields. First, we consider the quadratic action. The Feynman diagram and the corresponding Feynman integral are the same as in Section \ref{sec:OPE_HCS_2}. We can immediately write down the corresponding OPE using the deformed quadratic action $I^{(1,k)}_C$ \ref{O1k}. We have
\begin{equation}\label{OPE_beta_1k}
	\beta[p,q](z)\beta[q,p](0) \sim (-1)^q\frac{p!q!}{(p+q - 1)!}\frac{N^{p+q}}{z^{p+q}}.
\end{equation}

On the B-brane side, this OPE corresponds to a full contraction, which is also the two-point correlation function of two $A$ tower operators. Using the generating fields, we have
	\begin{equation}
\langle	A^{(n)}(\lambda;z)A^{(n)}(\lambda';0) \rangle = nN^{n}\frac{[\lambda\lambda']^n}{z^{n}}.
	\end{equation}
Expanding $A^{(n)}(\lambda;z)$ into $A^{(p,q)}(z)$ we get exactly the OPE \ref{OPE_beta_1k}.

As a next step, we consider the cubic action in flat space obtained in Section \ref{sec:KK_KS}. We find the following OPE
\begin{equation}\label{OPE_beta_20}
	\beta[p,q](z)\beta[r,s](0) \sim (ps - qr)\beta[p+r-1,q+s - 1](0).
\end{equation} 
Similarly, we can proceed to consider deformed cubic action $I_C^{(2,k)}$ obtained in \ref{coup_C_2k}. 
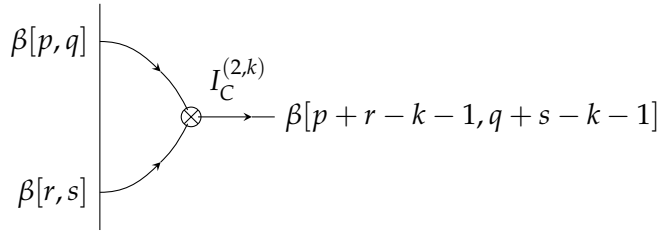
\begin{center}
	\begin{tikzpicture}[>=stealth]
		\draw (0,-1.5) -- (0,1.5);
		\draw (1.2,0) node[cross,label=above right :$I^{(2,k)}_C$] (O) {} ;
		\draw (O) circle (0.13);
		\draw (0,1) node[left] (B1) {$\beta[p,q]$};
		\draw (0,-1) node[left] (B2) {$\beta[r,s]$};
		\draw[->] (B2) .. controls (0.3,-1) and (0.5,-0.9) ..  (0.8,-0.6);
		\draw (0.8,-0.6).. controls (0.9,-0.5) and (1,-0.4) ..  (O);
		\draw[->] (B1) .. controls (0.3,1) and (0.5,0.9) ..  (0.8,0.6);
		\draw (0.8,0.6).. controls (0.9,0.5) and (1,0.4) ..  (O);
		\draw (2.3,0) node[right] (B3) {$\beta[p+r - k - 1,q+s - k - 1]$};
		\draw[->] (O) -- (2,0);
		\draw(2,0) -- (B3);
	\end{tikzpicture}
\end{center}
The leading order (without any $z$ derivatives) coefficient is computed as follows
\begin{equation}
	\begin{aligned}
		&\sum_{k_1+k_2 = k}\frac{k!(p+q)(r+s)R_{k+1}(p,q,r,s)}{k_1!k_2![p+q+r+s-k]_{k+1}}\left( \frac{1}{[p+q]_{k_1 + 1}[r+s]_{k_2}}+ \frac{1}{[p+q]_{k_1 }[r+s]_{k_2+1}}\right)\\
		& = \sum_{k_1+k_2 = k}\frac{k!R_{k+1}(p,q,r,s)}{k_1!k_2![p+q+r+s-k]_{k+1}}\left( \frac{r+s}{[p+q - 1]_{k_1}[r+s]_{k_2}}+ \frac{p+q}{[p+q]_{k_1 }[r+s - 1]_{k_2}}\right)\\
		&= \frac{R_{k+1}(p,q,r,s)}{[p+q - 1]_k[r+s- 1]_k}.
	\end{aligned}
\end{equation}
Therefore, we have the following OPE
\begin{equation}\label{OPE_beta_2k}
	\beta[p,q](z)\beta[r,s](0) \sim \frac{N^k}{z^{k+1}}\frac{R_{k+1}(p,q,r,s)}{[p+q - 1]_k[r+s- 1]_k}\beta[p+r - k -1,q+s - k - 1] (0)+\dots
\end{equation} 
Note that by letting $k = 0$ in the above formula, th result is the same as the flat space OPE \ref{OPE_beta_20}.

The corresponding B-brane OPE is given by $k+1$ adjacent contraction of $Z$ fields. For example, we have the following contraction between $A^{(n)}(\lambda;z)$ and $A^{(n')}(\lambda';0)$
\begin{equation}
\begin{aligned}
		&\wick{\Tr(Z(\lambda;z) \mycdots \c3 Z(\lambda;z)\c2\mycdots \c1 Z(\lambda;z)  )  \Tr( \c1 Z(\lambda';0) \c2\mycdots \c3 Z(\lambda';0)\mycdots Z(\lambda';0))}\\
		 =  &N^k\frac{[\lambda\lambda']^{k+1}}{z^{k+1}}\Tr Z(\lambda;z)^{n- k - 1}Z(\lambda';0)^{n'-k-1}.
\end{aligned}
\end{equation}
There are $nn'$ such contractions in total. As in the previous section, we use the BRST relation to symmetrize the operator and discard operators in different BRST representatives. Omitting terms with derivatives of $A$ field, we have 
\begin{equation}
	A^{(n)}(\lambda;z)A^{(n')}(\lambda';0) \sim nn'\frac{N^k[\lambda\lambda']^{k+1}}{z^{k+1}}\frac{(n'  - k - 1)!}{(n+n' - k - 2)!}(\lambda\cdot\pa_{\lambda'})^{n - k - 1}A^{(n+n' - k-2)}(\lambda';0) + \dots 
\end{equation}
Expanding $A^{(n)}(\lambda;z)$ into $A^{(p,q)}(z)$ in the above formula, we get exactly the OPE \ref{OPE_beta_2k}.

We leave it to future work to analyze the terms with derivatives.

We can also look at the coupling between holomorphic Chern-Simons theory and the Kodaira-Spencer gravity studied in Section \ref{sec:coup_OC}. When we take the Lie algebra to be the abelian Lie algebra $\Pi\C^2$ as in the last section, this coupling actually reflects the self-interaction of Kodaira-Spencer gravity. For example, from the coupling \ref{coup_OC_0} we obtain the following
\begin{equation}
	\frac{1}{2}\int \bm{\beta}[p+r,q+s]\mathbf{A}_1[p,r]\pa_z\mathbf{A}_2[r,s] + \dots
\end{equation}
This gives us the following boundary OPE
\begin{equation}
	B_1'[p,q](z)B_2'[r,s] \sim \frac{1}{z^2} \beta[p+r,q+s]
\end{equation}
On the B brane side, this OPE is given by a contraction of $\pa c$ and $b$ fields
	\begin{equation}
	\wick{\Tr(\c b  Z^{(i_1}\dots  Z^{i_n)} )(z)  \Tr(\c \pa c Z^{(i'_1} \dots  Z^{i'_{n'})}))(0)} \sim \frac{1}{z^2}\Tr Z^{(i_1}Z^{i_n}\dots Z^{i'_1}  Z^{i'_{n'})}.
\end{equation}
We will study OPE of this type in more detail in the next section.

	\subsection{Chiral from open-closed coupling}
	\label{sec:OPE_OC}
	In this section, we analyze the OPE coming from the coupling between holomorphic Chern-Simons and the Kodaira-Spencer gravity. 
	\paragraph{closed + open  $\rightarrow$ open}
	First, we look at the flat space coupling between $\mathbf{B},\mathbf{A}$ and $\bm{\alpha}$ fields \ref{coup_OC_0}. This coupling gives us the following OPE
	\begin{equation}\label{bdy_OPE_OC_21}
	\beta[p,q](z)	B_a[r,s](0) \sim \frac{ps - qr}{z}B_a[p+r - 1,q+s - 1](0).
	\end{equation}
Similarly, we can consider the deformed coupling $I_{OC}^{(2,k)}[\mathbf{B},\mathbf{A},\bm{\alpha}]$ in \ref{coup_OC_2k1}. 
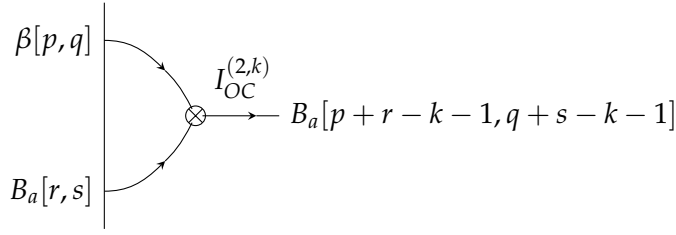
\begin{center}
	\begin{tikzpicture}[>=stealth]
		\draw (0,-1.5) -- (0,1.5);
		\draw (1.2,0) node[cross,label=above right :$I^{(2,k)}_{OC}$] (O) {} ;
		\draw (O) circle (0.13);
		\draw (0,1) node[left] (B1) {$\beta[p,q]$};
		\draw (0,-1) node[left] (B2) {$B_a[r,s]$};
		\draw[->] (B2) .. controls (0.3,-1) and (0.5,-0.9) ..  (0.8,-0.6);
		\draw (0.8,-0.6).. controls (0.9,-0.5) and (1,-0.4) ..  (O);
		\draw[->] (B1) .. controls (0.3,1) and (0.5,0.9) ..  (0.8,0.6);
		\draw (0.8,0.6).. controls (0.9,0.5) and (1,0.4) ..  (O);
		\draw (2.3,0) node[right] (B3) {$B_a[p+r - k - 1,q+s - k - 1]$};
		\draw[->] (O) -- (2,0);
		\draw(2,0) -- (B3);
	\end{tikzpicture}
\end{center}
The coefficient for the leading order term (without $z$ derivatives term) is computed by
\begin{equation}
\begin{aligned}
	& \sum_{k_1+k_2 = k}\frac{(p+q-k-1)}{(p+q - 1)}\frac{k!R_{k+1}(p,q,r,s)}{k_1!k_2![p+q - 2]_{k_1}[r+s]_{k_2}[p+q+r+s - k - 1]_{k}} \\
	& = \frac{R_{k+1}(p,q,r,s)}{[p+q - 1]_k[r+s]_k}.
\end{aligned}
\end{equation} 
Therefore we have the following OPE
	\begin{equation}\label{bdy_OPE_OC_21k}
	\beta[p,q](z)	B_a[r,s](0) \sim \frac{N^{k}}{z^{k+1}}\frac{R_{k+1}(p,q,r,s)}{[p+q - 1]_k[r+s]_k}B_a[p+r - k - 1,q+s - k - 1]+ \dots\;,
\end{equation}
where we omitted terms with derivatives of $B$.
		
On the B-brane side, this correspond to the wick contraction of $k+1$ adjacent $Z$ fields between $A^{(n)}(\lambda;z)$ and $E^{(n')}(\lambda',0)$ operators. For example,
\begin{equation}
	\begin{aligned}
		&\wick{\Tr(Z(\lambda;z) \mycdots \c3 Z(\lambda;z)\c2\mycdots \c1 Z(\lambda;z) )  I(0) \c1 Z(\lambda;0) \c2\mycdots \c3 Z(\lambda;0)\mycdots Z(\lambda;0)J(0)}\\
		=  &N^k\frac{[\lambda\lambda']^{k+1}}{z^{k+1}} I(0) Z(\lambda;z)^{n- k - 1}Z(\lambda';0)^{n'-k-1} J(0).
	\end{aligned}
\end{equation}
There are $n(n' - k)$ such contractions. As in the previous section, we use the BRST relation to symmetrize the operator and discard operators in different BRST representatives. Omitting terms with derivatives of $E$ field, we have 
\begin{equation}
	A^{(n)}(\lambda;z)E^{(n')}_{a_1a_2}(\lambda';0) \sim n(n' - k )\frac{N^{k}[\lambda\lambda']^{k+1}}{z^{k+1}}\frac{(n'  - k - 1)!}{(n+n' - k - 2)!}(\lambda\cdot\pa_{\lambda'})^{n - k - 1}E^{(n+n' - 2k - 2)}_{a_1a_2}(\lambda';0) + \dots
\end{equation}  
Expanding the generating fields into component, we find that this OPE is exactly the boundary OPE \ref{bdy_OPE_OC_21k}

\paragraph{open + open  $\rightarrow$ closed}
We can also consider the coupling between $\mathbf{A},\mathbf{A}$ and $\bm{\beta}$ fields. The flat space coupling gives us the following OPE
\begin{equation}
	B_a[p,q](z)B_b[r,s](0) \sim -\frac{K_{ab}}{z^2} \beta[p+r,q+s](0) + \dots
\end{equation}
The deformed coupling $I_{OC}^{(2,k)}[\mathbf{A},\mathbf{A},\bm{\beta}]$ \ref{coup_OC_2k2}gives us the following OPE
\begin{equation} \label{bdy_OPE_OC_2k}
		B_a[p,q](z)B_b[r,s](0) \sim -\frac{N^kK_{ab}}{z^{k+2}} \frac{(k+1)R_{k}(p,q,r,s)}{[p+q]_k[r+s]_k} \beta[p+r - k,q+s - k](0) + \dots
\end{equation}
where we omitted terms with derivatives of $\beta$.

On the B-brane side, this OPE corresponds to the wick contraction of two $I-J$ fields and $k$ pairs of $Z$ fields. The contracted $Z$ fields must be alongside one of the $I$ or $J$ fields. As a illustration, we list all possible contractions for $2$ pair of $Z$ fields
 \begin{equation}
\wick{(\c4 I \dots \c3 Z \c2 Z  \c1 J )(z) (\c1 I \c2 Z \c3 Z \dots \c4 J)(0)}\quad \wick{(\c4 I  \c3 Z \dots \c2 Z  \c1 J )(z) (\c1 I \c2 Z \dots \c3 Z  \c4 J)(0)} \quad \wick{(\c4 I  \c3 Z \c2 Z \dots   \c1 J )(z) (\c1 I  \dots \c2 Z \c3 Z  \c4 J)(0)}
 	\end{equation}
 For contraction of $k$ pairs of $Z$ fields, there are $k+1$ numbers of such contractions in total. This gives us the following
\begin{equation}
	E^{(n)}_{a_1a_2}(\lambda;z)E^{(n')}_{b_1b_2}(\lambda';0) \sim - \frac{N^{k}[\lambda\lambda']^{k}}{z^{k+2}}\delta_{a_1b_2}\delta_{a_2b_1}(k+1)\frac{(n'  - k )!}{(n+n' - k )!}(\lambda\cdot\pa_{\lambda'})^{n - k}A^{(n+n' - 2k)}(\lambda';0) + \dots
\end{equation}
Expanding this OPE we get exactly \ref{bdy_OPE_OC_2k}.

Finally, we briefly analyze boundary OPE corresponding to the quartic open-closed coupling. 
	
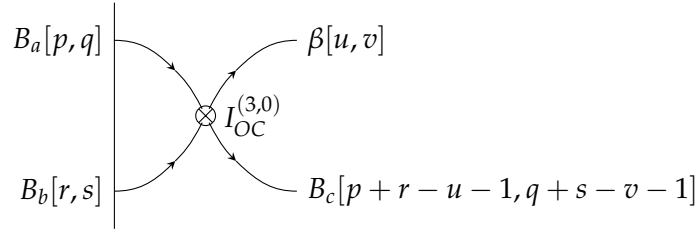
\begin{center}
	\begin{tikzpicture}[>=stealth]
		\draw (0,-1.5) -- (0,1.5);
		\draw (1.2,0) node[cross,label=right :$I^{(3,0)}_{OC}$] (O) {} ;
		\draw (O) circle (0.13);
		\draw (0,1) node[left] (B1) {$B_a[p,q]$};
		\draw (0,-1) node[left] (B2) {$B_b[r,s]$};
		\draw[->] (B2) .. controls (0.3,-1) and (0.5,-0.9) ..  (0.8,-0.6);
		\draw (0.8,-0.6).. controls (0.9,-0.5) and (1,-0.4) ..  (O);
		\draw[->] (B1) .. controls (0.3,1) and (0.5,0.9) ..  (0.8,0.6);
		\draw (0.8,0.6).. controls (0.9,0.5) and (1,0.4) ..  (O);
		\draw (2.4,1) node[right] (B3) {$\beta[u,v]$};
		\draw (2.4,-1) node[right] (B4) {$B_c[p+r-u-1,q+s - v- 1]$};
		\draw[->] (O) .. controls (1.4,0.4) and (1.5,0.5) .. (1.6,0.6);
		\draw (1.6,0.6) .. controls (1.9,0.9) and (2.1,1) .. (B3);
		\draw[->] (O) .. controls (1.4,-0.4) and (1.5,-0.5) .. (1.6,-0.6);
		\draw (1.6,-0.6) .. controls (1.9,-0.9) and (2.1,-1) .. (B4);
	\end{tikzpicture}
\end{center}
The coupling \ref{coup_OC_3} lead us the following boundary OPE
\begin{equation}
	B_a[p,q](z)B_b[r,s](0) \sim \frac{f_{ab}^c}{z^2} \sum_{u,v}\left( (m_3)^{p,q;r,s}_{u,v} - (m_3)^{r,s;p,q}_{u,v} \right) \beta[u,v] B_c[p+r-u-1,q+s - v- 1](0) + \dots
\end{equation}
In the special case $p = 1,q= 0$, this is computed in \ref{m3_10} and gives us the following
\begin{equation}\label{bdy_OPE_OOCO_1}
\begin{aligned}
			B_a[1,0](z)B_b[r,s](0) \sim &\frac{f_{ab}^c}{z^2}\sum_{u,v}\binom{r+s}{r}^{-1}\binom{u+v}{u} \binom{r+s - u-v-1}{r-u}\\
			&\times \beta[u,v] B_c[p+r-u-1,q+s - v- 1](0)\dots
\end{aligned}
\end{equation}

The corresponding wick contraction on B-brane is given by a single $I-J$ contraction, together with a single $Z-Z$ contraction that at least one $Z$ is not adjacent to the contracted $I-J$. This gives us the following
\begin{equation}
	E^{(1)}_{a_1a_2}(\lambda,z)E^{(n')}_{b_1b_2}(\lambda',z) \sim \sum_{n'' \geq 1}\frac{[\lambda\lambda']}{z^2}A^{(n'')}(\lambda',0)(\delta_{a_2b_1} E^{(n' - n'' - 1)}_{a_1b_2}(\lambda',0) - \delta_{a_1b_2}E^{(n' - n'' - 1)}_{a_2b_1}(\lambda',0)) + \dots
\end{equation}
Expanding the generating fields into component, we get the same constant coefficient as the expansion of \ref{OPE_ZZ}. We can check that this OPE match precisely the boundary OPE \ref{bdy_OPE_OOCO_1}.

\subsection{Comments on loops and non-planar contributions}
In this section, we briefly analyze a simple loop diagram that contributes to the OPE. We consider the following simplest loop diagrams.
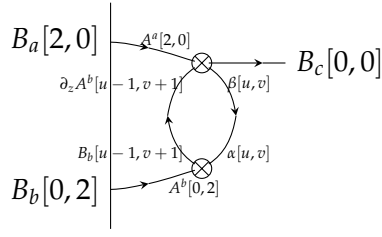
\begin{center}
	\begin{tikzpicture}[>=stealth]
		\draw (0,-1.5) -- (0,1.5);
		\draw (1.2,0.7) node[cross] (I1) {} ;
		\draw (1.2,1) node[left] (B1) {\scalebox{0.55}{$A^a[2,0]$}};
		\draw (1.1,0.42) node[left] (B1) {\scalebox{0.55}{$\pa_zA^b[u - 1,v+1]$}};
		\draw (1.4,0.42) node[right] (B1) {\scalebox{0.55}{$\beta[u ,v ]$}};
		\draw (I1) circle (0.13);
		\draw (1.2,-0.7) node[cross] (I2) {} ;
		\draw (1.1,-0.7) node[below] (B1) {\scalebox{0.55}{$A^b[0,2]$}};
		\draw (1.1,-0.5) node[left] (B1) {\scalebox{0.55}{$B_b[u - 1,v+1]$}};
		\draw (1.4,-0.5) node[right] (B1) {\scalebox{0.55}{$\alpha[u,v]$}};
		\draw (I2) circle (0.13);
		\draw (0,1) node[left] (B1) {$B_a[2,0]$};
		\draw (0,-1) node[left] (B2) {$B_b[0,2]$};
		\draw (2.3,0.7) node[right] (B3) {$B_c[0,0]$};
		\draw[->] (B1) .. controls (0.2,0.97) and (0.3,0.94) ..  (0.5,0.9);
		\draw(0.5,0.9) .. controls (0.8,0.83) and (1,0.75) ..  (I1);
		\draw[->] (B2) .. controls (0.2,-0.97) and (0.3,-0.94) ..  (0.5,-0.9);
		\draw(0.5,-0.9) .. controls (0.8,-0.83) and (1,-0.75) ..  (I2);
		\draw[->] (I1) .. controls (1.5,0.5) and (1.6,0.3) ..  (1.65,0);
		\draw (1.65,0) .. controls (1.6,-0.3) and (1.5,-0.5) ..  (I2);
		\draw (I1) .. controls (0.9,0.5) and (0.8,0.3) ..  (0.75,0);
		\draw[->] (I2).. controls (0.9,-0.5) and (0.8,-0.3)  ..  (0.75,0) ;
		\draw[->] (I1) -- (2,0.7);
		\draw(2,0.7) -- (B3);
	\end{tikzpicture}
\end{center}

By a dimensional analysis \cite{Costello:2020ndc}, one can show that the Feynman integral gives us a pole of order $3$. The Lie algebra factor gives us $f_{ab}^c$. Therefore the corresponding OPE takes the following form
\begin{equation}
	B_a[2,0](z)B_b[0,2](0) \sim \frac{f_{ab}^c}{z^3} B_c[0,0](0).
\end{equation}
We will see in Section \ref{sec:cele_loop} that the summation over the internal KK modes in the above diagram actually diverges. We assume that a proper regularization can be employed to make sense of the summation. 

There is a natural candidate for the corresponding OPE on the B brane side. We consider the following non-planar wick contraction
\begin{equation}
	\wick{(I_{a_1} \c3 Z^1 \c2 Z^1 \c1 J_{a_2} )(z)  ( \c1 I_{b_1} \c3 Z^{2} \c2 Z^2 J_{b_2})(0)} =  \frac{1}{z^3}\delta_{a_2b_1}I_{a_1} J_{b_2} (0).
\end{equation}
If we contract instead $I_{a_1}-J_{b_2}$, we get the term $\frac{1}{z^3}\delta_{a_1b_2}I_{a_2} J_{b_1} (0)$. Together with the above term, they reproduce exactly the boundary 1-loop OPE.

As we have discussed in Section \ref{sec:twist_higer}, the tree-level OPE actually defined a limit of the full chiral algebra and corresponds to the structure of a Poisson vertex algebra (or coisson algebra). Computing loop Feynman diagrams in the holomorphic topological theory can be understood mathematically as solving the deformation quantization problem of coisson algebra outlined in \cite{tamarkin2003deformations}.
\begin{equation*}
	\framebox[1.1\width]{Coisson algebra}
		\overset{\text{"deformation quantization"}}{\Longrightarrow}
	\framebox[1.1\width]{Full chiral algebra}\overset{?}{\leftrightsquigarrow}\framebox[1.1\width]{Full B-brane chiral algebra}
\end{equation*}
Suppose we can find a consistent regularization to deal with the divergence in summing the KK modes, then it will be interesting to match loops contribution to the boundary OPE with non-planar wick contractions of the B-brane chiral algebra.

	\section{Celestial holography from boundary chiral algebra}
	\label{sec:cele_holo}
	\subsection{Celestial holography and twistor correspondence}
	In this section, we review the various relationships between celestial holography and twistor correspondence following \cite{Costello:2022wso}.
	
	Recall that the twistor space $\mathbb{PT}$ of flat Euclidean spacetime can be identified with the total space of the rank $2$  bundle
	\begin{equation}
		\mathcal{O}(1)\oplus \mathcal{O}(1) \to \mathbb{CP}^1.
	\end{equation}
More invariantly, let $S_+$,$S_-$ denote the two spinor representation of $\mathrm{Spin}(4) = SU(2)\times SU(2)$. Then we can identify $\mathbb{PT}$ with the total space of the bundle $\mathcal{O}(1)\otimes S_{-} \to \mathbb{P}(S_+)$. The action of $\mathrm{Spin}(4)$ on $\mathbb{PT}$ is thus evident.

There is a non-holomorphic fibration 
\begin{equation}
	\pi : \mathbb{PT} \to \R^4
\end{equation}
called twistor fibration. It has the property that for any $u \in \R^4$, $\pi^{-1}(u) \cong \mathbb{CP}_u^1$ and is called a twistor line. The twistor fibration can be explicitly written in a local coordinate. Let $z$ be a coordinates on $\mathbb{CP}^1$ and $v_i$, $i = 1,2$ be the coordinates on the $\mathcal{O}(1)$ fiber. Then a point $(u_1,\dots,u_4) \in \R^4$ correspond to the twistor line $\mathbb{CP}_u^1$ defined by
\begin{equation}
	\begin{aligned}
		&v_1 = x_1 + ix_2 + z (x_3 - ix_4),\\
		&v_2 = x_3 + ix_4 -z(x_1 - ix_2),
	\end{aligned}
\end{equation}
which is also called the incidence relations.

The original Penrose-Ward correspondence \cite{ward_wells} states that there is a natural bijection between vector bundle on $\R^4$ equipped with self-dual connection and holomorphic vector bundles on $\mathbb{PT}$.

A more refined version of Penrose-Ward correspondence (see \cite{Boels:2006ir,Movshev:2008fu}) states that the self-dual Yang-Mills theory on $\R^4$ is equivalent to the holomorphic BF theory on $\mathbb{PT}$.

There are many variations of Penrose-Ward correspondence for holomorphic theory on twistor theory and four dimensional quantum field theories. We list some of them:
\begin{enumerate}
	\item Holomorphic Chern-Simons theory on twistor space and four dimensional WZW model \cite{Donaldson,Losev:1995cr,Bittleston:2020hfv}
	\item Holomorphic Chern-Simons theory on super twistor space $\mathbb{PT}^{3|4}$ and four dimensional self-dual $\mathcal{N} = 4$ SYM \cite{Witten:2003nn}.
	\item Holomorphic Poisson BF theory on twistor space and self-dual Einstein gravity \cite{Penrose:1976js,Mason:2007ct}.
\end{enumerate}

The relationships between celestial holography and twistor correspondence start by looking at the algebra of the twistor theory $\mathrm{Obs}_{6d}|_{\pi^{-1}(0)}$ restricted to the twistor line $\pi^{-1}(0) \cong \mathbb{CP}^1$. This is a dg chiral algebra, and its Koszul dual chiral algebra $(\mathrm{Obs}_{6d}|_{\pi^{-1}(0)})^{!}$ plays the fundamental role in this story. As is explained in \cite{Costello:2022wso}, we have the following dictionary between the chiral algebra $(\mathrm{Obs}_{6d}|_{\pi^{-1}(0)})^{!}$ and the four dimensional theory
\begin{enumerate}
	\item The generators of the chiral algebra are in bijection with single-particle conformal primary states of the four-dimensional theory.
	\item Conformal blocks of the chiral algebra are in bijection with local operators in the four-dimensional theory.
	\item Scattering amplitudes of the $4d$ theory in the presence of a chosen local operator at a fixed
	position are equal to correlation functions of the chiral algebra defined using the corresponding conformal block.
\end{enumerate}

In this section, we apply our results of the KK theory to compute the Koszul dual chiral algebra $(\mathrm{Obs}_{6d}|_{\pi^{-1}(0)})^{!}$ in holomorphic BF theory and holomorphic Poisson BF theory. We work in the open patch $\C^3 \subset \mathbb{PT}$ of twistor space. Note that the twistor line $\pi^{-1}(0)$ is exactly the boundary $\mathbb{CP}^1$ at $r = 0$ in our previous study. Therefore the chiral algebra $(\mathrm{Obs}_{6d}|_{\pi^{-1}(0)})^{!}$ can be computed by the boundary chiral algebra of the KK theory with transversal boundary condition. This is the same boundary condition we analyzed in the twisted holography chiral algebra.
	\subsection{Chiral algebra for holomorphic BF theory}
	\label{sec:cele_HBF}
	In this section, we study the (tree-level) boundary OPE for the holomorphic BF theory.
	
	The $3d$ KK theory of holomorphic BF theory have fields $\mathbf{A},\tilde{\mathbf{A}}, \mathbf{B},\tilde{\mathbf{B}}$. The original $6d$ fields are related to the KK fields as follows
		\begin{equation}
\begin{aligned}
			\pmb{\EuScript{A}} &=  \sum_{p,q\geq0} \mathbf{A}[p,q]z_1^pz_2^q + \frac{(p+q+1)!}{p!q!}\tilde{\mathbf{B}}[p,q]\frac{\bar{z}_1^p\bar{z}_2^q}{r^{2(p+q)}} \frac{\epsilon_{ij}\bar{z}_id\bar{z}_j}{r^4},\\
					\pmb{\EuScript{B}}& =  \sum_{p,q\geq0} \tilde{\mathbf{A}}[p,q]z_1^pz_2^q + \frac{(p+q+1)!}{p!q!}\mathbf{B}[p,q]\frac{\bar{z}_1^p\bar{z}_2^q}{r^{2(p+q)}} \frac{\epsilon_{ij}\bar{z}_id\bar{z}_j}{r^4}.
\end{aligned}
	\end{equation}
	We choose the boundary condition such that $\mathbf{A},\tilde{\mathbf{A}} = 0$ at the boundary. The boundary chiral algebra is generated by fields $B,\tilde{B}$ and their derivatives. In the celestial holography context, the $6d$ geometry is not deformed. The flat space action functional for the KK theory has all the information we need to compute the boundary chiral algebra. Given the interactions \ref{def_BF}, the tree-level Feynman diagrams computing the OPE of two boundary $B$ fields read as follows
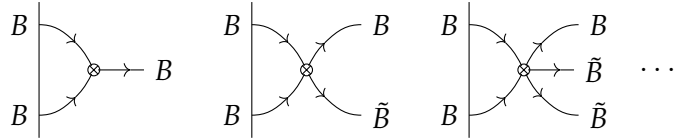
\begin{center}
	\begin{tikzpicture}[scale=0.6][>=stealth]
		\draw (0,-1.5) -- (0,1.5);
		\draw (1.2,0) node[crosss] (O) {} ;
		\draw (O) circle (0.13);
		\draw (0,1) node[left] (B1) {$B$};
		\draw (0,-1) node[left] (B2) {$B$};
		\draw[->] (B2) .. controls (0.3,-1) and (0.5,-0.9) ..  (0.8,-0.6);
		\draw (0.8,-0.6).. controls (0.9,-0.5) and (1,-0.4) ..  (O);
		\draw[->] (B1) .. controls (0.3,1) and (0.5,0.9) ..  (0.8,0.6);
		\draw (0.8,0.6).. controls (0.9,0.5) and (1,0.4) ..  (O);
		\draw (2.3,0) node[right] (B3) {$B$};
		\draw[->] (O) -- (2,0);
		\draw(2,0) -- (B3);
	\end{tikzpicture}\quad
		\begin{tikzpicture}[scale=0.6][>=stealth]
	\draw (0,-1.5) -- (0,1.5);
	\draw (1.2,0) node[crosss] (O) {} ;
	\draw (O) circle (0.13);
	\draw (0,1) node[left] (B1) {$B$};
	\draw (0,-1) node[left] (B2) {$B$};
	\draw[->] (B2) .. controls (0.3,-1) and (0.5,-0.9) ..  (0.8,-0.6);
	\draw (0.8,-0.6).. controls (0.9,-0.5) and (1,-0.4) ..  (O);
	\draw[->] (B1) .. controls (0.3,1) and (0.5,0.9) ..  (0.8,0.6);
	\draw (0.8,0.6).. controls (0.9,0.5) and (1,0.4) ..  (O);
	\draw (2.4,1) node[right] (B3) {$B$};
	\draw (2.4,-1) node[right] (B4) {$\tilde{B}$};
	\draw[->] (O) .. controls (1.4,0.4) and (1.5,0.5) .. (1.6,0.6);
	\draw (1.6,0.6) .. controls (1.9,0.9) and (2.1,1) .. (B3);
	\draw[->] (O) .. controls (1.4,-0.4) and (1.5,-0.5) .. (1.6,-0.6);
	\draw (1.6,-0.6) .. controls (1.9,-0.9) and (2.1,-1) .. (B4);
\end{tikzpicture}\quad
		\begin{tikzpicture}[scale=0.6][>=stealth]
	\draw (0,-1.5) -- (0,1.5);
	\draw (1.2,0) node[crosss] (O) {} ;
	\draw (O) circle (0.13);
	\draw (0,1) node[left] (B1) {$B$};
	\draw (0,-1) node[left] (B2) {$B$};
	\draw[->] (B2) .. controls (0.3,-1) and (0.5,-0.9) ..  (0.8,-0.6);
	\draw (0.8,-0.6).. controls (0.9,-0.5) and (1,-0.4) ..  (O);
	\draw[->] (B1) .. controls (0.3,1) and (0.5,0.9) ..  (0.8,0.6);
	\draw (0.8,0.6).. controls (0.9,0.5) and (1,0.4) ..  (O);
	\draw (2.4,1) node[right] (B3) {$B$};
	\draw (2.4,-1) node[right] (B4) {$\tilde{B}$};
	\draw[->] (O) .. controls (1.4,0.4) and (1.5,0.5) .. (1.6,0.6);
	\draw (1.6,0.6) .. controls (1.9,0.9) and (2.1,1) .. (B3);
	\draw[->] (O) .. controls (1.4,-0.4) and (1.5,-0.5) .. (1.6,-0.6);
	\draw (1.6,-0.6) .. controls (1.9,-0.9) and (2.1,-1) .. (B4);
	\draw (2.3,0) node[right] (B3) {$\tilde{B}$};
	\draw[->] (O) -- (2,0);
	\draw(2,0) -- (B3);
	\draw (3.5,0) node[right] (D) {$\cdots$};
\end{tikzpicture}
\end{center}
The corresponding boundary OPE takes the following schematic form 
	\begin{equation}
			B(z)B(0) \sim \frac{1}{z} \left( B(0) + (m_3)B\tilde{B}(0) + (m_4)B\tilde{B}\tilde{B}(0) \dots \right).
	\end{equation}
More precisely, the first term 
	\begin{equation}
	B_a(z)B_b(0) \sim \frac{f^{c}_{ab}}{z}  B_c(0)
\end{equation}
is the Kac-Moody chiral OPE. It is shown in \cite{Costello:2022wso} that the corresponding correlation function reproduces the Parke-Taylor formula. The second term is given by the quadratic interaction, a similar term also appears in the twisted holography chiral algebra in Section \ref{sec:OPE_HCS_4}. More explicitly, this term can be written as follows
	\begin{equation}
	B_a[p,q](z)B_b[r,s](0) \sim \sum_{\substack{u_1 + u_2 = p+r - 1\\v_1+v_2 = q+s - 1}}\sum_{c,d,e,f}\frac{K^{fc}}{2z} \left( (m_3)^{p,q;r,s}_{u_1,v_1}f_{bf}^ef_{ae}^d - (m_3)^{r,s;p,q}_{u_1,v_1}f_{af}^ef_{be}^d\right)  B_c[u_1,v_1]\tilde{B}_d[u_2,v_2](0),
\end{equation}
where we used the shorthand $(m_3)^{p,q;r,s}_{u,v} : = (m_3)^{p,q;r,s}_{u,v;p+r-1-u,q+s-1-v}$, which can be computed by the formula \ref{m3} and \ref{m3_2}. A special case of this formula is obtained in \cite{Costello:2022wso}.

Using the full action functional of the KK theory \ref{sec:flat_full} (though we write the action for holomorphic Chern-Simons, it is easy to generalize it to holomorphic BF theory), one can explicitly write down the remaining OPE. 
\begin{equation}
\begin{aligned}
&	B_a[p,q](z)B_b[r,s](0)\sim \sum_{n \geq 2} \frac{1}{(n-2)!} \sum_{\substack{u_1 + \dots u_{n-1} = p+r -  n +2\\v_1+\dots v_{n-1} = q+s - n+2}}\sum_{c_1,\dots,c_{n-1}}\sum_{\sigma \in S_{n-1}}(m_n)^{p,q;r,s}_{u_{\sigma(1)},v_{\sigma(1)};\dots;u_{\sigma(n-1)},v_{\sigma(n-1)}}\\
	&\times K(t_a,[...[[t_b,t^{c_{\sigma(1)}}],t^{c_{\sigma(2)}}],...,t^{c_{\sigma(n-1)}}])  B_{c_1}[u_1,v_1] \tilde{B}_{c_2}[u_2,v_2] \cdots \tilde{B}_{c_{n-1}}[u_{n-1},v_{n-1}].
\end{aligned}
\end{equation}
We recollect our results from Section \ref{sec:all_mn} that compute the constant $(m_n)^{p,q;r,s}_{u_{1},v_{1};\dots;u_{n-1},v_{n-1}}$. First, we change our variable as follows
\begin{equation}
	\begin{cases}
		&p = j_1 + m_1\\&q = j_1 - m_1
	\end{cases},\quad
	\begin{cases}
		&r = j_2 + m_2\\&s = j_2 - m_2
	\end{cases},\quad
	\begin{cases}
		&u_i = \bar{j}_i  - \bar{m}_i\\&v_i = \bar{j}_i + \bar{m}_i
	\end{cases}.
\end{equation}
Note that the constraint $u_1 + \dots u_{n-1} = p+r -  n +2$, $v_1+\dots v_{n-1} = q+s - n+2$ is equivalent to 
\begin{equation}
	\bar{j}_1 +\bar{j}_2 + \dots + \bar{j}_{n-1} = j_1 + j_2 - n + 2,\quad \bar{m}_1+ \bar{m}_2 + \dots \bar{m}_n = -(m_1 + m_2).
\end{equation}
Then we have
\begin{equation}
	\begin{aligned}
		&(m_n)^{p,q;r,s}_{u_{1},v_{1};\dots;u_{n-1},v_{n-1}} = \sqrt{\frac{(j_1 - m_1)!(j_1 + m_1)!(j_2 - m_2)!(j_1 + m_2)!}{(2j_1 + 1)!(2j_2 +1)!}}\left( \prod_{l = 1}^{n-1}\sqrt{\frac{(2\bar{j}_l + 1)!}{(\bar{j}_l + \bar{m}_l)!(\bar{j}_l - \bar{m}_l)!}}\right) 	 \\
		&\sum_{\substack{i_1,\dots,i_{n-2}\\i_{l-1}\leq i_l\leq \min\{2j_2 - l,2\bar{J}_l + l - 1\}}}\sqrt{2j_2 - n+3}\prod_{l = 2}^{n - 1} \sqrt{\frac{(2j_2 - l + 3)(2\bar{j}_l + 1)(2J_{l - 1} - 2i_{l-1} + 1)}{(2j_2 - l + 2 - i_{l-1})(2\bar{J}_{l-1}+ l - 1 - i_{l-1})}}\\
	&C^{j_2,\bar{j}_1;j_2+\bar{j}_1 - i_1}_{m_2,\bar{m}_1;m_2+\bar{m}_1}\left.\left( \prod_{l = 2}^{n - 1}\begin{Bmatrix}
		\bar{J}_{l-1} + \frac{l-1}{2}&j_2 - \frac{l-1}{2}& J_{l-1} - i_{l-1}\\ J_l - i_l& \bar{j}_l& \bar{J}_{l} + \frac{l-1}{2}
	\end{Bmatrix} C^{J_{l-1} - i_{l-1},\bar{j}_{l},J_{l} - i_{l}}_{m_2+\bar{M}_{l-1},m_{l},m_2+\bar{M}_{l}}\right)\right|_{i_{n - 1} = 2j_2 - n + 2} .
	\end{aligned}
\end{equation}
where 
\begin{equation}
	\begin{aligned}
		&\bar{J}_l = \bar{j}_1+ \dots+ \bar{j}_l,\\
		&J_l = j_2 + \bar{J}_l,\\
		&\bar{M}_l = \bar{m}_1 + \dots + \bar{m}_l.
	\end{aligned}
\end{equation}

Though our current computation of the boundary OPE is completely at the tree-level, the interaction vertices we used are effective interactions that already contain Feynman diagrams in themselves. These tree-level results should reproduce loop level computations in the Koszul dual picture. Indeed, the $BB \sim \frac{1}{z}B\tilde{B}$ OPE given above is from a one loop computation in \cite{Costello:2022wso}. More generally, the $BB \sim B\tilde{B}^{n}$ OPE should correspond to a $n$ loop computation in the Koszul dual picture and also a $n$ loop result in the form factor computation (However they are not the full loop results).

Similarly, The tree-level Feynman diagrams for a boundary $B$ field and a $\tilde{B}$ field consist of the following
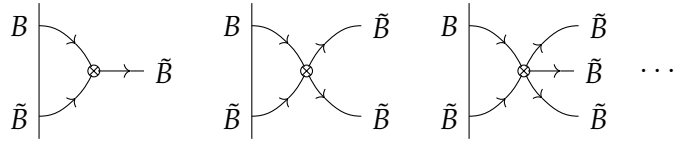
\begin{center}
	\begin{tikzpicture}[scale=0.6][>=stealth]
		\draw (0,-1.5) -- (0,1.5);
		\draw (1.2,0) node[crosss] (O) {} ;
		\draw (O) circle (0.13);
		\draw (0,1) node[left] (B1) {$B$};
		\draw (0,-1) node[left] (B2) {$\tilde{B}$};
		\draw[->] (B2) .. controls (0.3,-1) and (0.5,-0.9) ..  (0.8,-0.6);
		\draw (0.8,-0.6).. controls (0.9,-0.5) and (1,-0.4) ..  (O);
		\draw[->] (B1) .. controls (0.3,1) and (0.5,0.9) ..  (0.8,0.6);
		\draw (0.8,0.6).. controls (0.9,0.5) and (1,0.4) ..  (O);
		\draw (2.3,0) node[right] (B3) {$\tilde{B}$};
		\draw[->] (O) -- (2,0);
		\draw(2,0) -- (B3);
	\end{tikzpicture}\quad
	\begin{tikzpicture}[scale=0.6][>=stealth]
		\draw (0,-1.5) -- (0,1.5);
		\draw (1.2,0) node[crosss] (O) {} ;
		\draw (O) circle (0.13);
		\draw (0,1) node[left] (B1) {$B$};
		\draw (0,-1) node[left] (B2) {$\tilde{B}$};
		\draw[->] (B2) .. controls (0.3,-1) and (0.5,-0.9) ..  (0.8,-0.6);
		\draw (0.8,-0.6).. controls (0.9,-0.5) and (1,-0.4) ..  (O);
		\draw[->] (B1) .. controls (0.3,1) and (0.5,0.9) ..  (0.8,0.6);
		\draw (0.8,0.6).. controls (0.9,0.5) and (1,0.4) ..  (O);
		\draw (2.4,1) node[right] (B3) {$\tilde{B}$};
		\draw (2.4,-1) node[right] (B4) {$\tilde{B}$};
		\draw[->] (O) .. controls (1.4,0.4) and (1.5,0.5) .. (1.6,0.6);
		\draw (1.6,0.6) .. controls (1.9,0.9) and (2.1,1) .. (B3);
		\draw[->] (O) .. controls (1.4,-0.4) and (1.5,-0.5) .. (1.6,-0.6);
		\draw (1.6,-0.6) .. controls (1.9,-0.9) and (2.1,-1) .. (B4);
	\end{tikzpicture}\quad
	\begin{tikzpicture}[scale=0.6][>=stealth]
		\draw (0,-1.5) -- (0,1.5);
		\draw (1.2,0) node[crosss] (O) {} ;
		\draw (O) circle (0.13);
		\draw (0,1) node[left] (B1) {$B$};
		\draw (0,-1) node[left] (B2) {$\tilde{B}$};
		\draw[->] (B2) .. controls (0.3,-1) and (0.5,-0.9) ..  (0.8,-0.6);
		\draw (0.8,-0.6).. controls (0.9,-0.5) and (1,-0.4) ..  (O);
		\draw[->] (B1) .. controls (0.3,1) and (0.5,0.9) ..  (0.8,0.6);
		\draw (0.8,0.6).. controls (0.9,0.5) and (1,0.4) ..  (O);
		\draw (2.4,1) node[right] (B3) {$\tilde{B}$};
		\draw (2.4,-1) node[right] (B4) {$\tilde{B}$};
		\draw[->] (O) .. controls (1.4,0.4) and (1.5,0.5) .. (1.6,0.6);
		\draw (1.6,0.6) .. controls (1.9,0.9) and (2.1,1) .. (B3);
		\draw[->] (O) .. controls (1.4,-0.4) and (1.5,-0.5) .. (1.6,-0.6);
		\draw (1.6,-0.6) .. controls (1.9,-0.9) and (2.1,-1) .. (B4);
		\draw (2.3,0) node[right] (B3) {$\tilde{B}$};
		\draw[->] (O) -- (2,0);
		\draw(2,0) -- (B3);
		\draw (3.5,0) node[right] (D) {$\cdots$};
	\end{tikzpicture}
\end{center}
The corresponding boundary OPE is given by
\begin{equation}
\begin{aligned}
		&B_a[p,q](z)\tilde{B}_b[r,s](0)\sim \frac{f^c_{ab}}{z}\tilde{B}_c[p+r,q+s](0)\\
		+&\sum_{n \geq 3} \frac{1}{(n-1)!} \sum_{\substack{u_1 + \dots u_{n-1} = p+r -  n +2\\v_1+\dots v_{n-1} = q+s - n+2}}\sum_{c_1,\dots,c_{n-1}}\sum_{\sigma \in S_{n-1}}(m_n)^{p,q;r,s}_{u_{\sigma(1)},v_{\sigma(1)};\dots;u_{\sigma(n-1)},v_{\sigma(n-1)}}\\
		& K(t_a,[...[[t_b,t^{c_{\sigma(1)}}],t^{c_{\sigma(2)}}],...,t^{c_{\sigma(n-1)}}])  \tilde{B}_{c_1}[u_1,v_1] \tilde{B}_{c_2}[u_2,v_2] \cdots \tilde{B}_{c_{n-1}}[u_{n-1},v_{n-1}].
\end{aligned}
\end{equation}

As is explained in \cite{Costello:2022wso}, holomorphic BF theory alone suffers from an anomaly on twistor space. As a result, we might eventually get into trouble with the boundary chiral algebra at the full quantum level. For holomorphic BF theory, the anomaly can be canceled by coupling with the free limit of Kodaira-Spencer theory via a Green-Schwarz mechanism \cite{Costello:2019jsy}. In four dimension, the Kodaira-Spencer field introduces an axion field coupled with the self-dual gauge theory. For the boundary chiral algebra, the Kodaira Spencer fields introduce two towers of boundary generators. We can easily write their OPE via the coupling of KK theory we obtained in Section \ref{sec:coup_HBF}
\begin{equation}
\begin{aligned}
		B_a[p,q](z)\beta[r,s](0) &\sim \frac{ps-qr}{z}\tilde{B}_a[p+r-1,q+s - 1],\\
		B_a[p,q](z)\tilde{\beta}[r,s](0) &\sim  - \frac{(p+q+r+s+2)}{z^2}\tilde{B}_a[p+r,q+s] + \frac{r+s+2}{z}\pa_z\tilde{B}_a[p+r,q+s],\\
		B_a[p,q]B_b[r,s] &\sim -\frac{K_{ab}}{z^2}\beta[p+r,q+s] + \frac{r+s}{p+r+q+s}\frac{K_{ab}}{z}\pa_z\beta[p+r,q+s]\\
		&- \frac{ps - qr}{p+q+r+s}\frac{K_{ab}}{z}\tilde{\beta}[p+r-1,q+s-1].
\end{aligned}
\end{equation} 
Under a rescaling $\beta[r,s] \to (r+s)\beta[r,s],\tilde{\beta}[r,s] \to (r+s+2)\tilde{\beta}[r,s]$, we reproduce the OPE in \cite{Costello:2022wso} computed using Koszul duality \footnote{In \cite{Costello:2022wso}, the authors used a different notation for the fields, which is related to ours by $B \to J,\tilde{B} \to \tilde{J},\beta \to  E,\tilde{\beta} \to F$.} 

It won't be hard to generalize our method to compute all the coupling of the KK theory of Kodaira-Spencer gravity and holomorphic BF theory. This should lead to a complete understanding of the "tree-level" correction to the above OPE. We leave this problem to future work.

	\subsection{Comments on loop corrections}
	\label{sec:cele_loop}
	
	So far, we have only considered tree-level boundary OPE. Loop corrections to the boundary OPE are also important. In this section, we briefly analyze some loop effects. We will see that they does not behave so well in the KK theory. Divergence happens when we sum over the infinite KK tower and appropriate regularization is needed.
	
	As an example, we consider the following loop diagrams.
	
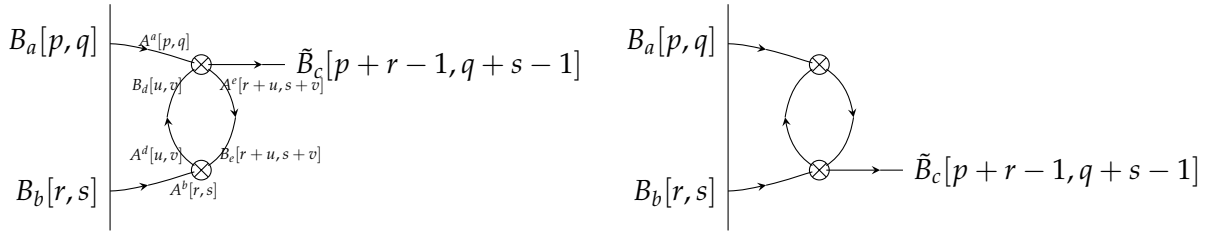
\begin{center}
		\begin{tikzpicture}[>=stealth]
	\draw (0,-1.5) -- (0,1.5);
	\draw (1.2,0.7) node[cross] (I1) {} ;
	\draw (1.2,1) node[left] (B1) {\scalebox{0.55}{$A^a[p,q]$}};
	\draw (1.1,0.42) node[left] (B1) {\scalebox{0.55}{$B_d[u,v]$}};
	\draw (1.3,0.42) node[right] (B1) {\scalebox{0.55}{$A^e[r+u,s+v]$}};
	\draw (I1) circle (0.13);
	\draw (1.2,-0.7) node[cross] (I2) {} ;
	\draw (1.1,-0.7) node[below] (B1) {\scalebox{0.55}{$A^b[r,s]$}};
	\draw (1.1,-0.5) node[left] (B1) {\scalebox{0.55}{$A^d[u,v]$}};
	\draw (1.3,-0.5) node[right] (B1) {\scalebox{0.55}{$B_e[r+u,s+v]$}};
	\draw (I2) circle (0.13);
	\draw (0,1) node[left] (B1) {$B_a[p,q]$};
	\draw (0,-1) node[left] (B2) {$B_b[r,s]$};
	\draw (2.3,0.7) node[right] (B3) {$\tilde{B}_c[p+r - 1,q + s - 1]$};
	\draw[->] (B1) .. controls (0.2,0.97) and (0.3,0.94) ..  (0.5,0.9);
	\draw(0.5,0.9) .. controls (0.8,0.83) and (1,0.75) ..  (I1);
	\draw[->] (B2) .. controls (0.2,-0.97) and (0.3,-0.94) ..  (0.5,-0.9);
	\draw(0.5,-0.9) .. controls (0.8,-0.83) and (1,-0.75) ..  (I2);
	\draw[->] (I1) .. controls (1.5,0.5) and (1.6,0.3) ..  (1.65,0);
	\draw (1.65,0) .. controls (1.6,-0.3) and (1.5,-0.5) ..  (I2);
	\draw (I1) .. controls (0.9,0.5) and (0.8,0.3) ..  (0.75,0);
	\draw[->] (I2).. controls (0.9,-0.5) and (0.8,-0.3)  ..  (0.75,0) ;
	\draw[->] (I1) -- (2,0.7);
	\draw(2,0.7) -- (B3);
\end{tikzpicture}\;\;
\begin{tikzpicture}[>=stealth]
	\draw (0,-1.5) -- (0,1.5);
	\draw (1.2,0.7) node[cross] (I1) {} ;
	\draw (I1) circle (0.13);
	\draw (1.2,-0.7) node[cross] (I2) {} ;
	\draw (I2) circle (0.13);
	\draw (0,1) node[left] (B1) {$B_a[p,q]$};
	\draw (0,-1) node[left] (B2) {$B_b[r,s]$};
	\draw (2.3,-0.7) node[right] (B3) {$\tilde{B}_c[p+r - 1,q+s - 1]$};
	\draw[->] (B1) .. controls (0.2,0.97) and (0.3,0.94) ..  (0.5,0.9);
	\draw(0.5,0.9) .. controls (0.8,0.83) and (1,0.75) ..  (I1);
	\draw[->] (B2) .. controls (0.2,-0.97) and (0.3,-0.94) ..  (0.5,-0.9);
	\draw(0.5,-0.9) .. controls (0.8,-0.83) and (1,-0.75) ..  (I2);
	\draw[->] (I1) .. controls (1.5,0.5) and (1.6,0.3) ..  (1.65,0);
	\draw (1.65,0) .. controls (1.6,-0.3) and (1.5,-0.5) ..  (I2);
	\draw (I1) .. controls (0.9,0.5) and (0.8,0.3) ..  (0.75,0);
	\draw[->] (I2).. controls (0.9,-0.5) and (0.8,-0.3)  ..  (0.75,0) ;
	\draw[->] (I2) -- (2,-0.7);
	\draw(2,-0.7) -- (B3);
\end{tikzpicture}
	\end{center}
The corresponding Feynman integral is actually difficult to evaluate and we won't compute it in this paper. Nevertheless, simply by a dimensional analysis \cite{Costello:2020ndc}, one can show that it gives us a pole of order $2$. So we expect to have OPE of the following form
\begin{equation}
	B(z)B(0) \sim \frac{1}{z^2}\tilde{B}(0).
\end{equation}
The corresponding color factor and the KK modes factor can be evaluated by contracting the cubic and quartic interaction vertices. The first diagram gives us the following 
\begin{equation}\label{coef_one_loop}
	\left( \sum_{u,v}\sum_{d,e,f,g}	(m_3)^{p,q;r+u,s+v}_{u,v} K^{df}f_{ef}^gf_{ag}^c + (m_3)^{r+u,s+v;p,q}_{u,v}K^{df}f_{af}^gf_{eg}^c\right)f_{bd}^e.
\end{equation}
The first term gives the Lie algebra constant $2h^\vee f_{ab}^c $, where $h^\vee$ is the dual coxter number. To compute the constant factor from the KK modes, it will be easier to label the KK modes using quantum mechanics notation. We change our notation by the following
\begin{equation}
	\begin{cases}
		&p = j_1 + m_1\\&q = j_1 - m_1
	\end{cases},\quad
\begin{cases}
	&r = j_2 + m_2\\&s = j_2 - m_2
\end{cases},\quad
\begin{cases}
	&u = \bar{j}  - \bar{m}\\&v = \bar{j} + \bar{m}
\end{cases}.
\end{equation}
Recall that we denoted $N(j,m) = \sqrt{\frac{(j - m)!(j+m)!}{(2j+1)!}}$. Then we have
\begin{equation}
	\begin{aligned}
		&\sum_{u,v}	(m_3)^{p,q;r+u,s+v}_{u,v} = \sum_{\bar{j},\bar{m}}\sum_i\frac{(-1)^{\bar{j} - \bar{m}}N(j_1,m_1)N(j_2 + \bar{j},m_2 - \bar{m})}{N(\bar{j},\bar{m})N(j_1+j_2 - 1,m_1 + m_2)}\begin{Bmatrix}
			j_2 + \bar{j} - \frac{1}{2}&\bar{j} + \frac{1}{2}&j_2 + 2\bar{j} - i\\j_1 + j_2 - 1&j_1&j_1 + j_2 + \bar{j} - \frac{1}{2}
		\end{Bmatrix}\\
	&(-1)^{2j_1 + 2j_2  -i + 1}\sqrt{\frac{(2j_1 + 1)(2j_2 + 2\bar{j})(2j_2 + 2\bar{j} + 1)(2j_2 + 4\bar{j} - 2i + 1)}{(2j_2 +2\bar{j} - i)(2\bar{j} - i +1)}} C^{j_2 + \bar{j},\bar{j};j_2+2\bar{j} - i}_{m_2 - \bar{m},\bar{m};m_2}C^{j_2+2\bar{j} - i,j_1;j_1+j_2 - 1}_{m_2,m_1;m_1+m_2}.\\
	\end{aligned}
\end{equation}
Note that $\frac{N(j_2 + \bar{j},m_2 - \bar{m})}{N(j_2,m_2)N(\bar{j},\bar{m})} = (-1)^{\bar{j} + \bar{m}} \sqrt{2\bar{j} + 1} C^{j_2 + \bar{j},\bar{j};j_2}_{m_2 - \bar{m},\bar{m};m_2}$. Moreover, $\sum_{\bar{m}}C^{j_2 + \bar{j},\bar{j};j_2}_{m_2 - \bar{m},\bar{m};m_2}  C^{j_2 + \bar{j},\bar{j};j_2+2\bar{j} - i}_{m_2 - \bar{m},\bar{m};m_2} =  \delta_{j_2,j_2 + 2\bar{j} - i}$. Therefore, summation over $i$ in the above formula can be simplified to letting $i = 2\bar{j}$. We also have
\begin{equation}
	\begin{Bmatrix}
		j_2 + \bar{j} - \frac{1}{2}&\bar{j} + \frac{1}{2}&j_2 \\j_1 + j_2 - 1&j_1&j_1 + j_2 + \bar{j} - \frac{1}{2}
	\end{Bmatrix} = \frac{(-1)^{2j_1 + 2j_2 + 2\bar{j} - 1}\sqrt{(2j_1)(2\bar{j} + 1)}}{\sqrt{(2j_1 + 2j_2)(2j_1 + 2j_2 - 1)(2j_2 + 2\bar{j})(2j_2 + 2\bar{j} + 1)}}.
\end{equation}
Consequently,
\begin{equation}
	\begin{aligned}
	&\sum_{u,v}	(m_3)^{p,q;r+u,s+v}_{u,v}  \\
	& =\sum_{\bar{j}\in\frac{1}{2}\Z_{\geq 0}}(2\bar{j} + 1)\frac{N(j_1,m_1)N(j_2,m_2)}{N(j_1+j_2 - 1,m_1 + m_2)}\sqrt{\frac{(2j_1)(2j_1 + 1)(2j_2  + 1)}{(2j_2)(2j_1+2j_2)(2j_1 +2j_2 - 1)}}C^{j_2,j_1;j_1+j_2 - 1}_{m_2,m_1;m_1+m_2}\\
	& = \left( \sum_{n \geq 0}(n + 1)\right) \frac{(ps - qr)}{(r+s)(p+q + r +s)}.
	\end{aligned}
\end{equation}
Unfortunately, the summation $ \sum_{n \geq 0}(n + 1)$ does not converge. Nevertheless, it still gives us the dependence on the $p,q,r,s$ indices. 

The second term in \ref{coef_one_loop} gives a Lie algebra factor $h^\vee f_{ab}^c$. The KK modes part is computed by 
\begin{equation}
	\begin{aligned}
		&\sum_{u,v}	(m_3)^{r+u,s+v;p,q}_{u,v} = 	 \sum_{\bar{j},\bar{m}}\sum_i\frac{(-1)^{2j_1 + 2j_2 + 2\bar{j} - i + 1}N(j_1,m_1)N(j_2,m_2)}{N(j_1+j_2 - 1,m_1 + m_2)}\begin{Bmatrix}
			j_1 - \frac{1}{2}&\bar{j} + \frac{1}{2}&j_1 + \bar{j} - i\\j_1 + j_2 - 1&j_2 + \bar{j}&j_1 + j_2 + \bar{j} - \frac{1}{2}
		\end{Bmatrix}\\
		&\sqrt{\frac{(2j_2 + 2\bar{j} + 1)(2j_1)(2j_1 + 1)(2j_1 + 2\bar{j} - 2i + 1)(2\bar{j} + 1)}{(2j_1  - i)(2\bar{j} - i +1)}} C^{j_2 + \bar{j},\bar{j};j_2}_{m_2 - \bar{m},\bar{m};m_2} C^{j_1,\bar{j};j_1+\bar{j} - i}_{m_1,\bar{m};m_1+\bar{m}}C^{j_1+\bar{j} - i,j_2 + \bar{j};j_1+j_2 - 1}_{m_1+\bar{m},m_2 - \bar{m};m_1+m_2}\\	
		& = \sum_{\bar{j}}\sum_i\frac{(-1)^{2\bar{j} - i + 1}N(j_1,m_1)N(j_2,m_2)}{N(j_1+j_2 - 1,m_1 + m_2)}\sqrt{\frac{(2j_2 + 2\bar{j} + 1)(2j_1)(2j_1 + 1)(2j_2 + 1)(2\bar{j} + 1)}{(2j_1  - i)(2\bar{j} - i +1)}} C^{j_1,j_2;j_1+j_2-1}_{m_1,m_2;m_1+m_2}\\
		&(2j_1 + 2\bar{j} - 2i + 1)\begin{Bmatrix}
			j_1 - \frac{1}{2}&\bar{j} + \frac{1}{2}&j_1 + \bar{j} - i\\j_1 + j_2 - 1&j_2 + \bar{j}&j_1 + j_2 + \bar{j} - \frac{1}{2}
		\end{Bmatrix}\begin{Bmatrix}
			\bar{j}&j_2 + \bar{j}&j_2\\j_1+j_2 - 1&j_1&j_1 + \bar{j} - i
		\end{Bmatrix}.
	\end{aligned}
\end{equation}
In this case, summation over $\bar{j}$ is also divergent. We let
\begin{equation}
		\sum_{u,v}	(m_3)^{r+u,s+v;p,q}_{u,v}  = D_{p+q,r+s}\times(ps - qr),
\end{equation}
where $D_{p+q,r+s}$ is a constant not determined. The second Feynman diagrams need no extra computation. We simply exchange the Lie algebra indices $a \leftrightarrow b$ and the KK modes $(p,q) \leftrightarrow (r,s)$ in the above formulas. In summary, the one loop computation gives us the following OPE
\begin{equation}
	B_a[p,q] B_b[r,s] \sim D_{p+q,r+s}'\frac{(ps - qr) h^\vee f^c_{ab}}{z^2}B_c[p+r-1,q + s - 1],
\end{equation}
where $D_{p+q,r+s}'$ are regularization constants that cannot be determined by our method. Note that this OPE is derived in \cite{Costello:2022upu} that correspond to the one loop splitting amplitudes of self-dual Yang-Mills. The coefficient of this OPE is also constrained in \textit{loc. cit.} by the associativity of the chiral algebra after adding the axion part. 

It is an interesting problem to compute other one loop Feynman diagrams. They take the following form
	
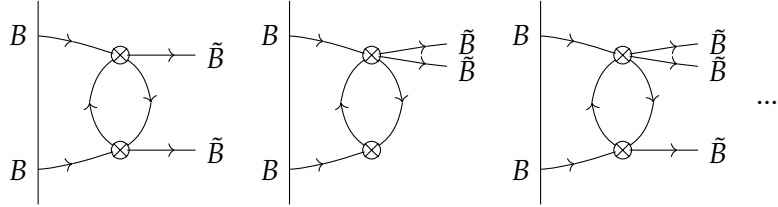
\begin{center}
	\begin{tikzpicture}[scale=0.9][>=stealth]
		\draw (0,-1.5) -- (0,1.5);
		\draw (1.2,0.7) node[cross] (I1) {} ;
		\draw (I1) circle (0.13);
		\draw (1.2,-0.7) node[cross] (I2) {} ;
		\draw (I2) circle (0.13);
		\draw (0,1) node[left] (B1) {$B$};
		\draw (0,-1) node[left] (B2) {$B$};
		\draw (2.3,-0.7) node[right] (B3) {$\tilde{B}$};
		\draw (2.3,0.7) node[right] (B4) {$\tilde{B}$};
		\draw[->] (B1) .. controls (0.2,0.97) and (0.3,0.94) ..  (0.5,0.9);
		\draw(0.5,0.9) .. controls (0.8,0.83) and (1,0.75) ..  (I1);
		\draw[->] (B2) .. controls (0.2,-0.97) and (0.3,-0.94) ..  (0.5,-0.9);
		\draw(0.5,-0.9) .. controls (0.8,-0.83) and (1,-0.75) ..  (I2);
		\draw[->] (I1) .. controls (1.5,0.5) and (1.6,0.3) ..  (1.65,0);
		\draw (1.65,0) .. controls (1.6,-0.3) and (1.5,-0.5) ..  (I2);
		\draw (I1) .. controls (0.9,0.5) and (0.8,0.3) ..  (0.75,0);
		\draw[->] (I2).. controls (0.9,-0.5) and (0.8,-0.3)  ..  (0.75,0) ;
		\draw[->] (I2) -- (2,-0.7);
		\draw(2,-0.7) -- (B3);
		\draw[->] (I1) -- (2,0.7);
		\draw(2,0.7) -- (B4);
	\end{tikzpicture}\;
	\begin{tikzpicture}[scale=0.9][>=stealth]
	\draw (0,-1.5) -- (0,1.5);
	\draw (1.2,0.7) node[cross] (I1) {} ;
	\draw (I1) circle (0.13);
	\draw (1.2,-0.7) node[cross] (I2) {} ;
	\draw (I2) circle (0.13);
	\draw (0,1) node[left] (B1) {$B$};
	\draw (0,-1) node[left] (B2) {$B$};
	\draw (2.3,0.5) node[right] (B4) {$\tilde{B}$};
	\draw (2.3,0.9) node[right] (B5) {$\tilde{B}$};
	\draw[->] (B1) .. controls (0.2,0.97) and (0.3,0.94) ..  (0.5,0.9);
	\draw(0.5,0.9) .. controls (0.8,0.83) and (1,0.75) ..  (I1);
	\draw[->] (B2) .. controls (0.2,-0.97) and (0.3,-0.94) ..  (0.5,-0.9);
	\draw(0.5,-0.9) .. controls (0.8,-0.83) and (1,-0.75) ..  (I2);
	\draw[->] (I1) .. controls (1.5,0.5) and (1.6,0.3) ..  (1.65,0);
	\draw (1.65,0) .. controls (1.6,-0.3) and (1.5,-0.5) ..  (I2);
	\draw (I1) .. controls (0.9,0.5) and (0.8,0.3) ..  (0.75,0);
	\draw[->] (I2).. controls (0.9,-0.5) and (0.8,-0.3)  ..  (0.75,0) ;
	\draw[->] (I1) -- (2,0.57);
	\draw  (2,0.57)-- (B4);
	\draw[->] (I1) -- (2,0.83);
	\draw  (2,0.83)-- (B5);
\end{tikzpicture}\;
	\begin{tikzpicture}[scale=0.9][>=stealth]
	\draw (0,-1.5) -- (0,1.5);
	\draw (1.2,0.7) node[cross] (I1) {} ;
	\draw (I1) circle (0.13);
	\draw (1.2,-0.7) node[cross] (I2) {} ;
	\draw (I2) circle (0.13);
	\draw (0,1) node[left] (B1) {$B$};
	\draw (0,-1) node[left] (B2) {$B$};
	\draw (2.3,-0.7) node[right] (B3) {$\tilde{B}$};
	\draw (2.3,0.5) node[right] (B4) {$\tilde{B}$};
	\draw (2.3,0.9) node[right] (B5) {$\tilde{B}$};
	\draw[->] (B1) .. controls (0.2,0.97) and (0.3,0.94) ..  (0.5,0.9);
	\draw(0.5,0.9) .. controls (0.8,0.83) and (1,0.75) ..  (I1);
	\draw[->] (B2) .. controls (0.2,-0.97) and (0.3,-0.94) ..  (0.5,-0.9);
	\draw(0.5,-0.9) .. controls (0.8,-0.83) and (1,-0.75) ..  (I2);
	\draw[->] (I1) .. controls (1.5,0.5) and (1.6,0.3) ..  (1.65,0);
	\draw (1.65,0) .. controls (1.6,-0.3) and (1.5,-0.5) ..  (I2);
	\draw (I1) .. controls (0.9,0.5) and (0.8,0.3) ..  (0.75,0);
	\draw[->] (I2).. controls (0.9,-0.5) and (0.8,-0.3)  ..  (0.75,0) ;
	\draw[->] (I2) -- (2,-0.7);
	\draw(2,-0.7) -- (B3);
	\draw[->] (I1) -- (2,0.57);
	\draw  (2,0.57)-- (B4);
	\draw[->] (I1) -- (2,0.83);
	\draw  (2,0.83)-- (B5);
	\draw (3,0) node[right] {$...$};
\end{tikzpicture}
\end{center}
The Feynman integrals are the same for all the one loop diagrams and produce a $\frac{1}{z^2}$ pole. Therefore, we expect that one loop corrections of the boundary algebra take the following schematic form
\begin{equation}
	B(z)B(0) \sim \frac{1}{z^2}(\tilde{B}(0) + \tilde{B}\tilde{B}(0) + \dots ).
\end{equation}
The precise coefficient can be accessed by contracting the corresponding higher order interactions vertices. We should expect that they all diverge when we sum over the KK modes of the internal line. On the one hand, we might hope that there is a regularization procedure that can make scenes of all the divergence. On the other hand, as our previous example suggests, the divergent answer still tell us some useful information about the OPE coefficient up to the regularization constant. We might hope that the associativity of the chiral algebra is strong enough to determine all the one loop OPE.

	\subsection{Chiral algebra from Poisson BF theory}
	The boundary chiral algebra for Poisson BF theory can also be obtained from the interactions vertices of the KK theory. Using the expression \ref{KK_Pois}, we have the following tree-level boundary OPE
	\begin{equation}
\begin{aligned}
			&w[p,q](z)w[r,s] \sim \frac{pr-qs}{z} w[p+r - 1,q+s - 1] + \\
			&\frac{1}{z} \sum_{n \geq 3} \frac{1}{2(n-2)!} \sum_{\substack{u_1 + \dots u_{n-1} = p+r -  2n +3\\v_1+\dots v_{n-1} = q+s - 2n+3}}\left( \sum_{\sigma \in S_{n-1}}(\pi_n)^{p,q;r,s}_{u_{\sigma(1)},v_{\sigma(1)},\dots,u_{\sigma(n-1)},v_{\sigma(n-1)}}\right)w[u_1,v_1]\tilde{w}[u_2,v_2]\dots \tilde{w}[u_{n-1},v_{n-1}]
\end{aligned}
	\end{equation} 
and 
	\begin{equation}
\begin{aligned}
		&w[p,q](z)\tilde{w}[r,s] \sim \frac{pr-qs}{z} \tilde{w}[p+r - 1,q+s - 1] + \\
		& \frac{1}{z}\sum_{n \geq 3} \frac{1}{(n-1)!} \sum_{\substack{u_1 + \dots u_{n-1} = p+r -  2n +3\\v_1+\dots v_{n-1} = q+s - 2n+3}}\left( \sum_{\sigma \in S_{n-1}}(\pi_n)^{p,q;r,s}_{u_{\sigma(1)},v_{\sigma(1)},\dots,u_{\sigma(n-1)},v_{\sigma(n-1)}}\right)\tilde{w}[u_1,v_1]\dots \tilde{w}[u_{n-1},v_{n-1}].
\end{aligned}
\end{equation} 
We compute the constant $(\pi_3)^{p,q;r,s}_{u_1,v_1;u_2,v_2}$ in Appendix \ref{apx:Poi_3}. The leading order term corresponds to the Kac-Moody algebra for the Lie algebra of Hamiltonian vector fields \cite{Costello:2022wso}. OPE's corresponding to the $\pi_3$ terms are also studied in \cite{Bittleston:2022jeq} as loop corrections.

Poisson BF theory have anomaly on twistor space. This is analyzed in detail in \cite{Bittleston:2022nfr}, where the authors also proposed a twistor theory to cancel the anomaly. It will be interesting to study the corresponding KK theory. 

One can also study loop correction to the above tree-level OPE. A naive computation of the loop diagram should diverge as in holomorphic BF theory. However, as shown in \cite{Bittleston:2022jeq}, associativity imposes strong constraints on the OPE. It will be interesting to study the constraint of associativity on higher order corrections to the OPE.

\section{Applications to $4d$ theories}
\label{sec:4d}
	\subsection{$4d$ holomorphic theories}
	\label{sec:4dKK}
	Our results on the dimensional reduction can also be applied to $4d$ holomorphic theory build on the Dolbeault complex $(\Omega^{k,\sbullet}(\C^2) , \pab)$. We can perform KK reduction along the unit sphere of $\C^2$.
	\begin{center}
		\begin{tikzpicture}[scale=0.7]
			\node (0) at (0,0) {$S^3$};
			\node (1) at (2,0) {$\C^2\backslash\{0\}$};
			\node (2) at (2,-2) {$ \R_{ > 0}$};
			\draw [->][shorten >=1pt,shorten <=1pt] (1) -- (2);
			\draw [->][shorten >=1pt,shorten <=1pt] (0) -- (1);
		\end{tikzpicture}
	\end{center}
Analogous to Proposition \ref{K}, we have the following isomorphism
	\begin{equation}
		(\Omega^{0,\sbullet}(\C^2\backslash \{0\}) , \pab, \wedge ) \approx (\Omega^{\sbullet}(\R_{ > 0})\otimes \Omega_b^{0,\sbullet}(S^3) , d_t + \pacr, \wedge ),
	\end{equation}
	and an $A_\infty$ quasi isomorphism
	\begin{equation}
		(\Omega^{\sbullet}(\R_{ > 0})\otimes \Omega_b^{0,\sbullet}(S^3) , d_t + \pacr, \wedge ) \to (\Omega^{\sbullet}(\R_{ > 0})\otimes H_{b}^{0,\sbullet}(S^3) , d_t , \{m_n\}_{n \geq 2} ),
	\end{equation}
	
	This tells us that for $4d$ holomorphic theory, the KK reduction along the unit sphere will produce a topological quantum mechanical system on $\R_{> 0 }$.
	
	There is a natural boundary condition on $t = 0$, chosen so that the fields extend to the whole $\C^2$.  Under this boundary condition, the boundary operators at $t = 0$ should reproduce the local operators of the original $4d$ theory restricted at the origin.
	
	The factorization algebra structure on the half line $\R_{\geq 0 }$ should produce an $A_\infty$ algebra together with an $A_\infty$ module. 
	
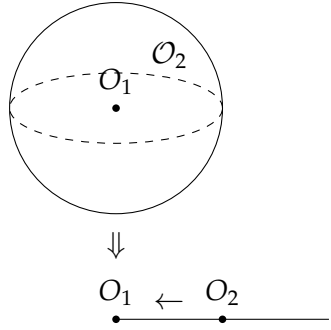
\begin{center}
		\begin{tikzpicture}[scale=0.7]
			
			\def\r{2}
			
			\draw (0,0) node[circle, fill, inner sep=1,label=above:$O_1$] (O1) {} ;
			\draw (\r/2,\r/2) node (O2) {$\mathcal{O}_2$};
			
			\draw (O1) circle (\r);
			\draw[dashed] (O1) ellipse (\r{} and \r/3);
			
			\draw (0,-2*\r) node[circle, fill, inner sep=1,label=above:$O_1$] (O3) {} -- (\r,-2*\r)
			node[circle, fill,inner sep=1,label=above:$O_2$] (O4) {} -- (2*\r,-2*\r);
			\draw (0.5*\r,-2*\r) node[above] {$\leftarrow$};
			\draw (0,-1.3*\r) node {$\Downarrow$};
		\end{tikzpicture}
		\end{center}
	We expect the following dictionary between the $1d$ KK theory and the original $4d$ pictures
\begin{center}
	\begin{tikzpicture}[scale=0.7]
		\node (10) at (-4,0) {$ 1d$ theory};
		\node (11) at (4,0) {$4d$ chiral algebra };
		\node (20) at (-4,-1.5) {boundary module $M_\pa$};
		\node (21) at (4,-1.5) {vacuum module $V$};
		\node (30) at (-4,-3) {bulk algebra $A$};
		\node (31) at (4,-3) {mode algebra $\oint_{S^3} V$};
		\node (4) at (0,-1.5) {$\cong$ };
		\node (5) at (0,-3) {$\cong$ };
	\end{tikzpicture}
\end{center}
		
	\subsection{$4d$ holomorphic BF theory}
	As an example, we consider the $4d$ holomorphic BF theory. This is also the holomorphic twist of $4d$ $\mathcal{N} = 1$ SYM studied in \cite{Johansen:1994aw,BGKWWY}. In the BV formalism, this theory has field content
	\begin{equation}
	\begin{aligned}
		&\pmb{\EuScript{A}}^{4d} \in \Omega^{0,\sbullet}(X,\mathfrak{g})[1],\\
		&\pmb{\EuScript{B}}^{4d}\in \Omega^{2,\sbullet}(X,\mathfrak{g}^*).
	\end{aligned}
\end{equation}
The BV action functional is given by
\begin{equation}\label{4d_BF}
	\int \Tr(\pmb{\EuScript{B}}^{4d}(\pab\pmb{\EuScript{A}}^{4d} + [\pmb{\EuScript{A}}^{4d} ,\pmb{\EuScript{A}}^{4d} ] ).
\end{equation}

As we have discussed, the Dolbeault complex of $\C^3\backslash \C$ decomposes into de Rham complex of $\R_{>0}$ and the tangential Cauchy Riemann complex. Performing (cohomological) KK reduction is analogous to our previous $6d$ examples. The field $\pmb{\EuScript{A}}^{4d}$ give us the following KK tower
	\begin{equation}
	\begin{aligned}
		&\mathbf{Q} \in	\Omega^{\sbullet}(\R_{>0})\otimes \mathfrak{g}[w_1,w_2][1],\\
		& \tilde{\mathbf{P}} \in	\Omega^{\sbullet}(\R_{>0})\otimes \mathfrak{g}[\bar{w}_1,\bar{w}_2]\epsilon.\\
	\end{aligned}
\end{equation}
Similarly, the field $\pmb{\EuScript{B}}^{4d}$ give us the following KK tower
		\begin{equation}
		\begin{aligned}
		& \mathbf{P} \in \Omega^{\sbullet}(\R_{>0})\otimes \mathfrak{g}^*[\bar{w}_1,\bar{w}_2]\epsilon[-1],\\
		& \tilde{\mathbf{Q}} \in \Omega^{\sbullet}(\R_{>0})\otimes \mathfrak{g}^*[w_1,w_2].
	\end{aligned}
\end{equation}
The BV action functional gives us the one dimension Chern-Simons (BF) theory studied in \cite{gwilliam2014one}, associated to the $L_\infty$ algebra $\otimes H_{b}^{0,\sbullet}(S^3)\otimes\mathfrak{g}$. Up to cubic term, the action functional can be directly obtained from the $4d$ action \ref{4d_BF}
\begin{equation}
	\int \Tr \mathbf{P} (d_t\mathbf{Q} + [\mathbf{Q},\mathbf{Q}])  +  \tilde{\mathbf{P}} (d_t \tilde{\mathbf{Q}} +  [\mathbf{Q}, \tilde{\mathbf{Q}}]) .
\end{equation}
Higher order effective interactions take the same form as the $6d$ cases:
		\begin{equation}\label{KK_4d_int}
	I^{\geq 3}  = \sum_{n = 3}^{\infty}\int \Tr( \mathbf{P} \wedge l_n(\mathbf{Q},\mathbf{Q},\underbrace{\tilde{\mathbf{P}} ,\dots,\tilde{\mathbf{P}} }_{n - 2}))+ \Tr(\tilde{\mathbf{Q}} \wedge l_n(\mathbf{Q} ,\underbrace{\tilde{\mathbf{P}} ,\dots,\tilde{\mathbf{P}} }_{n-1})).
\end{equation}
An explicit expansion of the above interaction is similar to \ref{act_full}.

The natural boundary condition we should choose is the Neumann boundary condition
\begin{equation}
\mathbf{P}|_{t = 0} = 0,\;\;	\tilde{\mathbf{P}}|_{t = 0} = 0.
\end{equation}
The corresponding boundary algebra is generated by the bottom components $Q,\tilde{Q}$ of the BV fields $\mathbf{Q}, \tilde{\mathbf{Q}}$. The action functional up to cubic term gives us the following (zeroth order) boundary BRST differential
\begin{equation}
	Q_{\text{BRST}}^{0}Q = [Q,Q],\;\;Q_{\text{BRST}}^{0}\tilde{Q} = [Q,\tilde{Q}].
\end{equation}
We find that the space of boundary operators can be identified with the following \CE complex 
\begin{equation}
	C^*(\mathfrak{g}[w_1,w_2],\mathrm{Sym}\,\mathfrak{g}[w_1,w_2]^*).
\end{equation}
This coincides with the classical local operators of the $4d$ holomorphic BF theory. 

At the quantum level, there are many corrections to classical result. First, as is standard in topological quantum mechanics model, the graded commutative product of the bulk algebra is deformed into the Weyl/Clifford product $\star$. Its (graded) commutator is given by
\begin{equation}
	[P[p,q],Q[p,q]]_\star = 1,\;\; [\tilde{P}[p,q],\tilde{Q}[p,q]]_\star = 1.
\end{equation}

The KK modes $Q,\tilde{Q}$ correspond to the $4d$ bulk local operators, and the KK modes $P,\tilde{P}$ correspond to the $S^3$ integration of descent operators. Therefore, the above commutation relation is precisely the mode expansion of $\lambda$ bracket of the $4d$ theory. 
\begin{equation}
	\text{Weyl/Clifford product }\rightsquigarrow \lambda-\text{ bracket }.
\end{equation}

It will be interesting to explore other quantum effects induced by the higher order effective interaction.

\subsection{A minimal model for higher Kac-Moody algebra}

In fact, the $L_\infty$ algebra 
\begin{equation}
	H_b^{0,\sbullet}(S^3)\otimes \mathfrak{g}
\end{equation}
we constructed in Section \ref{sec:flat_full} can be regarded as a minimal model of the higher dimensional ($4d$ in our case) Kac-Moody algebra
\begin{equation}
	\mathrm{R}\Gamma(\C^2\backslash\{0\},\mathcal{O}\otimes \mathfrak{g})
\end{equation}
introduced in \cite{FAONTE2019389} and \cite{Gwilliam:2018lpo}. Indeed, it is shown in \cite{FAONTE2019389} that the tangential Cauchy Riemann complex $\Omega_{b}^{0,\sbullet}(S^3)$ is isomorphic to the Jouanolou model of $\mathrm{R}\Gamma(\C^2\backslash\{0\},\mathcal{O})$. So $\Omega_b^{0,\sbullet}(S^3)\otimes \mathfrak{g}$ is actually a dg Lie algebra model of $\mathrm{R}\Gamma(\C^2\backslash\{0\},\mathcal{O})$. 

In \cite{FAONTE2019389}, the authors also considered a central extension of the higher Kac-Moody algebra. Translating to our notation, this central extension can be expressed as follows. Let $\theta \in \mathrm{Sym}^{3}(\mathfrak{g}^\vee)^{\mathfrak{g}}$ be a degree $3$ $\mathfrak{g}$-invariant polynomial on $\mathfrak{g}$. Then we have the following cocycle
\begin{equation}
	\gamma_{\theta}(f_0\otimes x_0,f_1\otimes x_1,f_2\otimes x_0) = \theta(x_0,x_1,x_2)\Tr_{S^3}(f_0\{f_1,f_2\}_{\pi}),
\end{equation}
where $\{-,-\}_{\pi}$ is the bracket on $\Omega_b^{0,\sbullet}(S^3)$ induced by the Poisson bracket defined in Section \ref{sec:KK_Pois}. This cocycle induced an $L_\infty$ central extension of $\Omega_b^{0,\sbullet}(S^3)\otimes \mathfrak{g}$. It will be interesting to look at the transferred structure on $H_b^{0,\sbullet}(S^3)\otimes \mathfrak{g}$, which will give a minimal model of the central extended $4d$ Kac-Moody algebra. According to \cite{Gwilliam:2018lpo}, this algebra can be regarded as the symmetry algebra of the holomorphic theory and its KK theory analyzed in the last section.

		\section{Other examples}
		\label{sec:other}
		\subsection{Theories on ADE singularity}
		There is a family of variants of the B-model twisted holography by considering B-model on ADE singularity $\C\times (\C^2/\Gamma)$ \cite{Costello:2018zrm}. By placing branes on $\C\times \{0\}$, the "gravity" side will becomes B-model on $SL_2(\C)/\Gamma$.
		\begin{prop}
			Let $\Gamma$ be ADE type subgroup. Then $H_{b}^{0,\sbullet}(S^3)^{\Gamma}$ is a $A_\infty$ subalgebra of $H_{b}^{0,\sbullet}(S^3)$.
		\end{prop}
	\begin{proof}
		Recall that $\Gamma$ are subgroup of $SU(2)$ that acts on $\Omega_{b}^{0,\sbullet}(S^3)$. By definition, the operator $h$ and $M$ are compatible with the $SU(2)$ action. Therefore, for any $g \in \Gamma$, it commute with the operator $h$ and satisfy
\begin{equation}
	M(ga,gb) = gM(a,b).
\end{equation}
By the construction of the higher product $m_n$, we have
\begin{equation}
	gm_n(a_1,\dots,a_n) = m_n(ga_1,\dots,ga_n ).
\end{equation}
	\end{proof}
As a simplest example, for $A_1$ type singularity $\Gamma = \Z_2$, we have $H_{b}^{0,\sbullet}(S^3)^{\Gamma} = \C[u,v,w]/(uv = w^2) \oplus \C[\bar{u},\bar{v},\bar{w}]/(\bar{u}\bar{v} = \bar{w}^2)\epsilon$. It will be interesting to explore the corresponding KK theory and the boundary chiral algebra.

	\subsection{Theories on superspace}
		In this section, we briefly comment on the KK reduction of $6d$ holomorphic theories on superspace. We mainly consider holomorphic theories on $\C^{3|N}\backslash\C$. The key to our analysis is the following Dolbeault complex
		\begin{equation}
			(\Omega^{0,\sbullet}(\C^{3}\backslash\C)\otimes\C[\theta_i]_{i=1,\dots,N},\pab ).
		\end{equation}
		Our previous results directly apply to this case. We have the following quasi-isomorphic complex
		\begin{equation}
			(\Omega^{\sbullet}_{3d}(\C\R_{ > 0})\otimes H_{b}^{0,\sbullet}(S^3)\otimes\C[\theta_i]_{i=1,\dots,N},\hat{d} ).
		\end{equation}
		The $A_\infty$ structure on $H_{b}^{0,\sbullet}(S^3)\otimes\C[\theta_i]$ is induced from the $A_\infty$ structure on $H_{b}^{0,\sbullet}(S^3)$. Let
		\begin{equation}
			\{\theta_A\} = \{1,\theta_i,\theta_i\wedge\theta_j, \dots \}
		\end{equation} 
		be a basis of $\C[\theta_i]$. Then we can write the $A_\infty$ structure on $H_{b}^{0,\sbullet}(S^3)\otimes\C[\theta_i]$ explicitly as
		\begin{equation}
			m_n(a_1\theta_{A_1},a_2\theta_{A_2},\dots,a_n\theta_{A_n}) = m_n(a_1,a_2,\dots,a_n)\theta_{A_1}\wedge \theta_{A_2} \dots \wedge \theta_{A_n}.
		\end{equation}
		
		KK theory can be built on this $A_\infty$ algebra. For instance, if we consider holomorphic Chern-Simons theory on $\C^{3|N}\backslash\C$. The resulting KK theory has field content
		\begin{equation}
			\begin{aligned}
				\mathbf{A} \in \Omega^{\sbullet}_{3d}(\C\R_{ > 0})\otimes \C[w_1,w_2]\otimes\C[\theta_i]\otimes \mathfrak{g}[1],\\
				\mathbf{B} \in \Omega^{1,\sbullet}_{3d}(\C\R_{ > 0})\otimes \C[\bar{w}_1,\bar{w}_2]\otimes\C[\theta_i]\otimes \mathfrak{g}.
			\end{aligned}
		\end{equation}
		The BV action functional is given by
		\begin{equation}
			\int_{3d}\int_{\C^{0|N}}\Tr \mathbf{B}(\hat{d}\mathbf{A} + \frac{1}{2}[\mathbf{A},\mathbf{A}]),
		\end{equation}
		with deformation
		\begin{equation}
			\sum_{n = 3}^{\infty}\frac{1}{2(n - 1)!}\int_{3d}\int_{\C^{0|N}} \Tr(\mathbf{B}\wedge l_n(\mathbf{A},\mathbf{A},\mathbf{B},\dots,\mathbf{B})).
		\end{equation}
		The Berezinian integral $\int_{\C^{0|N}}$ pairs $\theta_A$ with its Hodge dual $*\theta_A$. Therefore it will be convenient to label the fields by
		\begin{equation}
			\begin{aligned}
				\mathbf{A}& = \sum_A \mathbf{A}_A\theta_A,\\
				\mathbf{B}& = \sum_A \mathbf{B}_A(*\theta_A).
			\end{aligned}
		\end{equation}
		In this way, the propagator pairs $\mathbf{B}_A$ with $\mathbf{A}_A$. One easily reads off the chiral algebra structure (at tree-level) from the action functional. If we take the Dirichlet boundary condition, the boundary chiral algebra is generated by fields $B_A$. Omitting the color factor, the tree-level boundary OPE can be schematically written as follows
		\begin{equation}
			B_A(z)B_B(0) \sim m^{C}_{AB} \frac{1}{z} B_C(0) + (m_3)m^{C_1C_2}_{AB}\frac{1}{z} B_{C_1}B_{C_2} +  \dots (m_n)m^{C_1\dots C_n}_{AB} \frac{1}{z}B_{C_1}\dots B_{C_n}+ \dots
		\end{equation}
		where the constants $m^{C_1\dots C_n}_{AB}$ are defined by
		\begin{equation}
			\begin{aligned}
				m_{AB}^C &= \int_{\C^{0|N}}\theta_A\wedge\theta_B\wedge*\theta_C\\
				m_{AB}^{C_1C_2}& = \int_{\C^{0|N}}\theta_A\wedge\theta_B\wedge*\theta_{C_1}\wedge*\theta_{C_2}\\
				&\dots\\
				m_{AB}^{C_1\dots C_n}& = \int_{\C^{0|N}}\theta_A\wedge\theta_B\wedge*\theta_{C_1}\wedge\dots\wedge*\theta_{C_n}.\\
			\end{aligned}
		\end{equation}

	\appendix

\addtocontents{toc}{\protect\setcounter{tocdepth}{0}}

\section{tangential Cauchy Riemann complex}

\begin{defn}
	Let $M$ be an real manifold and $ T^{\C}M : = TM \otimes_{\R}\C$ its complexified tangent bundle. A CR structure on $M$ is a subbundle $\mathbb{L} \subset  T^{\C}M $ such that
	\begin{enumerate}
		\item[$\sbullet$] $\mathbb{L}\cap \bar{\mathbb{L}} = \{0\}$.
		\item[$\sbullet$] $[\mathbb{L},\mathbb{L}] \subset \mathbb{L}$, that is, $\mathbb{L}$ is an integrable distribution.
	\end{enumerate}
\end{defn}

Suppose $M$ is a real submanifold of $\C^n$ locally defined by real valued functions $\{\rho_i: \C^n \to \R\}_{i = 1,\dots d}$ that satisfy the independence condition:
\begin{equation}
	\pab \rho_1 \wedge \pab \rho_2 \dots \pab \rho_d \neq 0.
\end{equation}
Then $M$ is a CR manifold called embedded CR submanifold.

We introduce the tangential Cauchy Riemann complex only for embedded CR submanifold. 

First we define $\Omega^{p,q}(\C^n)|_M$ be the restriction of the bundle $\Omega^{p,q}(\C^n)$ to $M$. Suppose $\{\rho_i: \C^n \to \R\}_{i = 1,\dots d}$ is a local defining system of $M$. We consider 
\begin{equation}
	I^{p,q} := \left\lbrace  \begin{array}{c}
		\text{the ideal in } \Omega^{p,q}(\C^n) \text{ which is generated by} \\
		\rho_i,\pab \rho_i, i = 1,\dots d 
	\end{array}\right\rbrace.
\end{equation} 

Then we define the tangential Cauchy Riemann complex as the following complex of bundle
\begin{equation}
	\Omega_b^{p,q}(M) = \{\text{The orthogonal complement of } I^{p,q} \text{ in } \Omega^{p,q}(\C^n)|_M\}.
\end{equation}
The  tangential Cauchy Riemann differential is defined as follows. For an open subset $U \subset M$ and $f \in \Omega_b^{p,q}(U)  $, Let $\tilde{U}$ be an open subset in $\C^n$ with $U = \tilde{U}\cap M$. We choose a $\tilde{f} \in  \Omega^{p,q}(\tilde{U})$ such that $p_M(\tilde{f}|_M) = f$, where $p_M: \Omega^{p,q}(\C^n)|_M \to \Omega_b^{p,q}(M) $ is the orthogonal projection map. We define
\begin{equation}
	\pacr f : = p_M(\pab \tilde{f}).
\end{equation}
One can check that this definition is independent of the choice of $\tilde{f}$.

	\section{Harmonic polynomials on $S^3$}
	\label{apx:Har}
	
	In this appendix, we review some basic facts about harmonic polynomials. Though much of the results hold in other dimensions (see e.g \cite{knapp2013lie}), we focus on $S^3$.

	Let $V_N$ be the space of polynomials in $z_1,z_2,\bar{z}_1,\bar{z}_2$ that are homogeneous of degree $N$. Let $V_{p,q}$ be the space of homogeneous polynomial of bi-degree $(p,q)$ in $z_1,z_2$ and $\bar{z}_1,\bar{z}_2$ respectively. We have $V_N = \oplus_{p+q = N}V_{p,q}$. We consider the Laplacian
	\begin{equation}
		\Delta = \frac{\pa^2}{\pa \bar{z}_1 \pa z_1} + \frac{\pa^2}{\pa \bar{z}_1 \pa z_1},
	\end{equation}
	and define the space of harmonic polynomials
	\begin{equation}
		H_{p,q} = \{f \in V_{p,q}\mid \Delta f = 0 \}.
	\end{equation}
	We emphasize that we used a different notation in the main text, where used half-integer to label the space of harmonic polynomials $\mathcal{H}_{j,\bar{j}} = H_{2j,2\bar{j}}$. 
	
	For any homogeneous polynomial $f = \sum_{k}c_kz_1^{k_1}z_2^{k_2}\bar{z}_1^{\bar{k}_1}\bar{z}_2^{\bar{k}_2}$, we define a differential operator $\pa_f$ as follows
	\begin{equation}
		\pa_f = \sum_{k}c_k \frac{\pa^{k_1+k_2+\bar{k}_1+\bar{k}_2}}{\pa z_1^{k_1} \pa z_2^{k_2} \pa\bar{z}_1^{\bar{k}_1} \pa \bar{z}_2^{\bar{k}_2} }.
	\end{equation}
Denote $||z||^2 = z_1\bar{z}_z + z_2\bar{z}_2$. We have $\Delta = \pa_{||z||^2}$.

Suppose we have $f\in V_{p,q}$ and $g\in V_{q,p}$, then $\pa_f\bar{g}$ is a constant. We define an inner product $\langle\hspace{-2pt}\langle f,g\rangle\hspace{-2pt}\rangle = \pa_f\bar{g}$. Under this inner product, we have
\begin{equation}
\langle\hspace{-2pt}\langle z_1^{k_1}z_2^{k_2}\bar{z}_1^{\bar{k}_1}\bar{z}_2^{\bar{k}_2},\bar{z}_1^{k_1}\bar{z}_2^{k_2}z_1^{\bar{k}_1}z_2^{\bar{k}_2}\rangle\hspace{-2pt}\rangle = k_1!k_2!\bar{k}_1!\bar{k}_2!.
\end{equation}
It follows that this inner product is Hermitian and $SU(2)$ invariant. A useful property of this inner product is that 
\begin{equation}
\langle\hspace{-2pt}\langle f, \pa_gh\rangle\hspace{-2pt}\rangle = \pa_f \pa_{\bar{g}} \bar{h} = \langle\hspace{-2pt}\langle f \bar{g},  h\rangle\hspace{-2pt}\rangle.
\end{equation}
As a consequence, we have
\begin{equation}
	\langle\hspace{-2pt}\langle f, \Delta g\rangle\hspace{-2pt}\rangle = \langle\hspace{-2pt}\langle ||z||^2 f, g\rangle\hspace{-2pt}\rangle.
\end{equation}

\begin{prop}
	\begin{enumerate}
		\item  Under the Hermitian form $\langle\hspace{-2pt}\langle -, -\rangle\hspace{-2pt}\rangle$, the orthogonal complement of $||z||^2V_{p - 1,q - 1}$ in $V_{p,q}$ is $H_{p,q}$.
		\item  We have an orthogonal direct sum decomposition
		\begin{equation}
			V_{p,q} = H_{p,q} \oplus ||z||^2H_{p - 1,q - 1} \oplus ||z||^4H_{p - 1,q - 1} \oplus \dots
		\end{equation}
	\end{enumerate}
\end{prop}
\begin{proof}
	1. First we prove that $||z||^2V_{p - 1,q - 1}$ is orthogonal to $H_{p,q}$. Let $h \in H_{p,q}$ and $||z||^2 f \in ||z||^2V_{p - 1,q - 1}$. We have
	\begin{equation}
		\langle\hspace{-2pt}\langle ||z||^2 f,  h\rangle\hspace{-2pt}\rangle = \langle\hspace{-2pt}\langle  f, \Delta h \rangle\hspace{-2pt}\rangle = 0.
	\end{equation}
For $p,q \geq 1$, $\Delta:V_{p,q} \to V_{p-1,q-1}$ is surjective. Moreover $H_{p,q} = \ker \Delta|_{V_{p,q}}$. Therefore, $\dim V_{p,q} = \dim H_{p,q} + \dim V_{p-1,q-1}$. As a consequence, we have the following direct sum decomposition
\begin{equation}
	V_{p,q} = H_{p,q} \oplus ||z||^2V_{p - 1,q - 1}.
\end{equation}

2. By induction.
\end{proof}

The above result tells us that for any polynomial $f \in V_{p,q}$, $f$ have the following decomposition into harmonic polynomials
\begin{equation}
	f = h_0 + ||z||^2h_1 + ||z||^4h_2 \dots
\end{equation}
where $h_i \in H_{p - i,q-i}$.

\begin{cor}
	The restriction to $S^3$ of every polynomial is a sum of restrictions to $S^3$ of harmonic polynomials.
\end{cor}
Since the space of polynomials is dense in $L^2(S^3)$, we have the following
\begin{cor}
		\begin{equation}
		L^2(S^3) = \bigoplus_{p,q\geq 0} H_{p,q}.
	\end{equation}
\end{cor}
\begin{remark}
	This fact can also be deduced from the Peter-Weyl theorem for $SU(2)$.
\end{remark}
We further obtain the harmonic decomposition of tangential Cauchy Riemann complex on $S^3$
	\begin{equation}
\begin{aligned}
		\Omega_b^{0,0}(S^3)&= \bigoplus_{p,q \geq 0} H_{p,q},\\
	\Omega_b^{0,1}(S^3)& = \bigoplus_{p,q \geq 0} H_{p,q} \epsilon.
\end{aligned}
\end{equation}

	\section{Homotopy algebra and Homotopy transfer}
	Since this paper heavily uses techniques from homotopy algebra. We briefly review this topic in this appendix. We recommend the survey \cite{vallette2014algebra+} for a detailed review.
	
	\subsection{Convention and Koszul sign rule} 
	First, we fix the convention for our discussion. We work with $\mathbb{Z}$-graded $\C$-vector space
\begin{equation}
		V = \bigoplus\limits_{n\in \mathbb{Z}}V_n.
\end{equation}
	The grading $n$ is related to the ghost number in physics. The degree of an element $v \in V_n$ is denoted by $|v| = n$, and such a $v$ is called a homogeneous element.
		
	For $V$ and $W$ two graded vector spaces,  the tensor product $V\otimes W$ and the Hom space $\Hom(V,W)$ has the following grading
	\begin{equation*}
		(V\otimes W)_n = \bigoplus_{i+j = n} V_{i}\otimes W_j, \quad \Hom(V,W)_{n} = \bigoplus_{i}\Hom(V_i,W_{i+n}).
	\end{equation*}
	
	We denote the Koszul sign braiding on tensor products to be
	\begin{equation*}
		\begin{aligned}
			\tau_{V,W}:V\otimes W &\to W\otimes V,\\
			v\otimes w &\mapsto (-1)^{|v||w|} w\otimes v.
		\end{aligned}
	\end{equation*}
	The above sign rule induces naturally a sign rule for the action of the symmetric group $S_n$ on the $n$-th tensor product $V^{\otimes n}$
	\begin{equation*}
		\sigma:\;v_1\otimes v_2\otimes \cdots \otimes v_n \to \epsilon(\sigma,v)v_{\sigma(1)}\otimes v_{\sigma(2)} \otimes \dots v_{\sigma(n)},
	\end{equation*}
	where $\epsilon(\sigma,v)$ is called the Koszul sign.
	
		For $V$ a $\mathbb{Z}$ graded vector space, we denote $V[n]$ the degree $n$-shifted space such that
	\begin{equation}
		V[n]_{m} := V_{n+m}.
	\end{equation}
	We also use the notation of suspension $sV$ and desuspension $s^{-1}V$ as follows
	\begin{equation}
		sV := V[1],\quad s^{-1}V: = V[-1].
	\end{equation}
	We can also regard $s$ as a degree $-1$ linear map $s: V \to V[1]$. For a homogeneous $a \in V$, we have $sa \in V[1]$ and $|sa| = |a| - 1$. Similarly, $s^{-1}$ can be regarded as a degree $1$ linear map, such that $s^{-1}s = ss^{-1} = 1$. 
	
	\subsection{Homotopy algebra}
	In this appendix, we review the definition of various homotopy algebras including $A_\infty$, $C_\infty$ and $L_\infty$ algebras. 
	\paragraph{$A_\infty$ algebra}
	\begin{defn}
		An $A_\infty$ algebra is a graded vector space $A = \{A_n\}_{n\in \Z}$ with a collection of multi-linear operations
		\begin{equation}
			m_n:A^{\otimes n} \to A\; \text{ of degree } n - 2 \text{ for all } n\geq 1,
		\end{equation}
		which satisfy the following relations:
		\begin{equation}\label{A_inf_rel}
			\sum_{k = 1}^{n}\sum_{j = 0}^{n - k}(-1)^{jk + (n-j-k)}m_{n - k + 1}\circ(\mathrm{id}^{\otimes j} \otimes m_k \otimes \mathrm{id}^{\otimes n - j - k } ) = 0.
		\end{equation}
	\end{defn}
	
	Let's demonstrate the above relations for small values of $n$:
	\begin{enumerate}
		\item $n = 1$. We have $m_1\circ m_1 = 0$, which means that $m_1$ is a differential on $A$. We also denote $d = m_1$.
		\item $n = 2$. We have
		\begin{equation}
			d m_2(x_1, x_2) = m_2(dx_1, x_2) + (-1)^{|x_1|}m_2(x_1, dx_2).
		\end{equation}
		This relation implies $m_1$ is a derivation with respect to the binary product $m_2$.
		\item $n = 3$. The relation yields
		\begin{equation}
			\begin{aligned}
				m_2&(m_2(x_1,x_2),x_3) - m_2(x_1,m_2(x_2,x_3)) =\\
				& dm_3(x_1,x_2,x_3) + m_3(dx_1,x_2,x_3) + m_3(x_1,dx_2,x_3)  + m_3(x_1,x_2,dx_3) .
			\end{aligned}
		\end{equation}
	\end{enumerate}
	
	An $A_\infty$ algebra with $m_k = 0$ for $k \geq 3$ is also called a differential graded associative (dga) algebra. For example, the tangential Cauchy-Riemann complex $(\Omega_{b}^{0,\sbullet}(S^3),\pacr, \cdot)$ is a dga algebra

	There is an equivalent definition of $A_\infty$ algebra in terms of coderivation. We introduce the reduced tensor coalgebra
	\begin{equation}
		\bar{T}^c(V) = \bigoplus_{n\geq 1}V^{\otimes n},
	\end{equation}
with comultiplication given by 
\begin{equation}
	\bar{\Delta}(v_1\otimes v_2\otimes \dots \otimes v_n) = \sum_{i = 1}^{n-1}(v_1\otimes \dots  \otimes v_i)\otimes (v_{i+1}\otimes \dots \otimes v_n).
\end{equation}
Recall that a coderivation on a coalgebra $(C,\Delta)$ is a map $L: C \to C$ such that $\Delta\circ L = (L\otimes 1 + 1\otimes L)\Delta$.

For the (reduced) tensor coalgebra $\bar{T}^c(V) $, a coderivation on it is completely determined by its projection $p_V\circ L : \bar{T}^c(V)  \to \bar{T}^c(V) \to V$. To see this, we first notice that $p_V\circ L$ is given by a set of maps $L_k  \in \Hom(V^{\otimes k},V),k\geq 1$. Given this set of maps, the coderivation is uniquely given by 
	\begin{equation}\label{coder_tensorco}
	L = \sum_{i\geq 1}^n\sum_{j = 0}^{n-i} \mathds{1}^{\otimes j}\otimes L_{i}\otimes \mathds{1}^{n - i - j}.
\end{equation}
The structure of an $A_\infty$ algebra on $A$ can be compactly organized into the structure of a square zero coderivation on $\bar{T}^c(sA)$.
\begin{prop}\label{A_inf_coder}
			The following data are equivalent
		\begin{itemize}
			\item  A collection of linear maps $m_k: A^{\otimes k} \to A$ of degree $2 - k$ satisfying $A_\infty$ relation.
			\item A degree $1$ coderivation $b$ on $\bar{T}^c(A[1])$ satisfying $b^2 = 0$.
		\end{itemize}
\end{prop}
\begin{proof}
	We only sketch the proof here and refer to \cite{getzler1990algebras} for more details. Given linear maps $m_k:A^{\otimes k} \to A$, we define maps $b_k : (sA)^{\otimes k} \to  sA$ by 
	\begin{equation}
		b_k = s\circ m_k\circ (s^{-1})^{\otimes k}.
	\end{equation}
The maps $b_k$ further define a coderivation $b$ on $\bar{T}^c(A[1])$ through \ref{coder_tensorco}. One can check that the requirement $b^2 = 0$ is equivalent to the $A_\infty$ relations \ref{A_inf_rel}.
\end{proof}
	
	\paragraph{$C_\infty$ algebra}
	In this paper, the dga algebras that we studied satisfy additional properties of being graded commutative. 
	\begin{equation}
		m_2(a,b) = (-1)^{|a||b|} m_2(b,a).
	\end{equation}
	Such algebras are called differential graded commutative (dgc) algebra. The homotopy version of dgc algebra is called $C_\infty$ algebra, which we now define.
	
	A $(p,q)$-shuffle is a permutation $\sigma \in S_{p+q}$ such that
	\begin{equation}
		\sigma(1)<\sigma(2) < \dots < \sigma(p),\;\;\sigma(p + 1)<\sigma(p + 2) < \dots < \sigma(p+q).
	\end{equation}
	We denote by $Sh(p,q)$ the subset of $(p,q)$-shuffles in $S_{p+q}$.
	
	We have introduced the reduced tensor coalgebra $\bar{T}^c(V) = \bigoplus_{n\geq 1}V^{\otimes n}$. It becomes a Hopf algebra when equipped with the multiplication map called shuffle product 
	\begin{equation}
		sh((a_1,\dots,a_p)\otimes (a_{p+1},\dots a_{p+q})) = \sum_{\sigma \in Sh(p,q)}\epsilon(\sigma,a)(a_{\sigma^{-1}(1)},a_{\sigma^{-1}(2)},\dots a_{\sigma^{-1}(p+q)}).
	\end{equation}
	
	\begin{defn}
		A $C_{\infty}$-algebra structure on a graded vector space $A = \{A_n\}_{n\in \Z}$ is an $A_\infty$ structure $(A,\{m_n\}_{n\geq 1})$ such that the set of maps $\{b_k = s\circ m_k \circ (s^{-1})^{\otimes k}, k \geq 1\}$ vanish on the image of the shuffle product $sh:T^c(sA)\otimes T^c(sA) \to T^c(sA)$.
	\end{defn}
For example, the element $sa\otimes sb + (-1)^{(|a|+1)(|b| + 1)}sb\otimes sa$ is in the image of the shuffle product. Vanishing of $b_2$ on this element is the same as the graded commutativity of $m_2$.

\paragraph{$L_\infty$ algebra}
	
	We also introduce the notion of $L_\infty$ algebra.
	\begin{defn}
		Let $\mathfrak{g} = \{\mathfrak{g}^n\}_{n \in \Z}$ be a graded vector space. An $L_\infty$ structure on $\mathfrak{g}$ is a collection of multi-linear maps
		\begin{equation}
			l_n:\mathfrak{g}^{\otimes n} \to \mathfrak{g}\; \text{ of degree } n - 2 \text{ for all } n\geq 1,
		\end{equation}
		that are graded skew-symmetric:
		\begin{equation}
			l_n(x_{\sigma^{-1}(1)},\dots,x_{\sigma^{-1}(n)}) = (-1)^{\sigma}\epsilon(\sigma,x) l_n(x_1,\dots,x_n),\;\;\text{ for all } \sigma \in S_n,
		\end{equation}
		and satisfy the following relations:
		\begin{equation}
			\sum_{k = 1}^{n}(-1)^{k}\sum_{\sigma \in Sh(k,n - k)}(-1)^{\sigma}\epsilon(\sigma,x) l_{n - k - 1}(l_{k}(x_{\sigma^{-1}(1)},\dots,x_{\sigma^{-1}(k)}), x_{\sigma^{-1}(k + 1)},\dots,x_{\sigma^{-1}(n)}) = 0.
		\end{equation}
	\end{defn}
		Let us analyze the defining relations for small values of $n$:
\begin{enumerate}
	\item $n = 1$. The relation is $l_1\circ l_1 = 0$, which means that $l_1$ is a differential on $\mathfrak{g}$.
	\item $n = 2$. We have
	\begin{equation}
		l_1(l_2(x_1,x_2)) = l_2(l_1(x_1),x_2) + (-1)^{|x_1|}l_2(x_1,l_1(x_2))
	\end{equation}
	which says that $l_1$ is a derivation with respect to the binary map $l_2$.
	\item $n = 3$. The relations yields
	\begin{equation}
		\begin{aligned}
			& l_2(l_2(x_1,x_2),x_3) + (-1)^{(|x_1|+|x_2|)|x_3|}l_2(l_2(x_3,x_1),x_2)  + (-1)^{(|x_2|+|x_3|)|x_1|}l_2(l_2(x_2,x_3),x_1) \\
			&= l_1l_3(x_1,x_2,x_3) + l_3(l_1(x_1),x_2,x_3) + (-1)^{|x_1|}l_3(x_1,l_1(x_2),x_3) + (-1)^{|x_1| + |x_2|}l_3(x_1,x_2,l_1(x_3)).
		\end{aligned}
	\end{equation}
	which says that $l_2$ satisfies Jacobi identities up to homotopy given by $l_3$.
\end{enumerate}

There is a similar characterization of $L_\infty$ algebra in terms of a coderivation. Instead of the tensor coalgebra, we consider the reduced symmetric coalgebra $\bar{S}^c(V)$ where
\begin{equation*}
	\bar{S}^c(V) = \bigoplus_{n \geq 1}\mathrm{Sym}^n(V).
\end{equation*}
		The coproduct $\bar \Delta: \bar S^c(V) \to \bar S^c(V)\otimes \bar S^c(V)$  is defined by
\begin{equation}
	\bar \Delta(v_1\cdot v_2\dots  v_n) = \sum_{i = 1}^{n-1}\sum_{\sigma \in \text{Sh}(i,n-i)}\epsilon(\sigma,v)(v_{\sigma^{-1}(1)}\cdot v_{\sigma^{-1}(2)}\dots v_{\sigma^{-1}(i)})\otimes(v_{\sigma^{-1}(i+1)}\cdot\cdot\cdot v_{\sigma^{-1}(n)}).
\end{equation}
Then we have
\begin{prop}
			The following data are equivalent
		\begin{itemize}
			\item  A collection of linear maps $l_k: \mathfrak{g}^{\otimes k} \to \mathfrak{g}$ of degree $2 - k$  satisfying $L_\infty$ relation.
			\item  A degree $1$ coderivation $Q$ on $\bar S^c(\mathfrak{g}[1])$ satisfying $Q^2 = 0$.
		\end{itemize}
\end{prop}

\subsection{Homological perturbation lemma}
We introduce an important technical tool called the homological perturbation lemma. We refer to \cite{2004math3266C} for a more detailed discussion.

Let us first consider the following homotopy data of chain complexes.
\begin{defn}
	A special deformation retract (SDR) from a cochain complex $(A,d_A)$ to $(H,d_H)$ consists of the following data
	\begin{equation}\label{SDR_data}
		h\curved (A,d_A)\overset{p}{\underset{i}\rightleftarrows} (H,d_H),
	\end{equation}
	where $i, p$ are cochain maps and $h$ is a degree $-1$ map on $A$, such that
	\begin{equation}
		i\circ p-\mathds{1}_A  = d_A\circ h + h\circ d_A,\;\quad p\circ i  = \mathds{1}_{H},
	\end{equation}
	and
	\begin{equation}
		h\circ i = 0,\; p\circ h = 0,\; h\circ h = 0.
	\end{equation}
\end{defn}

Consider a perturbation $\delta$ to the differential on $A$:
\begin{equation}
	d_A' = d_A + \delta,\;\;d_A'^2 = 0
\end{equation}
The perturbation is called small if $(1 - \delta  h) $ is invertible.
\begin{lem}
	(Homological perturbation lemma) Given a SDR data as \ref{SDR_data} and a small perturbation, there is a new SDR: 
	\begin{equation}\label{SDR_data}
		h\curved (A,d_A')\overset{p'}{\underset{i'}\rightleftarrows} (H,d_H')
	\end{equation}
	where the maps above are defined by
	\begin{equation}
		\begin{aligned}
			d_H' &= d_H+ p(1 - \delta h)^{-1}\delta i,\\
			h' &= h + h(1 - \delta h)^{-1}\delta h,\\
			p' &= p + p(1 - \delta h)^{-1}\delta h,\\
			i' &= i + h(1 - \delta h)^{-1}\delta i.
		\end{aligned}
	\end{equation}
\end{lem}

The homological perturbation lemma can be regarded as a substitution of the spectral sequence techniques, which provides explicit formulae.
	
	\subsection{Homotopy transfer}
	\label{apx:Hom_trans}
	 Given a dga algebra (or an $A_\infty$ algebra in general) and a chain complex quasi-isomorphic to it, homotopy transfer theorem \cite{Kadeishvili1980ONTH} gives the complex an $A_\infty$ structure. In particular, one gets an $A_\infty$ structure on the cohomology of a dga algebra. We emphasize that there are different approaches to construct this $A_\infty$ structure. In this appendix, we take the approach using homological perturbation lemma \cite{berglund2014homological}. 
	
	Given a dga algebra $(A,d,\cdot)$. Suppose we can find a SDR to its cohomology $H = H^{\sbullet}(A)$
		\begin{equation}
		h\curved (A,d)\overset{p}{\underset{i}\rightleftarrows} (H,d_H = 0).
	\end{equation}
Recall that the dga algebra structure on $A$ is equivalent to a differential $b$ on $\bar{T}^c(sA)$. Therefore, we first extend the above SDR to the corresponding tensor coalgebra
\begin{prop}
	The following is a SDR
			\begin{equation}
		Th^s\curved (\bar{T}^c(sA), Td^s)\overset{Tp^s}{\underset{Ti^s}\rightleftarrows} (\bar{T}^c(sH),0),
	\end{equation}
where the differential $Td^s  $ is defined by $Td^s  = \sum_{n \geq 1} \sum_{i = 0}^{n-1} \mathds{1}^{i}\otimes (s\circ d\circ s^{-1}) \otimes \mathds{1}^{n - i - 1}$. The projection and inclusion maps are defined by $Tp^s = \sum_{n\geq1}(s\circ p\circ s^{-1})^{\otimes n}$ and $	Ti^s = \sum_{n\geq 1}(s\circ i\circ s^{-1})^{\otimes n}$. The deformation retract is defined as
\begin{equation*}
	Th^s = \sum_{n\geq 1} \sum_{i = 0}^{n-1} \mathds{1}^{\otimes i} \otimes (s\circ h\circ s^{-1}) \otimes (s\circ i\circ p\circ s^{-1})^{\otimes n - i - 1}.
\end{equation*}
\end{prop}
	
The product $\cdot$ on the dga algebra $A$ defined a map $b_2:(sA)^{\otimes 2} \to sA$ and extend to a map $\delta : \bar{T}^c(sA) \to \bar{T}^c(sA)$. Together with the differential $Td^s   $, the sum $b = Td^s   + \delta: \bar{T}^c(sA) \to \bar{T}^c(sA)$ encode the dga algebra structure $A$ in the sense of Proposition \ref{A_inf_coder}. Now we can regard $\delta $ as a perturbation to the differential and apply the homological perturbation lemma. We have the following new SDR 
\begin{equation}
		h'\curved (\bar{T}^c(sA), Td^s + \delta)\overset{p'}{\underset{i'}\rightleftarrows} (\bar{T}^c(sH),b').
\end{equation} 
The homological perturbation lemma provides us a formula for all the maps $h',p',i'$. However, only the differential $b_H$ matter to us as it encodes the transferred $A_\infty$ structure on the cohomology $H$. We have
\begin{equation}\label{coder_trans}
	b'= Tp^s\circ(1 - \delta \circ Th^s)^{-1}\circ\delta\circ Ti^s = \sum_{n \geq 0}Tp^s\circ( \delta \circ Th^s)^{n}\circ\delta\circ Ti^s .
\end{equation}

If we further expand the above formula into components, we find the usual tree description of the transferred $A_\infty$ structure on $H$. 	Let $\mathrm{PBT_n}$ be the set of planar binary rooted trees with $n$ leaves. We consider the following construction that assigns each $T \in \mathrm{PBT}_n$ an $n$ array operation $m_T$ on $H$. The operation $m_T$ is obtained by putting $i$ on the leaves, $m$ on the vertices, $h$ on the internal edges and $p$ on the root. Then we consider
\begin{equation}\label{mn_trans}
	m_n = \sum_{T \in PBT_n} (\pm) m_T,
\end{equation}
where the $(\pm)$ sign can be tracked by a careful analysis of the Koszul sign rule in \ref{coder_trans}.
	
	\begin{theorem}
		The operations $\{m_n\}_{n \geq 2}$ defined on $H$ by the formulae \ref{mn_trans} form an $A_\infty$-algebra structure on $H$.

		Moreover, the transferred $A_\infty$-algebra $(H,\{m_n\}_{n \geq 2})$ is $A_\infty$ quasi-isomorphic to the dg algebra $(A,d_A,\cdot)$.
	\end{theorem}
	
	In the example of our study, the tangential Cauchy-Riemann complex $(\Omega_{b}^{0,\sbullet}(S^3),\pacr,\cdot)$ is graded commutative. We are interested in the transferred structure for dgc algebra. This scenario is analyzed in \cite{2006math10912Z}. For $(A,d,\cdot)$ a dgc algebra, if we regard it as a dga algebra, the $A_\infty$ structure constructed by \ref{mn_trans} actually defines a $C_\infty$ structure. 
	
	For homotopy transfer of dg Lie algebra and $L_\infty$ algebra, a similar result can be established. We start with a dg Lie algebra $(L,d,[-,-])$ and consider the transferred structure on its cohomology $\mathfrak{g} = H^{\sbullet}(L)$. Suppose we are given the following SDR 
		\begin{equation}
	h\curved (L,d)\overset{p}{\underset{i}\rightleftarrows} (H,d_H = 0).
\end{equation}
The tensor trick can be extended to the symmetric case
			\begin{equation}
	Sh^s\curved (\bar{S}^c(sL), Sd^s)\overset{Sp^s}{\underset{Si^s}\rightleftarrows} (\bar{S}^c(s\mathfrak{g}),0),
\end{equation}
where the differential $Sd^s  $ is defined by $Sd^s  = \sum_{n \geq 1} \sum_{i = 0}^{n-1} \mathds{1}^{i}\otimes (s\circ d\circ s^{-1}) \otimes \mathds{1}^{n - i - 1}$. The projection and inclusion maps are defined by $Sp^s = \sum_{n\geq1}(s\circ p\circ s^{-1})^{\otimes n}$ and $Si^s = \sum_{n\geq 1}(s\circ i\circ s^{-1})^{\otimes n}$. The deformation retract is defined as
\begin{equation*}
	Sh^s = \sum_{n\geq 1}\frac{1}{n!}\sum_{\sigma \in S_n} \sigma^{-1}\left( \sum_{i = 0}^{n-1} \mathds{1}^{\otimes i} \otimes (s\circ h\circ s^{-1}) \otimes (s\circ i\circ p\circ s^{-1})^{\otimes n - i - 1}\right)\sigma.
\end{equation*}
The Lie bracket $[-,-]$ on $L$ defined a map $Q_2:(sL)^{\otimes 2} \to sL$ and extend to a map $\delta : \bar{S}^c(sL) \to \bar{T}^c(sL)$. We add this differential to the above SDR as a perturbation. Then we have a new SDR, with a new differential on $\bar{S}^c(s\mathfrak{g})$ given by the following
\begin{equation}
	Q' = Sp^s\circ(1 - \delta \circ Sh^s)^{-1}\circ\delta\circ Si^s = \sum_{n \geq 0}Sp^s\circ( \delta \circ Sh^s)^{n}\circ\delta\circ Si^s .
\end{equation}
We can expand the above formula into components. This gives us the usual tree description of the transferred $L_\infty$ structure on $\mathfrak{g}$. Let $\mathrm{BT_n}$ be the set of binary rooted trees with $n$ leaves. In this case, we need to consider trees not necessarily planar, which means edges can cross each other. We consider the following construction that assigns each $T \in \mathrm{BT}_n$ an $n$ array operation $l_T$ on $H$. The operation $l_T$ is obtained by putting $i$ on the leaves, $[-,-]$ on the vertices, $h$ on the internal edges and $p$ on the root. We consider
\begin{equation}\label{ln_trans}
	l_n = \sum_{T \in BT_n} (\pm) l_T.
\end{equation}
Then the operations $\{l_n\}_{n \geq 2}$ defined an $L_\infty$-algebra structure on $\mathfrak{g}$. Moreover, the $L_\infty$ algebra $(\mathfrak{g},l_2,l_3,\dots)$ is $L_\infty$ quasi-isomorphic to the dg Lie algebra $(L,d,[-,-])$. 

\section{Computation of (higher) products and brackets}

\subsection{Product of $S^3$ harmonics}
\label{apx:pro_Har}
In this section, we compute the product of two arbitrary $S^3$ harmonics. We first recall the formula \ref{Har_to_Pol} that decomposes a harmonic polynomial into sum of monomials
\begin{equation}
	e^{(j,\bar{j})}_{m} = \sum_{l}\lambda_{j,\bar{j},0}^{-1}C^{j,\bar{j};j+\bar{j}}_{m -l,l;m} e^{(j)}_{m - l} \bar{e}^{\bar{j}}_{l},
\end{equation}
where
\begin{equation}
	\lambda_{j,\bar{j},k} = (-1)^k\sqrt{\frac{(2j+1)!(2\bar{j} + 1)!}{k!(2j + 2\bar{j} - k + 1)!}}.
\end{equation}
Then we can write
\begin{equation}
			M(e^{(j_1,\bar{j}_1)}_{m_1},e^{(j_2,\bar{j}_2)}_{m_2}) = \sum_{l_1,l_2}\lambda_{j_1,\bar{j}_1,0}^{-1}\lambda_{j_2,\bar{j}_2,0}^{-1}C^{j_1,\bar{j}_1;j_1+\bar{j}_1}_{m_1 - l_1, l_1;m_1} C^{j_2,\bar{j}_2;j_2+\bar{j}_2}_{m_2 - l_2, l_2;m_2} M(e^{(j_1)}_{m_1 - l_1}\bar{e}^{(\bar{j}_1)}_{l_1},e^{(j_2)}_{m_2 - l_2}\bar{e}^{(\bar{j}_2)}_{l_2}).
\end{equation}
To compute $M(e^{(j_1)}_{m_1 - l_1}\bar{e}^{(\bar{j}_1)}_{l_1},e^{(j_2)}_{m_2 - l_2}\bar{e}^{(\bar{j}_2)}_{l_2})$, we consider the product $e^{(j_1)}_{m_1 - l_1}e^{(j_2)}_{m_2 - l_2}$ and $\bar{e}^{(\bar{j}_1)}_{l_1}\bar{e}^{(\bar{j}_2)}_{l_2}$ separately. We find
\begin{equation}
\begin{aligned}
		 &M(e^{(j_1)}_{m_1 - l_1}\bar{e}^{(\bar{j}_1)}_{l_1},e^{(j_1)}_{m_1 - l}\bar{e}^{(\bar{j}_1)}_{l_2}) \\
		 = &\sqrt{\frac{(2j_1+1)(2j_2+1)(2\bar{j}_1+1)(2\bar{j}_2+1)}{(2j_1 + 2j_2+1)(2\bar{j}_1+2\bar{j}_2 + 1)}} C^{j_1,j_2;j_1+j_2}_{m_1-l_1,m_2 - l_2,m_1+m_2-l_1-l_2}C^{\bar{j}_1,\bar{j}_2;\bar{j}_1+\bar{j}_2}_{l_1,l_2,l_1+l_2} M(e^{(j_1+j_2)}_{m_1+m_2-l_1-l_2},\bar{e}^{(\bar{j}_1+\bar{j}_2)}_{l_1+l_2})\\
		  = &\sum_{k} \lambda_{j_1+j_2,\bar{j}_1+\bar{j}_2,k}\sqrt{\frac{(2j_1+1)(2j_2+1)(2\bar{j}_1+1)(2\bar{j}_2+1)}{(2j_1 + 2j_2+1)(2\bar{j}_1+2\bar{j}_2 + 1)}} \\
		 & \times C^{j_1,j_2;j_1+j_2}_{m_1-l_1,m_2 - l_2;m_1+m_2-l_1-l_2}C^{\bar{j}_1,\bar{j}_2;\bar{j}_1+\bar{j}_2}_{l_1,l_2;l_1+l_2} C^{j_1+j_2,\bar{j}_1+\bar{j}_2;j_1+j_2+\bar{j}_1+\bar{j}_2 - k}_{m_1+m_2-l_1-l_2,l_1+l_2;m_1+m_2}e^{(j_1+j_2- \frac{k}{2},\bar{j}_1+\bar{j}_2 - \frac{k}{2})}_{m_1+m_2}.
\end{aligned}
\end{equation}
Therefore
\begin{equation}
\begin{aligned}
		& M(e^{(j_1,\bar{j}_1)}_{m_1},e^{(j_2,\bar{j}_2)}_{m_2})\\
	= &\sum_{k}\sum_{l_1,l_2}\lambda_{j_1,\bar{j}_1,0}^{-1}\lambda_{j_2,\bar{j}_2,0}^{-1} \lambda_{j_1+j_2,\bar{j}_1+\bar{j}_2,k}\sqrt{\frac{(2j_1+1)(2j_2+1)(2\bar{j}_1+1)(2\bar{j}_2+1)}{(2j_1 + 2j_2+1)(2\bar{j}_1+2\bar{j}_2 + 1)}} \\
\times& C^{j_1,\bar{j}_1;j_1+\bar{j}_1}_{m_1 - l_1, l_1;m_1} C^{j_2,\bar{j}_2;j_2+\bar{j}_2}_{m_2 - l_2, l_2;m_2}  C^{j_1,j_2;j_1+j_2}_{m_1-l_1,m_2 - l_2;m_1+m_2-l_1-l_2}C^{\bar{j}_1,\bar{j}_2;\bar{j}_1+\bar{j}_2}_{l_1,l_2;l_1+l_2} C^{j_1+j_2,\bar{j}_1+\bar{j}_2;j_1+j_2+\bar{j}_1+\bar{j}_2 - k}_{m_1+m_2-l_1-l_2,l_1+l_2;m_1+m_2}e^{(j_1+j_2- \frac{k}{2},\bar{j}_1+\bar{j}_2 - \frac{k}{2})}_{m_1+m_2}\\
= &\sum_{k}\lambda_{j_1,\bar{j}_1,0}^{-1}\lambda_{j_2,\bar{j}_2,0}^{-1} \lambda_{j_1+j_2,\bar{j}_1+\bar{j}_2,k}\sqrt{(2j_1+1)(2j_2+1)(2\bar{j}_1+1)(2\bar{j}_2+1)(2j_1 + 2\bar{j}_1+1)(2j_2+2\bar{j}_2 + 1)} \\
\times &\begin{Bmatrix}
	j_1&j_2&j_1+j_2\\\bar{j}_1&\bar{j}_2&\bar{j}_1+\bar{j}_2\\j_1+\bar{j}_1&j_2+\bar{j}_2&j_1+j_2+\bar{j}_1+\bar{j}_2 - k
\end{Bmatrix}C^{j_1+\bar{j}_1,j_2+\bar{j}_2;j_1+j_2+\bar{j}_1+\bar{j}_2 - k}_{m_1,m_2;m_1+m_2}e^{(j_1+j_2- \frac{k}{2},\bar{j}_1+\bar{j}_2 - \frac{k}{2})}_{m_1+m_2},
\end{aligned}
\end{equation}
where $\begin{Bmatrix}
	j_1&j_2&j_3\\j_4&j_5&j_6\\j_7&j_8&j_9
\end{Bmatrix}$ is the Wigner $9-j$ symbol.

In our study of the higher product on the CR cohomology, a constantly appearing computation is the product of the form $M(e^{(j_1 - \frac{i}{2},\bar{j}_1 - \frac{i}{2})}_{m_1},\bar{e}^{(\bar{j_2})}_{m_2})$. One can use the above general formula to compute this. Here, we derive an alternative formula that is more succinct. The key is that we use a variation of \ref{Har_to_Pol} to expand the harmonics polynomial $e^{(j_1 - \frac{i}{2},\bar{j}_1 - \frac{i}{2})}_{m_1}$
\begin{equation}
	e^{(j_1 - \frac{i}{2},\bar{j}_1 - \frac{i}{2})}_{m_1} = \sum_{l}\lambda_{j_1,\bar{j}_1,i}^{-1} C^{j_1,\bar{j}_1;j_1+\bar{j}_1 - i}_{m_1 - l,l;m_1}e^{(j_1)}_{m_1 - l}\bar{e}^{(\bar{j}_1)}_{l}.
\end{equation}
Therefore, we have
\begin{equation}\label{Prod_har_2}
\begin{aligned}
		&M(e^{(j_1 - \frac{i}{2},\bar{j}_1 - \frac{i}{2})}_{m_1},\bar{e}^{(\bar{j}_2)}_{m_2}) = \sum_{l}\lambda_{j_1,\bar{j}_1,i}^{-1}\sqrt{\frac{(2\bar{j}_1 + 1)(2\bar{j}_2 + 1)}{(2\bar{j}_1 + 2\bar{j}_2 + 1)}} C^{j_1,\bar{j}_1;j_1+\bar{j}_1 - i}_{m_1 - l,l;m_1} C^{\bar{j}_1,\bar{j}_2;\bar{j}_1+\bar{j}_2}_{l,m_2;l+m_2}M(e^{(j_1)}_{m_1 - l},\bar{e}^{(\bar{j}_1+\bar{j}_2)}_{l+m_2} )\\
		& = \sum_{k\geq 0}\sum_{l}\lambda_{j_1,\bar{j}_1,i}^{-1}\lambda_{j_1,\bar{j}_1+\bar{j}_2;k}\sqrt{\frac{(2\bar{j}_1 + 1)(2\bar{j}_2 + 1)}{(2\bar{j}_1 + 2\bar{j}_2 + 1)}} C^{j_1,\bar{j}_1;j_1+\bar{j}_1 - i}_{m_1 - l,l;m_1} C^{\bar{j}_1,\bar{j}_2;\bar{j}_1+\bar{j}_2}_{l,m_2;l+m_2}C^{j_1,\bar{j_1}+\bar{j}_2;j_1+\bar{j}_1+\bar{j}_2 - k}_{m_1 - l,m_2+l;m_1+m_2}e^{(j_1 - \frac{k}{2},\bar{j}_1+\bar{j}_2 - \frac{k}{2})}_{m_1 + m_2}\\
		& = \sum_{k\geq 0} (-1)^{2(j_1+\bar{j}_1 + \bar{j}_2) - k}\lambda_{j_1,\bar{j}_1,i}^{-1}\lambda_{j_1,\bar{j}_1+\bar{j}_2;k}\sqrt{(2\bar{j}_1 + 1)(2\bar{j}_2 + 1)(2j_1 + 2\bar{j}_1 - 2i + 1)}\\
		&\times \begin{Bmatrix}
			\bar{j}_1&j_1&j_1+\bar{j}_1 - i\\j_1+\bar{j}_1 + \bar{j}_2 - k&\bar{j}_2& \bar{j}_1+\bar{j}_2
		\end{Bmatrix}C^{j_1+\bar{j}_1 - i,\bar{j}_2;j_1+\bar{j}_1+\bar{j}_2 - k}_{m_1,m_2;m_1+m_2}e^{(j_1 - \frac{k}{2},\bar{j}_1+\bar{j}_2 - \frac{k}{2})}_{m_1 + m_2}.\\
\end{aligned}
\end{equation}
Though we write the summation range as $k \geq 0$, the Wigner $6j$ symbol actually constraint it such that $k \geq i$ and $k \leq \min\{2j_1,2\bar{j}_1+2\bar{j}_2\}$.

\subsection{$3$-brackets of Poisson BF theory}
\label{apx:Poi_3}

In this Appendix, we give a general formula for the $3$-bracket in the Poisson BF theory. We compute the constant
\begin{equation}
\begin{aligned}
			(\pi_3)^{p,q;r,s}_{u_1,v_1,u_{2},v_{2}} &=   \frac{(u_1+v_1+1)!(u_2+v_2+1)!}{u_1!v_1!u_2!v_2!}\times\Tr_{S^3}(w_1^pw_2^q\times p\{ h\{w_1^rw_2^s,\bar{w}_1^{u_1}\bar{w}_2^{v_1}\epsilon\},\bar{w}_1^{u_2}\bar{w}_2^{v_2}\epsilon\} )\\
			& = \frac{(u_1+v_1+1)!(u_2+v_2+1)!}{u_1!v_1!u_2!v_2!}\times\Tr_{S^3}(\bar{w}_1^{u_2}\bar{w}_2^{v_2}\epsilon\times p\{w_1^pw_2^q, h\{w_1^rw_2^s,\bar{w}_1^{u_1}\bar{w}_2^{v_1}\epsilon\}\} )
\end{aligned}
\end{equation}
Since $(\pi_3)^{p,q;r,s}_{u_1,v_1,u_{2},v_{2}}$ is only nonzero when $u_1+u_2 = p+r -3,v_1+v_2 = q+s - 3$, we denote $	(\pi_3)^{p,q;r,s}_{u,v} := (\pi_3)^{p,q;r,s}_{u,v,p+r - u - 3,q+s-v-3}$ throughout this section.

First, we need to give a $SU(2)$ decomposition of the two bracket $\{-,-\}_{\bar{\pi}}$ on the CR complex. From \ref{Poi_jjbar}, we see that the induced bracket on the CR complex restricted to $\mathcal{H}_{j,0}\otimes \mathcal{H}_{0,\bar{j}}\epsilon$ gives a map
\begin{equation}
\begin{aligned}
		\{-,-\} : \mathcal{H}_{j,0}\otimes \mathcal{H}_{0,\bar{j}}\epsilon &\to \mathcal{H}_{j-\frac{1}{2},0}\otimes \mathcal{H}_{0,\bar{j}+\frac{1}{2}}\epsilon\\
		& \cong \mathcal{H}_{j - \frac{1}{2},\bar{j}+\frac{1}{2}}\epsilon \oplus \mathcal{H}_{j - 1,\bar{j}}\epsilon \oplus \dots
\end{aligned}
\end{equation}
	We can analyze this map using the same techniques as in Section \ref{sec:pro_Har}.
		\begin{equation}
		\begin{aligned}
		&\{-,-\}\circ CG^{-1} (e^{(j + \bar{j} - k)}_{j + \bar{j}  - k}\epsilon)\\
		= & \sum_{l = 0}^{k}(-1)^{k} C^{j,\bar{j};j + \bar{j} - k}_{j,\bar{j} - k;j + \bar{j} - k}  \frac{k!}{(k - l)! l!  }\sqrt{\frac{(2j+1)(2\bar{j} + 1)!}{k!( 2\bar{j}-k)!}} \\
			&\times (2\bar{j} + 2)(lw_1\bar{w}_1 - (2j - l)w_2\bar{w}_2)w_1^{2j-l - 1}w_2^{l - 1}\bar{w}_1^{ k - l}\bar{w}_2^{ 2\bar{j} - k + l}\epsilon\\
			& = (-1)^{k+1}\sqrt{\frac{(2j+1)!(2\bar{j}+1)!}{k!(2j+2\bar{j} - k + 1)!}}\sqrt{(2\bar{j} - k + 1)( 2j - k)}(2\bar{j} + 2) e^{(j - \frac{1}{2} - \frac{k}{2},\bar{j} + \frac{1}{2} -  \frac{k}{2} )}_{j + \bar{j} - k}.
		\end{aligned}
	\end{equation}
Therefore, the map 	$\{-,-\} : \mathcal{H}_{j,0}\otimes \mathcal{H}_{0,\bar{j}}\epsilon \to \bigoplus_{k = 0}^{\min\{2j - 1,2\bar{j} + 1\}}\mathcal{H}_{j - \frac{1}{2} - \frac{k}{2},\bar{j}+\frac{1}{2} - \frac{k}{2}}\epsilon $ is given by
\begin{equation}\label{2Poi_general}
		\{e^{(j)}_m,\bar{e}^{(\bar{j})}_{\bar{m}}\epsilon\}  = \sum_{k} \tilde{\lambda}_{j,\bar{j},k}C^{j,\bar{j};j + \bar{j} - k}_{m,\bar{m};m+\bar{m}} e^{(j - \frac{1}{2} - \frac{k}{2},\bar{j}+\frac{1}{2} - \frac{k}{2})}_{m+\bar{m}}\epsilon,
\end{equation}
where
\begin{equation}
\begin{aligned}
		 \tilde{\lambda}_{j,\bar{j},k} &= (-1)^{k+1}\sqrt{\frac{(2j+1)!(2\bar{j}+1)!}{k!(2j+2\bar{j} - k + 1)!}}\sqrt{(2\bar{j} - k + 1)( 2j - k)}(2\bar{j} + 2) \\
	& = -\sqrt{(2\bar{j} - k + 1)( 2j - k)}(2\bar{j} + 2) \lambda_{j,\bar{j},k}.
\end{aligned}
\end{equation}
We emphasis that the bracket restricted to $\mathcal{H}_{j,0}\otimes \mathcal{H}_{0,\bar{j}}$ is different. We have
\begin{equation}
\{w_1^pw_2^q,\bar{w}_1^r\bar{w}_2^s\} = 	(r+s)(qw_1\bar{w}_1 -  pw_2\bar{w}_2 ) w_1^{p-1}w_2^{q - 1}\bar{w}_1^r\bar{w}_1^s\epsilon.
\end{equation}
As a result, we have
\begin{equation}
	\{e^{(j)}_m,\bar{e}^{(\bar{j})}_{\bar{m}}\}  = \sum_{k} -\sqrt{(2\bar{j} - k + 1)( 2j - k)}(2\bar{j} ) \lambda_{j,\bar{j},k}C^{j,\bar{j};j + \bar{j} - k}_{m,\bar{m};m+\bar{m}} e^{(j - \frac{1}{2} - \frac{k}{2},\bar{j}+\frac{1}{2} - \frac{k}{2})}_{m+\bar{m}}.
\end{equation}
To compute the constant $\pi^{p,q;r,s}_{u,v}$ we compute the map $ p\{w_1^pw_2^q, h\{w_1^rw_2^s,\bar{w}_1^{u_1}\bar{w}_2^{v_1}\epsilon\}\} $. First, we compute $h\{e^{(j_2)}_{m_2}, \bar{e}^{\bar{(j)}}_{\bar{m}}\epsilon\}$. Using \ref{2Poi_general}, we have
\begin{equation}
h\{e^{(j_2)}_{m_2}, \bar{e}^{\bar{(j)}}_{\bar{m}}\epsilon\} = \sum_{i} \tilde{\lambda}_{j_2,\bar{j},i}h_{j_2 - \frac{1}{2} - \frac{i}{2},\bar{j}+\frac{1}{2} - \frac{i}{2}}C^{j_2,\bar{j};j_2 + \bar{j} - i}_{m,\bar{m};m_2+\bar{m}} e^{(j_2 - 1 - \frac{i}{2},\bar{j}+1 - \frac{i}{2})}_{m_2+\bar{m}}.
\end{equation}
Then we compute $p\{e^{(j_1)}_{m_1},e^{(j_2 - 1 - \frac{i}{2},\bar{j}+1 - \frac{i}{2})}_{m_2 + \bar{m}}\}$. Using \ref{Prod_har_2} we have
\begin{equation}
	\begin{aligned}
		 &p\{e^{(j_1)}_{m_1},e^{(j_2 - 1 - \frac{i}{2},\bar{j}+1 - \frac{i}{2})}_{m_2+\bar{m}}\}= \sum_{m'}\lambda_{j_2-1,\bar{j}+1,i}^{-1}C^{j_2-1,\bar{j}+1;j_2+\bar{j} - i}_{m',m_2+\bar{m} - m';m_2+\bar{m}}p\{e^{(j_1)}_{m_1},e^{(j_2 - 1)}_{m'} \bar{e}^{(\bar{j}+1)}_{m_2+\bar{m} - m'}\}\\
		& = \sum_{m'}\lambda_{j_2-1,\bar{j}+1,i}^{-1}C^{j_2-1,\bar{j}+1;j_2+\bar{j} - i}_{m',m_2+\bar{m} - m';m_2+\bar{m}}\left(pM( \{e^{(j_1)}_{m_1},e^{(j_2 - 1)}_{m'}\},\bar{e}^{(\bar{j}+1)}_{m_2+\bar{m} - m'})+pM( e^{(j_2 - 1)}_{m'},\{e^{(j_1)}_{m_1},\bar{e}^{(\bar{j}+1)}_{m_2+\bar{m} - m'} \})\right).  \\
	\end{aligned}
\end{equation}
The first term in the above formula is given by
\begin{equation}
\begin{aligned}
		&pM( \{e^{(j_1)}_{m_1},e^{(j_2 - 1)}_{m'}\},\bar{e}^{(\bar{j}+1)}_{m_2+\bar{m} - m'})\\
		= &(-1)^{2\bar{j}}\sqrt{\frac{[2j_1 + 1]_2[2j_2 - 1]_2(2j_1 + 2j_2  - 2)(2\bar{j}+3)}{2j_1 + j_2 - 3}} C^{j_1,j_2 - 1;j_1 + j_2 - 2}_{m_1,m';m_1+m'}C^{j_1+j_2 - 2,\bar{j} + 1,j_1 + j_2 - \bar{j} - 3}_{m_1+m',m_2+\bar{m} - m';m_1+m_2 + \bar{m}} e^{(j_1+j_2 - \bar{j} - 3)}_{m_1+m_2+m_3}.
\end{aligned}
\end{equation}
The second term can be computed by
\begin{equation}
\begin{aligned}
		&pM( e^{(j_2 - 1)}_{m'},\{e^{(j_1)}_{m_1},\bar{e}^{(\bar{j}+1)}_{m_2+\bar{m} - m'} \})\\
		= & \sum_{k}  -\sqrt{(2\bar{j} - k + 3)( 2j_1 - k)}(2\bar{j} + 2)\lambda_{j_1,\bar{j} + 1,k} C^{j_1,\bar{j}+1,j_1+\bar{j}+1 - k}_{m_1,m_2+\bar{m} - m';m_1+m_2 + \bar{m} - m'}pM(e^{(j_2 - 1)}_{m'},e^{(j_1 - \frac{1}{2} - \frac{k}{2},\bar{j}+ 1 + \frac{1}{2} - \frac{k}{2})}_{m_1+m_2 + \bar{m} - m'})\\
		=& \sum_{k}  (-1)^{2(j_1+j_2) - k}\sqrt{(2\bar{j} - k + 3)( 2j_1 - k)}(2\bar{j} + 2)\lambda_{j_1,\bar{j} + 1,k} \lambda_{j_1 - \frac{1}{2},\bar{j}+\frac{3}{2},k}^{-1} \sqrt{2j_1(2j_2 + 1)(2\bar{j}+4)(2j_1 + 2\bar{j} - 2k + 3)}\\
		&\times \begin{Bmatrix}
			j_1 - \frac{1}{2}&\bar{j} + \frac{3}{2}&j_1 + \bar{j} + 1 - k\\j_1+j_2 - \bar{j} - 3&j_2 - 1 &j_1+j_2 - \frac{3}{2}
		\end{Bmatrix}  C^{j_1,\bar{j}+1,j_1+\bar{j}+1 - k}_{m_1,m_2+\bar{m} - m';m_1+m_2 + \bar{m} - m'} C^{j_1+\bar{j}+1 - k,j_2 - 1;j_1+j_2 - \bar{j} - 3}_{m_1+m_2+\bar{m} - m',m';m_1+m_2+\bar{m}} e^{(j_1+j_2 - \bar{j} - 3)}_{m_1+m_2+\bar{m}}.
\end{aligned}
\end{equation}
We find that 
\begin{equation}
	p\{e^{(j_1)}_{m_1},e^{(j_2 - 1 - \frac{i}{2},\bar{j}+1 - \frac{i}{2})}_{m_2+\bar{m}}\} = \lambda_{j_2-1,\bar{j}+1,i}^{-1}\Pi_{j_1,j_2,\bar{j};i} C^{j_2+\bar{j} - i,j_1;j_1+j_2 - \bar{j} - 3}_{m_2+\bar{m},m_1,m_1+m_2+\bar{m}} e^{(j_1+j_2 - \bar{j} - 3)}_{m_1+m_2+\bar{m}},
\end{equation}
where 
\begin{equation}
	\begin{aligned}
		&\Pi_{j_1,j_2,\bar{j};i} = (-1)^{2(j_1+j_2) - i}\sqrt{[2j_1 + 1]_2[2j_2 - 1]_2(2j_1 + 2j_2  - 2)(2\bar{j}+3)(2j_1 + 2\bar{j} - 2i +1)} \\
		&\times \begin{Bmatrix}
			j_2 - 1&\bar{j} + 1&j_2 + \bar{j} - i \\j_1+j_2 - \bar{j} - 3&j_1 &j_1+j_2 - 2
		\end{Bmatrix} \\
	& + \sum_{k}  (-1)^{2(j_1+j_2) - k}\sqrt{[2j_1+1]_2(2j_2 + 1)(2j_2+2\bar{j} - 2i +1)(2\bar{j} - k + 3)( 2j_1 - k)}(2\bar{j} + 2) \\
	&\times(2j_1 + 2\bar{j} - 2k + 3) \begin{Bmatrix}
		j_1 - \frac{1}{2}&\bar{j} + \frac{3}{2}&j_1 + \bar{j} + 1 - k\\j_1+j_2 - \bar{j} - 3&j_2 - 1 &j_1+j_2 - \frac{3}{2}
	\end{Bmatrix} \begin{Bmatrix}
	\bar{j} + 1 &j_2 - 1 &j_1 + \bar{j} -i\\j_1+j_2 - \bar{j} - 3&j_1 &j_1 + \bar{j} + 1 - k
\end{Bmatrix} .
	\end{aligned}
\end{equation}

We have
\begin{equation}
\begin{aligned}
		&(\pi_3)^{j_1+m_1,j_1-m_1;j_2+m_2,j_2 - m_2}_{\bar{j} - \bar{m},\bar{j}+\bar{m}} := \frac{(-1)^{\bar{j} + \bar{m}}N(j_1,m_1)N(j_2,m_2)}{N(\bar{j},\bar{m})N(j_1+ j_2 - \bar{j} - 3,m_1+m_2+\bar{m})}\\
		&\times \sum_{i} \tilde{\lambda}_{j_2,\bar{j},i}\lambda_{j_2-1,\bar{j}+1,i}^{-1}h_{j_2 - \frac{1}{2} - \frac{i}{2},\bar{j}+\frac{1}{2} - \frac{i}{2}}\Pi_{j_1,j_2,\bar{j};i}C^{j_2,\bar{j};j_2 + \bar{j} - i}_{m_2,\bar{m};m_2+\bar{m}} C^{j_2+\bar{j} - i,j_1;j_1+j_2 - \bar{j} - 3}_{m_2+\bar{m},m_1,m_1+m_2+\bar{m}}.
\end{aligned}
\end{equation}
This constant gives the quartic interaction of the KK theory of Poisson BF theory.

\section{Some identities involving Pochhammer symbols}
In this appendix, we review some identities involving Pochhammer symbols that are used in the calculation of holography chiral algebra. In the main text, we introduced the descending Pochhammer symbols
\begin{equation}
	[a]_n : = a(a - 1)\dots (a - n + 1) =  \frac{(a)!}{(a - n)!}.
\end{equation}
We also introduce the ascending Pochhammer symbol 
\begin{equation}
	(a)^{(n)} : = a(a + 1)\dots (a + n - 1) =  \frac{(a + n - 1)!}{(a - 1)!}.
\end{equation}
The descending and ascending Pochhammer symbols are related to one another by 
\begin{equation}
	(a)^{(n)} = [a+n-1]_n.
\end{equation}
The hypergeometric function $\prescript{~}{2}{F}_1$ is defined as a power series using the ascending Pochhammer symbol
\begin{equation}
	\prescript{~}{2}{F}_1(a,b,c,z) = \sum_{i = 0}^{\infty}\frac{(a)^{(i)} (b)^{(i)}}{(c)^{(i)}} \frac{1}{i!}z^i.
\end{equation}
The series terminates if either $a$ or $b$ is a nonpositive integer, in which case the function reduces to a polynomial:
\begin{equation}
	\prescript{~}{2}{F}_1(-k,b,c,z) = \sum_{i = 0}^{k}(-1)^i\binom{k}{i}\frac{ (b)^{(i)}}{(c)^{(i)}} z^i.
\end{equation}
The following result is important in obtaining various generalizations of the Chu–Vandermonde’s identity.
\begin{prop}[\cite{Favaro}]
	For any $k\geq 1$, $x,y \in \R_{+}$, and $a,b > 0$, we have 
	\begin{equation}\label{G_Chu_Vand}
		\sum_{i = 0}^k\binom{k}{i}x^{i}y^{n - i}(a)^{(i)}(b)^{(n - i)} = y^k(a+b)^{(k)}	\prescript{~}{2}{F}_1(-k,a,a+b,1 - \frac{x}{y} )	.
	\end{equation}
\end{prop}
Taking $x,y = 1$ in the above formula we obtain the Chu–Vandermonde’s identity
\begin{equation}\label{Chu_Vand}
	\sum_{i = 0}^k\binom{n}{i}(a)^{(i)}(b)^{(k - i)} = (a+b)^{(k)}.
\end{equation}

\begin{cor}
		\begin{equation}\label{Chu_Vand_1}
		\sum_{i = l}^n\binom{i}{l}\binom{k}{i}(a)^{(i)}(b)^{(k - i)} =\binom{k}{l}\frac{ (a+b)^{(k)}}{(a+b)^{(l)}}(a)^{(l)}	.
	\end{equation}
\end{cor}
\begin{proof}
	Letting $x \to 1 + x, y \to 1$ in the formula \ref{G_Chu_Vand}, we obtain the following
	\begin{equation}
		\sum_{i = 0}^k\binom{k}{i}(1+x)^{i}(a)^{(i)}(b)^{(k - i)} = (a+b)^{(k)}	\prescript{~}{2}{F}_1(-n,a,a+b,-x )	.
	\end{equation}
Expanding both side into a series of $x$ we obtain the formula \ref{Chu_Vand_1}.
\end{proof}

\begin{cor}
	We have the following identity
	\begin{equation}\label{Chu_Vand_var}
			\sum_{i = 0}^k \binom{k}{i}\frac{1}{[a]_i[b]_{k - i}}  = \frac{[a+b-k+1]_k}{[a]_k[b]_k}.
	\end{equation}
\end{cor}
\begin{proof}
	\begin{equation}
		\begin{aligned}
			\sum_{i = 0}^k \binom{k}{i}\frac{1}{[a]_i[b]_{k - i}} &= \sum_{i = 0}^k \binom{k}{i}\frac{(a - k + k - i)!(b-k +i)!}{a!b!} \\
			& = \sum_{i = 0}^k \binom{k}{i}\frac{(a - k + 1)^{(k - i)}(a-k)!(b-k +1)^{(i)}(b-k)!}{a!b!}\\
			& = (a+b-2k + 2)^{(k)}\frac{(a - k)!(b-k)!}{a!b!}\\
			& = \frac{[a+b-k+1]_k}{[a]_k[b]_k},
		\end{aligned}
	\end{equation}
where we used the Chu–Vandermonde’s identity \ref{Chu_Vand} in the third line.
\end{proof}
\begin{cor}
		We have the following identity
	\begin{equation}\label{Chu_Vand_var_1}
		\sum_{i = l}^k \binom{i}{l}\binom{k}{i}\frac{1}{[a]_i[b]_{k - i}}  = \binom{k}{l} \frac{[a+b - k+1]_k}{[a]_{k-l}[b]_k[a+b - 2k+l+1]_l}.
	\end{equation}
\end{cor}
\begin{proof}
	\begin{equation}
\begin{aligned}
		\sum_{i = l}^k \binom{i}{l}\binom{k}{i}\frac{1}{[a]_i[b]_{k - i}}  & = \sum_{i = l}^k \binom{k}{i} \binom{i}{l}\frac{(a - k + 1)^{(k - i)}(a-k)!(b-k +1)^{(i)}(b-k)!}{a!b!}\\
		& = \binom{k}{l}\frac{(a - k)!(b-k)!}{a!b!}\frac{(a+b-2k + 2)^{(k)}}{(a+b - 2k + 2)^{(l)}}(a - k + 1)^{(l)}	\\
		& = \binom{k}{l} \frac{[a+b - k+1]_k}{[a]_{k-l}[b]_k[a+b - 2k+l+1]_l}.
\end{aligned}
\end{equation}
\end{proof}

	\section{tree-level Feynman integrals}
	\label{apx:Feyn}
	In this appendix, we evaluate the tree-level Feynman diagrams that appear in Section \ref{sec:twist_hol}.
	\begin{equation}
		I_k(z) = \int_{z' \in \C,s\geq 0}\frac{1}{(|z - z'|^2 + s^2)^{\frac{3}{2}}}\pa_{z'}^k\frac{1}{(|z'|^2 + s^2)^{\frac{3}{2}}} \bar{z}s ds dz'd\bar{z}'.
	\end{equation}

First we note that $z^{k+1}I_k(z)$ is a constant independent of $z$. To see this we shift $z \to \alpha z,\bar{z} \to \bar{\alpha}\bar{z}$. We also shift the integration variable by $z' \to \alpha z', \bar{z}' \to \bar{\alpha}\bar{z}',s\to |\alpha| s$. $z^{k+1}I_k(z)$ is invariant under this shift. Therefore $z^{k+1}I_k(z)$ is a constant and we have
\begin{equation}
	I_k(z) = \frac{I_k(1)}{z^{k+1}}.
\end{equation}
To efficiently evaluate all $I_k(1)$, we make the following generating function
\begin{equation}
	I(1,\lambda) : = \sum_{k = 0}^\infty\frac{\lambda^k}{k!}I_k(1).
\end{equation}
By definition we have
\begin{equation}
	I(1,\lambda)  = 	 \int_{z' \in \C,s\geq 0}\frac{1}{(|1 - z'|^2 + s^2)^{\frac{3}{2}}}\frac{1}{((z'+\lambda)\bar{z}' + s^2)^{\frac{3}{2}}} s ds dz'd\bar{z}'.
\end{equation}
	We use Schwinger parametrization to rewrite the above integral as follows
	\begin{equation}
		I(1,\lambda) = \frac{4}{\pi}\int_{u_1,u_2\geq 0}du_1du_2\int_{z' \in \C,s\geq 0} ds dz'd\bar{z}'(u_1u_2)^{\frac{1}{2}}se^{-u_1(|1 - z'|^2 + s^2) - u_2((z'+\lambda)\bar{z}' + s^2)}.
	\end{equation}
We first evaluate the integral over $s$ and find 
	\begin{equation}
	I(1,\lambda) = \frac{4}{\pi}\int_{u_1,u_2\geq 0}du_1du_2\int_{z' \in \C} dz'd\bar{z}'\frac{1}{2(u_1 + u_2)}(u_1u_2)^{\frac{1}{2}}e^{-u_1|1 - z'|^2  - u_2(z'+\lambda)\bar{z}'}.
\end{equation}
Evaluate the integral over $z' \in \C$ gives
\begin{equation}
I(1,\lambda) = 2\int_{u_1,u_2\geq 0}du_1du_2 \frac{1}{(u_1 + u_2)^2}(u_1u_2)^{\frac{1}{2}}e^{-\frac{u_1u_2(1 + \lambda)}{(u_1 + u_2)}}.
\end{equation}
We make the following change of variable to the Schwinger parameter
\begin{equation}
	u_1 = u\xi,\; u_2 = u(1 - \xi),\; \text{ for } u\geq 0,1\leq \xi \leq 1.
\end{equation}
Then we have
\begin{equation}
\begin{aligned}
	I(1,\lambda) & = 2\int_0^1d\xi\int_0^\infty du (\xi(1 - \xi))^{\frac{1}{2}}e^{-(1+ \lambda)u\xi(1 - \xi)}\\
	&= 2\int_0^1d\xi\frac{1}{(1+ \lambda)\sqrt{\xi(1 - \xi)}}\\
	& = \frac{2\pi}{1 + \lambda}.
\end{aligned}
\end{equation}
We find that
\begin{equation}
	I_k(z) = (-1)^k\frac{k!}{z^{k+1}}.
\end{equation}
More generally, we can consider the following Feynman integral
	\begin{equation}
	I_{l,k}(z) = \int_{z' \in \C,s\geq 0}\pa_{z'}^l\frac{1}{(|z - z'|^2 + s^2)^{\frac{3}{2}}}\pa_{z'}^k\frac{1}{(|z'|^2 + s^2)^{\frac{3}{2}}} \bar{z}s ds dz'd\bar{z}'.
\end{equation}
This can be evaluated using integration by part to move the $\pa_{z'}^l $ derivatives to the second position and using the previous result. We find
\begin{equation}
I_{l,k}(z) = (-1)^k\frac{(k+l)!}{z^{k+l+1}}.
\end{equation}
\bibliographystyle{JHEP}
\bibliography{KK3d}
	\Addresses
\end{document}